\documentclass[letterpaper,11pt]{scrartcl}

\usepackage[left=1 in,right=1 in,top=1.21 in,bottom=1.21 in]{geometry} 

\usepackage[utf8]{inputenx}
\usepackage[T1]{fontenc}
\usepackage[greek, ngerman, english]{babel}

\usepackage
[
final,
babel = true, 
]
{microtype}           
\usepackage[autostyle, english=american]{csquotes} 

\usepackage{amsmath}
\usepackage{amssymb}
\usepackage{amsthm}
\usepackage{thmtools}
\usepackage{mathtools}
\usepackage{thm-restate}
\usepackage{dsfont}        
\usepackage{enumerate}
\usepackage{enumitem}   
\usepackage{longtable}
\usepackage{easybmat} 

\usepackage{lmodern}
\usepackage[only, llbracket, rrbracket]{stmaryrd} 
\usepackage{bbm} 
\usepackage{url} 
\usepackage{bm}  

\usepackage{graphicx}

\usepackage{multirow} 

\usepackage{tikz}
\usetikzlibrary{decorations.pathmorphing}
\usetikzlibrary{decorations.pathreplacing} 
\usetikzlibrary{calc}
\usetikzlibrary{shapes.symbols}
\usetikzlibrary{arrows} 
\usetikzlibrary{decorations.markings}

\usepackage{listings}

\usepackage[sort,numbers]{natbib}

\usepackage{datetime}
\usepackage{authblk} 
\usepackage{comment} 
\usepackage{float}	
\usepackage{xparse} 
\usepackage{color,soul}

\definecolor{stroke1}{HTML}{2574A9} 
\usepackage
[
bookmarks = true,                 
bookmarksopen = false,            
bookmarksnumbered = true,         
pdfstartpage = 1,                 
breaklinks = true,                
colorlinks = true,            
allcolors = stroke1,          
]
{hyperref} 

\usepackage
[
noabbrev,   
nameinlink, 
]
{cleveref} 

\ExplSyntaxOn
\cs_set_eq:NN \IfEmptyTF \tl_if_blank:nTF
\ExplSyntaxOff

\newcommand{\forThesis}[1]{}

\newcommand{\Lovasz}{Lov\'asz}
\newcommand{\Nesetril}{Ne\v{s}et\v{r}il}

\DeclarePairedDelimiter\abs{\lvert}{\rvert}
\DeclarePairedDelimiter\doublesquarebrackets{\llbracket}{\rrbracket}
\DeclarePairedDelimiter\sqBrackets{\lbrack}{\rbrack}
\DeclarePairedDelimiter\parenthesis{(}{)}

\providecommand\given{}
\newcommand\SetSymbol[1][]{%
	\nonscript\,#1\colon
	\allowbreak
	\nonscript\:
	\mathopen{}}
\DeclarePairedDelimiterX\set[1]\{\}{%
	\renewcommand\given{\SetSymbol[\delimsize]}
	#1
}

\DeclareMathOperator{\bip}{bip}

\DeclareMathOperator{\dom}{dom}
\DeclareMathOperator{\dist}{dist}
\DeclareMathOperator{\IN}{in}
\DeclareMathOperator{\OUT}{out}

\newcommand{\family}[1]{\ensuremath{\mathcal{#1}}}

\newcommand{\Z}{\ensuremath{\mathbb{Z}}}
\newcommand{\Q}{\ensuremath{\mathbb{Q}}}
\newcommand{\C}{\ensuremath{\mathbb{C}}}

\newcommand{\Zsp}{\ensuremath{\mathbb{Z}^*_p}}
\renewcommand{\L}{\ensuremath{L}}
\newcommand{\R}{\ensuremath{R}}

\DeclareDocumentCommand{\MetaMathOperator}{m O{} O{}}{ 
	\IfEmptyTF{#2}{
		\IfEmptyTF{#3}{
			\ensuremath{#1}
		}{
			\ensuremath{{#1}^{#3}}
		}
	}{
		\IfEmptyTF{#3}{
			\ensuremath{{#1}_{#2}}
		}{
			\ensuremath{{#1}_{#2}^{#3}}
		}
	}
}
\DeclareDocumentCommand{\eqrel}{O{}}{\ensuremath{\sim_{#1}}}
\DeclareDocumentCommand{\eqclass}{m O{}}{
	\ensuremath{\doublesquarebrackets{#1}_{#2}}
}
\DeclareDocumentCommand{\relArrow}{O{} O{}}{
	\ensuremath{\,\MetaMathOperator{\Rightarrow}[#1][#2]}\,
}

\newcommand{\cupdot}{\mathbin{\dot{\cup}}}

\DeclareDocumentCommand{\vector}{m}{\ensuremath{\bar{#1}}}

\newcommand{\GroupFont}[1]{\mathcal{#1}}
\DeclareDocumentCommand{\group}{O{G} O{}}{
	\IfEmptyTF{#2}{\ensuremath{\GroupFont{#1}}}{\ensuremath{(\GroupFont{#1}, #2)}}
}

\DeclareDocumentCommand{\isomorphic}{O{}}{
	\IfEmptyTF{#1}{isomorphic}{{#1}-isomorphic}%
}
\newcommand{\congbip}{\ensuremath{\cong_{\bip}}}
\newcommand{\congdist}{\ensuremath{\cong_{\dist}}}

\DeclareMathOperator{\Ord}{Ord}
\DeclareDocumentCommand{\Stab}{m O{}}{ 
	\IfEmptyTF{#2}{\MetaMathOperator{\GroupFont{G}}[#1][]}{\MetaMathOperator{#2}[#1][]}
}

\DeclareDocumentCommand{\Orb}{O{} O{} O{}}{ 
	\IfEmptyTF{#1}{
		\MetaMathOperator{\operatorname{Orb}}[#2][#3]
	}{
		\ensuremath{ \MetaMathOperator{\operatorname{Orb}}[#2][#3] [#1]}
	}
}
\DeclareDocumentCommand{\OrbBip}{O{} O{}}{  
	\Orb[#1][#2][\bip]
}
\DeclareDocumentCommand{\OrbDist}{O{} O{}}{  
	\Orb[#1][#2][\dist]
}

\DeclareDocumentCommand {\vertexset} {O{}}{ 
	\IfEmptyTF{#1}{\ensuremath{V}}{\ensuremath{V(#1)}}
}
\DeclareDocumentCommand {\edgeset} {O{}}{ 
	\IfEmptyTF{#1}{\ensuremath{E}}{\ensuremath{E({#1})}}
}
\newcommand{\distVertices}{\ensuremath{V^{\dist}}}
\newcommand{\distEdges}{\ensuremath{E^{\dist}}}

\DeclareDocumentCommand {\lpart} {O{G}}{ 
	{\ensuremath{\L(#1)}}
}
\DeclareDocumentCommand {\rpart} {O{G}}{ 
	{\ensuremath{\R(#1)}}
}
\DeclareDocumentCommand {\bipGraph} {O{}} {
	\IfEmptyTF{#1}{
		\ensuremath{(\lpart, \rpart, \edgeset[G])}
	}{
		\ensuremath{(\lpart[#1], \rpart[#1], \edgeset[#1])}
	}
}
\DeclareDocumentCommand {\partof} {O{}}{ 
	\ensuremath{\MetaMathOperator{\textrm{part}}[#1][]}
}
\DeclareDocumentCommand {\npartof} {O{}}{ 
	\ensuremath{\MetaMathOperator{\overline{\textrm{part}}}[#1][]}
}

\DeclareDocumentCommand {\neigh} {m O{}}{ 
	\ensuremath{\MetaMathOperator{\Gamma}[#2][] (#1) }
}
\newcommand{\twoneigh}[1]{\ensuremath{B_2(#1)}}

\DeclareDocumentCommand {\numWalks} {m m O{} O{}}{ 
	\ensuremath{\MetaMathOperator{\operatorname{W}}[#3][#4] (#1, #2) }
}
\newcommand{\quantum}[1]{\ensuremath{\bar{#1}}}

\newcommand{\forbiddenA}{\ensuremath{K_{3,3}\backslash\set{e}}}
\newcommand{\forbiddenB}{\ensuremath{{domino}}}
\newcommand{\graphclass}{($\forbiddenA{}$, $\forbiddenB{}$)-free}
\newcommand{\Graphclass}{\texorpdfstring{(\boldmath{\forbiddenA{}}, \boldmath{\forbiddenB{}})-Free}{(K3,3\textbackslash\string{e\string}, domino)-Free}}
\newcommand{\nice}{$p$-square-free}
\newcommand{\Nice}{\texorpdfstring{\boldmath{$p$}-Square-Free}{p-Square-Free}}

\newcommand{\redbipartites}{\ensuremath{\family{G}^{\ast p}_{\bip}}}

\DeclareDocumentCommand{\reducedForm}{m O{} O{p}}{ 
	\MetaMathOperator{#1}[#2][\ast #3]
}
\DeclareDocumentCommand{\normreduced}{m O{p}}{
	\reducedForm{#1}[][#2]
}
\DeclareDocumentCommand{\bipreduced}{m} {
	\reducedForm{#1}[\bip]
}

\newcommand{\gadget}[1]{(\ensuremath{#1)}-gadget}

\newcommand{\GadgetPart}[1]{\ensuremath{(J_{#1}, y_{#1})}}
\newcommand{\GadgetEdge}{\ensuremath{(J_{E}, y_{\L}, y_{\R})}}
\newcommand{\Gadget}{
	\ensuremath{(o_{\L}, o_{\R},\allowbreak \GadgetPart{\L}, \GadgetPart{\R},\allowbreak \GadgetEdge)}
}
\DeclareDocumentCommand{\selectSet}{O{} O{}}{
	\MetaMathOperator{\Omega}[#1][#2]
}

\newcommand{\lweight}{\ensuremath{\lambda_\ell}}
\newcommand{\rweight}{\ensuremath{\lambda_r}}
\newcommand{\lOUTweight}{\ensuremath{\kappa_\ell}}
\newcommand{\rOUTweight}{\ensuremath{\kappa_r}}
\DeclareDocumentCommand \numWeightBIS {m o o}{ 
	\ensuremath{ \MetaMathOperator{Z}[#2][#3] (#1) }
}

\newcommand\restr[2]{{
		\left.\kern-\nulldelimiterspace 
		#1 
		\vphantom{\big|} 
		\right|_{#2} 
}}

\newcommand{\FontMorphism}[1] {\ensuremath{\operatorname{#1}}} 
\DeclareDocumentCommand {\METAMorphism} {m O{} O{}}{ 
	\MetaMathOperator{\FontMorphism{#1}}[#2][#3]
}
\DeclareDocumentCommand {\Morph} {m O{} O{} O{}}{
	\IfEmptyTF{#2}{
		\ensuremath{\METAMorphism{#1}[#3][#4]}
	}{
		\ensuremath{\METAMorphism{#1}[#3][#4][ #2 ]}
	}
}

\DeclareDocumentCommand {\Hom} {O{}O{}}{ 
	\Morph{Hom}[#1][][#2]
}
\DeclareDocumentCommand {\numHom} {O{}O{}O{}}{ 
	\Morph{hom}[#1][#2][#3]
}

\DeclareDocumentCommand {\HomBip} {O{}}{ 
	\Morph{Hom}[#1][][\bip]
}

\DeclareDocumentCommand {\numHomBip} {O{}O{}}{ 
	\Morph{hom}[#1][#2][\bip]
}

\DeclareDocumentCommand {\Surj} {O{}O{}}{ 
	\Morph{Surj}[#1][][#2]
}
\DeclareDocumentCommand {\numSurj} {O{}O{}O{}}{ 
	\Morph{surj}[#1][{#2}][#3]
}

\DeclareDocumentCommand {\SurjBip} {O{}}{ 
	\Morph{Surj}[#1][][\bip]
}
\DeclareDocumentCommand {\numSurjBip} {O{}O{}}{ 
	\Morph{surj}[#1][{#2}][\bip]
}

\DeclareDocumentCommand {\Inj} {O{}O{}}{ 
	\Morph{Inj}[#1][][#2]
}
\DeclareDocumentCommand {\numInj} {O{}O{}O{}}{ 
	\Morph{inj}[#1][{#2}][#3]
}

\DeclareDocumentCommand {\InjBip} {O{}}{ 
	\Morph{Inj}[#1][][\bip]
}
\DeclareDocumentCommand {\numInjBip} {O{}O{}}{ 
	\Morph{inj}[#1][{#2}][\bip]
}

\DeclareDocumentCommand {\Aut} {O{}O{}}{ 
	\Morph{Aut}[#1][][#2]
}
\DeclareDocumentCommand {\numAut} {O{}O{}O{}}{ 
	\Morph{aut}[#1][{#2}][#3]
}

\DeclareDocumentCommand {\AutBip} {O{}}{ 
	\Morph{Aut}[#1][][\bip]
}
\DeclareDocumentCommand {\numAutBip} {O{}O{}}{ 
	\Morph{aut}[#1][{#2}][\bip]
}

\DeclareDocumentCommand {\AutDist} {O{}}{ 
	\Morph{Aut}[#1][][\dist]
}

\DeclareDocumentCommand {\PartSurj} {O{}O{} O{}}{ 
	\Morph{PartSurj}[#1][][#2]
}
\DeclareDocumentCommand {\numPartSurj} {O{}O{}O{}}{ 
	\Morph{p-surj}[#1][{#2}][#3]
}

\newcommand{\classfont}{\mathsf}
\newcommand{\class}[1]{\mbox{{\(\classfont{\,#1}\)}}}

\newcommand{\classNP}{\ensuremath{\class{NP}}}
\DeclareDocumentCommand {\classNumP} {o}{
	\IfNoValueTF{#1}{\ensuremath{\#\class{P}}}{\ensuremath{\#_{#1}\class{P}}}
}

\newcommand{\problemFont}[1]{\normalfont\textsc{#1}}

\DeclareDocumentCommand {\ProblemName} {m O{} O{}}{ 
	\MetaMathOperator{#1}[#2][#3]
}

\DeclareDocumentCommand {\ModularCountingProblem} {m O{} O{} O{} O{}} { 
	\IfEmptyTF{#2}{
		\IfEmptyTF{#3}{
			\#\ProblemName{#1}[#4][#5]
		}{
			\#\ProblemName{#1}[#4][#5] {#3}
		}
	}{
		\IfEmptyTF{#3}{
			\#_{#2}\ProblemName{#1}[#4][#5]
		}{
			\#_{#2}\ProblemName{#1}[#4][#5] {#3}	
		}
	}
}

\DeclareDocumentCommand {\prob} {m m m m}{
	\IfEmptyTF{#2}{
		\begin{problem}
			\begin{description}
				\item \emph{Name.} #1
				\item \emph{Input.} #3
				\item \emph{Output.} #4
			\end{description}
		\end{problem}
	}{
		\begin{problem}
			\begin{description}
				\item \emph{Name.} #1
				\item \emph{Parameter.} #2
				\item \emph{Input.} #3
				\item \emph{Output.} #4	
			\end{description}
		\end{problem}
	}
}

\DeclareDocumentCommand {\probNumIS} {O{}}{
	\ensuremath{\ModularCountingProblem{\problemFont{IS}}[#1][][][] }
}

\DeclareDocumentCommand {\probNumBIS} {O{} O{} O{}}{ 	
	\ensuremath{\ModularCountingProblem{\problemFont{BIS}}[#3][][#1][#2] }
}

\DeclareDocumentCommand{\probHom}{}{\ensuremath{\ProblemName{\problemFont{Hom}}}}
\DeclareDocumentCommand {\probNumHom} {m O{} O{} O{}}{
	\ensuremath{\ModularCountingProblem{\probHom}[#2][(#1)][#3][#4]}
}

\newcommand{\probPartLabHom}{\ensuremath{\problemFont{PartLabHom}}}
\DeclareDocumentCommand {\probNumPartLabHom} {m O{} O{} O{}}{
	\ensuremath{\ModularCountingProblem{\probPartLabHom}[#2][(#1)][#3][#4]}
}

\DeclareDocumentCommand {\probNumBipHom} {m O{} O{}}{
	\ensuremath{\ModularCountingProblem{\probHom}[#2][(#1)][#3][\bip]}
}
\DeclareDocumentCommand {\probNumPartLabBipHom} {m O{} O{} }{
	\ensuremath{\ModularCountingProblem{\probPartLabHom}[#2][(#1)][#3][\bip]}
}

\newcommand{\probPartSurjHom}{\ensuremath{\problemFont{PartSurjHom}}}
\DeclareDocumentCommand {\probNumPartSurjHom} {m O{} O{} O{}}{
	\ensuremath{\ModularCountingProblem{\probPartSurjHom}[#2][(#1)][#3][#4]}
}

\newcommand{\probVertSurjHom}{\ensuremath{\problemFont{VertSurjHom}}}
\DeclareDocumentCommand {\probNumVertSurjHom} {m O{} O{} O{}}{
	\ensuremath{\ModularCountingProblem{\probVertSurjHom}[#2][(#1)][#3][#4]}
}

\newcommand{\probComp}{\ensuremath{\problemFont{Comp}}}
\DeclareDocumentCommand {\probNumComp} {m O{} O{} O{}}{
	\ensuremath{\ModularCountingProblem{\probComp}[#2][(#1)][#3][#4]}
}

\theoremstyle{plain}
\newtheorem{theorem}{Theorem}[section]
\newtheorem{corollary}[theorem]{Corollary}
\newtheorem{lemma}[theorem]{Lemma}
\newtheorem{proposition}[theorem]{Proposition}
\theoremstyle{definition}
\newtheorem{definition}[theorem]{Definition}

\newtheorem{problem}[theorem]{Problem}
\newtheorem{observation}[theorem]{Observation}
\newtheorem{example}[theorem]{Example}
\newtheorem{notation}[theorem]{Notation}
\newtheorem{remark}[theorem]{Remark}

\title{On Counting (Quantum-)Graph Homomorphisms\\in Finite Fields of Prime Order}
\subtitle{Collapse of Complexity to Bipartite Graphs in Dimension 1\\and a Dichotomy for \Graphclass{} Graphs}
\author{J. A. Gregor Lagodzinski}
\author{Andreas G\"obel}
\author{\authorcr Katrin Casel}
\author{Tobias Friedrich\thanks{\{gregor.lagodzinski, andreas.goebel, katrin.casel, tobias.friedrich\}@hpi.de}}
\affil{\normalsize
	Hasso Plattner Institute\authorcr
	University of Potsdam\\
	Potsdam, Germany
}
\renewcommand\footnotemark{} 
\date{}

\begin{document}
\pagenumbering{gobble}
\clearpage
\thispagestyle{empty}

\maketitle
\begin{abstract}
	We study the problem of counting the number of homomorphisms from an input graph $G$ to a fixed (quantum) graph $\quantum{H}$ in any finite field of prime order $\Z_p$. The subproblem with graph $H$ was introduced by Faben and Jerrum~[ToC'15] and its complexity is subject to a growing series of research articles, e.g. the work of Focke, Goldberg, Roth, and Zivný~[SIDMA'21] and the work of Bulatov and Kazeminia~[STOC'22], subsequent to this article's conference version. Our contribution is threefold.
	
	First, we introduce the study of quantum graphs to the study of modular counting homomorphisms. We show that the complexity for a quantum graph $\quantum{H}$ collapses to the complexity criteria found at dimension 1: graphs.
	Second, in order to prove cases of intractability we establish a further reduction to the study of bipartite graphs.
	Lastly, we establish a dichotomy for all bipartite \graphclass{} graphs by a thorough structural study incorporating both local and global arguments. This result subsumes all previous results on bipartite graphs known for all prime moduli and extends them significantly. Even for the subproblem with $p$ equal to $2$, this establishes new results.\footnote{An extended abstract of this article appeared at ICALP 2021~\cite{Lagodzinski:21:On_Counting_Quantum-Graph_Homomorphisms}}
%
%
%
\end{abstract}

\clearpage
\pagenumbering{arabic}
\section{Introduction}
\label{sec:intro}
The study of graph homomorphisms represents one of the classic bodies of work in both discrete mathematics and computer science but remains a very active research area.
These homomorphisms play a crucial role in the study of graph limits and networks~\cite{Borgs:06:Graph_Limits_and_Parameter_Testing,  Elenberg:15:Beyond_Triangles, Elenberg:16:Distributed_Estimation_of_Graph_4-Profiles, Ugander:13:Subgraph_Frequencies}, in the study of databases~\cite{ Chandra:77:Optimal_Implementation_of_Conjunctive_Queries_in_Relational_Data_Bases, Grohe:01:When_is_the_Evaluation_of_Conjunctive_Queries_Tractable, Rossman:08:Homomorphism_Preservation_Theorems, Rossman:17:An_Improved_Homomorphism_Preservation_Theorem}, and in the study of spin-systems in statistical physics~\cite{Borgs:06:Counting_Graph_Homomorphisms,Brightwell:99:Graph_Homomorpisms_and_Phase_Transitions}.
Formally, a graph-homomorphism from $G$ to $H$ is a map from the vertex set of $G$ to the vertex set of $H$ that preserves edges.
Many classic problems studied in computer science can be expressed with graph homomorphisms. Examples range from the \emph{decision} problem of determining the chromatic number of a graph,
through the problem of \emph{counting} the number of independent sets, to the problem of \emph{counting} the number of $k$-colourings using all $k$ colours. The latter can be expressed by a \emph{linear combination} of the number of graph homomorphisms to a set of non-isomorphic graphs. 

Graph homomorphisms are a prime example of a very general class of problems that frequently yields complexity dichotomies with structural characterizations, where the properties of a graph implying (in)tractibility are easily computable. However, the dichotomy itself is hard to establish and by Ladner~\cite{Ladner:75:On_the_Structure} not obvious to exist. Hell and \Nesetril~studied the \emph{decision} problem $\probHom(H)$ with fixed image graph $H$, that asks whether there exists a homomorphism from an input graph $G$ to $H$. In~\cite{Hell:90:On_the_Complexity_of_H-Coloring} they showed that the problem $\probHom(H)$ can be solved in polynomial time if $H$ contains a loop or is bipartite; otherwise, it is $\classNP$-complete. Dyer and Greenhill introduced the \emph{counting} problem $\probNumHom{H}$ with fixed image graph $H$, that asks for the number of homomorphisms from an input graph $G$ to $H$. In their seminal work~\cite{Dyer:2000:Counting_Graph_Homs} they showed that $\probNumHom{H}$ can be solved in polynomial time if the connected components of $H$ are complete bipartite graphs or reflexive complete graphs; otherwise, it is $\classNumP$-complete. 

\Lovasz~\cite{Lovasz:12:book:Large_Networks_Graph_Limits} observed that many graph parameters can only be expressed by a linear combination of computational problems $\probNumHom{H}$ for a set of at least two graphs $H \in \family{H}$. Examples are the class of \emph{vertex surjective homomorphisms} and \emph{compactions} studied in this context by Focke, Goldberg, and Zivný~\cite{Focke:19:The_Complexity_of_Counting_Surjective_Homomorphisms_and_Compactions}. \Lovasz~\cite{Lovasz:12:book:Large_Networks_Graph_Limits} introduced the notion of a \emph{quantum graph} for a linear combination of finitely many graphs called its \emph{constituents}. We refer by the dimension of a quantum graph to its number of constituents and find the class of graphs at dimension $1$. With every increase of dimension, the class of graph parameters expressible by $\probNumHom{H}$ increases as well. For a quantum graph $\quantum{H}$, the counting problem $\probNumHom{\quantum{H}}$ denotes the \emph{linear combination} of problems $\probNumHom{H}$ for all constituents $H$ of $\quantum{H}$. Chen, Curticapean and Dell~\cite{Chen:19:The_Exponential-Time_Complexity} studied the complexity of $\probNumHom{\quantum{H}}$ and showed that the complexity is inherited from the complexity of $\probNumHom{H}$ for all constituents $H$ of $\quantum{H}$, which is given by the criterion of Dyer and Greenhill. Motivated by this strong connection, Chen et al. raised the question if techniques based on quantum graphs can advance the state of the art of open problems regarding modular counting homomorphisms. 

We study the complexity of the problem $\probNumHom{\quantum{H}}[p]$ for any prime $p$ and answer the question of Chen et al. in the affirmative, where the problem $\probNumHom{\quantum{H}}[p]$ asks for the value of $\probNumHom{\quantum{H}}$ in the finite field $\Z_p$. Our contribution is threefold. First, we obtain results for the whole class of quantum graphs by showing that the complexity of $\probNumHom{\quantum{H}}[p]$ is inherited from the complexity $\probNumHom{H}[p]$.  Second, we reduce the study of $\probNumHom{H}[p]$ to a study of bipartite graphs by establishing a reduction to a restricted homomorphism problem. Finally, we employ a structural analysis on the set of \graphclass{} graphs and establish a dichotomy for these.

The line of research on modular counting homomorphisms was initiated with the study of the problem $\probNumHom{H}[2]$ by Faben and Jerrum~\cite{Faben:15:Parity_Graph_Homs}.
The modulus implies additional cases of tractibility as structures in $H$ implying intractibility for $\probNumHom{H}$ get \enquote{cancelled} when counting in a finite field $\Z_p$. Faben and Jerrum~\cite{Faben:15:Parity_Graph_Homs} showed that automorphisms of order~$p$ capture a subset of these \enquote{cancellations} and reduced the study to a structural analysis of parameter graphs $H$ that do not admit such automorphisms. These graphs are called \emph{order~$p$ reduced}. In particular, for $p=2$ they conjectured that automorphisms of order $2$ capture all cancellations and that $\probNumHom{H}[2]$ for an order~$2$ reduced graph admits the same complexity criterion as the non-modular version $\probNumHom{H}$ given by Dyer and Greenhill.
Progress toward proving the conjecture has been made by Göbel, Goldberg, and Richerby~\cite{Goebel:14:Cactus, Goebel:16:Square-Free} and the recent work of Focke, Goldberg, Roth, and Zivný~\cite{Focke:21:Counting_Homomorphisms_to_K_4-Minor-Free_Graphs_Mod_2}. The body of work is dominated by a study of structures as the modulus commands incorporating not only the local but also global properties of the graph $H$.

The research on $\probNumHom{H}[p]$ for arbitrary primes $p$ was already suggested by Faben and Jerrum~\cite{Faben:15:Parity_Graph_Homs} as they showed that their results concerning automorphisms of order $p$ apply for any prime $p$. However, Valiant~\cite{Valiant:06:Accidental_Algorithms} showed the existence of computational counting problems that feature a change of complexity concerning different moduli. Therefore, a uniform complexity criterion for $\probNumHom{H}[p]$ would emphasize the special role of graph homomorphisms even more. The study of $\probNumHom{H}[p]$ was finally initiated by Göbel, Lagodzinski, and Seidel~\cite{Goebel:21:Counting_Homomorphisms_Trees} and followed by Kazeminia and Bulatov~\cite{Kazeminia:19:Count_Homs_Square_Free_Mod_Prime}. In light of the richer structure due to the higher moduli, less is known about the complexity of $\probNumHom{H}[p]$ compared to $\probNumHom{H}[2]$. 
Even though Faben and Jerrum~\cite{Faben:15:Parity_Graph_Homs} as well as Göbel et al.~\cite{Goebel:21:Counting_Homomorphisms_Trees} considered an extension of the conjecture to all prime moduli and the results suggested it, no one has gone that far before this article. 
We illustrate the individual contributions in Table~\ref{tab:comparison_chart}.
\begin{table}
	\begin{center}
		\small
		\begin{tabular}{ccccccc}
			\noalign{\hrule height 1pt}
			&\!\multirow{2}{*}{Mod}\!
			&\!\multirow{2}{*}{Trees}\!
			&\!\multirow{2}{*}{Cactus}\!
			&\!Square-\!
			&\!$K_4$-minor-\!
			&\!${(\forbiddenA,\forbiddenB)}$-\\
			&&&& free & free & free\\ 
			\noalign{\hrule height 1pt}
			Faben and Jerrum~\cite{Faben:15:Parity_Graph_Homs}\!
			& 2
			& $\bm\times$\\ 
			Göbel et al.~\cite{Goebel:14:Cactus}\!
			& 2
			& $\bm\times$	
			& $\bm\times$ \\ 
			Göbel et al.~\cite{Goebel:16:Square-Free}\!
			& 2
			& $\bm\times$
			&  
			& $\bm\times$ \\ 
			Focke et al.~\cite{Focke:21:Counting_Homomorphisms_to_K_4-Minor-Free_Graphs_Mod_2}\!
			& 2 
			& $\bm\times$ 
			& $\bm\times$ 
			& {\bfseries($\bm\times$)}
			& $\bm\times$ \\
			\noalign{\hrule height 1pt}
			Göbel et al.~\cite{Goebel:21:Counting_Homomorphisms_Trees}\!
			& $p$
			& $\bm\times$ \\
			Kazeminia and Bulatov~\cite{Kazeminia:19:Count_Homs_Square_Free_Mod_Prime}\!
			& $p$ 
			& $\bm\times$ 
			&& $\bm\times$ \\ 
			\noalign{\hrule height 1pt}
			\bfseries This paper 
			& $p$ 
			& $\bm\times$ 
			& $\bm\times$ 
			& $\bm\times$ 
			&& $\bm\times$ \\
			\noalign{\hrule height 1pt}
		\end{tabular} 
	\end{center}
	\caption{History of the study of $\probNumHom{H}[p]$ on \textbf{bipartite} graphs $H$. (Note that the complexity study can be restricted to bipartite graphs by the bipartization result of this paper.) Crosses denote that the result incorporates the dichotomy for the graph class and a $p$ denotes that the result holds for all primes. Parentheses denote that the result is not intrinsic but given by additional argumentation. }
	\label{tab:comparison_chart}
\end{table}
\subsection{Contribution}
We establish a plethora of technical results, which we believe to be a major asset to future works on the complexity of $\probNumHom{H}[p]$ and may be of independent interest to different lines of research. The main contributions are given in the following and discussed in more depth in the subsequent subsection.

\paragraph*{Quantum Homomorphisms}
We introduce the study of quantum graphs to the study of $\probNumHom{H}[p]$. For any quantum graph $\quantum{H}$, we find that $\probNumHom{\quantum{H}}[p]$ is congruent modulo $p$ to $\probNumHom{\quantum{H}'}[p]$, where $\quantum{H}'$ is a quantum graph whose constituents are order~$p$ reduced with coefficients in $\Z_p \setminus \set{0}$. We call such a quantum graph \emph{order~$p$ reduced}. Focusing on these quantum graphs, we obtain in a unified way the following inheritance theorem, which also incorporates the case of distinguished vertices. For a non-negative integer $k$, a graph has $k$ distinguished vertices if it features a sequence of $k$ vertices, and a quantum graph has $k$ distinguished vertices if every constituent has $k$ distinguished vertices. A homomorphism between graphs with $k$ distinguished vertices is restricted to map the sequences of distinguished vertices component-wise.
\begin{restatable}[]{theorem}{InheritenceQuantumHomsModp}
	\label{thm:homs_to_quantum_graph_mod_p}
	Let $p$ be a prime, $k$ be a non-negative integer, and $\quantum{H}$ be an order~$p$ reduced quantum graph with $k$ distinguished vertices.
	\begin{itemize}
		\item If there exists a constituent $H$ of $\quantum{H}$ such that the problem $\probNumHom{H}[p]$ is $\classNumP[p]$-hard, then $\probNumHom{\quantum{H}}[p]$ is $\classNumP[p]$-hard. 
		\item If for each constituent $H$ of $\quantum{H}$, the problem $\probNumHom{H}[p]$ is solvable in polynomial time, then $\probNumHom{\quantum{H}}[p]$ is also solvable in polynomial time.
	\end{itemize}
\end{restatable}

This shows that the complexity of $\probNumHom{\quantum{H}}[p]$ collapses to the complexity of $\probNumHom{H}[p]$. Even though the set of graph parameters expressible by $\probNumHom{\quantum{H}}$ is arbitrarily larger compared to the parameters expressible by $\probNumHom{H}$, the complexity behaviour is captured at dimension $1$, i.e. graphs. 

We show that the reduction technique applied to show Theorem~\ref{thm:homs_to_quantum_graph_mod_p} yields a universal technique that can be applied to obtain so-called \emph{pinning} in classes of graph-homomorphisms closed under composition. This technique is helpful for our study as we also obtain pinning for the restricted class of homomorphisms introduced in the following.

\paragraph*{Bipartization}
We restrict the study of $\probNumHom{H}[p]$ to the study of bipartite graphs by a restricted class of homomorphisms. We call a bipartite graph $G$ with fixed bipartition a \emph{bip-graph}.
For two bip-graphs $G$ and $H$, we say that a homomorphism from $G$ to $H$ is a \emph{bip-homomorphism} if it preserves the order of the fixed bipartition. The problem $\probNumBipHom{H}[p]$ with fixed bip-graph $H$ then asks for the number of bip-homomorphisms to $H$. It allows us to restrict the study of $\probNumHom{H}[p]$ to the study of bipartite graphs by the following theorem, where $\otimes$ denotes the tensor product of graphs. 
\begin{restatable}[]{theorem}{Bipartization}
	\label{thm:bipartization}
	Let $H$ be a graph and $\vertexset[K_2]$ consist of the two vertices $\set{u_\L,u_\R}$. If $H'$ is the bip-graph $H \otimes K_2$ with bipartition $(\vertexset[H] \times \set{u_\L}, \vertexset[H] \times \set{u_\R})$, then $\probNumBipHom{H'}$ reduces to $\probNumHom{H}$ under parsimonious reduction.
\end{restatable}

As observed by Faben~\cite[Theorem~3.1.17.]{Faben:12:thesis:Complexity_Modular_Counting_CSP}, for any positive integer $k$, any parsimonious reduction is parsimonious modulo $k$. This implies that a dichotomy for $\probNumBipHom{H'}[p]$ yields a dichotomy for $\probNumHom{H}[p]$. As we show later, the graph $H'$ is a collection of complete bipartite graphs if and only if $H$ satisfies the Dyer and Greenhill criterion. An additional feature of Theorem~\ref{thm:bipartization} is that it allows for the graph $H$ to contain loops whereas the bipartite graph $H'$ is always loop-less by definition. So far, no study of $\probNumHom{H}[p]$ allowed for loops. The structural implications of a bipartite graph $H$ are also heavily exploited in the following analysis.

\paragraph*{Hardness in Bipartite \Graphclass{} Graphs}
In the longest and most technically involved part of the paper, we study bipartite graphs $H$ not satisfying the Dyer and Greenhill criterion with the goal of finding enough structural information to establish hardness of $\probNumBipHom{H}[p]$.
We find that it suffices to study the class of bip-graphs without bip-automorphisms of order~$p$, where bip-automorphisms are automorphisms that are also bip-homomorphisms. Such a bip-graph is called \emph{order~$p$ bip-reduced}. We conduct a rigorous structural analysis of the class of bipartite graphs that contain no induced subgraph isomorphic to $\forbiddenA$ or $\forbiddenB$ (see Figure~\ref{fig:intro_example} for an illustration). Our insights into the structure of bipartite graphs allow us to establish the following theorem.
\begin{figure}[t]
	\centering
	\includegraphics[]{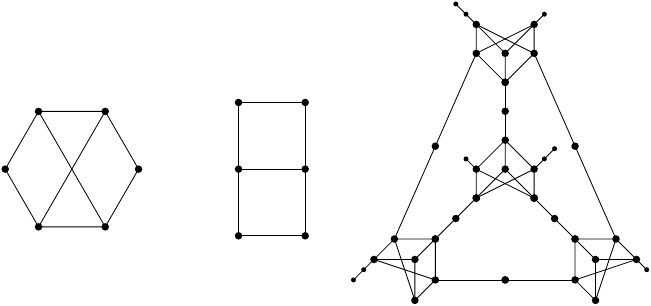}
	\caption{From left to right: $\forbiddenA$; $\forbiddenB$; Example of a bipartite \graphclass{} and asymmetric graph containing locally and globally $K_4$ as a minor.}
	\label{fig:intro_example}
\end{figure}
\begin{restatable}[]{theorem}{HardnessGraphclass}
	\label{thm:hardness_graphclass}
	Let $p$ be a prime and $H$ be an order~$p$ bip-reduced bip-graph that is \graphclass{}. If there exists a connected component of $H$ that is not a complete bipartite graph, then $\probNumBipHom{H}[p]$ is $\classNumP[p]$-hard.
\end{restatable}

In many cases, a $\forbiddenB$ as induced subgraph yields a pair of vertices $x$ and $y$ where $x$ \emph{dominates} $y$. The class of bipartite domination-free $\forbiddenA$-free graphs is one of the focal points of the seminal work by Feder and Vardi~\cite{Feder:98:The_Computational_Structure_of_Monotone_Monadic_SNP_and_Constraint_Satisfaction}. They showed that the class of \emph{graph retract} problems, a notion equivalent to a partially labelled graph homomorphism, to the class of bipartite domination-free $\forbiddenA$-free graphs contains as much computational power as the whole class of \emph{constraint satisfaction problem}s (CSPs), i.e. every CSP is polynomially equivalent to a partially labelled graph homomorphism problem, where the image is a bipartite domination-free $\forbiddenA$-free graph. 

Consider the graphs studied in the work of Brightwell and Winkler~\cite{Brightwell:99:Graph_Homomorpisms_and_Phase_Transitions} shown in Figure~\ref{fig:small_graphs_loops_example}. The set of graph homomorphisms to these graphs played a key role in their study of spin systems in statistical physics. Prior results incorporate only two out of the seven minimal fertile graphs: \enquote{the stick} and \enquote{the key}. 
Following the line of argumentation, our results incorporate the previous and three additional minimal fertile graphs. 
The only missing ones are \enquote{the hinge} and \enquote{the gun} as the construction used for bipartization yields graphs that are not $\forbiddenB$-free. 
The class of bipartite \graphclass{} graphs captures all the classes of bipartite graphs studied in previous works on $\probNumHom{H}[2]$ and $\probNumHom{H}[p]$ except for the recent work by Focke et al.~\cite{Focke:21:Counting_Homomorphisms_to_K_4-Minor-Free_Graphs_Mod_2} on $K_4$-minor-free graphs. Every biclique with at least $3$ vertices in each part contains a $K_4$ as minor, as is the case with $\forbiddenA$. A $\forbiddenB$ is $K_4$-minor-free. Hence, our result given by a local property is orthogonal to the result of Focke et al.~\cite{Focke:21:Counting_Homomorphisms_to_K_4-Minor-Free_Graphs_Mod_2} given by a global property. An example is depicted in Figure~\ref{fig:intro_example}.
\begin{figure}[t]
	\centering
	\includegraphics[]{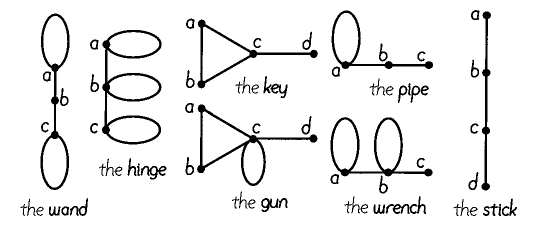}
	\caption{Depiction of the seven minimal fertile graphs as given in Brightwell and Winkler~\cite[Fig.~6.]{Brightwell:99:Graph_Homomorpisms_and_Phase_Transitions}.}
	\label{fig:small_graphs_loops_example}
\end{figure}

\subsection{Results Obtained after This Work}
Subsequent to the conference version~\cite{Lagodzinski:21:On_Counting_Quantum-Graph_Homomorphisms} and during the preparation of this journal version Bulatov and Kazeminia~\cite{Bulatov:22:Complexity_Classification_of_Homs_mod} proved the conjecture for all prime moduli. Utilizing the results from the \enquote{Bipartization}-part they considered the problem $\probNumBipHom{H}[p]$ in the context of CSPs. They proved hardness not by a structural analysis but by an extension of algebraic techniques regularly used when studying CSPs to the modular case; some of these techniques have only been shown to hold for graphs. In light of the proven conjecture for $\probNumHom{H}[p]$ Bulatov and Kazeminia have now rephrased and extended the conjecture to general CSPs.

With the proven complexity classification of $\probNumHom{H}[p]$ at hand, we obtain by Theorem~\ref{thm:homs_to_quantum_graph_mod_p} a dichotomy for quantum homomorphisms. Further, the results of Section~\ref{sec:surjective_homomorphisms}, previously stated and proved conditionally on the correctness of the conjectured dichotomy for $\probNumHom{H}[p]$, are now known to hold in general. This incorporates classifications for vertex surjective homomorphisms and compactions. For future studies, we adjusted these results in light of the proven conjecture, whereas the other parts of the paper were retained unchanged. Lastly, we note that a structural argumentation for the complexity dichotomy of $\probNumHom{H}[p]$ is not known yet, and it is unclear whether the arguments in~\cite{Bulatov:22:Complexity_Classification_of_Homs_mod} can be used to obtain more structural insight.

\subsection{Technical Overview}
In this work, hardness for modular counting problems is indicated by reducing from problems that are $\classNumP[p]$-hard. The class $\classNumP[p]$ contains functions of the form \enquote{$f \pmod p$}, where $f$ is in $\classNumP$. Notably, for the case that $p$ is $2$, the whole polynomial hierarchy reduces to problems in $\classNumP[2]$ by Toda~\cite{Toda:91:PP_is_as_Hard_as_the_Polynomial-Time_Hierarchy}.

We briefly discuss the insights by Faben and Jerrum~\cite{Faben:15:Parity_Graph_Homs}.
For a pair of graphs $G$ and $H$, we denote by $\Hom[G,H]$ the set of homomorphisms from $G$ to $H$. By $\numHom[{G},{H}]$ we denote the cardinality $\abs{\Hom[G,H]}$ and, for a modulus $p$, by $\numHom[{G},{H}][p]$ we denote $\numHom[{G},{H}]  \pmod p$. The computational problem $\probNumHom{H}$ with parameter $H$ then asks to compute $\numHom[{G},{H}]$ for an input $G$. Similarly, $\probNumHom{H}[p]$ asks to compute $\numHom[{G},{H}][p]$.
A central point in the study of $\probNumHom{H}[p]$ is the (non)-existence of automorphisms of order $p$, where for the case that $p$ is $2$, these automorphisms are called \emph{involutions}. 
Given an automorphism $\varrho$ of order $p$ acting as a derangement on the subset $V'$ of $\vertexset[H]$, Faben and Jerrum~\cite{Faben:15:Parity_Graph_Homs} showed that the number of homomorphisms $f$ from any input graph $G$ to $H$ is congruent modulo $p$ to $0$ if the image of $f$ intersects $V'$. They deduced that, for the subgraph $H^\varrho$ of $H$ induced by the fixed points $\vertexset[H] \setminus V'$, there exists a parsimonious reduction from $\probNumHom{H^\varrho}[p]$ to $\probNumHom{H}[p]$. Iteratively applying this reduction, one ends up with a subgraph $\normreduced{H}$ of $H$ that admits no automorphism of order $p$ called the \emph{order~$p$ reduced form} of $H$. This subgraph is unique up to isomorphism and thus well-defined by~\cite[Theorem~3.7.]{Faben:15:Parity_Graph_Homs}. The study of $\probNumHom{H}[p]$ focuses on graphs $H$ that do not admit automorphisms of order $p$, which are called \emph{order~$p$ reduced}.

We now discuss our technical contributions and argumentative routes in more detail.

\subsubsection{Quantum Homomorphisms}
It has been observed by Borgs, Chayes, Kahn, and \Lovasz~\cite{Borgs:13:Convergence_Graphs_Bounded_Degree} that the study of linear combinations of homomorphisms provides great insights, especially on the comparability of pairs of graphs, for instance, if one is a subgraph of the other. \Lovasz~\cite{Lovasz:12:book:Large_Networks_Graph_Limits} introduced the term \emph{quantum graph}, denoted $\quantum{H}$, for a linear combination of finitely many graphs. For a quantum graph $\quantum{H}$, the set of pairwise non-isomorphic graphs is denoted by $\family{H}$ and every graph $H$ in $\family{H}$ has an associated coefficient $\alpha_H$. If $\alpha_H\neq 0$, then $H$ is called a \emph{constituent} of $\quantum{H}$. The quantum graph $\quantum{H}$ is given by
\[
\quantum{H} = \sum_{H \in \family{H}} \alpha_H \cdot H .
\]
A computational problem on $\quantum{H}$ translates into the linear combination of computational problems on entities $H$ in $\family{H}$ with coefficient $\alpha_H$. By \Lovasz~\cite{Lovasz:12:book:Large_Networks_Graph_Limits} every graph parameter has -- if any -- a unique expression by a linear combination of finitely many graph homomorphisms up to isomorphisms.

Towards studying $\probNumHom{\quantum{H}}[p]$ we naturally assume that the coefficients of $\quantum{H}$ are integers. For a graph $G$, we denote by $\numHom[G, \quantum{H}][p]$ the value $\numHom[G, \quantum{H}] \pmod p$. Following the findings of Faben and Jerrum~\cite[Theorem~3.7.]{Faben:15:Parity_Graph_Homs}, it suffices to assume that $\quantum{H}$ consists of order~$p$ reduced constituents with coefficients in $\Z \setminus \set{0}$. Such a quantum graph is called \emph{order~$p$ reduced}. We establish a polynomial-time reduction from $\probNumHom{H}[p]$ to $\probNumHom{\quantum{H}}[p]$, for any order~$p$ reduced quantum graph $\quantum{H}$ and any constituent $H$ in $\quantum{H}$. 
Such a reduction is commonly referred to as a \emph{pinning}-reduction as it enables us to consider the subproblem where a partial mapping is already fixed. One of the main problems of reduction algorithms on modular counting problems is the loss of control of summations in a finite field because we cannot infer from non-zero summands that the sum is non-zero. 
For instance, let $p$ be $2$ and $\quantum{H}$ be the quantum graph consisting of the two graphs $H_1$ and $H_2$ with coefficients $\alpha_{H_1}$ and $\alpha_{H_2}$ equal to $1$, where $H_1$ is an asymmetric tree and $H_2$ is the disjoint union of a copy of $H_1$ and an isolated vertex. Let $G$ be a connected graph and input for $\probNumHom{\quantum{H}}$, then we obtain $\numHom[G,H_2] = \numHom[G,H_1] + \numHom[G,K_1]$. Consequently, when computing $\numHom[G,H_1] + \numHom[G,H_2]$ in $\Z_2$ the term referring to $H_1$ vanishes and this amounts to computing $\numHom[G,K_1]$, which is solvable in polynomial time. However, Theorem~\ref{thm:homs_to_quantum_graph_mod_p} yields that $\probNumHom{\quantum{H}}[2]$ is $\classNumP[2]$-hard. The reason is that the split into $\numHom[G,H_1] + \numHom[G,K_1]$ only works if $G$ is connected and by utilizing disconnected graphs the additional vertex of $H_2$ yields enough information to distinguish between $H_1$ and $H_2$. Therefore, we can extract $\numHom[G,H_1][p]$ from $\numHom[G,\quantum{H}][p]$.

In finite fields, reduction algorithms usually rely heavily on multiplication. We find that the beautiful insight on specific matrices defined on families $\family{F}$ of simple graphs provided by Borgs, Chayes, Kahn, and \Lovasz~\cite[Lemma~4.2.]{Borgs:13:Convergence_Graphs_Bounded_Degree}\footnote{The lemma is stated for possibly infinite families in~\cite[Proposition~5.43.]{Lovasz:12:book:Large_Networks_Graph_Limits}} is able to lift us above this hurdle. In order to adapt this result, we first show that the original proof extends naturally to graphs that feature loops or distinguished vertices. Then, we translate the result to counting in a finite field of prime order. A straightforward application of the modulo operator is not sufficient as the graphs in $\family{F}$ might contain a multiple of $p$ automorphisms. We restrict to order~$p$ reduced graphs and argue why this allows for an application of the modulo operator. In this way, we show the following.
\begin{restatable}[]{corollary}{HomMatrixNonSingularModp}
	\label{cor:hom_matrix_nonsingular_mod_p}
	Let $p$ be a prime, $k$ be a non-negative integer, and $\family{F}$ be a family $\set{F_i}_{i \in I}$ with index set $I$, where $\family{F}$ consists of pairwise non-\isomorphic{} order~$p$ reduced graphs with $k$ distinguished vertices and without multi-edges. If $\family{F}$ is closed under surjective homomorphic image, then the matrix
	\[
		M_{\numHom[][p]} = \sqBrackets[\big]{\numHom[F_i,F_j][p]}_{i,j \in I}
	\]
	is nonsingular.
\end{restatable}

The strength of this result for our purposes is twofold. First, it allows us to show Theorem~\ref{thm:homs_to_quantum_graph_mod_p} in a concise manner. Given an order~$p$ reduced quantum graph $\bar{F}$ with set of constituents $\family{F}$ closed under surjective homomorphic image, we obtain by Corollary~\ref{cor:hom_matrix_nonsingular_mod_p} that any system of linear equations of the form $\bar{x} \cdot M_{\numHom} = \bar{v}$ has a unique solution in the field $\Z_p$. Therefore, for any vector $\bar{v}$, there exists a unique linear combination of entities in $\family{F}$ with coefficients $\alpha_F$ that yield the vector $\bar{v}$. In fact, we observe that this corresponds to a quantum graph $\bar{F}'$ with $\numHom[\bar{F}',F_i][p] = v_i$ that \emph{implements} the vector $\bar{v}$, where $F_i$ and $v_i$ is the $i$-th entry of $\family{F}$ and $\vector{v}$, respectively. In particular, there exists a quantum graph $\bar{F}'$ implementing the $i$-th standard vector allowing us to \enquote{pick} the $i$-th entry of $\family{F}$, i.e. $\numHom[\bar{F}',F_j][p]$ is equal to $1$ if $j$ is equal to $i$, otherwise $\numHom[\bar{F}',F_j][p]$ is equal to $0$. Given an input graph $G$ for $\numHom[G,\bar{F}][p]$, we construct a quantum graph $\bar{F}^\ast$ from $G$ and $\bar{F'}$ such that $\numHom[\bar{F}^\ast, \bar{F}][p]$ is equal to $\numHom[G, F_i][p]$. The main problem for this application is that the set $\family{F}$ of constituents might not be closed under surjective homomorphic image. Given any quantum graph $\quantum{H}$ with set of constituents $\family{H}$, we need to define a suitable family $\family{F}$ that contains all the image graphs needed. We find that the subgraphs of the maximal closure are sufficient for this purpose and obtain Theorem~\ref{thm:homs_to_quantum_graph_mod_p}.

The second strength is the adaptability to subproblems of homomorphisms. The main property needed is that the subclass of homomorphisms has to be closed under composition provided the composition exists. Examples are \emph{vertex surjective homomorphisms} and \emph{compactions} as studied by Focke et al.~\cite{Focke:19:The_Complexity_of_Counting_Surjective_Homomorphisms_and_Compactions}. A homomorphism $f$ in $\Hom[G, H]$ is \emph{vertex surjective} if the image-set of $f$ is the whole set $\vertexset[H]$. The homomorphism $f$ is a \emph{compaction} if it is vertex surjective and every non-loop edge $e$ is in the image of $f$. 
A closely related example is the problem of counting \emph{partially labelled} homomorphisms $\probNumPartLabHom{H}$, that are homomorphisms from an input graph $G$ to $H$ that have to respect a given mapping from a subset $V_G$ of $\vertexset[G]$ to a subset $V_H$ of $\vertexset[H]$. By Feder and Hell~\cite[Theorem~4.1.]{Feder:98:List_Homomorphisms_to_Reflexive_Graphs}, these are equivalent to \emph{retractions}. The reduction from $\probNumPartLabHom{H}[p]$ to $\probNumHom{H}[p]$ is a building stone of every paper in the study of $\probNumHom{H}[p]$ and can be obtained swiftly due to the strength of Corollary~\ref{cor:hom_matrix_nonsingular_mod_p}. A third example will be discussed in the next subsection.

\subsubsection{Bipartization}
Chen et al.~\cite{Chen:19:The_Exponential-Time_Complexity} employed the tensor product to construct the graph $H'$ by $H'=H \otimes K_2$. Subsequently, they reduced $\probNumHom{H'}$ to $\probNumHom{H}$. Let the vertex set $\vertexset[K_2]$ be $\set{u_\L, u_\R}$. The graph $H'$ is bipartite with parts $\lpart[H']$ and $\rpart[H']$ given by $\vertexset[H] \times \set{u_\L}$ and $\vertexset[H] \times \set{u_\R}$, respectively. 
The main problem when adapting this construction to modular counting $\probNumHom{H}[p]$ is that, for every graph $G$, the number of homomorphisms $\numHom[G, K_2][2]$ is $0$. Thus, the tensor product with $K_2$ seemingly annihilates any structure that might imply hardness. Instead of branching the study of $\probNumHom{H}[p]$ into one for the case that $p$ is $2$ and one for odd primes, we solve this issue in a uniform way for all prime moduli.

The key insight is that for an involution-free graph $H$, the tensor product $H \otimes K_2$ only yields involutions on $H \otimes K_2$ that exchange the parts of the bipartition. An important example is the graph $H$ consisting of a single edge with one loop, for which it is known that $\probNumHom{H}$ is equivalent to counting the number of independent sets. Then, the graph $H'$ given by $H \otimes K_2$ is the path with $4$ vertices (see Figure~\ref{fig:small_graphs_loops_example}), that admits only the reflection across the middle edge as a non-trivial automorphism. It is known that $\probNumHom{H'}$ is equivalent to counting the number of bipartite independent sets $\probNumBIS$ and also that $\probNumHom{H'}[4]$ is $\classNumP[2]$-hard (see~\cite{Goebel:21:Counting_Homomorphisms_Trees}) whereas $\probNumHom{H'}[2]$ is polynomial-time solvable. In order to evade the artificial involutions yielded by the tensor product with $K_2$ we introduce the study on the problem of counting bip-homomorphisms between bip-graphs denoted $\probNumBipHom{H}$. For example, if $H$ is the path with $4$ vertices then $\probNumBipHom{H}[2]$ is equivalent to $\probNumHom{H}[4]$. We note that the graph $H \otimes K_2$ is a collection of complete bipartite graphs if and only if $H$ satisfies the Dyer and Greenhill criterion. For such bip-graphs $H'$, the problem $\probNumBipHom{H'}[p]$ is solvable in polynomial time.

The reduction from $\probNumBipHom{H}[p]$ in~Theorem~\ref{thm:bipartization} has the downside that the machinery developed over the course of multiple papers on $\probNumHom{H}[p]$ is not stated for the subclass of homomorphisms counted by $\probNumBipHom{H}[p]$. We remedy this. First, by the strong adaptability of Corollary~\ref{cor:hom_matrix_nonsingular_mod_p} and the subsequent reduction algorithm we obtain pinning for the problem $\probNumBipHom{H}[p]$. Second, analogue to the order~$p$ reduced form $\normreduced{H}$, a bip-graph $H$ is \emph{order~$p$ bip-reduced} if it does not admit a bip-automorphisms of order $p$. We reduce the bip-graph $H$ to its \emph{order~$p$ bip-reduced form} $\bipreduced{H}$. The goal toward a dichotomy for $\probNumHom{H}$ is captured in the following corollary; the applied chain of reductions is displayed in Figure~\ref{fig:reduction_chain_BIP}.
\begin{figure}[t]
	\centering
	\includegraphics[width=\textwidth]{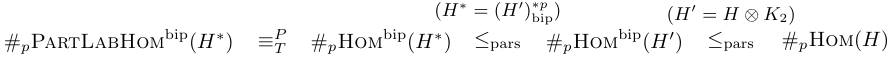}
	\caption{For $p$ a prime and $H$ a graph, chain of reductions employed for bipartization.}
	\label{fig:reduction_chain_BIP}
\end{figure}
\begin{corollary}\label{cor:bip_homs_non-bip_graph}
	Let $p$ be a prime.
	If, for every order~$p$ bip-reduced bip-graph $H'$ that is not a collection of complete bipartite graphs, $\probNumBipHom{H'}[p]$ is $\classNumP[p]$-hard, then, for every order~$p$ reduced graph $H$ that is not a collection of reflexive complete and complete bipartite graphs, $\probNumHom{H}[p]$ is $\classNumP[p]$-hard.
\end{corollary}

We employ a gadget that establishes a reduction to $\probNumBipHom{H}[p]$ from a variant of $\probNumBIS$ with weights on both types of vertices. We call such a gadget yielding hardness a \emph{$p$-hardness gadget}. By an adaptation of the dichotomy for $\probNumBIS$ with weights on the vertices \textit{in} the independent set given in~\cite{Goebel:21:Counting_Homomorphisms_Trees} this reduction establishes hardness when counting in $\Z_p$ if and only if none of the weights is congruent modulo $p$ to $0$. 
The problem of $\probNumBIS$ with weights on the vertices in the independent set is established as the terminal problem yielding hardness in the study of $\probNumHom{H}[p]$~\cite{Goebel:21:Counting_Homomorphisms_Trees, Kazeminia:19:Count_Homs_Square_Free_Mod_Prime} as the bigger modulus implies a richer structure compared to the study of $\probNumHom{H}[2]$ that traditionally focusses on counting the number of independent sets.

\subsubsection{Hardness in Bipartite \Graphclass{} Graphs}

A central argument in the work of Chen et al.~\cite{Chen:19:The_Exponential-Time_Complexity} is that, for a bipartite graph $H$, there exists a simple reduction from $\probNumHom{B}$ to  $\probNumHom{H}$. Here, $B$ is the ball of radius $2$ around a vertex $v$ of $H$ denoted by $\twoneigh{v}$, i.e. vertices of distance at most $2$ from $v$. By an iterative application of this argument, they establish a reduction from $\probNumHom{P}$, where $P$ is a generalization of the path with $4$ vertices. Even though the reduction argument can be made valid for $\probNumHom{H}[p]$ and $\probNumBipHom{H}[p]$ the restriction to a substructure might destroy the properties that yield hardness already for trees $H$, a class of graphs for which the dichotomy is proved (see \cite{Faben:15:Parity_Graph_Homs, Goebel:21:Counting_Homomorphisms_Trees}). An example is depicted in Figure~\ref{fig:example_failure_2-neighbourhood}.
\begin{figure}[t]
	\centering
	\includegraphics[]{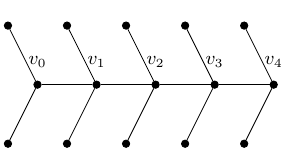}
	\caption{For $p$ equal to $3$, the tree $H$ contains no ball of radius $2$ around any vertex with enough structure to yield hardness even though $\probNumHom{H}[3]$ is $\classNumP[{3}]$-hard.}
	\label{fig:example_failure_2-neighbourhood}
\end{figure}

In a nutshell, the induced subgraphs of radius at most $2$ can admit too many automorphisms of order $p$. In Figure~\ref{fig:example_failure_2-neighbourhood} we observe that the problem originates from too many instances of complete bipartite graphs $K_{1,b}$, where $b$ is congruent modulo $3$ to $0$ or $1$. The way to overcome this is to also consider the global structure. In the case displayed in Figure~\ref{fig:example_failure_2-neighbourhood}, the number of walks of length $4$ from $v_4$ to a vertex $v$ in the neighbourhood of $v_1$ is congruent modulo $3$ to $0$ if $v$ is equal to $v_2$ and equal to $1$ otherwise. The goal is then a reduction restricting the study to $\twoneigh{v_0} \setminus \set{v_2}$, a graph that yields hardness. We do this in a general form by the second type of gadgetry called \emph{$(B,p)$-gadget} that reduces $\probNumBipHom{H}[p]$ from $\probNumBipHom{B}[p]$, where $B$ is an induced sub-bip-graph of $H$, an induced subgraph inheriting the fixed bipartition.

As we have argued, some of the main obstacles in the study of $\probNumHom{H}[p]$ are complete bipartite graphs. The graph \forbiddenA{} denotes the graph obtained from $K_{3,3}$ by deleting an edge, and the graph $\forbiddenB$ denotes the graph obtained from $K_{3,3}$ by deleting two edges without introducing a cut-vertex (see Figure~\ref{fig:intro_example}). By the restriction to \graphclass{} bipartite graphs, we study exactly the case of a great many complete bipartite induced subgraphs. To this end, we observe that, for every vertex $v$ of $H$, the induced subgraph $\twoneigh{v}$ splits into connected components obtained by deleting $v$. The \emph{split} of $\twoneigh{v}$ at $v$ is given by the set of these connected components, where every component contains a copy of $v$. By the absence of induced subgraphs isomorphic to $\forbiddenA$ or $\forbiddenB$, we deduce that the blocks that contain $v$ in these components are complete bipartite.

The overall line of argumentations toward Theorem~\ref{thm:hardness_graphclass} is the following. First, we establish the dichotomy for all \graphclass{} bipartite graphs $H$ of radius at most $2$ by a combination of $(B,p)$- and $p$-hardness gadgets. This is done by a careful structural study of the split of $H$ at a central vertex $v$. An important first result is that, for any order~$p$ bip-reduced bip-graph $H$, if $H$ contains a vertex $v$ where $\twoneigh{v}$ falls into the hard cases of the dichotomy, then $\probNumBipHom{H}[p]$ is $\classNumP[{p}]$-hard. Second, we study graphs $H$ of radius larger than $2$ in order to establish enough structural information of $H$ allowing us to construct either a $p$-hardness gadget for $H$ or a $(B,p)$-gadget such that $\probNumBipHom{B}[p]$ is $\classNumP[{p}]$-hard. This second step is very long and technically involved because the class of \graphclass{} graphs allows for many cases commanding us to explore the global structure of $H$. The chain of reduction arguments is displayed in Figure~\ref{fig:reduction_chain_hard}.
\begin{figure}[t]
	\centering
	\includegraphics[]{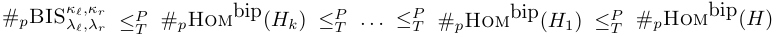}
	\caption{For $p$ a prime and $H$ a bip-graph, chain of reduction arguments. The intermediate bip-graphs $H_i$ refer to $H$ itself or an induced sub-bip-graph obtained by an $(H_i,p)$-gadget.}
	\label{fig:reduction_chain_hard}
\end{figure}

Informally, we split the study of bipartite \graphclass{} graphs of radius larger than $2$ into two broad cases: graphs with no pair of vertices that have a multiple of $p$ common neighbours, and graphs with such a pair of vertices. We recall Figure~\ref{fig:example_failure_2-neighbourhood} and the challenges imposed by subgraphs isomorphic to $K_{a, p}$ or $K_{a, p+1}$. The first case is dedicated to the study of only instances of $K_{a,p+1}$, and the second case adds the possibility of instances of $K_{a,p}$. In both cases, the challenging situations arise from closed walks and paths that yield hardness depending on the endvertices.

We initialize with a study of closed walks, generalizing the situation of cycles frequently encountered (see~\cite{Kazeminia:19:Count_Homs_Square_Free_Mod_Prime, Focke:21:Counting_Homomorphisms_to_K_4-Minor-Free_Graphs_Mod_2}). A walk $W$ given by $W=(w_0, w_1, \dots, w_k)$ is square-free if, for every index $i \in \sqBrackets{k-1}$, the common neighbourhood $\neigh{w_{i-1}} \cap \neigh{w_{i+1}}$ contains only $w_i$. We find that a closed square-free walk in $H$ with more than $2$ different vertices is sufficient to show hardness. Moreover, a \emph{thick} walk is obtained from $W$ by replacing every vertex $w_i$ with $b_i$ copies, where $b_i$ is a positive integer. If for every index $i \in \sqBrackets{k}$, the number of copies $b_i$ is not congruent modulo $p$ to $0$ and the thick walk is closed, then we call the thick walk \emph{$p$-hard} and show that it also yields hardness. 

Focusing on the first case, we show that it suffices to study a path $P$ in $H$ constructed from concatenating instances of $K_{a,k \cdot p+1}$ with $a$ not congruent modulo $p$ to $0$. We call such a path $P$ a \emph{$p$-hardness path}; an example is depicted in Figure~\ref{fig:generalized_hardness_path_intro_example}.
The path $P$ is constructed such that it allows us to establish hardness depending only on the local properties of its endvertices. In the case that one endvertex does not satisfy any of the properties allowing to establish hardness, we show that there exists a way to extend $P$. By the finiteness of $H$, we derive a closed walk or the existence of a $p$-hardness path $P$ whose endvertices are such that $P$ yields hardness. A closed walk constructed in this way has to be $p$-hard, which concludes the case.

Turning toward the second broad case, we traverse the graph $H$ again along a path $Q$. This path $Q$ is constructed by concatenating instances of $K_{a,k \cdot p}$ carefully such that it evades a set of common neighbours of cardinality congruent modulo $p$ to $0$. This leads to a structure we call \emph{$p$-mosaic path}; an example is shown in Figure~\ref{fig:generalized_hardness_path_intro_example}. Similar to $p$-hardness paths, a $p$-mosaic path provides enough structural information toward establishing hardness depending only on the local properties of its endvertices. Also, similar to $p$-hardness paths, we show that if an endvertex $v$ of $Q$ does not satisfy such property, then one option is that $Q$ can be extended. However, by the finiteness of $H$, this yields a $p$-hard walk. The only remaining option for $v$ is to also be the endvertex of a $p$-hardness path. We conclude that $H$ has to contain a concatenation of $p$-hardness paths and $p$-mosaic paths. By the finiteness of $H$, we show that this concatenation is a $p$-hard walk.
\begin{figure}[t]
	\centering
	\includegraphics[]{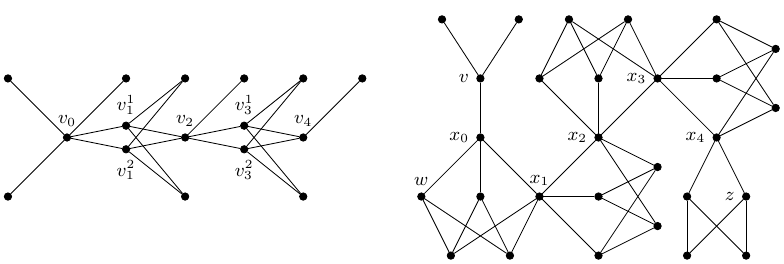}
	\caption{For $p$ equal to $3$, the left figure depicts an example of a $p$-hardness path and the right figure depicts an example of a $p$-mosaic path.}
	\label{fig:generalized_hardness_path_intro_example}
\end{figure}

Before we proceed, we note that many of our structural findings and gadget constructions are given for general graphs $H$, not only \graphclass{} ones. Examples are $p$-hard walks and $p$-hardness paths. This is done to support future studies of $\probNumHom{H}[p]$. Despite these more general findings, we only give complexity results for \graphclass{} graphs as, in particular, our arguments do not yield a straightforward extension to general graphs $H$ of radius $2$ as explained in the following.

\subsubsection{Beyond \Graphclass{} Graphs}
In light of our findings we conjecture that, for an order~$p$ bip-reduced bip-graph $H$, the problem $\probNumBipHom{H}[p]$ is $\classNumP[p]$-hard if $H$ is not a collection of complete bipartite graphs. This extends toward a conjecture on $\probNumHom{H}[p]$ and also incorporates the conjecture of Faben and Jerrum.
\begin{restatable}[]{conjecture}{FabenJerrump}
	\label{conj:faben-jerrum_p}
	Let $p$ be a prime and $H$ be a graph with order~$p$ reduced form $\normreduced{H}$. If the connected components of $\normreduced{H}$ are complete bipartite or reflexive complete, then $\probNumHom{H}[p]$ is solvable in polynomial time. Otherwise, $\probNumHom{H}[p]$ is $\classNumP[p]$-complete.
\end{restatable}

We highlight the implications of this conjecture by a study of the set of \emph{partially surjective homomorphisms} from a graph $G$ to a graph $H$ denoted $\PartSurj[G, H]$. Partially surjective homomorphisms have to be surjective on a set of distinguished vertices $\distVertices$ and a set of distinguished edges $\distEdges$, where $\distVertices$ and $\distEdges$ are contained in $\edgeset[H]$ and $\vertexset[H]$, respectively. We deduce that it suffices to study graphs $H$ without automorphisms of order $p$ acting bijectively on $\distVertices$ and $\distEdges$. However, this reduction does not capture all cancellations because the graph $H$ might still admit too many general automorphisms.

By our results on quantum graphs and an application of the inclusion-exclusion principle, we find that the dichotomy presented in Conjecture~\ref{conj:faben-jerrum_p} extends to a dichotomy on the whole class of partially surjective homomorphisms. Contrary to the non-modular version established in~\cite{Chen:19:The_Exponential-Time_Complexity}, this dichotomy
does not state a clear structural characterization of the hard instances with respect to $H$. Due to the mentioned possibility of additional cancellations, the complexity criterion is of algorithmic type. For the special cases in which the parameter graph $H$ is order~$p$ reduced, we amplify the dichotomy such that it states clear structural characterizations. Two examples for this case are the problems $\probNumVertSurjHom{H}[p]$ and $\probNumComp{H}[p]$ of counting in $\Z_p$ the number of vertex surjective homomorphisms and compactions, respectively. We obtain the following criteria analogous to the criteria in the non-modular setting given by Focke et al.~\cite{Focke:19:The_Complexity_of_Counting_Surjective_Homomorphisms_and_Compactions}.
\begin{restatable}[]{corollary}{SurjCompDichotomyModp} 
	\label{cor:surj_comp_dichotomy_mod_p}
	Let $p$ be a prime and $H$ be a graph. If $H$ is a collection of complete bipartite graphs and reflexive complete graphs or $H$ admits an automorphism of order $p$, then $\probNumVertSurjHom{H}[p]$ is solvable in polynomial time. Otherwise, $\probNumVertSurjHom{H}[p]$ is $\classNumP[p]$-hard.
	
	If $H$ is a collection of irreflexive stars and reflexive complete graphs of size at most two or $H$ admits an automorphism of order $p$, then $\probNumComp{H}[p]$ is solvable in polynomial time. Otherwise, $\probNumComp{H}[p]$ is $\classNumP[p]$-hard.	
\end{restatable}

In order to prove Conjecture~\ref{conj:faben-jerrum_p}, we need to study bipartite graphs that contain $\forbiddenA$ or $\forbiddenB$ as an induced subgraph. 
The strong structural restrictions on the graphs under study are a double-edged sword. On the one hand, it is more plausible to find enough structure that yields hardness. On the other hand, it is more difficult to pin the structural analysis down to a handful of cases. We illustrate an especially difficult example in Figure~\ref{fig:5_catherine_wheel}.

\begin{figure}[t]
	\centering
	\includegraphics[]{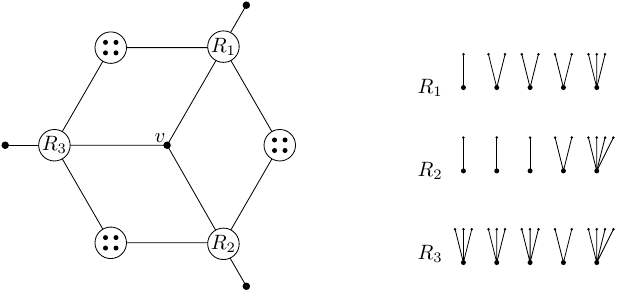}
	\caption{Illustration of a $5$-Catherine wheel. Edges to encircled sets illustrate edges to every vertex in the set. The smaller substructure in the sets $R_1, \dots, R_5$ is illustrated to the right, where the vertices in the sets $R_i$ are the more prominent ones at the bottom of the row.}
	\label{fig:5_catherine_wheel}
\end{figure}

We call such a graph as illustrated in Figure~\ref{fig:5_catherine_wheel} a \emph{$p$-Catherine wheel}. Even though these graphs are $2$-connected and of radius $2$, their highly symmetric global structure together with the lack of small structure in the sets $R_i$ makes it difficult to identify sources for hardness. Here, the case displayed in Figure~\ref{fig:5_catherine_wheel} where the sets $R_i$ only yield a collection of trees and the sum of degrees of the vertices in the sets is congruent modulo $p$ to $0$ is especially difficult. We note that such a graph cannot be order~$p$ bip-reduced for $p$ equal to $2$ or $3$, which highlights the gain of complexity due to higher moduli.

\subsection{Related Literature}
Before we conclude the introduction, we mention related bodies of work. The study of homomorphisms under the point of view of \emph{parameterized algorithms} has been long-established (see Diaz, Serna, and Thilikos~\cite{Diaz:02:Counting_H-Colorings_of_Partial_k-Trees}) but enriched by the work of Amini, Fomin, and Saurabh~\cite{Amini:12:Counting_Subgraphs_via_Homomorphisms} and by Curticapean, Dell, and Marx~\cite{Curticapean:17:Homomorphisms_Are_a_Good_Basis_for_Counting}, who also introduced linear combinations of graph homomorphisms to the study and motivated subsequent works, for instance, Roth and Wellnitz~\cite{Roth:20:Counting_and_Finding_Homomorphisms_is_Universal_for_Parameterized_Complexity}.

The study of homomorphisms from the point of view of \emph{extremal combinatorics} incorporates important conjectures like \emph{Sidorenko's conjecture}~\cite{Sidorenko:91:Inequalities_for_Functionals, Sidorenko:93:A_Correlation_Inequality}, which states a universal lower bound on the number of homomorphisms from a bipartite graph and, in a weaker version, can be found in the work of Simonovits~\cite{Simonovits:84:Extremal_Graph_Problems}. Until today the conjecture remains open but still enjoys new contributions like the recent article by Shams, Ruozzi, and Csikvári~\cite{Shams:19:Counting_Homomorphisms_in_Bipartite_Graphs}.

This leads to the body of work studying \emph{approximation algorithms} including the work of Goldberg and Jerrum~\cite{Goldberg:14:The_Complexity_of_Approximately_Counting_Tree_Homomorphisms} on tree homomorphisms and the work of Galanis, Goldberg, and Jerrum~\cite{Galanis:16:Approximately_Counting_H-Colorings}, who showed that approximating the number of homomorphisms to a fixed graph $H$ is $\probNumBIS$-hard, a notorious complexity class in this body of work. These findings yield an interesting connection to ours in the form of the reduction from (versions of) $\probNumBIS$.

The body of studies concerning different versions of homomorphism problems is vast. It contains dichotomies for the affiliated problem, where the pre-image is from a fixed class of graphs, given by Dalmau and Jonsson~\cite{Dalmau:04:The_Complexity_of_Counting_Homomorphisms_Seen_from_the_Other_Side} and Grohe~\cite{Grohe:07:The_Complexity_of_Homomorphisms_and_Constraint_Satisfaction_Problems}. Peyerimhoff, Roth, Schmitt, Stix, and Vdovina~\cite{Peyerimhoff:21:Parameterized_Modular_Counting_Cayley_Graph_Expanders} gave a dichotomy for the modular version with modulus $p$. Turning toward versions of the problem with fixed image, Focke, Goldberg, and Zivný~\cite{Focke:19:The_Complexity_of_Counting_Surjective_Homomorphisms_and_Compactions} gave a dichotomy for surjective homomorphisms and compactions, and Dyer, Goldberg, and Paterson~\cite{Dyer:07:Counting_Homs_Directed_Acyclic_Graphs} gave a dichotomy for directed homomorphisms if the target is acyclic. The line of research toward the dichotomy for the generalization of $\probNumHom{H}$ allowing weights by Cai, Chen, and Lu~\cite{Cai:13:Graph_Homomorphisms_with_Complex_Values} incorporates works by Bulatov and Grohe~\cite{Bulatov:05:The_Complexity_of_Partition_Functions} and Goldberg, Grohe, Jerrum, and Thurley~\cite{Goldberg:10:A_Complexity_Dichotomy_for_Partition_Functions_with_Mixed_Signs}. Recently, Govorov, Cai, and Dyer~\cite{Govorov:20:Dichotomy_Bounded_Degree_Homomorpisms} extended this research body to bounded degree graphs, for which Cai and Govorov~\cite{Cai:20:Graph_Homs_Bounded_Degree_Complex} gave a dichotomy allowing for complex weights.

The connection of homomorphisms and CSPs was already shown by Feder and Vardi~\cite{Feder:98:The_Computational_Structure_of_Monotone_Monadic_SNP_and_Constraint_Satisfaction}. Bulatov~\cite{Bulatov:13:The_Complexity_of_Counting_CSP} showed that the problem of counting satisfying assignments to a CSP features a dichotomy theorem, a result on which Dyer and Richerby~\cite{Dyer:13:An_Effective_Dichotomy_for_CSP} shed more light. Furthermore, a complete dichotomy for directed homomorphism can be found in the dichotomy on counting weighted versions of CSPs by Cai and Chen~\cite{Cai:17:Counting_CSP_Complex_Weights}. Guo, Huang, Lu, and Xia~\cite{Guo:11:Weighted_Bolean_CSP_Modulo} gave a dichotomy for the associated modular problem.

Finally, the recent work by Cai and Govorov~\cite{Cai:22:Perfect_Matchings} studied the power of expression of the class of homomorphisms. By studying algebras of quantum graphs, they provide a general technique and showed, for instance, that the problem of counting perfect matchings cannot be expressed by counting homomorphisms to a fixed graph $H$ regardless of possible weights in $\C$.

\subsection{Outline}
The paper focuses on counting (quantum-)homomorphisms in a finite field $\Z_p$ of prime order and is structured as follows. After stating the needed preliminaries, we study quantum graphs under the context of homomorphisms and restricted homomorphisms in Section~\ref{sec:quantum_graphs}. Section~\ref{sec:bipartization} presents the reduction to the study of bip-graphs. Subsequently, Section~\ref{sec:gadgets} is devoted to the gadgetry used in order to obtain hardness. The structural analysis on \graphclass{} graphs is conducted in Section~\ref{sec:hardness_graphclass}, where we establish the respective dichotomy. Finally, in Section~\ref{sec:surjective_homomorphisms} we show that Conjecture~\ref{conj:faben-jerrum_p} yields dichotomies for the whole class of partially surjective homomorphisms.

\section{Preliminaries}
\label{sec:prelims}
We denote the disjoint union of two sets $S_1$ and $S_2$ by $S_1 \cupdot S_2$. For a positive integer $k$, we denote by $\sqBrackets{k}$ the set $\set{1, \dots, k}$. Similarly, for two non-negative integers $a$ and $b$ with $a \leq b$, we denote by $\sqBrackets{a;b}$ the set of integers in the closed interval $\sqBrackets{a,b}$.

\subsection{Graphs}
For a graph $G$, we denote its vertex set by $\vertexset[G]$ and its edge set by $\edgeset[G]$. We write $G= \bipGraph$ for the bipartite graph $G$ with fixed bipartition of the vertex set $\lpart \cupdot \rpart$ and edge set $\edgeset[G]$. We recall that a triple $\bipGraph$ is called a \emph{bip-graph}. A \emph{sequence of vertices} $v_1,\dots, v_r$ may be abbreviated by $\vector{v}$. Such a sequence can contain multiple entries of the same vertex.
We call a graph with all loops present a \emph{reflexive} graph. A reflexive complete graph of size $q$ with $q \geq 1$ is denoted by $K^{\circ}_q$. In the same way, a graph is called \emph{irreflexive} if it contains no loop. We note that bipartite graphs are irreflexive by definition.
For a graph $G$ with a subgraph $H$ and a vertex $v$ in $\vertexset[H]$, we denote by $\neigh{v}[H]$ the \emph{neighbourhood of $v$ in $H$} given by the set $\set{u\in \vertexset[H] \given (u,v)\in \edgeset[H]}$ that contains all vertices in $\vertexset[H]$ adjacent to $v$ under edges in $H$, and we denote by $\deg_H(v)$ the size of $\neigh{v}[H]$. We also extend the definition of neighbourhood to sets, that is for the subset $V'$ of $\vertexset[H]$,  $\neigh{V'}[H]  = \set{u\in \vertexset[H] \given \text{there exists $v\in V'$ such that } (u,v)\in \edgeset[H]}$. When we observe the neighbourhood in the graph $G$ not restricted to a subgraph, we omit to state $G$ in the subscript.
Paths in graphs do not repeat vertices; walks may repeat both vertices and edges. The \emph{distance of two connected vertices} $u$ and $v$ in $G$, denoted by $\dist_G(u,v)$, is the length of a shortest path in $G$ connecting $u$ and $v$.

Furthermore, we use the following operations for graphs. Let $G$ and $H$ be graphs on disjoint vertex sets. The \emph{disjoint union} $G\cupdot H$ is the graph with vertex set $\vertexset[G] \cupdot \vertexset[H]$ and edge set $\edgeset[G] \cupdot \edgeset[H]$. For the subset $V'$ of $\vertexset[G]$, the \emph{subgraph of $G$ induced by $V'$} is denoted by $G[V']$, formally $G[V']$ is given by $(V',E')$ with $E'=\set{(v,w)\in \edgeset[G] \given v,w\in V'}$. For two subgraphs $G_1$ and $G_2$ of $G$, we use the shorthand $G_1\cap G_2$ to denote the graph with vertex set $\vertexset[G_1]\cap \vertexset[G_2]$ and edge set $\edgeset[G_1]\cap \edgeset[G_2]$.

\paragraph{Graph homomorphisms}
Let $G$ and $H$ be graphs.
A \emph{homomorphism from $G$ to $H$} is a function $f\colon \vertexset[G]\to \vertexset[H]$, such that $(v_1,v_2)\in \edgeset[G]$ implies $(f(v_1),f(v_2))\in \edgeset[H]$.
We denote by $\Hom[G,H]$ the set of homomorphisms from $G$ to $H$ and by $\numHom[G,H]$ the cardinality of $\Hom[G,H]$. Furthermore, for an integer $k$, $\numHom[G,H][k]$ is the integer in $\sqBrackets{k-1}$ such that $\numHom[G,H]$ is congruent modulo $k$ to $\numHom[G,H][k]$.
An \emph{isomorphism between $G$ and $H$} is a bijective function $\varrho \colon \vertexset[G]\to \vertexset[H]$ preserving the edge relation in both directions, meaning $(v_1,v_2)\in \edgeset[G]$ if and only if $(\varrho(v_1),\varrho(v_2))\in \edgeset[H]$.
If such an isomorphism exists, then we say that $G$ is \emph{isomorphic to} $H$ and denote this by $G\cong H$. 
An \emph{automorphism of $G$} is an isomorphism from the graph $G$ to itself. The \emph{automorphism group of $G$} is denoted by $\Aut[G]$.
An automorphism $\varrho$ is an \emph{automorphism of order~$k$} in case $k$ is the smallest positive integer such that $\varrho^{k}$ is the identity. 

\paragraph{Partially labelled graphs}
Let $H$ be a graph.
A \emph{partially $H$-labelled graph} $J$ consists of an \emph{underlying graph} $G$ denoted by $G(J)$ and a \emph{partial labelling} function $\tau \colon \vertexset[G]\to \vertexset[H]$ denoted by $\tau(J)$. Frequently, we apply for this the shorthand $J=(G,\tau)$.
Every vertex $v$ in the domain of $\tau$ is said to be \emph{$H$-pinned to $\tau(v)$}.
In case it is immediate from the context we omit $H$.
We denote a partial function $\tau$ with finite domain $\set{v_1,\dots,v_r}$ also in the form $\tau = \set{v_1\mapsto \tau(v_1),\dots,v_r\mapsto \tau(v_r)}$.
A \emph{homomorphism from a partially $H$-labelled graph~$J$ to a graph~$H$} is a homomorphism from $G(J)$ to $H$ that respects $\tau$, i.e. if $v$ is in the domain of $\tau$, then $f(v)$ is equal to $\tau(v)$.
We denote the set of homomorphisms from $J$ to $H$ that respect the labelling $\tau(J)$ by $\Hom[J,H]$ and the cardinality of $\Hom[J,H]$ by $\numHom[J,H]$. The value $\numHom[J,H] \pmod p$ is denoted by $\numHom[J,H][p]$.

\paragraph{Graphs with distinguished vertices}  
Let $G$ and $H$ be graphs.
It is often convenient to regard a graph with a sequence $v_1, \dots, v_k$ of not necessarily distinct distinguished vertices, which we denote by $(G, v_1, \dots, v_k)$. For two graphs with $k$ distinguished vertices $(G,\vector{x})$ to $(H,\vector{v})$, a \emph{homomorphism from $(G,\vector{x})$ to $(H,\vector{v})$} is a homomorphism $f$ from $G$ to $H$ that, for each $i\in \sqBrackets{k}$, satisfies $f(x_i)=v_i$. By definition, there does not exist a homomorphism between graphs with an unequal number of distinguished vertices.
We note that a homomorphism from $(G,\vector{x})$ to $(H,\vector{v})$ is equivalently a homomorphism from the partially $H$-labelled graph $(G, \set{x_1\mapsto v_1, \dots, x_k\mapsto v_k})$ to $H$ and vice versa.
For a partially $H$-labelled graph $J$ and vertices $x_1,\dots,x_k$ not in the domain of $\tau(J)$, we identify a homomorphism from $(J, \vector{x})$ to $(H,\vector{v})$ with the corresponding homomorphism from $(G(J), \tau(J) \cupdot \set{x_1\mapsto v_1, \dots, x_k\mapsto v_k})$ to $H$.
The two graphs with $k$ distinguished vertices $(G, \vector{x})$ and $(H, \vector{v})$ are \emph{\isomorphic{}} if there is an isomorphism $\phi$ from $G$ to $H$ such that, for each $i\in \sqBrackets{k}$, it holds $\phi(x_i)=v_i$. We denote this by $(G, \vector{x}) \cong (H, \vector{v})$.
An \emph{automorphism of $(G, \vector{x})$} is an isomorphism from $(G, \vector{x})$ to itself. The \emph{automorphism group of $(G,\vector{x})$} is denoted by $\Aut[(G,\vector{x})]$.

\paragraph{Graph algebra}
For a reference on the used operations, the reader is referred to \Lovasz~\cite{Lovasz:12:book:Large_Networks_Graph_Limits}. Let $G_1$ and $G_2$ be two graphs. The \emph{tensor product} of $G_1$ and $G_2$ is the graph on the vertex set $\vertexset[G_1] \times \vertexset[G_2]$, a pair of vertices $(u,v)$ and $(u',v')$ is adjacent if and only if $(u,u') \in \edgeset[G_1]$ and $(v,v') \in \edgeset[G_2]$. The tensor product of $G_1$ and $G_2$ is denoted by $G_1 \otimes G_2$, and we note that it has many alternative names, e.g. the \emph{direct product}, the \emph{weak graph product}, and the \emph{Kronecker product}.

Let $(G_1, \vector{x}_1)$ and $(G_2, \vector{x}_2)$ be two graphs with $k$ distinguished vertices. The \emph{dot product} $(G_1, \vector{x}_1) \odot (G_2, \vector{x}_2)$ is the graph obtained from taking the disjoint union $(G_1, \vector{x}_1) \cupdot (G_2, \vector{x}_2)$ and identifying $\vector{x}_1$ with $\vector{x}_2$ component-wise. In particular, if $k$ is $0$, then the dot product is identical to the disjoint union.

We obtain the following identities for all vertex-disjoint graphs $G$, $H_1$, and $H_2$:
\begin{align}
	\label{eq:homs_tensor_product_target}
	\numHom[G,(H_1\otimes H_2)] &= \numHom[G, H_1] \cdot \numHom[G, H_2],
	\shortintertext{if $G$ is connected, then we have also}
	\label{eq:homs_disjoint_union_target}
	\numHom[G, (H_1 \cupdot H_2)] &= \numHom[G, H_1] + \numHom[G, H_2].
	\intertext{Additionally, the following identity holds for all vertex-disjoint graphs with $k$ distinguished vertices $(G_1, \vector{x}_1)$, $(G_2, \vector{x}_2)$, and $(H, \vector{v})$:}
	\label{eq:homs_disjoint_union_origin}
	\numHom[((G_1, \vector{x}_1) \odot (G_2, \vector{x}_2)), (H, \vector{v})] &= \numHom[(G_1, \vector{x}_1), (H, \vector{v})] \cdot \numHom[(G_2, \vector{x}_2), (H, \vector{v})].
\end{align}

\paragraph{Graph constructions}
We frequently describe graph constructions with the dot product. For two graphs with $k$ distinguished vertices $(G_1, \vector{x})$ and $(G_2, \vector{y})$, the dot product $(G_1, \vector{x}) \odot (G_2, \vector{y})$ is by definition a graph with distinguished vertices $x_1 = y_1, \dots, x_k = y_k$. For any graph with $k$ distinguished vertices $(H, \vector{v})$, any homomorphism in $\Hom[((G_1, \vector{x}) \odot (G_2, \vector{y})), (H, \vector{v})]$ maps the sequence $x_1 = y_1, \dots, x_k = y_k$ coordinate-wise to the sequence $\vector{v}$.

By the equivalence between a partially labelled homomorphism and a homomorphism for graphs with distinguished vertices, this construction applies to partially labelled graphs, such that \eqref{eq:homs_disjoint_union_origin} holds also if $G_1$ and $G_2$ are partially $H$-labelled. We provide a formal argumentation how the definitions and \eqref{eq:homs_disjoint_union_origin} extend to labelled graphs.

Let $H$ be a graph and $J_1$ and $J_2$ be two partially $H$-labelled graphs with partial $H$-labelling $\tau_1 = \tau(J_1)$ and $\tau_2 = \tau(J_2)$. The disjoint union of $J_1$ and $J_2$ is the graph $(G(J_1) \cupdot G(J_2), \tau_1 \cupdot \tau_2)$ denoted by $J_1 \cupdot J_2$. Here, the disjoint union of labelling functions is given by the disjoint union of single vertex mappings, i.e., for a vertex $x$ in $\dom(\tau_1)$, the image is $(\tau_1 \cupdot \tau_2) (x) = \tau_1 (x)$, and for a vertex $x$ in $\dom(\tau_2)$, the image is $(\tau_1 \cupdot \tau_2) (x) = \tau_2 (x)$. 

We are not allowed to apply \eqref{eq:homs_disjoint_union_origin} now as $\tau_1$ and $\tau_2$ might have different range in $H$. However, we circumvent this issue by a slight alteration on the graphs. For $J_1$, we add the essentially void $H$-labelling $\tau_2$ by taking the disjoint union $G(J_1) \cupdot \dom(\tau_2)$ of the underlying graph $G(J_1)$ and the collection of isolated vertices $\dom(\tau_2)$. We denote this operation by $J_1 \cupdot \tau_2$. Since the labelling $\tau_2$ precisely states the mapping for any vertex in $\dom(\tau_2)$, we deduce $\numHom[(J_1 \cupdot \tau_2), H] = \numHom[J_1, H]$. Repeating the operation, we find the same $H$-labelling in $J_2 \cupdot \tau_1$ and deduce by \eqref{eq:homs_disjoint_union_origin}
\[
	\numHom[J_1, H] \cdot \numHom[J_2, H] = \numHom[J_1 \cupdot \tau_2, H] \cdot \numHom[J_2 \cupdot \tau_1, H]= \numHom[J_1 \cupdot J_2, H].
\]

With the disjoint union at hand, we obtain a dot product of partially $H$-labelled graphs as before for unlabelled graphs. A partially $H$-labelled graph with $k$ distinguished vertices consists of a partially $H$-labelled graph $J$ and a tuple $\vector{x}$ of $k$ distinguished vertices in $\vertexset[G(J)])$. We denote this by $(J, \vector{x})$. Let $(J_1, \vector{x})$ and $(J_2, \vector{y})$ be two partially $H$-labelled graphs with $k$ distinguished vertices. The dot product of $(J_1, \vector{x})$ and $(J_2, \vector{y})$ is constructed from the disjoint union of $(J_1, \vector{x})$ and $(J_2, \vector{y})$ by identifying $\vector{x}$ with $\vector{y}$ component-wise.
Following our argumentation for the disjoint union, \eqref{eq:homs_disjoint_union_origin} yields, for every $k$-tuple $\vector{v}$ of distinguished vertices in $\vertexset[H]$,
\begin{align*}
	\numHom[(J_1, \vector{x}), (H, \vector{v})] \cdot \numHom[(J_2, \vector{y}), (H, \vector{v})] =& \numHom[(J_1 \cupdot \tau_2, \vector{x}), (H, \vector{v})] \cdot \numHom[(J_2 \cupdot \tau_1, \vector{y}), (H, \vector{v})] \\
	=& \numHom[(J_1 \cupdot \tau_2, \vector{x}) \odot (J_2 \cupdot \tau_1, \vector{y}), (H, \vector{v}) ] \\
	=& \numHom[(J_1, \vector{x}) \odot (J_2, \vector{y}), (H, \vector{v}) ].
\end{align*}
As we have seen, we are allowed to safely apply the dot product and \eqref{eq:homs_disjoint_union_origin} also for pairs of partially $H$-labelled graphs. We do so frequently when we construct partially $H$-labelled graphs in the later sections of the paper.

\subsection{Algebra}
An introduction to the fundamentals of algebra is provided in~\cite{Artin:11:book:Algebra}.
Algebraic concepts important for our study are \emph{groups} $\group$ with \emph{subgroups} $\group[H]$. We apply the property that the intersection of a pair of subgroups yields a subgroup. The \emph{order} of a group $\group$ is denoted by $\Ord(\group)$ or $\abs{\group}$ and the order of an element $g$ in $\group$ is denoted by $\Ord(g)$ or $\abs{g}$. For a prime $p$, the cyclic group $\Z_p$ is a \emph{finite field} and the set $\sqBrackets{p-1}$ of \emph{units} is denoted by $\Zsp$. For a subgroup $\group[H]$ of $\group$, the \emph{index} of $\group[H]$ is denoted by $\abs{\group \colon \group[H]}$.

\begin{theorem}[Lagrange]\label{thm:Lagrange}
	If $\group$ is a group of finite order with subgroup $\group[H]$, then
	\[
		\abs{\group} = \abs{\group \colon \group[H]} \cdot \abs{\group[H]}.
	\]
\end{theorem}

The groups of prime order are especially important for us. For these, Cauchy and Fermat gave the following set of fundamental theorems.
\begin{theorem}[Cauchy]\label{thm:Cauchy}
	Let $\group$ be a group of finite order and $p$ be a prime. If $p$ divides the order of $\group$, then there exists an element $g$ in $\group$ of order $p$.
\end{theorem}
\begin{theorem}[Fermat]\label{thm:Fermat}
	Let $p$ be a prime. If the interger $a$ is not congruent modulo $p$ to $0$, then $a^{p-1}$ is congruent modulo $p$ to $1$.
\end{theorem}

For a graph $H$, the set of automorphisms of $H$ denoted by $\Aut[H]$ is in fact a \emph{group action} acting on the set of vertices $\vertexset[H]$. Let a general group $\group$ act on a set $X$ by a group action $\alpha \colon \group \times X \to X$ shortened to $g(x)$, where $g$ is in $\group$ and $x$ is in $X$. The \emph{stabilizer} of an element $x \in X$ is the subgroup of $\group$ given by the set $\set{g \in \group \given g (x) = x}$. It is denoted by $\Stab{x}$. The \emph{orbit} of $x$ is the subset $\set{g (x) \given g \in \group}$ of $X$ and denoted by $\Orb[x][\group]$. When the group $\group$ is obvious from context, we omit it. 

Consequently, for a graph $H$, a positive integer $k$, and a tuple of vertices $\vector{v}$ in $\vertexset[H]^k$, the orbit of $\vector{v}$ is the set of images of $\vector{v}$ under automorphisms of $H$. Equivalently, a tuple $\vector{w}$ in $\vertexset[H]^k$ is in $\Orb[\vector{v}]$ if and only if $(H, \vector{v}) \cong (H, \vector{w})$. The set of tuples $\vertexset[H]^k$ decomposes into pairwise disjoint orbits. This yields the following well-known lemma frequently utilized for proving pinning in homomorphism problems (see e.g.~\cite[Lemma~5.12.]{Goebel:21:Counting_Homomorphisms_Trees}). For the sake of completeness, we state a short proof.
\begin{lemma}
	\label{lem:pinning_split}
	Let $H$ be a graph, $k$ be a positive integer, and $\nu$ be the number of disjoint orbits in $\vertexset[H]^k$. If $\vector{v}_1, \dots, \vector{v}_\nu$ are representatives of the orbits in $\vertexset[H]^k$, then, for every graph $G$ and every tuple $\vector{x}$ in $\vertexset[G]^k$,
	\[
	\numHom[G,H] = \sum_{i \in \sqBrackets{\nu}} \abs[\big]{\Orb[\vector{v}_i]} \cdot \numHom[(G, \vector{x}), (H, \vector{v}_i)].
	\]
\end{lemma}
\begin{proof}
	Every homomorphism in $\Hom[G,H]$ has to map the given tuple $\vector{x}$ to some tuple $\vector{v}$ in $\vertexset[H]^k$. By definition, there exists a unique representative $\vector{v}_i$ such that $(H, \vector{v} ) \cong (H, \vector{v}_i)$. The summation on the right-hand side accounts for all tuples in $\vertexset[H]^k$ exactly once, which concludes the proof.
\end{proof}

Although the orbit of a tuple of vertices is not a subgroup of the group of automorphisms, the cardinality of an orbit is still a divisor of the group's order by the Orbit-Stabilizer theorem.
\begin{theorem}[Orbit-Stabilizer]\label{thm:orbit_stabilizer}
	If $\group$ is a group acting on a set $X$ and $x$ is an element in $X$, then
	\[
	\abs[\big]{\group} = \abs[\big]{\Orb[x]} \cdot \abs[\big]{\Stab{x}} .
	\]
\end{theorem}

For a graph $H$, a positive integer $k$, an automorphism $\varrho$ of $H$, and a tuple $\vector{y}$ in $\vertexset[H]^k$, we denote by $\Orb[\vector{y}][\varrho]$ the set $\set{\varrho^i(\vector{y})}_{i \in [\Ord(\varrho)]}$. In this way, we obtain the following consequence.
\begin{lemma}\label{lem:tuple-orbit}
	Let $p$ be a prime, $H$ be a graph, $k$ be a positive integer, and $\varrho$ be an automorphism of $H$. Further, let there exists a tuple $\vector{y}$ in $\vertexset[H]^k$ with cardinality $\abs{\Orb[\vector{y}][\varrho]}$ denoted by $r$. If $r$ is larger than $1$ and there exists $j$ in $\sqBrackets{r-1}$ such that $\varrho^{p\cdot j}(\vector{y})$ is equal to $\vector{y}$, then $\Aut[H]$ contains an element of order $p$.
\end{lemma}
\begin{proof}
	The cyclic group $\group_\varrho$ generated by $\varrho$ acts on $H$ and is a subgroup of $\Aut[H]$. Due to Lagrange's Theorem~\ref{thm:Lagrange}, the order of $\group_\varrho$ divides the order of $\Aut[H]$. 
	By the Orbit-Stabilizer Theorem~\ref{thm:orbit_stabilizer}, $\abs{\Orb[\vector{y}][\varrho]}$ divides the order of $\group_\varrho$. Since $\varrho^{p \cdot j}(\vector{y})=\vector{y}$, we obtain that $\abs{\Orb[\vector{y}][\varrho]}$ divides $p$ or $j$. The first has to hold by $j\in \sqBrackets{r-1}$ and thus $\abs{\Aut[H]}\equiv 0 \pmod p$. By Cauchy's Theorem~\ref{thm:Cauchy}, there exists an automorphism $\psi \in \Aut[H]$ with $\Ord(\psi)=p$.
\end{proof}

\subsection{Reductions and Tractability}
We briefly discuss the tractability results for the following problem; the focus of this paper.
\prob
{$\probNumHom{H}[p]$.}
{Prime $p$ and graph $H$.}
{Graph $G$.}
{$\numHom[G,H][p]$.}

We identify the classes of graphs $H$ for which $\probNumHom{H}[p]$ can be solved in polynomial time. When counting graph homomorphisms modulo a prime $p$, the automorphisms of order $p$ of a target graph $H$ help us identify groups of homomorphisms that cancel out. More specifically, let the target graph $H$ admit an automorphism $\varrho$ of order $p$. For any 
homomorphism $f$ from the input graph $G$ to $H$, $f\circ\varrho$ is also a homomorphism from $G$ to $H$. Intuitively, if $f \circ \varrho$ is not equal to $f$, then the set $\set{f\circ\varrho^{j} \given j \in \sqBrackets{p}}$ has cardinality of a multiple of $p$ and cancel outs. This intuition is captured by the theorem of Faben and Jerrum~\cite[Theorem~3.4.]{Faben:15:Parity_Graph_Homs}. Before we formally state their theorem, we need the following definition.
\begin{definition}
	Let $H$ be a graph and $\varrho$ an automorphism of $H$. We denote by $H^\varrho$ the subgraph of $H$ induced by the fixed points of $\varrho$.
\end{definition}

\begin{theorem}[Faben and Jerrum]\label{thm:reduction_cancel_order_p_auto}
	Let $p$ be a prime and $G$ and $H$ be graphs. If $\varrho$ is an automorphism of $H$ of order $p$, then $\numHom[G,H]$ is congruent modulo $p$ to $\numHom[G, H^\varrho]$.
\end{theorem}

We repeat the above reduction of $H$ recursively in the following way.
\begin{definition}
	Let $p$ be a prime and $H$ and $H'$ be graphs. We write $H \relArrow[][p] H'$ if there exists an automorphism~$\varrho$ of~$H$ of order~$p$ such that $H^\varrho=H'$.	Additionally, we write $H \relArrow[][\ast p] H'$ if   there exists a sequence of graphs $H_1, \dots, H_k$ such that $H \relArrow[][p] H_1 \relArrow[][p] \cdots \relArrow[][p] H_k = H'$.
\end{definition}

We note that, contrary to previous work on $\probNumHom{H}[p]$, we adjusted notation and moved the \enquote{$p$} in the superscript. This allows us to denote restricted versions of the same reduction in a similar and uniform way.
Faben and Jerrum~\cite[Theorem~3.7.]{Faben:15:Parity_Graph_Homs} show that, for any choice of intermediate automorphisms of order~$p$, the reduction 
$H \relArrow[][\ast p]\!$ results in a unique graph up to isomorphism. 
\begin{theorem}[Faben and Jerrum]\label{thm:reduced_form_unique}
	Given a graph $H$ and a prime $p$, there is (up to isomorphism) exactly one graph $\normreduced{H}$ such that $\normreduced{H}$ has no automorphism of 
	order~$p$ and $H \relArrow[][\ast p] \normreduced{H}$.
\end{theorem}

The latter suggests the following definition.
\begin{definition}
\label{defn:reduced-form}
	Let $p$ be a prime and $H$ be a graph. We call the unique (up to isomorphism) graph $\normreduced{H}$ with $H \relArrow[][\ast p] \normreduced{H}$ the \emph{order~$p$ reduced form} of $H$.
\end{definition}

We observe that the same reduction is applicable to graphs $(H, \vector{v})$ with $k$ distinguished vertices by restricting our attention to automorphisms in $\Aut[(H, \vector{v})]$, that is automorphisms of $H$ that fix the sequence $\vector{v}$. The order~$p$ reduced form of $(H, \vector{v})$ is thus well-defined; a technicality we added for pedantic reasons that allows our argumentation in the next section to be more uniform.

In order to compute the number of homomorphisms from a graph $G$ to a graph $H$ modulo $p$ it suffices to compute the number of homomorphisms from $G$ to $\normreduced{H}$ modulo $p$. We refer to the dichotomy theorem by Dyer and Greenhill~\cite[Theorem~1.1.]{Dyer:2000:Counting_Graph_Homs} to obtain the graphs for which $\probNumHom{H}[p]$ is computable in polynomial time. 
\begin{theorem}[Dyer and Greenhill \cite{Dyer:2000:Counting_Graph_Homs}]\label{thm:dyer-greenhil}
	Let $H$ be a graph. If every connected component of $H$ is a complete bipartite graph or a reflexive complete graph, then $\probNumHom{H}$ can be solved in polynomial time. Otherwise, $\probNumHom{H}$ is $\classNumP$-complete.
\end{theorem}

We notice that a polynomial-time algorithm for $\probNumHom{H}$ gives a polynomial-time algorithm for $\probNumHom{H}[p]$ by applying the modulo~$p$ operation. This gives the following characterization for instances of $\probNumHom{H}[p]$ computable in polynomial time.
\begin{corollary}\label{cor:polyt-graphs}
	Let $p$ be a prime and $H$ be a graph. If every connected component of $\normreduced{H}$ is a complete bipartite graph or a reflexive complete graph, then $\probNumHom{H}[p]$ is computable in polynomial time.
\end{corollary}

\section{Quantum Graphs}
\label{sec:quantum_graphs}
We recall that a \emph{quantum graph} $\quantum{F}$ is a finite linear combination of pairwise non-\isomorphic{} graphs $\family{F}$ given by a set of coefficients $\set{\alpha_F}_{F \in \family{F}}$ in $\C$, i.e.
\[
	\quantum{F} = \sum_{F \in \family{F}} \alpha_F \cdot F.
\]
We follow \Lovasz~\cite{Lovasz:12:book:Large_Networks_Graph_Limits} and call the graphs $F$ in $\family{F}$ with non-zero coefficient $\alpha_F$ the \emph{constituents} of $\quantum{F}$. Further, we refer by \emph{the coefficients of $\quantum{G}$} to the set of coefficients $\set{\alpha_F}_{F \in \family{F}}$.

A graph parameter $f\colon \family{G} \to \C$ that maps a class of graphs $\family{G}$ to $\C$ can be extended to quantum graphs by the linear combination of function-values. Thus, for a graph $G$, the linear combination $\sum_{H \in \family{F}} \alpha_F \cdot \numHom[G, F]$ is equal to $\numHom[G, \quantum{F}]$. This also leads to the notion of a \emph{quantum graph with distinguished vertices}: Let $k$ be a non-negative integer, we say that a quantum graph $\quantum{F}$ has \emph{$k$ distinguished vertices} if every constituent of $\quantum{F}$ is a graph with $k$ distinguished vertices. When $k$ is $0$ this is equivalent to having no distinguished vertices, but it helps to state results uniformly. We omit the distinguished vertices and instead identify a constituent $(F_i, \vector{v}_i)$ with the short $F_i$ when stating the distinguished vertices explicitly is not necessary. In principle, the notion of a quantum graph allows for a linear combination of constituents with a different number of distinguished vertices, i.e. non-uniform number of distinguished vertices. However, since there are no homomorphisms between graphs with a different number of distinguished vertices, the number of homomorphisms between quantum graphs with a non-uniform number of distinguished vertices translates into the sum of homomorphisms between sub-quantum graphs with a uniform number of vertices. For the sake of this paper, such generality is not necessary.

\prob%
{$\probNumHom{\quantum{H}}$.}
{Non-negative integer $k$ and quantum graph $\quantum{H}$ with $k$ distinguished vertices and coefficients in $\C$.}
{Graph with $k$ distinguished vertices $G$.}
{$\numHom[G, \quantum{H}]$.}

In the same spirit, we extend graph homomorphisms to allow for pre-image quantum graphs $\quantum{G}$. For two quantum graphs $\quantum{G}$ and $\quantum{H}$ given by $\quantum{G} = \sum_{G \in \family{G}} \beta_G \cdot G$ and $\quantum{H}=\sum_{H \in \family{H}} \alpha_H \cdot H $, the number of homomorphisms $\numHom[\quantum{G}, \quantum{H}]$ is given by 
\[
	\numHom[\quantum{G}, \quantum{H}] = \sum_{G \in \family{G}} \sum_{H \in \family{H}} \alpha_G \cdot \alpha_H \cdot \numHom[G, H].
\]
For the case of graphs with distinguished vertices, we state the expression once in detail. Let $k$ and $k'$ be positive integers and let $\quantum{G}$ and $\quantum{H}$ be a quantum graph with $k$ and $k'$ distinguished vertices, respectively. The number of homomorphisms is given by
\[
	\numHom[\quantum{G},\quantum{H}] = \sum_{(G, \vector{x}) \in \family{G}} \sum_{(H, \vector{v}) \in \family{H}} \beta_{(G, \vector{x})} \cdot \alpha_{(H, \vector{v})} \cdot  \numHom[(G,\vector{x}),(H,\vector{v})].
\]
We note that this number of homomorphisms is $0$ if $k$ is unequal to $k'$, and hence we only study the case that $k$ is equal to $k'$ in the following.

Many combinatorial situations can be expressed as a quantum graph; for instance, the number of $3$-colourings of a graph $G$ that use all given colours is equal to
\[
	\numHom[G, \quantum{F}] = \numHom[G, K_3] - 3\cdot \numHom[G, K_2] - 3\cdot \numHom[G, K_1] .
\]
\forThesis{Another, more evolved quantum graph as discussed in \cite{Lovasz:12:Book_LargeNetworksGraphLimits} is given by, for any graph $F$ with signed edges $\edgeset[F] = E_+ \cupdot E_-$,
\[
	\quantum{F}=\sum_{F'} (-1)^{\abs{\edgeset[F']} - \abs{E_+}} \cdot F',
\]
where the summation extends over all simple graphs, such that $\vertexset[F']=\vertexset[F]$ and $E_+ \subseteq \edgeset[F'] \subseteq  E_+ \cupdot E_-$. This summation refers to the inclusion-exclusion principle. For any simple graph $G$,  $\numHom[\quantum{F}, G]$ is the number of mappings $\vertexset[F] \to \vertexset[G]$ that map positive edges to edges and negative edges to non-edges. Hence, we are allowed to regard $\quantum{F}$ by just its ground-graph $F$
}
A second example is given by the equation in Lemma~\ref{lem:pinning_split}, where the right side $\sum_{i \in \sqBrackets{\nu}} \abs{\Orb[\vector{v}_i]} \cdot \numHom[(G, \vector{x}), (H, \vector{v}_i)]$ is a quantum graph with distinguished vertices.

When studying graph homomorphisms, the notion of a quantum graph is of great value as it not only allows to study graph properties that are only expressible by a linear combination of homomorphisms and not just by a single homomorphism but also allows transforming the given graphs for a homomorphism problem in a controlled and concise manner using graph products.

Due to the definition of a quantum graph, the dot product as well as the tensor product straightforwardly extend to quantum graphs. We give the construction explicitly for the dot product. Let $k$ be a non-negative integer and $\quantum{F}$ and $\quantum{G}$ be two quantum graphs with $k$ distinguished vertices given by $\quantum{F}=\sum_{F \in \family{F}} \alpha_F \cdot F$ and $\quantum{G} = \sum_{G \in \family{G}} \beta_G \cdot G$. We derive
\[
	\quantum{F} \odot \quantum{G} = \sum_{F \in \family{F}} \sum_{G \in \family{G}} \alpha_F \cdot \beta_G \cdot F \odot G.
\]
Replacing the dot product with the tensor product yields the construction for the tensor product of quantum graphs. We obtain the following implication of \eqref{eq:homs_tensor_product_target} and \eqref{eq:homs_disjoint_union_origin}.
\begin{corollary}
	\label{cor:dot_product}
	Let $k$ be a non-negative integer. Given three quantum graphs with $k$ distinguished vertices $\quantum{F}$, $\quantum{G}$, and $\quantum{H}$, it follows
	\begin{align*}
	\numHom[\quantum{F} \odot \quantum{G}, \quantum{H}] &= \numHom[\quantum{F}, \quantum{H}] \cdot \numHom[\quantum{G}, \quantum{H}], \\
	\numHom[\quantum{G}, \quantum{F} \otimes \quantum{H}] &= \numHom[\quantum{G}, \quantum{F}] \cdot \numHom[\quantum{G}, \quantum{F}].
	\end{align*}
\end{corollary}

Chen et al.~\cite{Chen:19:The_Exponential-Time_Complexity} studied the problem $\probNumHom{\quantum{H}}$, where the quantum graph $\quantum{H}$ is given by $\quantum{H} = \sum_{H \in \family{H}} \alpha_H \cdot H$ and, for every constituent $H$ in $\family{H}$, the coefficient $\alpha_H$ is in $\Q_{\neq 0}$. The key insight of their work is that the problem $\probNumHom{\quantum{H}}$ inherits, for every constituent $H\in \family{H}$, the complexity of the problems $\probNumHom{H}$. They proved this via the existence of an invertible matrix $M$ that allows computing the number of homomorphism to a graph $H$ in $\family{H}$ with an oracle for $\probNumHom{\quantum{H}}$.

When we study the answers to the problem $\probNumHom{\quantum{H}}$ in $\Z_p$ we arrive at the problem $\probNumHom{\quantum{H}}[p]$, which is the focal point of this section. For this, we focus naturally on the case that the coefficients of $\quantum{H}$ are not congruent modulo $p$ to $0$, i.e. every coefficient $\alpha_H$ is in $\Zsp$. Unless otherwise specified, we assume the quantum graphs to have integer coefficients in the following. For $k$ a non-negative integer and two quantum graphs $\quantum{G}$ and $\quantum{H}$ with $k$ distinguished vertices, we denote by $\numHom[\quantum{G}, \quantum{H}][p]$ the value $\numHom[\quantum{G}, \quantum{H}] \pmod p$.

\prob%
{$\probNumHom{\quantum{H}}[p]$.}
{Non-negative integer $k$ and quantum graph $\quantum{H}$ with $k$ distinguished vertices and integer coefficients.}
{Graph with $k$ distinguished vertices $G$.}
{$\numHom[G, \quantum{H}][p]$.}

Let $\quantum{H}$ be a quantum graph $k$ distinguished vertices, constituents $\family{H}$, and set of coefficients $\set{\alpha_H}_{H \in \family{H}}$. In order to study $\probNumHom{\quantum{H}}[p]$ we follow the insight by Faben and Jerrum~\cite{Faben:15:Parity_Graph_Homs} and study the equivalent problem $\probNumHom{\normreduced{\quantum{H}}}[p]$. Here, $\normreduced{\quantum{H}}$ is the quantum graph obtained from $\quantum{H}$ in the following manner. First, for all constituents $H$ in $\family{H}$, we collect the order~$p$ reduced form $\normreduced{H}$ of $H$ in $\family{H}'$ and set $\alpha_{\normreduced{H}}$ equal to $\alpha_H$. Second, we collect in $\normreduced{\family{H}}$ the isomorphism classes of graphs in $\family{H}'$, where the coefficient of the representative $\normreduced{H}$ in $\normreduced{\family{H}}$ is given by the sum of the coefficients of the graphs $F$ in $\family{H}'$ that are \isomorphic{} to $H$, i.e.
\[
	\alpha_{\normreduced{H}} = \sum_{ {F} \in \family{H} \colon \normreduced{F} \cong \normreduced{H}} \alpha_{F}.
\]
By this reduction, there might be cases of constituents $H$ in $\family{H}$ whose coefficient satisfies $\alpha_H \not\equiv 0 \pmod p$ but after the reduction $\alpha_{\normreduced{H}} \equiv 0 \pmod p$. In a nutshell, these constituents \enquote{cancel} out. Therefore, we assume from this point onward, unless otherwise specified, that the parameter quantum graph $\quantum{H}$ for $\probNumHom{\quantum{H}}[p]$ satisfies that, for all constituents $H$ in $\family{H}$, the graph $H$ is order~$p$ reduced and the coefficient $\alpha_H$ is in $\Zsp$. In the same spirit as before, we call such a quantum graph \emph{order~$p$ reduced}.

We are going to show that the cancellations resulting from studying the order~$p$ reduced constituents are the only cancellations affecting a quantum graph. Additional cancellations have to be on a smaller scale, i.e. on individual constituents. We obtain the following dichotomy that also allows for disconnected constituents containing loops.

\InheritenceQuantumHomsModp*

The second case in Theorem~\ref{thm:homs_to_quantum_graph_mod_p} follows directly from the definition of $\probNumHom{\quantum{H}}$. In order to prove the first case, we show that a polynomial-time reduction exists by establishing the following proposition. 
\begin{proposition}\label{prop:reduction_quantum_homs}
	Let $p$ be a prime, $k$ be a non-negative integer, and $\quantum{H}$ be an order~$p$ reduced quantum graph with $k$ distinguished vertices. If we have access to an oracle for $\probNumHom{\quantum{H}}[p]$, then, for any input graph $G$ and each constituent $H$ of $\quantum{H}$, we can compute $\numHom[G, H][p]$ in polynomial time. 
\end{proposition}

With Proposition~\ref{prop:reduction_quantum_homs} at hand we obtain a straightforward proof of Theorem~\ref{thm:homs_to_quantum_graph_mod_p}.
It remains to prove Proposition~\ref{prop:reduction_quantum_homs}, to which we dedicate the following subsection.

\subsection{Pinning in Quantum Graphs}
\label{subsec:pinning_quantum_graphs}

The concepts and argumentations used in this subsection are unaffected by adding distinguished vertices if we assume the number of distinguished vertices to be the same. Therefore, we tacitly omit stating the tuples of vertices for graphs with distinguished vertices in the same way we did previously.

The key ingredient to our reduction is the insight provided by~\cite[Lemma~4.2.]{Borgs:13:Convergence_Graphs_Bounded_Degree}. Since the quantum graphs under study are allowed to contain graphs with loops and possible distinguished vertices, we are first going to extend~\cite[Lemma~4.2.]{Borgs:13:Convergence_Graphs_Bounded_Degree} to the case of such graphs but still without multi-edges. For two graphs $G$ and $H$ with $k$ distinguished vertices, we denote the set of injective and surjective homomorphisms by $\Inj[G, H]$ and $\Surj[G, H]$, respectively. We note that, as in \cite{Borgs:13:Convergence_Graphs_Bounded_Degree}, $\Surj[G, H]$ refers to homomorphisms surjective on vertices \emph{and} edges. Additionally, we abbreviate notation and denote the value of $\abs{\Inj[G, H]}$ and $\abs{\Surj[G, H]}$ by $\numInj[G, H]$ and $\numSurj[G, H]$, respectively. In the same way, we denote the values of $\numInj[G, H]$ and the value $\numSurj[G, H]$ in $\Z_p$ by $\numInj[G, H][p]$ and $\numSurj[G, H][p]$, respectively. Similarly, for a graph $G$, the value $\abs{\Aut[G]}$ is denoted by $\numAut[G]$ and the value $\numAut[G]$ in $\Z_p$ is denoted by $\numAut[G][p]$.

A family of graphs $\family{F}$ is said to be \emph{closed under the image of surjective homomorphic image} if, for every graph $F$ in $\family{F}$ and every graph $G$ with $\numSurj[F, G]>0$, then $G$ is in $\family{F}$. We obtain that the argumentation of Borgs et al.~\cite[Lemma~4.2.]{Borgs:13:Convergence_Graphs_Bounded_Degree} is still valid if we allow loops and distinguished vertices. For the sake of completeness, we provide the full proof. 
\begin{lemma}\label{lem:hom_matrix_nonsingular_loops}
	Let $k$ be a non-negative integer and let $\family{F}$ be a family $\set{F_i}_{i \in I}$ with index set $I$, where $\family{F}$ consists of pairwise non-\isomorphic{} graphs with $k$ distinguished vertices and without multi-edges. If $\family{F}$ is closed under surjective homomorphic image, then the matrix
	\[
	M_{\numHom} = \sqBrackets[\big]{\numHom[F_i,F_j] }_{i,j \in I}
	\]
	is nonsingular.
\end{lemma}
\begin{proof}
	All graphs in this proof are assumed to have $k$ distinguished vertices.
	For two indices $i$ and $j$ in $I$, let $f$ be a homomorphism in $\Hom[F_i, F_j]$. There exists a graph $J$ such that $f$ decomposes into a surjective homomorphism $g$ in $\Surj[F_i, J]$ followed by an injective homomorphism $h$ in $\Inj[J, F_j]$, i.e. $f = h \circ g$. In particular, $J$ is the image of $f$, $g$ is given by $f$, and $h$ is the embedding of $J$ into $H$. By assumption, $\family{F}$ is closed under surjective homomorphic image, and thus $J$ is an element in $\family{F}$. For each automorphism $\varrho$ in $\Aut[J]$, we find another pair of homomorphism $g'$ and $h'$ by $g'= \varrho \circ g$ and $h' = h \circ \varrho^{-1}$. Similar to $g$ and $h$, the homomorphism $g'$ is in $\Surj[F_i, J]$, the homomorphism $h'$ is in $\Inj[J, F_j]$, and $f$ is equal to $h' \circ g'$. In order to avoid double-counting, we divide by the $\numAut[J]$ options for $\varrho$ and obtain
	\begin{equation}\label{eq:decomposition_homs_matrix}
		\numHom[F_i,F_j] = \sum_{J \in \family{F}} \frac{\numSurj[F_i, J] \cdot \numInj[J, F_j] }{\numAut[J]} ,
	\end{equation}
	from which we obtain the decomposition
	\[
	M_{\numHom} = 	M_{\numSurj} \cdot \left(D_{\numAut}\right)^{-1} \cdot M_{\numInj} ,
	\]
	where the matrices $M_{\numSurj}$ and $M_{\numInj}$ are given by
	\begin{align*}
		M_{\numSurj} &= \sqBrackets[\big]{\numSurj[F_i, F_j] }_{i,j \in I}, \quad
		M_{\numInj} = \sqBrackets[\big]{\numInj[F_i, F_j] }_{i,j \in I},
	\end{align*}	
	and $D_{\numAut}$ is the $\abs{I} \times \abs{I}$ diagonal matrix with the values $\numAut[F_i]$ on the diagonal. The matrix $D_{\numAut}$ is invertible because $\Aut[F_i]$ contains at least the identity. If $M_{\numSurj}$ is lower-triangular and $M_{\numInj}$, then the lemma follows.
	
	We are going to construct a total ordering of $\family{F}$ such that, for two indices $i$ and $j$ in $I$ with $i < j$,
	\begin{itemize}
		\item $\numSurj[F_i, F_j] = 0$;
		\item $\numInj[F_j, F_i] = 0$.
	\end{itemize}
	We note that the diagonal entries of $M_{\numInj}$ and $M_{\numSurj}$ are given by the number of automorphisms and thus are non-zero. In this way, $M_{\numInj}$ is upper-triangular and $M_{\numSurj}$ is lower-triangular for a total ordering satisfying both properties.
	
	We order the family $\family{F}$ in the following way.
	First, we construct a partial ordering. For two indices $i$ and $j$ in $I$ with $i < j$, we assume a partial ordering of $\family{F}$ satisfying 
	\begin{enumerate}
		\item $\abs{\vertexset[F_i]} \leq \abs{\vertexset[F_j]}$;
		\item if $\abs{\vertexset[F_i]} = \abs{\vertexset[F_j]}$, then $\abs{\edgeset[F_i]} \leq \abs{\edgeset[F_j]}$.
	\end{enumerate}
	Second, we extend this partial ordering to a total ordering by taking an arbitrary order for every pair of graphs $F_i$ and $F_j$ with the same number of edges and vertices.
	
	Let $F_i$ and $F_j$ be a pair of graphs in $\family{F}$ with $F_i < F_j$ by the total ordering. If $\abs{\vertexset[F_i]} < \abs{\vertexset[F_j]}$ or $\abs{\edgeset[F_i]}<\abs{\edgeset[F_j]}$, then both $\Surj[F_i, F_j]$ and $\Inj[F_j, F_i]$ are empty.	
	It remains to argue on the case that $F_i$ and $F_j$ have the same number of vertices and edges.
	We recall that $F_i$ and $F_j$ are non-\isomorphic{}. However, assuming $\numInj[F_i, F_j]>0$ yields an isomorphism from $F_i$ to $F_j$, a contradiction. Similarly, assuming $\numSurj[F_i, F_j]>0$ yields an isomorphism. This concludes the proof.
\end{proof}

We note the subtlety of the statement. The family $\family{F}$ is allowed to contain two graphs $(F, \vector{x}_1)$ and $(F, \vector{x}_2)$ originating from the same graph $F$ but with different distinguished vertices $\vector{x}_1$ and $\vector{x}_2$. By assumption $(F, \vector{x}_1)$ is not \isomorphic{} to $(F, \vector{x}_2)$, i.e. $\vector{x}_1$ is not in $\Orb[\vector{x}_2]$. The non-existing isomorphism between $(F, \vector{x}_1)$ and $(F, \vector{x}_2)$ is what we utilized in the proof.

By the distributivity of the modulo operator, we adapt the proof of Lemma~\ref{lem:hom_matrix_nonsingular_loops} to modular counting. Important is the notion \emph{closed under surjective homomorphic image}, which played a crucial role due to the decomposition~\eqref{eq:decomposition_homs_matrix}. Towards an equivalent statement when counting modulo $p$, let us assume, for a pair of graphs $F_i$ and $F_j$, that there is a graph $J$ with both $\numSurj[F_i, J]$ and $\numInj[F_i, J]$ non-zero, but $J$ contains an automorphism of order $p$. By Lagrange's Theorem~\ref{thm:Lagrange}, $\numAut[J][p]$ is equal to $0$ and each of these automorphisms yields a family of \isomorphic{} homomorphisms from $F_i$ to $F_j$ via $J$. Therefore, these homomorphisms cancel out, and it suffices to account for order~$p$ reduced graphs $J$ satisfying that $\numSurj[F_i, J]$ is non-zero.

\HomMatrixNonSingularModp*
\begin{proof}
	All graphs in this proof are assumed to have $k$ distinguished vertices.
	By assumption, we have for all graphs $F_i$ in $\family{F}$ that $\numAut[F_i][p]$ is non-zero. Following the proof of Lemma~\ref{lem:hom_matrix_nonsingular_loops}, we deduce that the matrix $M_{\numInj[][p]}$ with
	\[
	M_{\numInj[][p]} = \sqBrackets[\big]{\numInj[F_i, F_j][p]}_{i,j \in I}
	\]
	is upper triangular and the matrix $M_{\numSurj[][p]}$ with
	\[
	M_{\numSurj[][p]} = \sqBrackets[\big]{\numSurj[F_i, F_j][p] }_{i,j \in I}
	\]		
	is lower triangular. Additionally, we deduce by the distributivity of the modulo operator 
	\[
		\numHom[F_i, F_j][p] = \sum_{J \in \family{F}} \frac{\numSurj[F_i, J][p] \cdot \numInj[J, F_j][p]  }{\numAut[J][p]} ,
	\]
	from which we obtain the decomposition 
	$M_{\numHom[][p]} = 	M_{\numSurj[][p]} \cdot \left(D_{\numAut[][p]}\right)^{-1} \cdot M_{\numInj[][p]}$,
	where $D_{\numAut[][p]}$ is the $\abs{I} \times \abs{I}$ diagonal matrix with the values $\numAut[F_i][p]$ on the diagonal. Since every graph in $\family{F}$ is order~$p$ reduced, we have that $\numAut[F_i][p]$ is non-zero. Since $\Z_p$ is a field, it does not contain zero-divisors and  $M_{\numHom[][p]}$ being non-singular over $\Z_p$ is equivalent to $\det(M_{\numHom[][p]}) \not\equiv 0 \pmod p$. The corollary follows.
\end{proof}

We note here the generality of Corollary~\ref{cor:hom_matrix_nonsingular_mod_p}, which is not restricted to the case of fixed image graphs. This follows remarks made by \Lovasz{}~\cite[Section~5.5.]{Lovasz:12:book:Large_Networks_Graph_Limits}. By taking for $\mathcal{F}$ the class of all order~$p$ reduced graphs with $k$ distinguished vertices, one immediately obtains that in $M_{\numHom[][p]}$ no pair of rows and no pair of columns is identical. This gives a large series of results that have been proved independently but with similar arguments, e.g. Peyerimhoff et al.~\cite[Lemma~1.9.]{Peyerimhoff:21:Parameterized_Modular_Counting_Cayley_Graph_Expanders} or Göbel et al.~\cite[Lemma~5.3.]{Goebel:21:Counting_Homomorphisms_Trees}. 

We note that our approach for Theorem~\ref{thm:homs_to_quantum_graph_mod_p} is similar to establishing pinning in $\probNumHom{H}[p]$ as we \enquote{pin} to the constituent. Instead of the approach in~\cite{Goebel:21:Counting_Homomorphisms_Trees} for pinning in $\probNumHom{H}[p]$, we are going to apply Corollary~\ref{cor:hom_matrix_nonsingular_mod_p} and the notion of a quantum graph also for the preimage of a homomorphism. Nonetheless, the approaches are strongly intertwined as we elaborate in Subsection~\ref{subsec:quantum_graph_part_lab_homs}. The following definition is derived from~\cite{Goebel:21:Counting_Homomorphisms_Trees} and proves to be useful. 
\begin{definition}
	Let $p$ be a prime, $k$ be a non-negative integer, and $\family{H}$ be a set $\set{H_1, \dots, H_r}$ consisting of pairwise non-\isomorphic{} order-$p$ reduced graphs with $k$ distinguished vertices. Let $\vector{v}$ be a $1 \times r$ vector $\sqBrackets{v_1, \dots, v_r}$ with entries in $\Z_p$. We call $\vector{v}$ \emph{$\family{H}$-implementable} if there exists a quantum graph $\quantum{F}$ with $k$ distinguished vertices and number of constituents depending only on $k$ and $\max_{H \in \family{H}} \abs{\vertexset[H]}$, such that for, all $i \in \sqBrackets{r}$,
	\[
	v_i = \numHom[\quantum{F},H_i][p] .
	\]
	In this case, we say that $\quantum{F}$ \emph{$\family{H}$-implements} $\vector{v}$.
\end{definition}

In fact, we show that the set of $\family{H}$-implementable vectors is the whole vector-space. Consequently, the quantum-graph that $\family{H}$-implements the $i$-th standard basis vector has to exist, which is going to yield Proposition~\ref{prop:reduction_quantum_homs}.

\begin{lemma}\label{lem:everything_implementable}
	Let $p$ be a prime, $k$ be a non-negative integer, and $\family{H}$ be a set $\set{H_1, \dots, H_r}$ of pairwise non-\isomorphic{} graphs with $k$ distinguished vertices. If every graph in $\family{H}$ is order-$p$ reduced, then every vector in $(\Z_p)^r$ is $\family{H}$-implementable.
\end{lemma}
\begin{proof}
	Let $\vector{v}$ be a given vector in $(\Z_p)^r$, and we denote $\vector{v} = \sqBrackets{v_i}_{i \in \family{F}}$, where we index over the graphs in $\family{H}$. This is equivalent to indices in $\sqBrackets{r}$  by the given enumeration of $\family{H}$. Let $H$ be a graph in $\family{H}$ such that $H$ contains the maximal number $c$ of vertices over graphs in $\family{H}$. Every graph in $\family{H}$ is equal to a graph with $k$ distinguished vertices obtained from a subgraph of the reflexive complete graph $K^\circ_c$. Let $\family{F}$ denote the set of pairwise non-isomorphic subgraphs of $K^\circ_c$ that are order-$p$ reduced and have $k$ distinguished vertices. In particular, every graph in $\family{H}$ is in $\family{F}$. We observe that $\family{F}$ is closed under surjective homomorphic image and the quadratic matrix $M_{\family{F}}$ with entries $\sqBrackets{\numHom[F , F'][p]}_{F,F' \in \family{F}}$ is non-singular in $\Z_p$ due to Corollary~\ref{cor:hom_matrix_nonsingular_mod_p}. Let $\family{F}$ be of cardinality $r'$, for which we know that $r \leq r'$. We define the vector $\vector{v}'$ in $(\Z_p)^{r'}$ with $\vector{v}' = \sqBrackets{v'_i}_{i \in \family{F}} $ by restricting $\vector{v}$ to entries in $\family{H}$, i.e.
	\[
	v'_i = \begin{cases}
		v_i &,\text{ if } i \in \family{H}, \\
		0 &,\text{ else } .
	\end{cases}
	\]
	The vector $\vector{v}'$ represents $\vector{v}$ in $(\Z_p)^{r'}$. In search of a quantum graph $\quantum{G}$ that $\family{H}$-implements $\vector{v}$, we observe the system of linear equations in $\Z_p$ given by
	\[
	\vector{v}' \equiv \vector{x} \cdot M_\family{H} \pmod p,
	\]
	which has a unique solution $\vector{x}^\ast$ with  $\vector{x}^\ast = \sqBrackets{x^\ast_i}_{i \in \family{F}}$ and entries in $\Z_p$ because $M_\family{F}$ is non-singular. We note that this yields, for a graph $J \in \family{F}$,
	\begin{equation}\label{eq:everything_implementable}
		v'_J \equiv \sum_{G \in \family{F} } x^\ast_G \cdot \numHom[G,J][p] \pmod p.
	\end{equation}
	By taking as constituents $\family{G}$ the set of graphs $\set{G \in \family{F} \given x^\ast_G \not\equiv 0 \pmod p}$ with set of coefficients $\set{x^\ast_G}_{G \in \family{G}}$ this gives a quantum graph $\quantum{G}$ implementing $\vector{v}'$, i.e. 
	\[
	\quantum{G} = \sum_{G \in \family{G}} x^\ast_G \cdot G .
	\]
	The equivalence~\eqref{eq:everything_implementable} remains true if we restrict the indices $i$ in $\family{F}$ to elements in $\family{H}$. 
	
	Finally, we confirm that $\abs{\family{G}}$ depends only on $c$ and $k$ because $\family{G}$ is a subset of $\family{F}$, which contains at most all possible subgraphs of $K^\circ_c$ with $k$ distinguished vertices. Therefore, $\quantum{G}$ $\family{H}$-implements $\vector{v}$.
\end{proof}

\begin{proof}[Proof of Proposition~\ref{prop:reduction_quantum_homs}]
	All graphs are assumed to have $k$ distinguished vertices.
	The assumptions for Lemma~\ref{lem:everything_implementable} are satisfied, which yields a set of quantum graphs $\set{\quantum{F}_1, \dots , \quantum{F}_r}$ such that $\quantum{F}_i$ $\family{H}$-implements the $i$-th standard basis vector. Let $H_i$ be a fixed constituent of $\quantum{H}$ and $\quantum{F}_i$ have constituents $\family{F}_i$ and set of coefficients $\set{\beta_{F}}_{F \in \family{F}_i}$. Utilizing $\quantum{F}_i$ and the dot product we construct $\quantum{F}_i \odot G$, which yields
	\begin{align*}
		\numHom[\quantum{F}_i \odot G, \quantum{H}][p] &\equiv \sum_{j \in \sqBrackets{r}} \alpha_{H_j} \cdot \numHom[\quantum{F}_i \odot G,H_j][p] \pmod p.
		\shortintertext{%
			By Corollary~\ref{cor:dot_product} we have%
		}
		&\equiv \sum_{j \in \sqBrackets{r}} \alpha_{H_j} \cdot \numHom[\quantum{F}_i,H_j][p] \cdot \numHom[G,H_j][p] \pmod p.
		\shortintertext{%
			Due to the construction of $\quantum{F}_i$, every term except for $j$ equal to $i$ vanishes, and we deduce%
		}
		\numHom[\quantum{F}_i \odot G,\quantum{H}][p] &\equiv \alpha_{H_i} \cdot \numHom[G,H_i][p] \pmod p.
	\end{align*}
	We can solve this linear congruence by $\abs{\family{F}_i}$ calls to the oracle for $\probNumHom{\quantum{H}}[p]$, one for every graph in $\family{F}_i$.	Finally, the cardinality of the set of constituents $\family{F}_i$ depends only on $k$ and $\abs{\vertexset[H_j]}$ for some constituent $H_j$ in $\family{H}$. 
\end{proof}

We note that the arguments are almost unaffected by counting in $\Z_p$. The crucial point for this was the assumption that the graphs under study are order~$p$ reduced. For the case of non-modular counting, the line of arguments yields the equivalent statement straightforwardly. This gives the slight generalization of \cite[Theorem~5]{Chen:19:The_Exponential-Time_Complexity} allowing for graphs with distinguished vertices. We state the result but omit the proof as it follows the same line of argumentation employed for Theorem~\ref{thm:homs_to_quantum_graph_mod_p} except for the assumption of order~$p$ reduced graphs.
\begin{theorem}\label{thm:quantum_homs_non-modular}
	Let $k$ be a non-negative integer and $\quantum{H}$ be a quantum graph with $k$ distinguished vertices that is given by $\quantum{H} = \sum_{H \in \family{H}} \alpha_H \cdot H$ and coefficients $\set{\alpha_H}_{H \in \family{H}}$ in $\Q_{\neq 0}$.
	\begin{itemize}
		\item If there exists a graph $H$ in $\family{H}$ such that the problem $\probNumHom{H}$ is $\classNumP$-hard, then $\probNumHom{\quantum{H}}$ is $\classNumP$-hard; 
		\item If for all graphs $H$ in $\family{H}$, the problem $\probNumHom{H}$ is solvable in polynomial time, then $\probNumHom{\quantum{H}}$ is also solvable in polynomial time.
	\end{itemize}
\end{theorem}

\subsection{Quantum Graphs and Partially Labelled Homomorphisms}
\label{subsec:quantum_graph_part_lab_homs}

The main reasons we allowed graphs with $k$ distinguished vertices in the previous subsection is that, first, the proofs are almost unaffected by it and, second, it allows us to show the strength of the technique by reproving pinning. 

We study the following set of problems.

\prob%
{$\probNumPartLabHom{H}$.}
{Positive integer $k$ and graph $H$.}
{Graph $G$ and partial mapping $\tau \colon \vertexset[G]^k \to \vertexset[H]^k$.}
{$\numHom[(G,\tau),H]$.}

\prob%
{$\probNumPartLabHom{H}[p]$.}
{Prime $p$, positive integer $k$, and graph $H$.}
{Graph $G$ and partial mapping $\tau \colon \vertexset[G]^k \to \vertexset[H]^k$.}
{$\numHom[(G,\tau),H][p]$.}

We note that the following result establishes \enquote{pinning} for the problem $\probNumHom{H}[p]$ and has been proved in~\cite{Goebel:21:Counting_Homomorphisms_Trees} in a different manner.

\begin{lemma}\label{lem:homs_and_partlab_homs_equivalent}
	Let $p$ be a prime and $H$ be a graph. If $H$ is order~$p$ reduced, then $\probNumHom{H}[p]$ and $\probNumPartLabHom{H}[p]$ are interreducible under polynomial-time Turing reductions.
\end{lemma}
\begin{proof}
	Let $k$ be the parameter for $\probNumPartLabHom{H}[p]$ that denotes the size of the domain for partial mappings under study. We recall Lemma~\ref{lem:pinning_split}. For every graph $G$ and every tuple $\vector{x}$ in $\vertexset[H]^k$, this yields
	\begin{align}
		\nonumber
		\numHom[G,H] &= \sum_{i \in \sqBrackets{\nu}} \abs[\big]{\Orb[\vector{v}_i]} \cdot \numHom[(G, \vector{x}), (H, \vector{v}_i)],
	\intertext{%
	where $\vector{v}_1, \dots, \vector{v}_\nu$ are representatives of the orbits in $\vertexset[H]^k$. We deduce%
	}
		\label{eq:pinning_split}
		\numHom[G,H][p] &= \sum_{i \in \sqBrackets{\nu}} \abs[\big]{\Orb[\vector{v}_i]} \cdot \numHom[(G, \vector{x}), (H, \vector{v}_i)][p] \pmod p.
	\end{align}
	By assumption, $H$ is order~$p$ reduced and by Lemma~\ref{lem:tuple-orbit}, for all $i \in \sqBrackets{\nu}$, the cardinality $\abs{\Orb[\vector{v}_i]}$ is larger than $0$.
	It suffices to establish a reduction from $\probNumPartLabHom{H}[p]$ to $\probNumHom{H}[p]$.
	
	Let $(G, \tau)$ be the input of an instance of $\probNumPartLabHom{H}[p]$, let $\vector{x}$ be the domain of $\tau$, and $\vector{v}$ be the range of $\tau$. Since $\vector{v}$ is in $\vertexset[H]^k$, there exists a representative $\vector{v}_i$ such that $\vector{v}$ is in $\Orb[\vector{v}_i]$. We are equivalently tasked with computing $\numHom[(G, \vector{x}), (H, \vector{v}_i)]$.
	
	Let $\quantum{H}$ be the quantum graph given by~$\eqref{eq:pinning_split}$, that is $\quantum{H} = \sum_{i \in \sqBrackets{\nu}} \abs{\Orb[\vector{v}_i]} \cdot (H, \vector{v}_i)$. We know by Proposition~\ref{prop:reduction_quantum_homs} that $\probNumHom{(H, \vector{v}_i)}[p]$ reduces to $\probNumHom{\quantum{H}}[p]$. It remains to translate the oracle calls for $\probNumHom{\quantum{H}}[p]$ into polynomial-many oracle calls for $\probNumHom{H}[p]$, where we emphasize that the constituents of $\quantum{H}$ have distinguished vertices.
	
	By assumption, $H$ is order~$p$ reduced. It follows from Lemma~\ref{lem:tuple-orbit} that $\quantum{H}$ is order~$p$ reduced as well. Let $\family{H}$ be the set $\set{(H, \vector{v}_i)}_{i \in \sqBrackets{\nu}}$ of constituents of $\quantum{H}$ and let $\quantum{F}$ be the order~$p$ reduced quantum graph with $k$ distinguished vertices given by Lemma~\ref{lem:everything_implementable} that $\family{H}$-implements the $i$-th standard basis vector. By definition, the number of constituents of $\quantum{F}$ depends only on $\abs{\vertexset[H]}$ and $k$. Let this number be $\mu$ and let $\quantum{F}= \sum_{j \in \sqBrackets{\mu}} \beta_j \cdot (F_j, \vector{y}_j)$. The quantum graph $\quantum{F}$ is the same we used in the proof of Proposition~\ref{prop:reduction_quantum_homs} for the reduction from $\probNumHom{\quantum{H}}[p]$ to $\probNumHom{(H, \vector{v}_i)}[p]$. On one hand, we have
	\begin{align*}
		\numHom[(G, \vector{x}) \odot \quantum{F}, \quantum{H}][p] &= \abs[\big]{\Orb[\vector{v}_i]} \cdot \numHom[(G, \vector{x}), (H, \vector{v}_i)][p].
	\intertext{%
	On the other, for each $j \in \sqBrackets{\mu}$, let $(G_j, \vector{x}_j)$ be the graph $(G, \vector{x}) \odot (F_j, \vector{y}_j)$. We apply \eqref{eq:pinning_split} on each $G_j$ and deduce
	}
		\numHom[(G, \vector{x}) \odot \quantum{F}, \quantum{H}][p] &= \sum_{i \in \sqBrackets{\nu}} \sum_{j \in \sqBrackets{\mu}} \abs[\big]{\Orb[\vector{v}_i]} \cdot \beta_j \cdot \numHom[(G, \vector{x}) \odot (F_j, \vector{y}_j), (H, \vector{v}_i)][p] \\
		&= \sum_{j \in \sqBrackets{\mu}} \beta_j \cdot \numHom[G_j, H][p].
	\end{align*}
	We can compute $\abs{\Orb[\bar{v}_i]}$ in constant time because $H$ is fixed. Since $\abs{\Orb[\vector{v}_i]}$ is larger than $0$ we can compute $\numHom[(G, \vector{x}), (H, \vector{v}_i)][p]$ using $\mu$ calls to an oracle for $\probNumHom{H}[p]$.
\end{proof}

\subsection{Investigation of Restricted Homomorphisms}
\label{subsec:restricted_homs}

In this subsection, we study a restricted class of homomorphisms that will be of high importance for restricting our study of $\probNumHom{H}[p]$ to bipartite graphs. A \emph{bip-graph} $G$ is given by a bipartite graph with a fixed bipartition, which we denote by the triple $\bipGraph$. Here, $\lpart[G]$ denotes the \enquote{left part} of $G$ and $\rpart[G]$ denotes the \enquote{right part} of $G$. A \emph{sub-bip-graph} of a bip-graph $G$ is a subgraph of $G$ that inherits the bipartition of $G$. Similarly, since a bip-graph is also a graph, the usual graph-theoretic notions apply to bip-graphs. We study the following set of homomorphisms. 
\begin{definition}\label{def:part_wise_homs}
	Given a pair of bip-graphs $G$ and $H$, we call a homomorphism $f$ in $\Hom[G ,H]$ a \emph{bip-homomorphism} if $f$ maps vertices in $\lpart[G]$ to vertices in $\lpart[H]$ and vertices in $\rpart[G]$ to vertices in $\rpart[H]$. The set $\HomBip[G, H]$ consists of all bip-homomorphisms in $\Hom[G ,H]$.
\end{definition}

The respective notion of isomorphism is called \emph{bip-isomorphism}, and we call two bip-graphs $H,H'$ \emph{\isomorphic[bip]} if there exists an isomorphism $\phi \colon H \to H'$ that is in $\HomBip[G, H]$. This is denoted by $H \congbip H'$. For instance, the simple path on $5$ vertices $P_5$ with 3 vertices in the left part is not {\isomorphic[bip]} to $P_5$ with 3 vertices in the right part.

Bip-homomorphisms allow naturally for graphs with $k$ distinguished vertices. Similar to our previous study of unrestricted homomorphisms, we tacitly omit the distinguished vertices when the arguments are unaffected by the existence of distinguished vertices.

Following the used notation for unrestricted homomorphisms, we denote the cardinality of $\HomBip[G, H]$ by $\numHomBip[G, H]$. Additionally, we denote by $\numHomBip[G, H][p]$ the value $\numHomBip[G, H] \pmod p$. The computational problems $\probNumHom{\quantum{H}}$ and $\probNumPartLabHom{H}$ extend naturally to the set of bip-homomorphisms. We study the following set of problems.
\prob%
{$\probNumBipHom{H}$.}
{Bip-graph $H$.}
{Bip-graph $G$.}
{$\numHomBip[G, H]$.}

\prob%
{$\probNumBipHom{H}[p]$.}
{Prime $p$ and bip-graph $H$.}
{Bip-graph $G$.}
{$\numHomBip[G, H][p]$.}

Our goal is to obtain similar technical results for $\probNumBipHom{H}[p]$ as we obtained for $\probNumHom{H}[p]$. In particular, we aim to establish both a notion equivalent to \enquote{order~$p$ reduced graph $H$} and pinning for the problem $\probNumBipHom{H}[p]$. By the adaptability of Corollary~\ref{cor:hom_matrix_nonsingular_mod_p} to restricted subclasses of homomorphisms, we solve both tasks in the same manner as for $\probNumHom{H}[p]$. The main task is thus to adjust the notions and concepts used in the previous part of this section to bip-graphs. For the sake of completeness, we provide all details.

\paragraph{Confluent Reduction} 
Regarding the first task, we apply almost the same argumentation as Faben and Jerrum~\cite{Faben:15:Parity_Graph_Homs} for their \enquote{confluent reduction} given by the order~$p$ reduction for $\probNumHom{H}[p]$. Let $H$ be a bip-graph. We denote by $\AutBip[H]$ the subgroup of automorphisms in $\Aut[H]$ that are also in $\HomBip[H, H]$, i.e. $\AutBip[H]$ is the intersection $\Aut[H] \cap \HomBip[H, H]$. For any prime $p$, any bip-graph $G$, and any automorphism $\varrho$ in $\AutBip[H]$ of order~$p$, the value $\numHomBip[G, H][p]$ is equal to the  value $\numHomBip[G, H^\varrho][p]$. Here, $H^\varrho$ has to be a bip-graph and so we let $H^\varrho$ denote the sub-bip-graph of $H$ induced by the vertices fixed under $\varrho$. In this way, for two bip-graphs $H$ and $H'$, we adjust the binary relation of Faben and Jerrum and denote $H \relArrow[\bip][p] H'$ if there exists an automorphism $\varrho \in \AutBip[H]$ of order $p$ such that $H'$ is equal to $H^\varrho$. Following a fixed enumeration of automorphisms $\varrho$ in $\AutBip[H]$ and since any graph $H$, parameter for $\probNumBipHom{H}[p]$, is of finite size, every iterative application of the reduction $H \relArrow[\bip][p] H^\varphi$ has to terminate.

We recall that $\normreduced{H}$ denotes the order~$p$ reduced form of $H$. By slightly modifying this notation, we denote by $\bipreduced{H}$ a terminal graph of the relation {\!$\relArrow[\bip][p]$\!}. It remains to show that $\bipreduced{H}$ is unique up to $\bip$-isomorphisms. For this purpose, we adapt Corollary~\ref{cor:hom_matrix_nonsingular_mod_p} to bip-homomorphisms.

For a bip-graph $G$, the subgroup $\AutBip[G]$ induces a respective notion of orbits. Furthermore, for a positive integer $k$ and a tuple $\vector{v}\in \vertexset[G]^k$, the set $\OrbBip[\vector{v}]$ denotes the orbit of $\vector{v}$ under automorphisms in $\AutBip[G]$. Following our general notation, for two bip-graphs $G$ and $H$, we denote by $\InjBip[G, H]$ and $\SurjBip[G, H]$ the restriction of homomorphisms in $\HomBip[G, H]$ to be injective and surjective on edges and vertices, respectively. We denote by $\numAutBip[G]$, $\numSurjBip[G, H]$, and $\numInjBip[G, H]$ the cardinality of the set $\AutBip[G]$, $\SurjBip[G, H]$, and $\InjBip[G, H]$, respectively. Sticking to the same order and taking the cardinalities modulo a prime $p$ yields the notations $\numAutBip[G][p]$, $\numSurjBip[G, H][p]$, and $\numInjBip[G, H][p]$. We call $G$ \emph{order~$p$ bip-reduced} if there exists no automorphism of order~$p$ in $\AutBip[G]$. Analogue to unrestricted homomorphisms, the same notation is also used for graphs with $k$ distinguished vertices.

%
The notion that a family of graphs is closed under surjective homomorphic image translates naturally to the notation that a family of bip-graphs is closed under surjective bip-homomorphic image. We now state the analogue to Corollary~\ref{cor:hom_matrix_nonsingular_mod_p}. For the sake of completeness, we provide a proof, although the arguments do not differ significantly from the arguments for Corollary~\ref{cor:hom_matrix_nonsingular_mod_p}. 
\begin{lemma}\label{lem:hom_matrix_nonsingular_bip2}
	Let $k$ be a non-negative integer and $\family{F}$ be a family $\set{F_i}_{i \in I}$ with index set $I$, where $\family{F}$ consists of pairwise non-\isomorphic[bip] bip-graphs with $k$ distinguished vertices and without multi-edges. If $\family{F}$ is closed under surjective bip-homomorphic image, then the matrix
	\[
	M_{\numHomBip} = \sqBrackets[\big]{\numHomBip[F_i,F_j]}_{i,j \in I}
	\]
	is nonsingular.
\end{lemma}
\begin{proof}
	All bip-graphs are assumed to have $k$ distinguished vertices. Analogue to the case of unrestricted homomorphisms, we observe that every homomorphism $f$ in $\HomBip[F_i,F_j]$ factors into a surjective homomorphism $g$ in $\SurjBip[F_i, J]$, where $J$ is the sub-bip-graph of $F_j$ given by the image of $f$ and $g$ is given by $f$, followed by an injective homomorphism $h$ in $\InjBip[J, F_j]$ given by the embedding of $J$ into $F_j$. By assumption, $J$ is an element of $\family{F}$. In order to avoid double counting, we divide by the number of elements in $\AutBip[J]$ and obtain
	\[
		\numHomBip[F_i, F_j] = \sum_{J \in \family{F}} \frac{\numSurjBip[F_i, J] \cdot \numInjBip[J, F_j]} {\numAutBip[J]}.
	\]
	The remainder of the proof follows one-to-one the proof of Lemma~\ref{lem:hom_matrix_nonsingular_loops}.
\end{proof}

\begin{corollary}\label{cor:hom_matrix_nonsingular_bip}
	Let $p$ be a prime, $k$ be a non-negative integer, and $\family{F}$ be a family $\set{F_i}_{i \in I}$ with index set $I$, where $\family{F}$ consists of pairwise non-\isomorphic[bip] order~$p$ bip-reduced bip-graphs with $k$ distinguished vertices and without multi-edges. If $\family{F}$ is closed under surjective bip-homomorphic image, then the matrix
	\[
		M_{\numHomBip[][p]} = \sqBrackets[\big]{\numHomBip[F_i,F_j][p] }_{i,j \in I}
	\]
	is nonsingular.
\end{corollary}
\begin{proof}
	All bip-graphs are assumed to have $k$ distinguished vertices. The proof follows the argumentation for Corollary~\ref{cor:hom_matrix_nonsingular_mod_p}. By assumption, for every bip-graph $F$ in $\family{F}$, the value $\numAutBip[F][p]$ is unequal to $0$. Hence, the matrices $M_{\numInjBip[][p]}$ and $M_{\numSurjBip[][p]}$ given by
	\[
		M_{\numInjBip[][p]} = \sqBrackets[\big]{\numInjBip[F_i,F_j][p] }_{i,j \in I}
		\quad \text{and} \quad
		M_{\numSurjBip[][p]} = \sqBrackets[\big]{\numSurjBip[F_i,F_j][p] }_{i,j \in I}
	\]
	are have non-zero entries on their diagonal, which yields that they are triangular. The $\abs{I} \times \abs{I}$ diagonal matrix $D_{\numAutBip[][p]}$ with diagonal entries $\numAutBip[F][p]$, for $F$ $\family{F}$ is invertible as are $M_{\numInjBip[][p]}$ and $M_{\numSurjBip[][p]}$. We derive
	\[
		M_{\numHomBip[][p]} = M_{\numSurjBip[][p]} \cdot \parenthesis[\big]{ D_{\numAutBip[][p]} }^{-1} \cdot M_{\numInjBip[][p]}.
	\]
	Since $Z_p$ is a field, we obtain that $M_{\numHomBip[][p]}$ is non-singular.
\end{proof}

An immediate consequence is the desired uniqueness of the order~$p$ bip-reduced form up to bip-isomorphisms. Let $p$ be a prime. We denote by $\redbipartites$ the class of all bip-graphs that are order~$p$ bip-reduced.
\begin{corollary}\label{cor:p-wise_reduced_unique}
	Let $p$ be a prime and let the two bip-graphs $H$ and $H'$ be in $\redbipartites$, where we assume a fixed enumeration of $\redbipartites$. The two bip-graphs $H$ and $H'$ are \isomorphic[bip] if and only if the vector $\sqBrackets[]{\numHomBip[G, H][p]}_{G \in \redbipartites}$ is equal to the vector $\sqBrackets[]{\numHomBip[G, H'][p]}_{G \in \redbipartites}$.
\end{corollary}
\begin{proof}
	If $H$ and $H'$ are \isomorphic[bip], then the result holds. We assume toward contradiction that $H$ and $H'$ are not \isomorphic[bip] but $\sqBrackets[]{\numHomBip[G, H][p]}_{G \in \redbipartites}$ is equal to $\sqBrackets[]{\numHomBip[G, H'][p]}_{G \in \redbipartites}$. Let $\family{G}$ be the subclass of $\redbipartites$ containing all pairwise non-\isomorphic[bip] graphs. The fixed enumeration on $\redbipartites$ yields an index set $I$ for $\family{G}$. Let $\family{G} = \set{G_i}_{i \in I}$. We observe that $\family{G}$ satisfies the prerequisites of Corollary~\ref{cor:hom_matrix_nonsingular_bip}. Thus, the matrix $M_{\numHomBip}$ given by $M_{\numHomBip}=\sqBrackets[]{\numHomBip[G_i, G_j][p]}_{i,j \in I}$ is non-singular.
	The bip-graphs $H$ and $H'$ correspond to two distinct rows in $M_{\numHomBip}$. However, by assumption we have that the row-vector $\sqBrackets[]{\numHomBip[G_i, H][p]}_{i \in I}$ is equal to the row-vector $\sqBrackets[]{\numHomBip[G_i, H'][p]}_{i \in I}$. Therefore, the matrix is singular, a contradiction.
\end{proof}
In a nutshell, Corollary~\ref{cor:p-wise_reduced_unique} establishes that, for each order~$p$ bip-reduced bip-graph $G$, the \enquote{type of bip-isomorphism is determined by the vector $\sqBrackets[]{\numHomBip[F, G][p]}_{F \in \redbipartites}$}.

For any prime $p$, any bip-graph $H$, and any automorphism $\varrho \in \AutBip[H]$ of order~$p$, the reduction $H \relArrow[\bip][p] H^\varrho$ preserves the vector
$\sqBrackets[]{\numHomBip[F, H][p]}_{F \in \redbipartites}$, i.e. $\sqBrackets[]{\numHomBip[F, H][p]}_{F \in \redbipartites}$ is equal to $\sqBrackets[]{\numHomBip[F, H^\varphi][p]}_{F \in \redbipartites}$. Assuming a fixed enumeration of $\redbipartites$, any pair of bip-graphs $H'$ and $H''$ that are terminal graphs obtained from $H$ under iterative application of the reduction {$\relArrow[\bip][p]$\!} have an identical vector, that is $\sqBrackets[]{\numHomBip[F, H'][p]}_{F \in \redbipartites}$ is equal to $\sqBrackets[]{\numHomBip[F, H''][p]}_{F \in \redbipartites}$. By Corollary~\ref{cor:p-wise_reduced_unique}, the bip-graphs $H$ and $H'$ have to be \isomorphic[bip].
\begin{observation}\label{obs:p-wise_reduced_unique}
	Let $p$ be a prime and $H$ be a bip-graph. The order~$p$ $\bip$-reduced form $\bipreduced{H}$ is unique up to $\bip$-isomorphisms. Further, for every bip-graph $G$, the value $\numHomBip[G,H][p]$ is equal to the value $\numHomBip[G,\bipreduced{H}][p]$.
\end{observation}

Before the solution of the second task, we note that the uniqueness of $\bipreduced{H}$ will be crucial for iterative reduction techniques applied for the study of $\probNumBipHom{H}[p]$ in the later sections of this paper.

\paragraph{Pinning}
Since a bip-homomorphism is a restricted homomorphism between bipartite graphs, the notion of a partially labelled bip-graph is naturally derived from the corresponding notion for graphs. In particular, for a positive integer $k$, two bip-graphs $G$ and $H$, and a partial mapping $\tau: \vertexset[G]^k \to \vertexset[H]^k$, we denote by $\HomBip[(G, \tau), H]$ the intersection $\Hom[(G, \tau), H] \cap \HomBip[G, H]$. In the usual way, we arrive at the notions $\numHomBip[(G, \tau), H]$ and $\numHomBip[(G, \tau), H][p]$.

Establishing pinning for the problem $\probNumBipHom{H}[p]$
translates to showing a polynomial-time Turing interreducibility with the problem $\probNumPartLabBipHom{H}[p]$, where the latter is defined in the following. 
\prob%
{$\probNumPartLabBipHom{H}$.}
{Positive integer $k$ and bip-graph $H$.}
{Bip-graph $G$ and partial mapping $\tau \colon \vertexset[G]^k \to \vertexset[H]^k$.}
{$\numHomBip[(G, \tau), H]$.}

The associated modular problem is then the following.
\prob%
{$\probNumPartLabBipHom{H}[p]$.}
{Prime $p$, positive integer $k$ and bip-graph $H$.}
{Bip-graph $G$ and partial mapping $\tau \colon \vertexset[G]^k \to \vertexset[H]^k$.}
{$\numHomBip[(G, \tau), H][p]$.}

Analogue to the previous subsection, we employ quantum graphs for the reduction. A quantum graph $\quantum{G}$ is bipartite if each constituent $G$ of $\quantum{G}$ is bipartite. Hence, $\quantum{G}$ is a \emph{quantum bip-graph} if every constituent of $\quantum{G}$ is a bip-graph. Given two quantum bip-graphs $\quantum{G}$ and $\quantum{H}$ with $k$ distinguished vertices given by $\quantum{H}=\sum_{H \in \family{H}} \alpha_H \cdot H$ and $\quantum{G}= \sum_{G \in \family{G}} \beta_G \cdot G$, we denote by $\numHomBip[\quantum{G}, \quantum{H}]$ the value of the linear combination
$
	\sum_{G \in \family{G}} \sum_{H \in \family{H}} \beta_G \cdot \alpha_H \cdot \numHomBip[G, H]
$
and by $\numHomBip[\quantum{G}, \quantum{H}][p]$ the value $\numHomBip[\quantum{G}, \quantum{H}] \pmod p$. 

We follow a similar path as in Section~\ref{subsec:pinning_quantum_graphs}.
Since $\AutBip[G]$ is a subgroup of $\Aut[G]$, the following is an immediate consequence of Lemma~\ref{lem:tuple-orbit}.
\begin{corollary}\label{cor:tuple-orbit_bip}
	Let $p$ be a prime, $H$ be a bip-graph, $k$ be a positive integer, and $\varrho$ be an automorphism in $\AutBip[H]$. Further, let there exist a tuple $\vector{y}$ in $\vertexset[H]^k$ with cardinality $\abs{\Orb[\vector{y}][\varrho]}$ denoted by $r$. If $r$ is larger than $1$ and there exists $j$ in $\sqBrackets{r-1}$ such that $\varrho^{p \cdot j}(\vector{y})$ is equal to $\vector{y}$, then $\AutBip[H]$ contains an element of order $p$.
\end{corollary}

The class of bipartite graphs is closed under both the dot product and the tensor product. This is still true when restricting the graphs to come with a fixed bipartition as the product graphs inherit the bipartition of their factor graphs.
This gives the following consequence of Corollary~\ref{cor:dot_product}.
\begin{corollary}
	\label{cor:dot_product_bip}
	Let $k$ be a non-negative integer. Given three quantum bip-graphs with $k$ distinguished vertices $\quantum{F}$, $\quantum{G}$, and $\quantum{H}$, it follows
	\begin{align*}
		\numHomBip[\quantum{F} \odot \quantum{G}, \quantum{H}] &= \numHomBip[\quantum{F}, \quantum{H}] \cdot \numHomBip[\quantum{G}, \quantum{H}], \\
		\numHomBip[\quantum{G}, \quantum{F} \otimes \quantum{H}] &= \numHomBip[\quantum{G}, \quantum{F}] \cdot \numHomBip[\quantum{G}, \quantum{F}].
	\end{align*}
\end{corollary}

The notion analogue to an $\family{H}$-implementable vector is the following.
\begin{definition}\label{def:bip_labelled_implementable}
	Let $p$ be a prime, $k$ be a non-negative integer, and $\family{H}$ be a set $\set{H_1, \dots, H_r}$ of pairwise non-\isomorphic[bip] order~$p$ bip-reduced bip-graphs with $k$ distinguished vertices. Let $\vector{v}$ be a $1 \times r$  vector $\sqBrackets[]{v_1, \dots, v_r}$ with entries in $\Z_p$. We call $\vector{v}$ \emph{$\family{H}$-$\bip$-implementable} if there exists a quantum bip-graph $\quantum{F}$ with $k$ distinguished vertices and number of constituents depending only on $k$ and $\max_{H \in \family{H}} \abs{\vertexset[H]}$, such that for all, $i \in \sqBrackets{r}$,
	\[
	v_i = \numHomBip[\quantum{F}, H_i][p] . 
	\]
	In this case, we say that $\quantum{F}$ \emph{bip-implements $\vector{v}$}.
\end{definition}

We recall Lemma~\ref{lem:everything_implementable}, which states for unrestricted homomorphisms that every vector is $\family{H}$-implementable. 
The key result used in the proof of Lemma~\ref{lem:everything_implementable} is the existence of a solution to a system of linear equations provided by Corollary~\ref{cor:hom_matrix_nonsingular_mod_p}. With Corollary~\ref{cor:hom_matrix_nonsingular_bip} at hand, the same argumentation also works for bip-homomorphisms.
\begin{corollary}\label{cor:everything_part-wise_implementable}
	Let $p$ be a prime, $k$ be a non-negative integer, and $\family{H}$ be a set $\set{H_1, \dots, H_r}$ of pairwise non-\isomorphic[bip] bip-graphs with $k$ distinguished vertices. If every bip-graph in $\family{H}$ is order~$p$ bip-reduced, then every vector in $(\Z_p)^r$ is $\family{H}$-$\bip$-implementable.
\end{corollary}
\begin{proof}
	We provide a less-detailed proof because we follow the line of arguments utilized for Lemma~\ref{lem:everything_implementable}. All bip-graphs are assumed to have $k$ distinguished vertices. Let $c$ be the maximal number of vertices of a bip-graph in $\family{H}$. Every graph in $\family{H}$ is equal to a bip-graph with $k$ distinguished vertices obtained from a bipartite subgraph of the complete reflexive graph $K_c^\circ$ by fixing a bipartition and distinguishing $k$ vertices. Let $\family{F}$ be the set of pairwise non-\isomorphic[bip] bip-graphs obtained in this manner from $K_c^\circ$. We restrict $\family{F}$ to only the entries that are order~$p$ $\bip$-reduced. Then, $\family{F}$ satisfies the prerequisites of Corollary~\ref{cor:hom_matrix_nonsingular_bip}.
	
	Let $\family{F}$ be of cardinality $r'$, $\vector{v}'$ be an arbitrary vector in $(\Z_p)^{r'}$, and $M_{\numHomBip[][p]}$ be the matrix with entries $\sqBrackets[]{\numHomBip[F, F'][p]}_{F,F' \in \family{F}}$. By Corollary~\ref{cor:hom_matrix_nonsingular_bip}, the matrix $M_{\numHomBip[][p]}$ is non-singular. Thus, the system of linear equations in $\Z_p$ given by
	\[
		\vector{v}' \equiv \vector{x} \cdot M_{\numHomBip[][p]} \pmod p
	\]
	has a unique solution $\vector{x}^\ast$; a vector in $(\Z_p)^{r'}$ with indices represented by the graphs in $\family{F}$. The subvector $\vector{v}$ that is obtained from $\vector{v}'$ by restricting the indices to elements in $\family{H}$ yields a quantum graph $\quantum{G}$ that $\family{H}$-$\bip$-implements $\vector{v}$ in just the same way as in the proof of Lemma~\ref{lem:everything_implementable}. Every vector in $(\Z_p)^r$ can be extended to a vector in $(\Z_p)^{r'}$. Finally, $r'$ depends only on $k$ and $c$, which concludes the proof.
\end{proof}

Now, pinning is proven in the same way as we did for unrestricted homomorphisms in Section~\ref{subsec:quantum_graph_part_lab_homs}.
\begin{lemma}\label{lem:pinning_bip}
	Let $p$ be a prime and $H$ be a bip-graph. If $H$ is is order~$p$ $\bip$-reduced, then $\probNumPartLabBipHom{H}[p]$ and $\probNumBipHom{H}[p]$ are interreducible under polynomial-time Turing reductions.
\end{lemma}
\begin{proof}
	Let $k$ be the number of distinguished vertices, parameter for $\probNumPartLabBipHom{H}[p]$. We recall from the proof of Lemma~\ref{lem:homs_and_partlab_homs_equivalent} that, for every pair of graphs $\hat{G}$ and $\hat{H}$ and every tuple $\vector{x}$ in $\vertexset[\hat{G}]^k$,
	\[
	\numHom[\hat{G},\hat{H}] = \sum_{i \in \sqBrackets{\kappa}} \abs[\big]{\Orb[\vector{w}_i]} \cdot \numHom[(\hat{G}, \vector{x}), (\hat{H}, \vector{w}_i)],
	\]
	where $\vector{w}_1, \dots, \vector{w}_\kappa$ are representatives of the orbits in $\vertexset[\hat{H}]^k$. The same holds also when we restrict to bip-graphs as we briefly argue.
	
	The set of tuples $\vertexset[H]^k$ is partitioned by the equivalence relation \enquote{$\eqrel$} defined as $\vector{v} \eqrel \vector{v}'$ if and only if $(H, \vector{v})$ and $(H, \vector{v}')$ are \isomorphic[bip]. We note that the tuples equivalent to $\vector{v}$ are the elements in the orbit $\OrbBip[\vector{v}]$. Let $\vector{v}_1, \dots, \vector{v}_\nu$ be representatives of bip-orbits in $\vertexset[H]^k$ and $\mathbf{v}$ be an enumeration of $\vertexset[H]^k$, which yields 
	\begin{equation}\label{eq:decomp_of_tuples_bip}
		\mathbf{v}=\bigcup_{i\in\sqBrackets{\nu}} \eqclass{\vector{v}_i} \quad \text{with} \quad \abs[\big]{\eqclass{\vector{v}_i} }= \abs[\big]{\OrbBip[\vector{v}_i]}.
	\end{equation}
	Let $G$ be an arbitrary bip-graph and let $\vector{x}$ be an arbitrary tuple in $\vertexset[G]^k$. Any homomorphism $f$ in $\HomBip[G, H]$ has to map $\vector{x}$ to a $k$-tuple in $\vertexset[H]^k$, i.e. to an element in $\mathbf{v}$. Thus, we derive by the decomposition~\eqref{eq:decomp_of_tuples_bip} of $\mathbf{v}$
	\begin{align}
		\nonumber
		\numHomBip[G, H] &= \sum_{i \in \sqBrackets{\nu}} \abs[\big]{\OrbBip[\vector{v}_i]} \cdot  \numHomBip[(G,\vector{x}), (H, \vector{v}_i)] .
		\intertext{%
			By Corollary~\ref{cor:tuple-orbit_bip}, it follows that, for all $i\in\sqBrackets{\nu}$, the cardinality $\abs{\OrbBip[\vector{v}_i]}$ is not congruent modulo $p$ to $0$. By the decomposition%
		}
		\label{eq:part_lab_homs_are_quantum_graphs_bip}
		\numHomBip[G, H][p] &\equiv \sum_{i \in \sqBrackets{\nu}} \abs[\big]{\OrbBip[\vector{v}_i]} \cdot  \numHomBip[(G,\vector{x}), (H, \vector{v}_i)][p] \pmod p
	\end{align}
	it suffices to show that $\probNumPartLabBipHom{H}[p]$ reduces to $\probNumBipHom{H}[p]$.
	
	The remainder of the proof is almost identical to the proof of Lemma~\ref{lem:homs_and_partlab_homs_equivalent}. Let $(G, \tau)$ be the input of an instance of $\probNumPartLabBipHom{H}[p]$, where $G$ is a bip-graph and $\tau \colon \vertexset[G]^k \to \vertexset[H]^k$ is a partial mapping. Further, let $\vector{x}$ be the domain of $\tau$ and $\vector{v}$ be the range of $\tau$. There has to exist a representative $\vector{v}_i$ such that $\vector{v}$ is in the bip-orbit $\OrbBip[\vector{v}_i]$. Thus, we are equivalently tasked with computing $\numHomBip[(G, \vector{x}), (H, \vector{v}_i)][p]$.
	
	The equivalence~$\eqref{eq:part_lab_homs_are_quantum_graphs_bip}$ yields the quantum bip-graph $\quantum{H}$ given by $\quantum{H} = \sum_{i \in \sqBrackets{\nu}} \abs{\OrbBip[\vector{v}_i]} \cdot (H, \vector{v}_i)$. The consituents $\family{H}$ are the set $\set{(H, \vector{v}_i)}_{i \in \sqBrackets{\mu}}$ of order~$p$ bip-reduced bip-graphs with $k$ distinguished vertices. Let $\quantum{F}$ be the quantum bip-graph that $\family{H}$-bip-implements the $i$-th standard basis vector that is given by Corollary~\ref{cor:everything_part-wise_implementable}. By definition, the number of constituents in $\quantum{F}$ depends only on $k$ and $\max_{H \in \family{H}} \abs{\vertexset{H}}$. Let this number be $\mu$ and let $\quantum{F}$ be $\sum_{j \in \sqBrackets{\mu}} \beta_j \cdot (F_j, \vector{y}_j)$.
	
	On one hand, we have by Corollary~\ref{cor:dot_product_bip}
	\begin{align*}
		\numHomBip[(G, \vector{x}) \odot \quantum{F}, \quantum{H}][p] & = \abs[\big]{\OrbBip[\vector{v}_i]} \cdot \numHomBip[(G, \vector{x}), (H, \vector{v}_i)][p].
	\intertext{%
		On the other, for each $j \in \sqBrackets{\mu}$, let $(G_j, \vector{x}_j)$ be the graph $(G, \vector{x}) \odot (F_j, \vector{y}_j)$. We have by \eqref{eq:part_lab_homs_are_quantum_graphs_bip}%
	}
		\numHomBip[(G, \vector{x}) \odot \quantum{F}, \quantum{H}][p] &= \sum_{i \in \sqBrackets{\nu}} \sum_{j \in \sqBrackets{\mu}} \abs[\big]{\OrbBip[\vector{v}_i]} \cdot \beta_j \cdot \numHomBip[(G, \vector{x}) \odot (F_j, \vector{y}_j), (H, \vector{v}_i)][p] \\
		&= \sum_{j \in \sqBrackets{\mu}} \beta_j \cdot \numHomBip[G_j, H][p].
	\end{align*}
	Computing $\abs{\OrbBip[\bar{w}_i]}$ can be done in constant time because $H$ is fixed and $\abs{\OrbBip[\vector{v}_i]}$ is not congruent modulo $p$ to $0$. Therefore, we can compute $\numHomBip[(G, \vector{x}), (H, \vector{v}_i)][p]$ using $\mu$ calls of an oracle to $\probNumBipHom{H}[p]$.
\end{proof}

\paragraph{Restricted Quantum Homomorphisms}
We state the following immediate results as they might be of interest in future studies. For the purposes of this paper, the reader can safely skip it.

Of course, bip-homomorphisms allow also for an extension of $\probNumBipHom{H}$ to quantum graphs.
\prob%
{$\probNumBipHom{\quantum{H}}$.}
{Non-negative integer $k$ and quantum bip-graph $\quantum{H}$ with $k$ distinguished vertices and coefficients in $\C$.}
{Bip-graph $G$ with $k$ distinguished vertices.}
{$\numHomBip[G, \quantum{H}]$.}

Restricting the coefficients in a quantum graph $\quantum{H}$ to be integers we can also state the modular counting variant.
\prob%
{$\probNumBipHom{\quantum{H}}[p]$.}
{Prime $p$, non-negative integer $k$, and quantum bip-graph $\quantum{H}$ with integer coefficients and $k$ distinguished vertices.}
{Bip-graph $G$ with $k$ distinguished vertices.}
{$\numHomBip[G, \quantum{H}][p]$.}

A quantum bip-graph is order~$p$ $\bip$-reduced if its constituents are order~$p$ bip-reduced and have coefficients in $\Zsp$.
Following our insights we know that the techniques employed to obtain the reductions for quantum homomorphisms in Theorem~\ref{thm:homs_to_quantum_graph_mod_p} and Theorem~\ref{thm:quantum_homs_non-modular} carry over to quantum homomorphisms counted by $\probNumBipHom{\quantum{H}}[p]$ and $\probNumBipHom{\quantum{H}}$.

For the non-modular version we have.
\begin{theorem}
	Let $k$ be a non-negative integer and $\quantum{H}$ be a quantum bip-graph with $k$ distinguished vertices and coefficients in $\Q$.
	\begin{itemize}
		\item If there exists a constituent $H$ of $\quantum{H}$ such that the problem $\probNumBipHom{H}$ is $\classNumP$-hard, then $\probNumBipHom{\quantum{H}}$ is $\classNumP$-hard; 
		\item If for every constituent $H$ of $\quantum{H}$, the problem $\probNumBipHom{H}$ is solvable in polynomial time, then $\probNumBipHom{\quantum{H}}$ is also solvable in polynomial time.
	\end{itemize}
\end{theorem}

Finally, for the modular version we have.
\begin{theorem}
	Let $p$ be a prime, $k$ be a non-negative integer, and $\quantum{H}$ be an order~$p$ bip-reduced quantum graph with $k$ distinguished vertices.
	\begin{itemize}
		\item If there exists a constituent $H$ of $\quantum{H}$ such that the problem $\probNumBipHom{H}[p]$ is $\classNumP[p]$-hard, then $\probNumBipHom{\quantum{H}}[p]$ is $\classNumP[p]$-hard; 
		\item If for every constituent $H$ of $\quantum{H}$, the problem $\probNumBipHom{H}[p]$ is solvable in polynomial time, then $\probNumBipHom{\quantum{H}}[p]$ is also solvable in polynomial time.
	\end{itemize}
\end{theorem} 

\section{Graph Bipartization and Bipartite Homomorphisms}
\label{sec:bipartization}
Let $H$ be an arbitrary graph, where we explicitly highlight the possibility of loops. Instead of studying the problem $\probNumHom{H}[p]$ directly on $H$, we  study the problem $\probNumBipHom{H'}[p]$, where $H'$ is the bip-graph obtained from $H\otimes K_2$ with an intrinsic bipartition obtained from a fixed bipartition of $K_2$. In particular, if $M_H$ is the $n\times n$ adjacency matrix of $H$, then the adjacency matrix~$M_{H'}$ of $H \times K_2$ is a $2n\times 2n$ matrix of the following form:
\[
	\begin{BMAT}{cc}{cc}
		&
		\begin{BMAT}(b){cc}{c}
 			\lpart[H'] & \rpart[H']
		\end{BMAT}
		\\
		\begin{BMAT}(b){c}{cc}
			\lpart[H'] \\ \rpart[H']
		\end{BMAT}
		&
		\left(\begin{BMAT}(b){cc}{cc}
 			\mathbf{0} & M_H\\  M_H & \mathbf{0}
		\end{BMAT}\right) ,
	\end{BMAT}
\]
where $\mathbf{0}$ denotes the $n\times n$ matrix containing only entries equal to $0$. From the structure of $M_{H'}$, we deduce that $H'$ is a bip-graph, where $\lpart[H']$ and $\rpart[H']$ indicate the bipartition. This bipartition is given by the projection $q: H \otimes K_2 \to K_2$, given by $q(v,w) = w$, and a fixed bipartition of $K_2$. In particular, the projection $q$ maps $\lpart[H']$ to $\lpart[K_2]$ and $\rpart[H']$ to $\rpart[K_2]$. Further, the graph $H'$ is a collection of complete bipartite graphs if and only if every connected component of $H$ is complete bipartite or reflexive complete.  
We always assume a fixed bipartition of $K_2$ and that the fixed bipartition of $H'$ is $(\lpart[H'], \rpart[H'])$ as given above. 

In order to utilize the pinning results in Section~\ref{subsec:restricted_homs}, we need the parameter bip-graph $H$ for $\probNumBipHom{H}[p]$ to be order~$p$ bip-reduced. This is a weaker restriction than the one for the study of $\probNumHom{H}[p]$, where the graph $H$ is assumed to be order~$p$ reduced.

\begin{lemma}\label{lem:automorphism_bipartisation}
Let $p$ be a prime, let $H$ be an arbitrary graph, and let $H'$ be the bip-graph obtained from $H\otimes K_2$ by fixing a bipartition of $K_2$. If $H$ is order~$p$ reduced, then $H'$ is order~$p$ bip-reduced.
\end{lemma}
\begin{proof}
	We assume toward contradiction that there is an automorphism $\varrho$ in $\AutBip[H']$ of order~$p$. Since $\varrho$ preserves the order of the bipartition, the permutation matrix $P_\varrho$ of $\varrho$ has the following form
	\[
		\begin{BMAT}{cc}{cc}
			&
			\begin{BMAT}(b){cc}{c}
				\lpart[H'] & \rpart[H']
			\end{BMAT}
			\\
			\begin{BMAT}(b){c}{cc}
				\lpart[H'] \\ \rpart[H'] 
			\end{BMAT}
			&
			\left(\begin{BMAT}(b){cc}{cc}
				~P~ & ~\mathbf{0} ~\\  ~\mathbf{0}~ & ~Q ~
			\end{BMAT}\right) .
		\end{BMAT}
	\]
	Let $M_{H'}$ be the adjacency matrix of $H'$. We recall the form of $M_{H'}$ given as a $2\times2$ block matrix consisting of $\mathbf{0}$ and $M_H$, where $M_H$ is the adjacency matrix of $H$. Since $\varrho$ is an automorphism of order~$p$, we have that $(P_\varrho)^p \cdot M_{H'}$ is equal to $M_{H'}$, that is
	\[
		\left(\begin{array}{cc}
			P & \mathbf{0}\\
			\mathbf{0} & Q
		\end{array}\right)^p \cdot
		\left(\begin{array}{cc}
			\mathbf{0} & M_H\\
			M_H & \mathbf{0}
		\end{array}\right)
		=
		\left(\begin{array}{cc}
			\mathbf{0} & P^p \cdot M_H\\
			Q^p \cdot M_H & \mathbf{0}
		\end{array}\right).
	\]
	The latter implies that $P^p \cdot M_H$ is equal to $M_H$ because $\varrho$ is of order~$p$. Since $P$ is also a permutation matrix, this contradicts the assumption that $H$ has no automorphism of order~$p$.
\end{proof}

We note that by Lemma~\ref{lem:automorphism_bipartisation} and the reduction of Faben and Jerrum~\cite{Faben:15:Parity_Graph_Homs} the initial parameter graph for $\probNumBipHom{H}[p]$ can be assumed to be order~$p$ bip-reduced even without Observation~\ref{obs:p-wise_reduced_unique}.
Observation~\ref{obs:p-wise_reduced_unique} is specifically needed later for an iterative reduction argument on $\probNumBipHom{H}[p]$, where we might lose the property of having an order~$p$ bip-reduced parameter graph already after the first reduction.

Lemma~\ref{lem:automorphism_bipartisation} is of particular interest when $p$ is $2$. In this case, the bip-graph $H'$ given by $H \otimes K_2$ has an involution exchanging the partitions of the graph. However, the lemma shows that there is no involution in $\AutBip[H']$. This enables us to analyse the complexity of $\probNumHom{H}[2]$ by analysing the complexity of $\probNumBipHom{H'}[2]$. In contrast, the problem $\probNumHom{H'}[2]$ yields as output a $0$ for every input. This would not have been an issue for any odd prime $p$, but we present a unified framework in our technique rather than distinguishing the two cases. To this end, we show the following, where we recall that a bip-graph is also a graph by essentially \enquote{forgetting} the fixed bipartition.

\begin{lemma}\label{lem:bipartite_homs_partitioning}
	Let $H$ be a graph and let $H'$ be the bip-graph obtained from $H\otimes K_2$ by fixing a bipartition of $K_2$. If $G$ is a connected bip-graph, then $\numHom[G, H]$ is equal to $\numHomBip[G, H']$.
\end{lemma}
\begin{proof}

Let $f'$ be a bip-homomorphism in $\numHomBip[G, H']$. For every vertex $v$ of $G$, the image $f'(v)$ is given by the image $f(v)$ of $v$ in $H$ and by the image $g(v)$ of $v$ in $K_2$. Hence, the image $f'(v)$ is equal to the pair of images $(f(v), g(v))$ and, for the functions $f \colon \vertexset[G] \to \vertexset[H]$ and $g \colon \vertexset[G] \to \vertexset[K_2]$ given in this way, $f$ is equal to $f \times g$. Since $f'$ is a homomorphism, it follows that both $f$ and $g$ are also homomorphisms. Therefore, $f$ is in $\numHomBip[G, H]$. Further, $g$ is the unique bip-homomorphism in $\numHomBip[G, K_2]$ and the homomorphism $f$ is uniquely given by $f'$. Conversely, every homomorphism $f$ in $\numHomBip[G, H]$ yields a unique bip-homomorphism $f'$ in $\numHomBip[G, H']$ by $f' = f \times g$.
The lemma follows.
\end{proof}

We recall equation \eqref{eq:homs_tensor_product_target} from Section~\ref{sec:quantum_graphs} regarding the tensor product and Corollary~\ref{cor:dot_product_bip} with its equation regarding the dot product. Additionally, we recall that the dot product of two graphs without distinguished vertices is the disjoint union of these two graphs.

We now prove the following main result, where we recall that by Faben~\cite[Theorem~3.1.17.]{Faben:12:thesis:Complexity_Modular_Counting_CSP}, for any positive integer $k$, any parsimonious reduction is parsimonious modulo $k$.
\Bipartization*
\begin{proof}
	Let $G$ be a bip-graph and input for $\probNumBipHom{H'}[p]$. Further, let $G$ contain $r$ connected components $G_1,\dots,G_r$, all of these are bip-graphs. By Corollary~\ref{cor:dot_product_bip}, we obtain
	\begin{align*}
		\numHomBip[G, H'] &= \prod_{i\in\sqBrackets{r}}\numHomBip[G_i, H'] .
		\intertext{Since every $G_i$ is connected, Lemma~\ref{lem:bipartite_homs_partitioning} yields}
		&=\prod_{i\in\sqBrackets{r}}\numHom[G_i, H].
	\end{align*}
	The latter is equal to $\numHom[G, H]$ due to \eqref{eq:homs_disjoint_union_origin}. The theorem follows.
\end{proof}

Faben and Jerrum~\cite[Theorem~8.6.]{Faben:15:Parity_Graph_Homs} show that, in order to obtain hardness for $\probNumHom{H}[2]$, it suffices to show that the problem is hard for any connected component of $H$. We note that their argumentation also applies to our setting.

\begin{lemma}\label{lem:bip_components}
	Let $p$ be a prime and let $H$ be an order~$p$ bip-reduced bip-graph. If the bip-graph $H_1$ is a connected component of $H$, then $\probNumBipHom{H_1}[p]$ reduces to $\probNumBipHom{H}[p]$ under polynomial-time Turing reduction.
\end{lemma}
\begin{proof}
	We show that $\probNumBipHom{H_1}[p]$ reduces to $\probNumPartLabBipHom{H}[p]$ via polynomial-time Turing reduction, then the proof follows from Lemma~\ref{lem:pinning_bip}. Let the bip-graph $G$ be the input for $\probNumBipHom{H_1}[p]$ and let $G$ consist of $r$ connected components $G_1,\dots,G_r$. We note that for each $i\in\sqBrackets{r}$, the connected component $G_i$ is an order~$p$ bip-reduced bip-graph. For each connected component $G_i$, we choose a vertex $x_i$ of $G_i$ in the left part $\lpart[G]$, and derive
	\begin{align*}
		\numHomBip[G, H_1]&= \prod_{i\in\sqBrackets{r}}\numHomBip[G_i, H_1] \\
		 &=\prod_{i\in\sqBrackets{r}}\sum_{v\in \lpart[H_1]}\numHomBip[(G_i,x_i), (H_1,v)],
	\end{align*}
	where the first equality follows from Corollary~\ref{cor:dot_product_bip} and the second equality comes from the fact that the vertex $x_i$ must be mapped to some vertex in $\lpart[H_1]$. Since $G_i$ and $H_1$ are connected, we observe that for each vertex $v$ in the left part $\lpart[H_1]$, the set of bip-homomorphisms $\HomBip[(G_i,x_i), (H,v)]$ is equal to $\HomBip[(G_i,x_i), (H_1,v)]$. By the second equality, this implies that we can compute $\numHomBip[G_i, H_1]$ with at most $\abs{\lpart[H_1]}$ oracle calls to $\probNumPartLabBipHom{H}[p]$. There are at most $\abs{\vertexset[G]}$ connected components in $G$. Therefore, we can compute $\numHomBip[G, H_1]$ with at most $\abs{\vertexset[G]} \cdot \abs{\lpart[H_1]}$ oracle calls to $\probNumPartLabBipHom{H}[p]$.
\end{proof}

We collect the gained insights.
For studying the complexity of $\probNumHom{H}[p]$ for a graph $H$, we bipartize $H$ by constructing the bip-graph $H'$ from $H \otimes K_2$ and a fixed bipartition of $K_2$. By Theorem~\ref{thm:bipartization}, it suffices to study $\probNumBipHom{H'}[p]$. We restrict our attention to the order~$p$ bip-reduced form $\bipreduced{(H')}$ due to Observation~\ref{obs:p-wise_reduced_unique}. We have by Lemma~\ref{lem:automorphism_bipartisation} that $\bipreduced{(H')}$ is a collection of complete bipartite graphs if and only if $\normreduced{H}$ is a collection of complete bipartite graphs and reflexive complete graphs. For this case, Corollary~\ref{cor:polyt-graphs} shows tractability. Towards proving intractability in the remaining cases, Lemma~\ref{lem:pinning_bip} states that  $\probNumBipHom{H'}[p]$ and $\probNumPartLabBipHom{H'}[p]$ are interreducible. This chain of arguments is displayed in Figure~\ref{fig:reduction_chain_BIP}.
Finally, lemma~\ref{lem:bip_components} suggests that we only have to show hardness for any connected component of $H'$.

\section{Gadgets and Hardness}
\label{sec:gadgets}
Toward a dichotomy for $\probNumHom{H}[p]$, we established in the previous section that it suffices to show that, for any connected order~$p$ bip-reduced bip-graph $H$, the problem $\probNumBipHom{H}[p]$ is $\classNumP[p]$-hard. Since $H$ is connected, by fixing for a single vertex $v$ of $H$ that $v$ belongs to one part we fix the bipartition of $H$. Therefore, when we give a partially $H$-labelled bipartite graph $J$ with $J=(G, \tau)$ such that $G$ is connected and $\tau$ is not empty, we obtain a fixed bipartition of $G$ due to $\tau$ and the fixed bipartition of $H$, i.e. $J$ is a partially $H$-labelled bip-graph. This observation allows us to avoid stating the bipartition explicitly.

In the spirit of gadgetry used for the study of $\probNumHom{H}[2]$ (see e.g.~\cite{Goebel:16:Square-Free}), we define the following gadget as a structural criterion whose existence yields hardness. 

\begin{definition}\label{def:hardness_gadget}
	Let $p$ be a prime and $H$ be a bip-graph. Let there exist a triple of partially $H$-labelled bip-graphs with distinguished vertices $\GadgetPart{\L}$, $\GadgetPart{\R}$, and $\GadgetEdge$. If, for the sets
	\begin{alignat*}{4}
		&\selectSet[\L] &&= \{v \in \vertexset[H] &&\SetSymbol \numHomBip[\GadgetPart{\L}, (H, v)][p] &&\neq 0 \} , \\
		&\selectSet[\R] &&= \{v \in \vertexset[H] &&\SetSymbol \numHomBip[\GadgetPart{\R}, (H, v)][p] &&\neq 0 \} , \\
		&\selectSet[E] &&= \{(u, v) \in \vertexset[H]^2 &&\SetSymbol \numHomBip[\GadgetEdge, (H, u, v)][p] &&\neq 0 \},	
	\end{alignat*}
	there exists a partitioning $\selectSet[\L] = i_\L \cupdot o_\L$ and $\selectSet[\R] = i_\R \cupdot o_\R$ such that 
\begin{enumerate}
	\item $\abs{i_\L}, \abs{i_\R}, \abs{o_\L}, \abs{o_\R} \not \equiv 0 \pmod p$,
	\item for $u \in \selectSet[\L]$ and $v \in \selectSet[\R]$ it holds $(u,v) \in \selectSet[E]$ if and only if $(u,v) \not\in i_\L \times i_\R$,
\end{enumerate}		
	then we call the triple $\GadgetPart{\L}$, $\GadgetPart{\R}$, and $\GadgetEdge$ along with the partitioning $\selectSet[\L] = i_\L \cupdot o_\L$ and $\selectSet[\R] = i_\R \cupdot o_\R$ a \emph{$p$-hardness gadget}, and we denote it by $\Gadget$. Moreover, in this case, we say that \emph{$H$ has a $p$-hardness gadget}.
\end{definition}

\begin{example}\label{ex:path_has_hardness_gadget}
	Let $p$ be a prime and $P$ be the simple path on $4$ vertices $(x_1, x_2, x_3, x_4)$ with bipartition $\lpart[P]$ and $\rpart[P]$ given by $\lpart[P]=\set{x_1, x_3}$ and $\rpart[P]=\set{x_2, x_4}$. Let $P_2$ be the single edge graph $(x,y)$. The bip-graph $P$ has a $p$-hardness gadget given by
	\begin{itemize}
		\item the partially $P$-labelled bip-graph $\GadgetPart{\L}$ consisting of the simple edge bip-graph $G(J_\L) = P_2$ with distinguished vertex $y_\L = x$ and partial labelling $\tau(J_\L) \colon y_\L \to x_2$, which gives $\selectSet[\L] = \set{x_1, x_3}$;
		\item the partially $P$-labelled bip-graph $\GadgetPart{\R}$ consisting of the simple edge bip-graph $G(J_\R) = P_2$ with distinguished vertex $y_\R = y$ and partial labelling $\tau (J_\R) \colon y_\R \to x_3$, which gives $\selectSet[\R] = \set{x_2, x_4}$;
		\item the partially $P$-labelled bip-graph $\GadgetEdge$ consisting of the simple edge bip-graph $G(J_E) = P_2$ with distinguished vertices $y_\L = x$ and $y_\R = y$ and partial labelling $\tau (J_E) = \tau(J_\L) \cupdot \tau (J_\R)$;
		\item the partitioning $i_\L =\set{x_1}$, $o_\L = \set{x_3}$, $i_\R= \set{x_4}$, and $o_\R = \set{x_2}$.
	\end{itemize}
\end{example}

In many following proofs, we identify $p$-hardness gadgets. To this end, we simplify notation using the following definition.
\begin{definition}\label{def:selection}
	Let $p$ be a prime, $k$ be a positive integer, $H$ be a bip-graph, and $(J,\vector{y})$ be a partially $H$-labelled bip-graph with $k$ distinguished vertices. We define the set
	\[
		\selectSet = \set{\vector{v} \in \vertexset[H]^k \given \numHomBip[(J,\vector{y}), (H, \vector{v})][p] \neq 0 },
	\]
	and say that $(J,\vector{y})$ \emph{$(H,p)$-selects} the set $\selectSet$. When $H$ and $p$ are clear from context, we  say that $(J,\vector{y})$ \emph{selects} $\selectSet$.
\end{definition}

Many of our gadgets will be constructed by combining graphs with the dot product. In order to utilize the dot product iteratively over the set of vertices or edges of a graph, it is useful to disregard distinguished vertices. For a graph with $k$ distinguished vertices $(G, \vector{x})$, we call the graph $G$ obtained by \enquote{annulling the distinguishing} the \emph{host graph of $(G, \vector{x})$}. 

When we say that a partially labelled bip-graph $(H,p)$-selects a set $\selectSet$, we only consider the vertices for which the number of homomorphisms is not congruent modulo $p$ to $0$. Different values of the number of homomorphisms yield different weights for each vertex in $\selectSet$. In order to make the argumentation easier, the following lemma shows that if we can select a set $\selectSet$, then we can select it such that the number of homomorphisms from our graph with distinguished vertices is congruent modulo $p$ to $1$. Intuitively, this is similar to the $0$-$1$ question of whether a vertex can be mapped to or not.
\begin{lemma}\label{lem:p_copies}
	Let $p$ be a prime, $k$ be a non-negative integer, $H$ be a bip-graph, and $(J,\vector{x})$ be a partially $H$-labelled bip-graph with $k$ distinguished vertices. There exists a partially $H$-labelled bip-graph with $k$ distinguished vertices $(J',\vector{x}')$, such that, for all tuples $\vector{x}$ in $\vertexset[H]^k$,
	\begin{itemize}
		\item
		$\numHomBip[(J,\vector{y}), (H,\vector{x})][p]\neq 0$, if and only if $\numHomBip[(J',\vector{y}'), (H,\vector{x})][p]= 1$;
		\item
		$\numHomBip[(J,\vector{y}), (H,\vector{x})][p]= 0$, if and only if
		$\numHomBip[(J',\vector{y}'), (H,\vector{x})][p]=  0$.
	\end{itemize}
\end{lemma}
\begin{proof}
	We take the dot product of $p-1$ disjoint copies of $(J,\vector{y})$. By Corollary~\ref{cor:dot_product_bip}, we obtain for all tuples $\vector{x}$ in $\vertexset[H]^k$
	\begin{align*}
		\numHomBip[(\bigodot_{i \in \sqBrackets{p-1}} (J,\vector{y}) ), (H, \vector{x})] = \parenthesis[\big]{\numHomBip[(J,\vector{y}), (H, \vector{x})] }^{p-1}.
	\end{align*}
	The lemma follows by Fermat's theorem (Theorem~\ref{thm:Fermat}). 
\end{proof}

Triples of partially $H$-labelled bip-graphs $\GadgetPart{\L}$, $\GadgetPart{R}$, and $\GadgetEdge$ allow us to restrict the problem $\probNumBipHom{H}[p]$ to the subproblem $\probNumBipHom{U}[p]$, where $U$ is a sub-bip-graph of $H$.

\begin{lemma}\label{lem:gadget_subgraph_reduction}
	Let $p$ be a prime and let $H$ be a bip-graph. Let there exist a triple of partially $H$-labelled bip-graphs $\GadgetPart{\L}$, $\GadgetPart{R}$, and $\GadgetEdge$ such that
	\begin{itemize}
		\item $\GadgetPart{\L}$ selects the subset $\selectSet[\L]$ of $\lpart[H]$;
		\item $\GadgetPart{\R}$ selects the subset $\selectSet[\R]$ of $\rpart[H]$;
		\item $\GadgetEdge$ selects the subset $\selectSet[E]$ of $\lpart[H] \times \rpart[H]$.
	\end{itemize}
	If $U$ is the bip-graph given by $U=(\selectSet[\L], \selectSet[\R], \selectSet[E] \cap (\selectSet[\L] \times \selectSet[\R]))$, 
	then, for every bip-graph $G$, there exists a partially $H$-labelled bip-graph $J$ such that $\numHomBip[J, H][p]$ is equal to $\numHomBip[G, U][p]$. Furthermore, $J$ is constructible in time $\mathcal{O}(\abs{\lpart[H]} + \abs{\rpart[H]} + \abs{\edgeset[H]})$.
\end{lemma}
\begin{proof}
	We construct the partially $H$-labelled bip-graph $J$ from $G$ and the triple of graphs $\GadgetPart{\L}$, $\GadgetPart{R}$, and $\GadgetEdge$. By Lemma~\ref{lem:p_copies}, we assume without loss of generality that, for every vertex $u$ in $\selectSet[\L]$ and every vertex $v$ in $\selectSet[\R]$, the values $\numHomBip[\GadgetPart{\L}, (H,u)][p]$ and $\numHomBip[\GadgetPart{\R}, (H,v)][p]$ are both equal to $1$. Similar, for every edge $(u,v)$ in the intersection $\selectSet[E] \cap (\lpart[H] \times \rpart[H])$, the value $\numHomBip[\GadgetEdge, (H,u, v)][p]$ is also equal to $1$.
	
	Utilizing the partially $H$-labelled bip-graphs, we construct the partially $H$-labelled bip-graph $J$ from $G$ iteratively by the following three steps. Let $J$ be given by the tuple $(G', \theta)$, where we initially set $G'$ to $G$ and $\theta$ to the empty labelling.
	
	First, for every vertex $x$ in the left part $\lpart$, we replace $G'$ with the host graph of $(G',x) \odot \GadgetPart{L}$ and add the $H$-labelling of the constructed copy of $J_\L$ to $\theta$. We recall the equivalence of $\HomBip[(G',\theta) , H]$ with $\HomBip[(G', \dom(\theta)), (H, \operatorname{range}(\theta))]$, where we assumed a one to one correspondence in the enumerations of $\dom(\theta)$ and $\operatorname{range}(\theta)$.
	By Corollary~\ref{cor:dot_product_bip} and $\numHomBip[\GadgetPart{\L}, (H, u)][p] =1$, 
	\begin{align*}
		\nonumber
		\numHomBip[(G', \theta), H][p] &\equiv \sum_{u \in \vertexset[H]} 	\numHomBip[(G', \theta, x), (H, u)][p] \pmod p \\
		&\equiv  \sum_{u \in \selectSet[\L]} \numHomBip[(G, x), (H, u)][p] \pmod p.
	\end{align*}	
	
	Second, for every vertex $y$ in the right part $\rpart$, we replace $G'$ with the host graph of $(G',y) \odot \GadgetPart{R}$ and add the $H$-labelling of the constructed copy of $J_\R$ to $\theta$. By $\numHomBip[\GadgetPart{\R}, (H, u)][p] =1$, we obtain for any $x$ in the left part $\lpart$
	\begin{align*}
		\numHomBip[(G', \theta), H][p] \equiv  \sum_{u \in \selectSet[\L]} \sum_{v \in \selectSet[\R]}	\numHomBip[(G, x, y), (H, u, v)][p] \pmod p.
	\end{align*}			
	Third, for every edge $(x,y)$ in $\edgeset[G]$ assuming without loss of generality that $x$ is in the left part $\lpart[G]$ and $y$ is in the right part $\rpart[G]$, we replace $G'$ with the host graph of $(G',x,y) \odot \GadgetEdge$ and add the partial $H$-labelling of the constructed copy of $J_E$ to $\theta$. By $\numHomBip[\GadgetEdge, (H, u, v)][p] =1$, we deduce
	\begin{align*}
		\numHomBip[(G', \theta), H][p] \equiv  \sum_{(u,v) \in \selectSet[E] \cap (\selectSet[\L] \times \selectSet[\R])} 	\numHomBip[(G, x,y), (H, u,v)][p] \pmod p.
	\end{align*}	
	
	We conclude that $\numHomBip[(G', \theta), H][p]$ is given by the homomorphisms that map every vertex $x$ in the left part $\lpart$ to $\selectSet[\L]$, every vertex $y$ in the right part $\rpart$ to $\selectSet[\R]$, and every edge $(x,y)$ in $\edgeset[G]$ to $\selectSet[E] \cap (\selectSet[\L] \times \selectSet[\R])$. This gives exactly the bip-graph $U$.
\end{proof}

We are going to show that if an order~$p$ bip-reduced bip-graph $H$ has a $p$-hardness gadget, then $\probNumBipHom{H}[p]$ is $\classNumP[p]$-hard. To this end, we reduce from a weighted bipartite independent sets problem. 
An \emph{independent set} $I$ of a graph $G$ is a set of vertices in $\vertexset[G]$, such that no pair of vertices in $I$ is adjacent in $G$. We denote the \emph{set of independent sets of $G$} by $\mathcal{I} (G)$. For a bip-graph $G$ and integer weights $\lweight, \rweight, \lOUTweight, \rOUTweight$, we denote by $\numWeightBIS{G}[\lweight, \rweight][\lOUTweight, \rOUTweight]$ the weighted number of bipartite independent sets $\sum_{I\in\mathcal{I}(G)}\lweight^{ \abs{\lpart \cap I}} \cdot \lOUTweight^{\abs{\lpart \setminus (\lpart \cap I)}} \cdot \rweight^{\abs{\rpart \cap I}} \cdot \rOUTweight^{\abs{\rpart \setminus (\rpart \cap I)}}$.

\prob%
{$\probNumBIS[\lweight, \rweight][\lOUTweight, \rOUTweight][p]$\label{prob:weighted_bis}.}
{Prime $p$ and weights $\lweight, \rweight, \lOUTweight, \rOUTweight$ in $\Z_p$.}
{Bip-graph $G$.}
{$\numWeightBIS{G}[\lweight, \rweight][\lOUTweight, \rOUTweight] \pmod p$.}

This problem was first considered by Göbel et al. They showed the following hardness result \cite[Theorem~1.6.]{Goebel:21:Counting_Homomorphisms_Trees}.
\begin{theorem}\label{thm:BIS_hardness}
 	Let $p$ be a prime and let $\lweight$ and $\rweight$ be weights in $\Z_p$.
	If at least one of $\lweight$ and $\rweight$ is congruent modulo $p$ to $0$, then $\probNumBIS[\lweight, \rweight][1,1][p]$ is computable in polynomial time.
	Otherwise, $\probNumBIS[\lweight, \rweight][1,1][p]$ is $\classNumP[p]$-complete.
\end{theorem}

For our purposes, we need the extension of this result in the form of the following corollary of Theorem~\ref{thm:BIS_hardness}.
\begin{corollary}\label{cor:BIS_hardness}
	Let $P$ be a prime and let $\lweight$, $\rweight$, $\lOUTweight$, and $\rOUTweight$ be weights in $\Z_p$. If at least one of $\lweight$, $\rweight$, $\lOUTweight$, and $\rOUTweight$ is congruent modulo $p$ to $0$, then $\probNumBIS[\lweight, \rweight][\lOUTweight, \rOUTweight][p]$ is computable in polynomial time.
	Otherwise, $\probNumBIS[\lweight, \rweight][\lOUTweight, \rOUTweight][p]$ is $\classNumP[p]$-complete.
\end{corollary}
\begin{proof}
	The statement holds if at least one of the weights is congruent modulo $p$ to $0$ because the output will always be $0$. We assume that none of the weights $\lweight$, $\rweight$, $\lOUTweight$, and $\rOUTweight$ is congruent modulo $p$ to $0$ and reduce from $\probNumBIS[\lweight^*, \rweight^*][1, 1][p]$, where $\lweight^*$ is $\lweight \cdot \lOUTweight^{-1}\pmod p$ and $\rweight^*$ is $\rweight \cdot \rOUTweight^{-1}\pmod p$. Theorem~\ref{thm:BIS_hardness} yields that $\probNumBIS[\lweight^*, \rweight^*][1, 1][p]$ is \classNumP{}-hard. Let $G$ be the input for $\probNumBIS[\lweight^*, \rweight^*][1, 1][p]$. We have 
	\begin{align*}
	\numWeightBIS{G}[\lweight^*, \rweight^*][1, 1] 
	& \equiv \sum_{I\in\mathcal{I}(G)} ( \lweight^*  )^{\abs{\lpart \cap I}} \cdot ( \rweight^*)^{\abs{\rpart \cap I}} \pmod p
	\\
	& \equiv \sum_{I\in\mathcal{I}(G)} \left( \frac{\lweight}{\lOUTweight} \right )^{\abs{\lpart \cap I}} \cdot \left( \frac{\rweight}{\rOUTweight} \right)^{\abs{\rpart \cap I}} \pmod p
	\\
	&\equiv \lOUTweight^{\abs{\lpart}} \cdot \rOUTweight^{\abs{\rpart}} \cdot \sum_{I\in\mathcal{I}(G)}\lweight^{\abs{\lpart \cap I}}\cdot \lOUTweight^{\abs{\lpart \setminus (\lpart \cap I)}}  \cdot \rweight^{\abs{\rpart \cap I}} \cdot \rOUTweight^{\abs{\rpart \setminus (\rpart \cap I)}} \pmod p
	\\
	& \equiv \lOUTweight^{\abs{\lpart}} \cdot \rOUTweight^{\abs{\rpart}} \cdot \numWeightBIS{G}[\lweight, \rweight][\lOUTweight, \rOUTweight] \pmod p.
	\end{align*}
	The latter can be computed using an oracle for $\probNumBIS[\lweight, \rweight][\lOUTweight, \rOUTweight][p]$. 
\end{proof}

Similar to the problem $\probNumIS$ of counting the independent sets in a graph, the problem $\probNumBIS$ is also characterized as a homomorphism problem to a specific graph. This graph is the simple path $P$ on $4$ vertices. We already know from Example~\ref{ex:path_has_hardness_gadget} that $P$ has a $p$-hardness gadget. This holds for both ways to bipartition $P$. 
\begin{observation}\label{obs:BIS_graph}
	If $P$ is a bip-graph obtained from the simple path on 4 vertices by choosing any bipartition, then $\probNumBipHom{P}$ and $\probNumBIS$ are interreducible under polynomial-time Turing reduction.
\end{observation}
\begin{proof}
	Let $P$ consist of the vertices $\set{u^{\IN}, v^{\OUT}, u^{\OUT}, v^{\IN}}$,  where we assume without loss of generality the bipartition by setting $\lpart[P]$ equal to $\set{u^{\IN}, u^{\OUT}}$ and $\rpart[P]$ equal to $\set{v^{\IN}, v^{\OUT}}$.
	
	Let $G$ be an input bip-graph. Every bip-homomorphism $f$ in $\HomBip[G,P]$ yields an independent set $I$ in $\mathcal{I}(G)$ given by the pre-images of $u^{\IN}$ and $v^{\IN}$. Conversely, every independent set $I$ in $\mathcal{I}(G)$ yields a bip-homomorphism $f$ in $\HomBip[G,P]$ defined by 
	\begin{equation} \label{eq:BIS_graph_mapping}
		f(x) = 
		\begin{cases}
			u^{\IN} &,\text{if } x \in \lpart[G] \cap I \\
			u^{\OUT} &,\text{if } x \in \lpart[G] \setminus I \\
			v^{\IN} &,\text{if } x \in \rpart[G] \cap I \\
			v^{\OUT} &,\text{if } x \in \rpart[G] \setminus I.			
		\end{cases} 
	\end{equation}
	This shows that $\numHomBip[G,P]$ is equal to $\abs{\mathcal{I}(G)}$.
\end{proof}

We construct graphs such that also the weighted versions of $\probNumBIS$ are characterized as homomorphism problems. Two vertices $u$ and $v$ in a graph $H$ are called \emph{twins} if they have the same neighbourhood but are not adjacent. By definition, if $H$ is bipartite, then an equal neighbourhood is sufficient. A tuple of vertices is said to be \emph{a tuple of twins} if every pair is a twin.

Given a graph $H$ and a vertex $v$ in $\vertexset[H]$, the \emph{cloning} of $v$ produces the graph $H'$, where $\vertexset[H']$ is $\vertexset[H]\cupdot \set{v'}$ and $\edgeset[H']$ is $\edgeset[H]\cupdot \set{(v',u) \given (v,u) \in \edgeset[H]}$. The new vertex $v'$ is a twin of $v$ and $v$ is a \emph{representative} of all its twins. 
The class of graphs characterizing the weighted versions of $\probNumBIS$ is then the following.
\begin{definition}
	Let $b_0$, $b_1$, $b_2$, and $b_3$ be positive integers and $P_4$ be the bip-graph obtained from the path on four vertices $(v_0, v_1, v_2, v_3)$ with left part $\lpart[P_4]$ equal to $\set{v_0, v_2}$ and right part $\rpart[P_4]$ equal to $\set{v_1, v_3}$. The bip-graph $P$ is obtained from $P_4$ by, for all indices $i \in \sqBrackets{0;3}$ in ascending order, repeatedly cloning the vertex $v_i$ for $b_i$ times yielding the set $v_i^{b_i}$ of twins $\set{v^j_i}_{j \in \sqBrackets{b_i}}$. We call $P$ a \emph{thick $4$-vertex path} and denote it by $(v_0^{b_0}, v_1^{b_1}, v_2^{b_2}, v_3^{b_3})$.
\end{definition}
An example of a thick $4$-vertex path is given in Figure~\ref{fig:terminal_tree_gen_path}. We emphasize that a thick $4$-vertex path is by definition a bip-graph with bipartition intrinsic from the fixed bipartition of $P_4$. The latter is arbitrary since $P_4$ is connected. The characterization of $\probNumBIS[b_1, b_4][b_2, b_3]$ by a homomorphism problem is as follows.

\begin{figure}[t]
	\centering
	\includegraphics[]{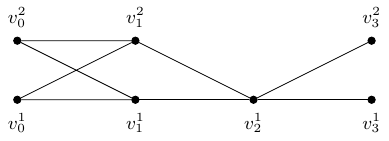}
	\caption{The figure depicts the thick $4$-vertex path $(v_0^{2}, v_1^{2}, v_2^{1}, v_3^{2})$.}
	\label{fig:terminal_tree_gen_path}
\end{figure}

\begin{lemma}\label{lem:BIS_graph}
	Let $b_0$, $b_1$, $b_2$, and $b_3$ be positive integers and let $P$ be the thick $4$-vertex path $(v_0^{b_0}, v_1^{b_1}, v_2^{b_2}, v_3^{b_3})$. The two problems $\probNumBipHom{P}$ and $\probNumBIS[b_0, b_3][b_1, b_2]$ are equal.
\end{lemma}
\begin{proof}
	Let $G$ be an input bipartite graph for the two problems. The bip-graph $P$ was obtained from the simple path on $4$ vertices $P_4 = (v_0, v_1, v_2, v_3)$ by cloning the vertices. Therefore, in $P$ all twins of $v_i$ are in the same orbit. Let $\set{v_0, v_1, v_2, v_3}$ be the representatives of these orbits. Every bip-homomorphism $g$ in $\HomBip[G,P]$ yields a homomorphism $f_g$ in $\HomBip[G, P_4]$ by identifying the twins in $P$, i.e., for each vertex $x$ in $\vertexset[G]$, if $g(x)$ is in $v^{b_i}_i$, then we assign $f_g(x)$ to be $v_i$. This gives an equivalence relation on $\HomBip[G,P]$ by $g \eqrel g'$ if and only if the bip-homomorphisms $f_g$ and $f_{g'}$ are equal. We denote by $\eqclass{f_g}$ the equivalence class in $\HomBip[G, P]$ of $g$ given by the representative $f_g \in \HomBip[G, P_4]$ and obtain
	\begin{align*}
	\numHomBip[G,P] & = \sum_{f \in \HomBip[G, P_4]} \abs[\big]{\eqclass{f}}.
	\intertext{%
		For any bip-homomorphism $f$ in $\HomBip[G, P_4]$ and $i \in \sqBrackets{0;3}$, we denote by $f^{-1}(v_i)$ the set of vertices $x$ in $\vertexset[G]$ with $f(x)$ equal to $v_i$. This yields%
	}
	&= \sum_{f \in \HomBip[G, P_4]} b_0^{\abs{f^{-1}(v_0)} }\cdot  b_1^{\abs{f^{-1}(v_1)}} \cdot   b_2^{\abs{f^{-1}(v_2)}} \cdot   b_3^{\abs{f^{-1}(v_3)}} .
	\intertext{%
		We recall Observation~\ref{obs:BIS_graph}. In particular, we recall that $\numHomBip[G,P_4]$ is equal to $\abs{\mathcal{I}(G)}$ and the corresponding map given in~\eqref{eq:BIS_graph_mapping}, where $(u^{\IN}, v^{\OUT}, u^{\OUT}, v^{\IN})$ is $(v_0, v_1, v_2, v_3)$. Therefore,%
	}
	\numHomBip[G,P] & = \sum_{I \in \mathcal{I}(G)} b_0^{\abs{I \cap \lpart}} \cdot b_1^{ \abs{\lpart \setminus I}}  \cdot b_2^{\abs{\rpart \setminus I}} \cdot b_3^{\abs{I \cap \rpart}} = \numWeightBIS{G}[b_0, b_3][b_1, b_2]. \qedhere
	\end{align*}
\end{proof}

We note that, for a prime $p$ and positive integers $b_0$, $b_1$, $b_2$, and $b_3$ that are not congruent modulo $p$ to $0$, the same gadget as employed in Example~\ref{ex:path_has_hardness_gadget} yields a $p$-hardness gadget for the thick $4$-vertex path $(v_0^{b_0}, v_1^{b_1}, v_2^{b_2}, v_3^{b_3})$. Indeed, in this case we obtain hardness.

\begin{corollary}\label{cor:gen_path_reduction_BIS}
	Let $p$ be a prime, $b_0$, $b_1$, $b_2$, and $b_3$ be positive integers, and $P$ be the thick $4$-vertex path $(v_0^{b_0}, v_1^{b_1}, v_2^{b_2}, v_3^{b_3})$. If, for each index $i\in\sqBrackets{0;3}$, the weight $b_i$ is not congruent modulo $p$ to $0$, then $\probNumBipHom{P}[p]$ is $\classNumP[p]$-hard.
\end{corollary}
\begin{proof}
	The order~$p$ bip-reduced form of $P$ is the thick $4$-vertex path $\bipreduced{P}$ equal to $(v_1^{c_1}, v_2^{c_2}, v_3^{c_3}, v_4^{c_4})$, where, for each index $i \in \sqBrackets{0;3}$, the weight $c_i$ is given by $c_i \equiv b_i \pmod p$. By Lemma~\ref{lem:BIS_graph}, the problem $\probNumBipHom{\bipreduced{P}}[p]$ is equal to $\probNumBIS[c_1, c_4][c_2, c_3]$. The latter is $\classNumP[p]$-hard by Corollary~\ref{cor:BIS_hardness}. Lastly, for every input bip-graph $G$, the output of $\probNumBipHom{\bipreduced{P}}[p]$ is equal to the output of $\probNumBipHom{P}[p]$ due to Observation~\ref{obs:p-wise_reduced_unique}.
\end{proof}

We are now ready to show that the existence of a $p$-hardness gadget yields indeed hardness. 
\begin{lemma}\label{lem:hardness_gadget}
	Let $p$ be a prime and $H$ be an order~$p$ bip-reduced bip-graph. If $H$ has a $p$-hardness gadget, then $\probNumBipHom{H}[p]$ is $\classNumP[p]$-hard.
\end{lemma}
\begin{proof}
	We are going to show that there exist four positive integers $b_0$, $b_1$, $b_2$, and $b_3$ satisfying the following. First, for each index $i \in \sqBrackets{0;3}$, the weight $b_i$ is not congruend modulo $p$ to $0$. Second, for the thick $4$ vertex path $P$ equal to $(v_0^{b_0}, v_1^{b_1}, v_2^{b_2}, v_3^{b_3})$, the problem $\probNumBipHom{H}[p]$ reduces to $\probNumBipHom{P}[p]$. The result follows then from Corollary~\ref{cor:gen_path_reduction_BIS}.
	
	Let $\Gadget$ be the $p$-hardness gadget in $H$ and $G$ be a bip-graph. We apply Lemma~\ref{lem:gadget_subgraph_reduction} and obtain a partially $H$-labelled bip-graph $(G', \theta))$ such that the values $\numHomBip[(G', \theta), H][p]$ and $\numHomBip[G, U][p]$ are equal, where $U$ is the bip-graph given by $U=(\selectSet[\L], \selectSet[\R], \selectSet[E] \cap (\selectSet[\L] \times \selectSet[\R]))$. 
	
	We recall from the properties of a $p$-hardness gadget that $\selectSet[\L]$ is equal to $i_\L \cupdot o_\L$ and $\selectSet[\R]$ is equal to $i_\R \cupdot o_\R$, where none of the parts has cardinality congruent modulo $p$ to $0$. We define the weights by $b_0 = \abs{i_\L}$, $b_1 = \abs{o_\R}$, $b_2 = \abs{o_\L}$, and $b_3 =\abs{i_\R}$. It suffices to show that $U$ is \isomorphic[bip] to $P$, the thick $4$-vertex path $(v_0^{b_0}, v_1^{b_1}, v_2^{b_2}, v_3^{b_3})$.
	
	Let $(u,v)$ be an edge in $\edgeset[H]$ with $(u,v)$ in $\selectSet[\L] \times \selectSet[\R]$. By property~2 of a $p$-hardness gadget, $(u,v))$ in in $\selectSet[E]$ if and only if $(u,v)$ is not in $i_\L \times i_\R$. The same characterization of edges holds for the thick $4$-vertex path $P$.
\end{proof}

\subsection{Recursive Gadgets}

Constructing $p$-hardness gadgets for a bip-graph $H$ is a challenging task. Frequently, it is easier to argue that $H$ contains an induced sub-bip-graph~$B$ for which it has been established that $\probNumBipHom{B}[p]$ is $\classNumP[p]$-hard. We define the following.

\begin{definition}
	Let $p$ be a prime, let $H$ be a bip-graph, and let $B$ be an induced sub-bip-graph of $H$. We say that $H$ \emph{has a \gadget{B,p}} if there exists a pair of partially $H$-labelled bip-graphs, each with a distinguished vertex, $(J_\L, y_\L)$ and $(J_\R, y_\R)$ such that the following hold:
	\begin{itemize}
	 \item $(J_\L, y_\L)$ selects the set $\selectSet[\L]$, that contains $\lpart[B]$;
	 \item $(J_\R, y_\R)$ selects the set $\selectSet[\R]$, that contains $\rpart[B]$;
	 \item the set of adjacent pairs $(u,v)$ in $\selectSet[\L] \times \selectSet[\R]$ is the set of edges $\edgeset[B]$.
	\end{itemize}
\end{definition}

A \gadget{B,p} allows us to restrict our study from the larger bip-graph $H$ to the induced sub-bip-graph $B$ by an application of Lemma~\ref{lem:gadget_subgraph_reduction}.
\begin{corollary}\label{cor:(B,p)_gadget_subgraph_reduction}
	Let $p$ be a prime, let $H$ be a bip-graph, and let $B$ be an induced sub-bip-graph of $H$. If $H$ has a \gadget{B,p}, then, for all bip-graphs $G$, there exists a partially $H$-labelled bip-graph $J$, such that $J$ is constructible from $G$ in time $\mathcal{O}(\abs{\lpart[H]} + \abs{\rpart[H]} + \abs{\edgeset[H]})$ and $\numHomBip[J,H][p]$ is equal to $\numHomBip[G,B][p]$.
\end{corollary}
\begin{proof}
	Let $\GadgetPart{\L}$ and $\GadgetPart{\R}$ be the given \gadget{B,p}. We have that $\GadgetPart{\L}$ selects the set $\selectSet[\L]$ containing the left part $\lpart[B]$ and $\GadgetPart{\L}$ selects the set $\selectSet[\R]$ containing the right part $\rpart[B]$. In order to apply Lemma~\ref{lem:gadget_subgraph_reduction}, we supplement the \gadget{B,p} with a partially $H$-labelled bip-graph $\GadgetEdge$. Since we aim toward an induced subgraph, the graph $\GadgetEdge$ has to select every edge in $\edgeset[H]$. To this end, the simple edge graph without labelling is sufficient, i.e. $\GadgetEdge$ is defined to be the bip-graph with left part $\lpart[G(J_E)]$ equal to $\set{y_\L}$, right part $\rpart[G(J_E)]$ equal to $\set{y_\R}$, and edge set $\edgeset[G(J_E)]$ equal to ${(y_\L, y_\R)}$. By the definition of a \gadget{B,p}, $\GadgetEdge$ selects exactly the edges in $\edgeset[B]$ from the pairs in $\selectSet[\L] \times \selectSet[\R]$. Lemma~\ref{lem:gadget_subgraph_reduction} yields a partially $H$-labelled bip-graph $J$ that satisfies the corollary.
\end{proof}

Let $p$ be a prime and $H$ be a bip-graph as in Corollary~\ref{cor:(B,p)_gadget_subgraph_reduction}.
If $H$ is order~$p$ bip-reduced, then $\probNumPartLabBipHom{H}[p]$ reduces to $\probNumBipHom{H}[p]$ due to Lemma~\ref{lem:pinning_bip}. By Corollary~\ref{cor:(B,p)_gadget_subgraph_reduction}, $\probNumBipHom{B}[p]$ reduces to $\probNumPartLabBipHom{H}[p]$.
\begin{corollary}\label{cor:gadget_reduction}
	Let $p$ be a prime, let $H$ be an order~$p$ bip-reduced bip-grap, and let $B$ be an induced sub-bip-graph of $H$. If $H$ has a \gadget{B,p}, then $\probNumBipHom{B}[p]$ reduces to $\probNumBipHom{H}[p]$ via polynomial-time Turing reduction.
\end{corollary}

Corollary~\ref{cor:gadget_reduction} gives us a recursive generalization of a $p$-hardness gadget. Here, the \gadget{B,p} is the advertized means for an iterative reduction argument and also the main reason for Observation~\ref{obs:p-wise_reduced_unique}.
\begin{definition}\label{def:p-gadget}
	Let $p$ be a prime and $H$ be an order~$p$ bip-reduced bip-graph. We say that $H$ \emph{admits a $p$-hardness gadget} if $H$ has a $p$-hardness gadget or there exists an induced sub-bip-graph $B$ of $H$, such that $H$ has a \gadget{B,p} and the order~$p$ bip-reduced form $\bipreduced{B}$ admits a $p$-hardness gadget.
\end{definition}

Now we conclude with the main result of this section.
\begin{corollary}\label{cor:hardness gadget}
	Let $p$ be a prime and $H$ be an order~$p$ bip-reduced bip-graph. If $H$ admits a $p$-hardness gadget, then $\probNumBipHom{H}[p]$ is $\classNumP[p]$-hard.
\end{corollary}
\begin{proof}
	We recall Definition~\ref{def:p-gadget} and Observation~\ref{obs:p-wise_reduced_unique}, the latter allows us to assume without loss of generality that every intermediate bip-graph is order~$p$ bip-reduced.
	There exist $k$ bip-graphs $H_1, \dots, H_k$, where $H_k$ is equal to $H$, such that the following hold. First, for every index $i\in \sqBrackets{k-1}$, the bip-graph $H_{i+1}$ has an induced sub-bip-graph $H_i$ and an $(H_i,p)$-gadget. Second, $H_1$ has a $p$-hardness gadget.
	We iteratively apply Corollary~\ref{cor:gadget_reduction} on the graphs $H_{i+1}$, which gives us a reduction from $\probNumBipHom{H_{i+1}}[p]$ to $\probNumBipHom{H_i}[p]$. It suffices to establish that $\probNumBipHom{H_1}[p]$ is $\classNumP[p]$-hard. Since $H_1$ admits a $p$-hardness gadget, this follows from Lemma~\ref{lem:hardness_gadget}.
\end{proof}

We collect our insights.
Towards a dichotomy for $\probNumHom{H}[p]$ we aim to show that the problem $\probNumPartLabBipHom{H'}[p]$ is $\classNumP[p]$-hard for every order~$p$ bip-reduced bip-graph $H'$ that is not complete bipartite. By Corollary~\ref{cor:hardness gadget} it suffices to establish that every such graph $H'$ admits a $p$-hardness gadget. Lemma~\ref{lem:bip_components} allows us to project hardness for connected components of $H$ to hardness for $H$ itself.

\section{Hardness for \Graphclass{} Graphs}
\label{sec:hardness_graphclass}
In the longest section of this paper, we are going to show that, for all primes $p$ and every connected \graphclass{} bip-graph $H$ that is order~$p$ bip-reduced and not complete bipartite, 
the problem $\probNumBipHom{H}[p]$ is $\classNumP[p]$-hard. To this end, we are going to establish that $H$ admits a $p$-hardness gadget via a careful structural analysis on the target graph~$H$. For this reason, we require the following definitions.

\begin{definition}
	Let $H$ be a graph and $v$ be a vertex of $H$. The \emph{2-neighbourhood} of $v$ is formally defined as the induced subgraph $H[\neigh{v} \cup \neigh{\neigh{v}}]$ of $H$ and denoted by $B_2(v)$. 
\end{definition}

\begin{definition}\label{def:split}
	Let $H$ be a graph and let $v$ be a vertex of $H$. Consider the set $\family{C}$ of connected subgraphs that result from deleting $v$ in $H$. Each subgraph $C$ in $\family{C}$ corresponds to an induced subgraph $H[V(C) \cupdot \set{v}]$, and we denote the set $\set{H[V(C)\cupdot\set{v}]\given C\in \family{C}}$ by $\family{U}$. 

	The set $\mathcal U$ decomposes into $\sigma$ equivalence classes $(\eqclass{U_i})_{i \in \sqBrackets{\sigma}}$ under the equivalence relation $(U_1,v)\cong (U_2,v)$, where $U_1$ and $U_2$ are in $\family{U}$. For each equivalence class $\eqclass{U_i}$, we refer by the cardinality $\abs{\eqclass{U_i}}$ to the \emph{count} of $\eqclass{U_i}$, which we denote by $\alpha_i$. The \emph{split of $H$ at $v$} is the set $\set{(\eqclass{U_i},\alpha_i)\given i \in \sqBrackets{\sigma}}$ of equivalence classes together with their counts $\alpha_i$. The set $\family{U}$ consists of \emph{components in the split of $H$ at $v$}.
\end{definition}

\begin{figure}[t]
	\centering
	\includegraphics[]{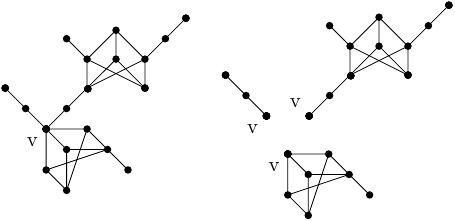}
	\caption{Example of a split. The left figure depicts the graph $H$. The right figure depicts the split of $H$ at the vertex $v$.}
	\label{fig:split_example}
\end{figure}
An example of a split is depicted in Figure~\ref{fig:split_example}.
The split of a graph $H$ is given by induced subgraphs of $H$ as is the $2$-neighbourhood of any vertex of $H$. These are sub-bip-graphs of $H$ provided that $H$ is a bip-graph, which is how we are applying these two notions. The definition of a split can also be applied for the $2$-neighbourhood $\twoneigh{v}$ of any vertex $v$ of $H$. In this case, the notion of the split of $\twoneigh{v}$ at $v$ is properly defined by Definition~\ref{def:split}. This is how we are going to apply the split.

The following definition allows us to study the splits of $2$-neighbourhoods concisely. It lies at the heart of this section and its second part already hints toward the property of \graphclass{} graphs we are going to prove immediately afterwards and employ frequently.
\begin{definition}\label{def:complete_core}
	Let $H$ be a bipartite graph with a pair of adjacent vertices $v$ and $u$ in $\vertexset[H]$. We denote the component $U$ in the split of $\twoneigh{v}$ at $v$ that contains $u$ by $U^{v,u}$. Moreover, we employ the following notation:
	\begin{itemize}
		\item if $\abs{\neigh{v}\cap U}$ is larger than $1$, then we call the maximal $2$-connected subgraph $K$ of $U^{v,u}$ the \emph{complete core of $U^{v,u}$} and denote it by $K^{v,u}$;
		\item if $\abs{\neigh{v}\cap U}$ is equal to $1$, then we call $U^{v,u}$ the \emph{complete core of $U^{v,u}$} and denote it by $K^{v,u}$.
	\end{itemize}
\end{definition}
The second point is justified because in the case that $\abs{\neigh{v}\cap U}$ is equal to $1$, the component $U^{v,u}$ itself is a star.  Figure~\ref{fig:complete_core_example} illustrates the definition.
We note that from the notation $K^{v,u}$ and $U^{v,u}$ it follows that both belong to the $2$-neighbourhood of $v$. Further, for a bip-graph $H$, both $U^{v,u}$ and $K^{v,u}$ are induced sub-bip-graphs of $H$, and thus inherit the bipartition from $H$. We note that the definition for components in the split of $H$ at $v$ as well as for complete core only refers to induced subgraphs contrary to sub-bip-graphs. In the case of a bip-graph $H$, we assume that these components and complete cores are indeed bip-graphs, i.e. they have a fixed bipartition inherited from $H$. Apart from the following lemma, we apply these concepts only for bip-graphs.

\begin{figure}[t]
	\centering
	\includegraphics[]{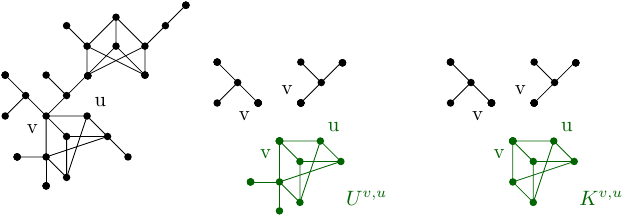}
	\caption{Example of a complete core. The first figure depicts the graph $H$. The second figure depicts the split of $\twoneigh{v}$ at the vertex $v$ with highlighted component $U^{v,u}$. The third figure depicts the complete cores in $\twoneigh{v}$ with highlighted $K^{v,u}$.}
	\label{fig:complete_core_example}
\end{figure}

We now identify the following structural property of \graphclass{} graphs, which is the key property used for our analysis and justifies the term \emph{complete core} in the case of \graphclass{} graphs. 

\begin{lemma}\label{lem:weirdness}
	If $H$ is a bipartite graph, then $H$ is \graphclass{} if and only if, for every pair of adjacent vertices $u$ and $v$ of $H$, the complete core $K^{v,u}$ is complete bipartite.  
\end{lemma}
\begin{proof}
	We observe that both $\forbiddenA$ and $\forbiddenB$ as an induced subgraph of $H$ yield a complete core in $H$ that is not complete bipartite. It remains to show the other direction.
	
	Let $H$ be \graphclass{}. We assume toward contradiction that there exists a pair of adjacent vertices $u$ and $v$ in $\vertexset[H]$ such that the complete core $K^{v,u}$ is not complete bipartite. By the definition of complete core, we deduce that $K^{v,u}$ is $2$-connected. Therefore, there exists a cycle $C$ in $K^{v,u}$ that contains $v$ such that the induced subgraph $H[C]$ is not complete bipartite. Let $C$ be $(w_0,w_1, w_2, \dots, w_0)$. Since $H$ is bipartite, $C$ is of even length at least $6$. Further, since $K^{v,u}$ is of radius $2$, every vertex $w_i$ of $C$ with even index is connected to $v$ or identical to $v$.	
	Without loss of generality, let $w_0$ be $v$ and $v$ be in the left part $\lpart[K^{v,u}]$. We deduce by $C$ that there exist three vertices $u_1$, $u_2$, and $u_3$ in the right part $\rpart[K^{v,u}]$ and adjacent to $v$. Additionally, there exist two vertices $v_1$ and $v_2$ in the left part $\lpart[K^{v,u}]$ and  unequal to $v$, such that the induced subgraph $H[\set{v,u_1, u_2, u_3, v_1,v_2}]$ is $2$-connected but not complete bipartite. We denote this induced subgraph by $U$, and claim that $U$ is isomorphic to either $\forbiddenA$ or $\forbiddenB$.
	
	If $v_1$ is adjacent to all three vertices $u_1$, $u_2$, and $u_3$, then $v_2$ has to be adjacent to exactly two vertices in the set $\set{u_1, u_2, u_3}$ because otherwise $U'$ is not $2$-connected. This yields that $U$ is isomorphic to $\forbiddenA$. The case that $v_2$ is adjacent to all three vertices follows from the analogue argument. It remains the case that both $v_1$ and $v_2$ are adjacent to exactly two vertices in the set $\set{u_1, u_2, u_3}$. Since $U$ is $2$-connected, the vertices $v_1$ and $v_2$ cannot be adjacent to the same two vertices in the set $\set{u_1, u_2, u_3}$ because this would yield a leaf in $U'$. Therefore, $U$ is isomorphic to $\forbiddenB$.
\end{proof}

In light of Lemma~\ref{lem:weirdness}, when working with \graphclass{} graphs $H$ we automatically assume that every complete core $K^{v,u}$, given by two adjacent vertices $u$ and $v$ of $H$, is a complete bipartite graph. By the definition of a component, this implies that $K^{v,u}$ is neither empty nor an isolated vertex. Further, $K^{v,u}$ is $2$-connected if it is not a star. The remainder $\neigh{u} \setminus \vertexset[K^{v,u}])$ in $U^{v,u}$  consists exclusively of leaves if it is not empty.

For a bip-graph $H$ and a vertex $v$ of $H$, we introduce the following two indicator functions that grant us access to the two parts of $H$ with respect to $v$. The function $\partof[H] \colon \vertexset[H] \to \lpart[H] \cupdot \rpart[H]$ is defined by $\partof[H] (v) = \lpart[H]$ if and only if $v$ is in $\lpart[H]$, otherwise $\partof[H] (v) = \rpart[H]$. The inverse is made available by the function $\npartof[H] \colon \vertexset[H] \to \lpart[H] \cupdot \rpart[H]$ defined by $\npartof[H] (v) = \rpart[H]$ if and only if $\partof[H] (v) = \lpart[H]$, otherwise $\npartof[H] (v) = \lpart[H]$. For two bip-graphs $H$ and $G$, a vertex $x$ of $G$, and a vertex $v$ of $H$, we say that \emph{$\partof[G](x)$ and $\partof[H](v)$ agree} if both vertices are in the same respective part of their bip-graph. 

We recall the definition of a bip-isomorphism. If $H$ is a bip-graph, then $K^{v,u}$ is \isomorphic[bip] to the complete bipartite graph $K_{a,b}$, where $a$ is the number of vertices in $\partof[K^{v,u}] (v)$ and $b$ is the number of vertices in $\npartof[K^{v,u}] (v)$, and the bipartition is fixed accordingly. This is denoted by $K^{v,u} \congbip K_{a,b}$.

\subsection{Hardness for \Graphclass{} Graphs of Small Radius}
In this subsection, we study the following subclass of \graphclass{} graphs.

\begin{definition}
	We say that a graph $H$ has \emph{radius at most 2} if there exists a vertex $v$ of $H$ such that $H$ is equal to the $2$-neighbourhood of $v$.
\end{definition}

By definition, a graph of radius at most $2$ is connected.
We have already seen examples of \graphclass{} graphs of radius at most 2 in the form of thick $4$-vertex paths. A thick $4$-vertex path $(v_0^{b_0}, v_1^{b_1}, v_2^{b_2}, v_3^{b_3})$ is order~$p$ bip-reduced if, for every index $i \in \sqBrackets{0;3}$, $b_i$ is smaller than $p$. In the case that, additionally, for all indices $b_i>0$, Corollary~\ref{cor:gen_path_reduction_BIS} applies and, as we have argued in the previous section, we find a $p$-hardness gadget.

Graphs of radius at most $2$ arise naturally from our applications of a \gadget{B,p}. At multiple points of our search for gadgets, the following observation will be useful.
\begin{observation}\label{obs:hard_vertex}
	Let $p$ be a prime and $H$ be a bip-graph. If $v$ is a vertex in the left part $\lpart[H]$, then we construct a \gadget{B,p} by
	\begin{itemize}
		\item $\GadgetPart{\L}$ is the single distinguished vertex $\lpart[G(J_\L)] = \set{y_\L}$ with empty labelling;
		\item  $\GadgetPart{\R}$ consists of the edge $(y,y_\R)$ with $\rpart[G(J_\R)] = \set{y_\R}$ and the labelling $\tau_R \colon y \mapsto v $.
	\end{itemize}  
	We observe that $\GadgetPart{\L}$ selects all vertices in $\lpart[H]$ and $\GadgetPart{\R}$ selects all vertices in $\neigh{v}$. This gives that the induced sub-bip-graph $B$ consists of the disjoint union of $\twoneigh{v}$ and the set $\lpart[H]\setminus \twoneigh{v}$ of isolated vertices. If $v$ is in the right part $\rpart[H]$, then we obtain a \gadget{B,p} by swapping $\GadgetPart{\L}$ and $\GadgetPart{\R}$ against each other.
	
	Iteratively applying this observation, we obtain a \gadget{B,p} for $H$, where $B$ is the disjoint union of a possibly empty set of isolated vertices and the intersection of 2-neighbourhoods of any number of vertices.
\end{observation}

For our purposes, we can safely disregard the isolated vertices. First, we show that a $p$-hardness gadget is void of isolated vertices.
\begin{lemma}\label{lem:no_isolated_vertices_hardness_gadget}
	Let $p$ be a prime and $H$ be a bip-graph. If $H$ has a $p$-hardness gadget $\Gadget$, then neither $\GadgetPart{\L}$ nor $\GadgetPart{\R}$ select isolated vertices of $H$.
\end{lemma}
\begin{proof}
	Without loss of generality, let $\GadgetPart{\L}$ select an isolated vertex $u^\ast$. There exists a vertex $v$ selected by $\GadgetPart{\R}$ such that $\numHomBip[\GadgetEdge, (H, u^\ast, v)][p]$ is not $0$. We consider the connected component $C_{y_\L}$ of $G(J_E)$ that contains $y_\L$. Since $u^\ast$ is an isolated vertex, we deduce that $C_{y_\L}$ is the isolated vertex $y_\L$. Let $v'$ be a vertex in $i_\R$, and we consider a pair of vertices $u^{\mathrm{in}}$ and $u^{\mathrm{out}}$, where $u^{\mathrm{in}}$ is in $i_\R$ and $u^{\mathrm{out}}$ is in $o_\R$. Since $C_{y_\L}$ is the isolated vertex $y_\L$, the numbers of homomorphisms $\numHomBip[\GadgetEdge, (H, u^{\mathrm{in}}, v')]$ and $\numHomBip[\GadgetEdge, (H, u^{\mathrm{out}}, v')]$ are equal, a contradiction to the definition of a $p$-hardness gadget.
\end{proof}

Second, we show that adding isolated vertices does not affect the hardness.
	
\begin{lemma}\label{lem:2-neighbourhood_isolated_vertices_gadgetry}
	Let $p$ be a prime and $H$ be an order~$p$ bip-reduced bip-graph. If there exists a set $V$ of vertices of $H$, such that for $B$ the intersection $\bigcap_{v \in V}\twoneigh{v}$, \forThesis{the bip-graph $B$ has a $p$-hardness gadget or }the order~$p$ bip-reduced form $\bipreduced{B}$ admits a $p$-hardness gadget, then $H$ admits a $p$-hardness gadget.
\end{lemma}
\begin{proof}		
	We consider the \gadget{B',p} given by Observation~\ref{obs:hard_vertex} for $B$. The induced sub-bip-graph $B'$ consists of $B$ and a possibly empty collection $C$ of isolated vertices in the same part of $B'$. \forThesis{If $B$ has a $p$-hardness gadget $\Gadget$, then any partially $B$-labelled bip-graph is also partially $B'$-labelled since $B$ is a sub-bip-graph of $B'$. Neither $\GadgetPart{\L}$ nor $\GadgetPart{\R}$ select isolated vertices by Lemma~\ref{lem:no_isolated_vertices_hardness_gadget}. Therefore, $\Gadget$ is a $p$-hardness gadget for $B'$ as well.
		
	The second part of the lemma demands a more evolved argument. }First, we show that removing isolated vertices of a bip-graph $G$ does not affect a partially $G$-labelled graph $J$ with $G$-labelling $\tau$ -- except for the case that a distinguished vertex $y$ is in the domain of $\tau$. This is necessary because we aim to use a gadget for $G$ on the subgraph $G'$ obtained by removing all isolated vertices of $G$ in one part.

	Let $\tau$ be the partial $G$-labelling of $J$ and, for a vertex $y$ of $G(J)$ not in the domain of $\tau$, let $(J,y)$ select the non-empty set $\selectSet$. Let $V$ be the set of isolated vertices of $G$ in the range of $\tau$ and let $X$ be the domain of $V$ under $\tau$. The set $X$ might contain vertices in $\selectSet$. There is at least one vertex $v$ in $\selectSet$ and hence at least one bip-homomorphism $f$ from $(J,y)$ to $(G,v)$. If $v$ is an isolated vertex, then it follows that $y$ is an isolated vertex in $G(J)$. Since $y$ is not in the domain of $\tau$, the unlabelled bip-graph $(J',y)$ consisting of the isolated vertex $y$ selects $\selectSet$ too. Otherwise, the partially $G$-labelled bip-graph $J'$ obtained by removing $X$ from $G(J)$ and removing $\set{x \mapsto \tau(x) \given x \in X}$ from $\tau$ is partially $G'$-labelled. We recall that $y$ is not in $X$. The partially $G'$-labelled bip-graph $(J',y)$ selects $\selectSet$ too, that is $\numHomBip[(J,y), (G,v)]$ is equal to $\numHomBip[(J',y), (G',v)]$ by Corollary~\ref{cor:dot_product}, which also holds for partially labelled graphs by Section~\ref{sec:prelims}, because $J$ is obtained from $J'$ by the disjoint union. The statement still holds when we distinguish multiple vertices of $J$.
		
	Second, we show the lemma by an inductive argument toward the following claim. Let $G$ be a bip-graph and $\bipreduced{G}$ admit a $p$-hardness gadget, where $G_1, \dots, G_k$ is a resulting series of minimal length $k$ of bip-graphs, that is, $G_k$ is equal to $\bipreduced{G}$, the order~$p$ reduced form $\bipreduced{(G_1)}$ has the $p$-hardness gadget $\Gadget$, and, for every $i \in \sqBrackets{k-1}$, the order~$p$ reduced bip-graph $\bipreduced{(G_{i+1})}$ has $G_{i}$ as induced sub-bip-graph and the \gadget{G_i,p} given by the pair of partially $\bipreduced{(G_{i+1})}$-labelled bip-graphs $(J^{i+1}_\L,y^{i+1}_\L)$ and $(J^{i+1}_\R,y^{i+1}_\R)$. We refer to $k$ by the \emph{length of the gadget-series}. Given the bip-graph $C$ that consists of only isolated vertex in one part, the bip-graph $\bipreduced{(\bipreduced{G} \cupdot C)}$ admits a $p$-hardness gadget.
		
	Without loss of generality, let all vertices of $C$ be in $\lpart[C]$. Before we prove the claim, we provide more insight into the structure of $\bipreduced{(\bipreduced{G} \cupdot C)}$. Since by Observation~\ref{obs:p-wise_reduced_unique} the order~$p$ bip-reduced form is unaffected by the order of the used automorphisms, we observe that $\bipreduced{(\bipreduced{G} \cupdot C)}$ is \isomorphic[bip] to the bip-graph obtained by first taking the order~$p$ bip-reduced form $\bipreduced{G}$ and then reducing by automorphisms that do not fix $C$. Let $\hat{C}$ be the possibly empty set of isolated vertices in $\bipreduced{G}$ that are in $\lpart[\bipreduced{G}]$, it follows
	\begin{equation}\label{eq:split_get_rid_of_isolated_vertices}
		\bipreduced{(\bipreduced{G} \cupdot C)} = \bipreduced{\parenthesis[\big]{\bipreduced{G} \setminus \hat{C} \cupdot (\hat{C} \cupdot C)}} \congbip \bipreduced{G}\setminus \hat{C} \cupdot \bipreduced{(\hat{C} \cupdot C)}.
	\end{equation}
		
	Now, we prove the claim via induction on $k$. If the length of the gadget-series $k$ is $1$, then $\bipreduced{G}$ has the $p$-hardness gadget $\Gadget$. As we have seen, we assume without loss of generality that all bip-graphs in $\Gadget$ are partially $(\bipreduced{G}\setminus C)$-labelled, and thus by \eqref{eq:split_get_rid_of_isolated_vertices} they are partially $\bipreduced{(\bipreduced{G} \cupdot C)}$-labelled. By Lemma~\ref{lem:no_isolated_vertices_hardness_gadget}, neither $\GadgetPart{\L}$ nor $\GadgetPart{\R}$ select isolated vertices. Therefore, $\Gadget$ is a $p$-hardness gadget for $\bipreduced{(\bipreduced{G} \cupdot C)}$.
		
	Assuming the length of the gadget-series $k$ is at least $2$ and the claim holds for all graphs $G'$ with length of the gadget-series at most $k-1$, we take the \gadget{G_{k-1}, p} of $\bipreduced{G}$ in the form of the pair of partially $\bipreduced{G}$-labelled bip-graphs $(J^{k}_\L,y^{k}_\L)$ and $(J^{k}_\R,y^{k}_\R)$. The bip-graph $G_{k-1}$ has length of the gadget-series $k-1$ and cannot be a collection of isolated vertices due to Lemma~\ref{lem:no_isolated_vertices_hardness_gadget}. Therefore, we assume without loss of generality that $(J^{k}_\L,y^{k}_\L)$ and $(J^{k}_\R,y^{k}_\R)$ are partially $\bipreduced{(\bipreduced{G} \cupdot C)}$-labelled. The statement will follow from the induced sub-bip-graph $G'_{k-1}$ of $\bipreduced{(\bipreduced{G} \cupdot C)}$ obtained by the \gadget{G'_{k-1}, p} given by $(J^{k}_\L,y^{k}_\L)$ and $(J^{k}_\R,y^{k}_\R)$. Since by \eqref{eq:split_get_rid_of_isolated_vertices} adding $C$ does not affect the right part, we have that $(J^{k}_\R,y^{k}_\R)$ selects from $\bipreduced{(\bipreduced{G} \cupdot C)}$ the same set $\selectSet[\R]$ it selects from $\bipreduced{G}$. Let $\selectSet[\L]$ and $\selectSet[\L][']$ be the sets of vertices selected by $(J^{k}_\L,y^{k}_\L)$ from $\bipreduced{G}$ and $\bipreduced{(\bipreduced{G} \cupdot C)}$, respectively. If $\selectSet[\L]$ contains no isolated vertex, then by \eqref{eq:split_get_rid_of_isolated_vertices} adding $C$ has no effect on the vertices selected, that is $\selectSet[\L]$ and $\selectSet[\L][']$ are equal. In this case, $G'_{k-1}$ is equal to $G_{k-1}$ and the statement follows. Otherwise, $(J^{k}_\L,y^{k}_\L)$ can be assumed to consist only of the isolated vertex $y^k_\L$. Hence, $(J^{k}_\L,y^{k}_\L)$ selects all vertices in the left part, that is $\selectSet[\L]$ is equal to $\lpart[\bipreduced{G}]$ and $\selectSet[\L][']$ is equal to $\lpart[\bipreduced{G}\setminus \hat{C} ] \cupdot \bipreduced{(\hat{C} \cupdot C)}$. We conclude that $G_{k-1}$ is the induced sub-bip-graph of $\bipreduced{G}$ given by $\lpart[\bipreduced{G}]$ and $\selectSet[\R]$. Further, $G'_{k-1}$ is the induced sub-bip-graph of $\bipreduced{G\cupdot C}$ given by $\rpart[G_{k-1}]$ and $\lpart[\bipreduced{(\bipreduced{(G_{k-1})} \cupdot C)} ]$. Since $C$ consists of only isolated vertices in the left part, adding $C$ to $G_{k-1}$ has no effect on the bip-orbits of vertices in $\selectSet[\R]$. It follows that $\bipreduced{(G'_{k-1})}$ is \isomorphic[bip] to $\bipreduced{(\bipreduced{(G_{k-1})} \cupdot C)}$. The lemma follows from the induction hypothesis on $G_{k-1}$.
\end{proof}	
Even though the \gadget{B,p} given in Observation~\ref{obs:hard_vertex} does not satisfy that $B$ is exactly the $2$-neighbourhood of a given vertex $v$, we still refer to it as the \gadget{\twoneigh{v}, p} despite the isolated vertices. Similar, for $B$ the intersection of $2$-neighbourhoods of any number of vertices.

We recall Definition~\ref{def:selection}. For a graph $H$ and a set of vertices $V'$ of $H$, we frequently construct a partially $H$-labelled graph $(J,y)$ that selects the common neighbourhood of the vertices in $V'$. 
\begin{lemma}\label{lem:complete_core_gadget}
	Let $H$ be a bip-graph. If $\set{u_1, \dots, u_k}$ is a set of vertices of $H$, then there exists a partially $H$-labelled bip-graph $(J, y)$ that selects the set of common neighbours of $\set{u_1, \dots, u_k}$, that is
	\begin{align*}
		\selectSet = \set{v \in \vertexset[H] \given \numHomBip[(J,y), (H,v)]=1 } = \bigcap_{i \in \sqBrackets{k}} \neigh{u_i}.
	\end{align*}
\end{lemma}
\begin{proof}
	For every index $i \in \sqBrackets{k}$, we take the partially $H$-labelled bip-graph $(J_i, y_i)$ that consists of the simple edge $(x_i, y_i)$ and the partial labelling $\tau_i \colon x_i \mapsto u_i$, which fixes the bipartition of $J_i$.
	The bip-graph $(J_i, y_i)$ selects the set of neighbours $\neigh{u_i}$, where, for every neighbour $v$ in $\neigh{u_i}$, there is exactly one homomorphism in $\HomBip[(J_i, y_i), (H,v)]$.	
	The partially $H$-labelled bip-graph $(J, y)$ is then given by the dot product 
	\[
		(J, y) = \bigodot_{i \in \sqBrackets{k}} (J_i, y_i),
	\]
	and the result follows from Corollary~\ref{cor:dot_product_bip}.
\end{proof}

We recall that, for a graph $H$ and a tuple of vertices $\vector{v}$ in $\vertexset[H]$, the stabilizer $\Stab{\vector{v}}[\Aut]$ of $\vector{v}$ is  the subgroup of automorphisms $\varrho \in \Aut[H]$ that fix $\vector{v}$, i.e. $\varrho(\vector{v})= \vector{v}$.

\begin{lemma}\label{lem:bound_tuple_same_neighbourhood}
	Let $p$ be a prime and $H$ be a bip-graph. If $H$ is order~$p$ bip-reduced, then there exists no tuple of at least $p$ twins in $\vertexset[H]$.
\end{lemma}
\begin{proof}
	We assume toward a contradiction that there exists a tuple $V'$ of $p$ twins in $\vertexset[H]$. Let $V'$ be $\set{v_1, \dots, v_k}$. The automorphism $\varrho$, that shifts the elements of $V'$ by one and fixes every vertex $v'$ in the complement $\vertexset[H] \setminus V'$ is an automorphism of order $p$. Moreover, since $\varrho$ is in the stabilizer of the complement $\vertexset[H] \setminus V'$, it preserves the order of the bipartition of $H$. Hence, $\varrho$ is in $\AutBip[H]$, a contradiction.
\end{proof}

For an order~$p$ bip-reduced graph, Lemma~\ref{lem:bound_tuple_same_neighbourhood} directly yields the following result to find vertices $w$ whose neighbourhood is not contained in the complete core under study. 

\begin{observation}\label{obs:nb_outside_core}
	Let $p$ be a prime and let $H$ be an order~$p$ reduced bip-graph that is \graphclass{}. If $H$ contains a pair of adjacent vertices $v$ and $u$ such that $K^{v,u} \congbip K_{a,b}$ and $b\geq p$, then there exists a vertex $w$ in $\partof[K^{v,u}](u)$ with $\deg (w) > a$.
\end{observation}

In the same way, for a connected bip-graph $H$, any automorphism $\varrho$ of $H$ that stabilizes at least one vertex of $H$ has to be in $\AutBip[H]$. This motivates the following general observation.
\begin{observation}\label{obs:AutId_and_stabilizer}
	Let $p$ be a prime and let $H$ be an order~$p$ bip-reduced bip-graph. If $H$ is connected and $v$ is a vertex of $H$, then there exists no automorphism $\varrho$  in $\Aut[H]$ of order~$p$ such that $\varrho$ is in $\Stab{v}[\Aut]$.
\end{observation}

We need the following structural result for graphs of radius at most $2$. We recall that, for a graph $H$, a subgraph $K$, and a vertex $v$ of $K$, $\neigh{v}$ denotes the neighbourhood of $v$ in $H$.

\begin{lemma}\label{lem:small_graphs_components_size_left}
	Let $p$ be a prime and $H$ be an order~$p$ bip-reduced and \graphclass{} bip-graph of radius at most $2$ that is not complete bipartite. If $v$ and $u$ are adjacent vertices of $H$ such that $H$ is equal to $\twoneigh{v}$, 
	then, for the component $U^{v,u}$ denoted $U$ and its complete core $K^{v,u}$ denoted $K$,
	\begin{enumerate}
		\item $K^{v,u}$ is \isomorphic[bip] to $K_{a,b}$;
		\item $a \in \sqBrackets{p}$ and, for all vertices $v'$ different from $v$ but in $\partof[K][v]$, the neighbourhood $\neigh{v'}$ is equal to both $\partof[U](u)$ and $\partof[K](u)$;
		\item if $\twoneigh{v}$ consists of the single component $U$, then $a \in \sqBrackets{2; p-1}$, $b \geq 2$, and there exists a neighbour $w$ of $v$ such that $\neigh{u} \not\subseteq \partof[K](v)$; 
		\item if $b \geq$, then, for every tuple $S_p$ of $p$ vertices in $\partof[K](u)$, there exists a vertex $w$ in $S_p$ such that $\neigh{u}$ is not contained in $K$;
		\item for every vertex $w$ in $\partof[K](u)$ with $\neigh{w} \not\subseteq \partof[K](v)$, the complement $\neigh{w} \setminus \partof[K](v)$ is of cardinality in $\sqBrackets{p-1}$ and consists of only leaves;
		\item if $a<p$, then $U$ is order~$p$ bip-reduced and \graphclass{} with radius at most $2$.
	\end{enumerate}
\end{lemma}
\begin{proof}
	By the radius of $H$, every vertex $w$ in $\npartof[H](v)$ has to be a neighbour of $v$ and $H$ has to be connected.
	The first property follows from Lemma~\ref{lem:weirdness}. Regarding the second property, let $v'$ be a vertex in $\partof[K](v)$ different from $v$. Every neighbour $w$ of $v'$ has to be adjacent to $v$ by the radius of $H$. If $w$ is not in $K$, then the graph $K\cupdot\set{w}$ is a larger $2$-connected graph, contradicting the maximality of $K$. We assume toward 
	contradiction that $a > p$ and find in $\partof[K](v)$ a set of $p$ vertices, all distinct from $v$. It follows that $H$ contains a tuple of $p$ twins, a contradiction by Lemma~\ref{lem:bound_tuple_same_neighbourhood}.
	
	Regarding the third property, we observe that every neighbour of $v$ is in $K$ because the split of $H$ at $v$ consists of a single component. Hence, every vertex in $\partof[K](v)$ is a twin of $v$. By Lemma~\ref{lem:bound_tuple_same_neighbourhood} we obtain that $a <p$. If $b=1$, then $H$ is a star, which contradicts $H$ not being a complete bipartite graph. Similarly, if $a=1$ then $K$ is equal to $U$. Since $H$ is equal to $U$, we obtain that $H$ has is the single edge $(v,u)$, which gives $b=1$, and by the previous argument a contradiction. Lastly, assuming that the neighbourhood of every vertex $w \in \partof[K](u)$ satisfies is contained in $K$ contradicts that $H$ is complete bipartite.
	
	The fourth property follows from  Lemma~\ref{lem:bound_tuple_same_neighbourhood} similar to Observation~\ref{obs:nb_outside_core}.
	Regarding the fifth property, let $w \in \partof[K](u)$ be a vertex whose neighbourhood is not contained in $K$, and we denote the complement $\neigh{w} \setminus K$ by $\L_u$. If $\L_u$ contains a vertex $v'$ with $\deg(v')>1$, then this contradicts the maximality of $K$. Hence, $\L_u$ consists of only leaves. In fact, $\L_u$ consists of twins, and by Lemma~\ref{lem:bound_tuple_same_neighbourhood} it follows $\abs{\L_u} < p$.
	
	Regarding the last property, $U$ is an induced sub-bip-graph of $H$ and thus a \graphclass{} bip-graph of radius at most $2$. It suffices to show that $U$ is order~$p$ bip-reduced as well. 
	For any automorphism $\varphi$ in $\AutBip[U]$, we observe that $\varphi(v)$ is in $\partof[K](v)$. By the Orbit-Stabilizer Theorem~\ref{thm:orbit_stabilizer}, $\varphi$ in $\AutBip[U]$ decomposes into an automorphism $\varphi_v$ in $\AutBip[U]$ that does not fix $v$, and an automorphism $\varrho$ in $\Aut[U] \cap \Stab{v}[\Aut]$ that fixes $v$. In particular, since $\varrho$ fixes $v$, it follows that $\varrho$ is in $\AutBip[H]$. We deduce $\varphi= \varphi_v \circ \varrho$. By assumption, $\abs{\lpart[K]}<p$, and thus $\Ord(\varphi_v)$ is not congruent modulo $p$ to $0$. Further, $\varrho$ is in $\AutBip[H]$. It follows that $\Ord(\varrho)$ is not congruent modulo $p$ to $0$. We conclude that also $\Ord(\varphi)$ is not congruent modulo $p$ to $0$.
\end{proof}

One structure we are going to encounter frequently also beyond small graphs is studied in the following lemma.
\begin{lemma}\label{lem:gadget_adjacent_components}
	Let $p$ be a prime and $H$ be an order~$p$ bip-reduced bip-graph. Further, let $H$ contain a vertex $v$ and the split of $\twoneigh{v}$ contain two distinct components $U^{v,u_1}$ and $U^{v,u_2}$, where $u_1$ and $u_2$ are vertices adjacent to $v$. If $\deg(v) \not \equiv 0 \pmod p$ and, for both indices $i \in \sqBrackets{2}$, the complete core $K^{v,u_i}$ satisfies $\abs{\partof[K^{v,u_1}](v)} \not \equiv 1 \pmod p$ and $\abs{\npartof[K^{v,u_1}](v)} \not \equiv 0 \pmod p$, then $H$ admits a $p$-hardness gadget.
\end{lemma}
\begin{proof}
	Let $\GadgetPart{\L}$ and $\GadgetPart{\R}$ be the partially $H$-labelled bip-graphs given by Lemma~\ref{lem:complete_core_gadget} such that the set $\selectSet[\L]$ of vertices selected by $\GadgetPart{\L}$ is $\partof[K^{v,u_1}](v)$ and the set $\selectSet[\R]$ of vertices selected by equal $\GadgetPart{\L}$ is $\partof[K^{v,u_2}](v)$. We denote $o_\L = o_\R = \set{v}$, $i_\L = \selectSet[\L] \setminus \set{v}$, and $i_\R = \selectSet[\R] \setminus \set{v}$. Since by assumption $\partof[K^{v,u_1}](v)$ and $\partof[K^{v,u_2}](v)$ are of cardinality not congruent modulo $p$ to $1$ we deduce that none of the sets $i_\L, i_\R, o_\L, o_\R$ is of cardinality congruent modulo $p$ to $0$. It remains to construct a suitable partially $H$-labelled bip-graph $\GadgetEdge$.
	
	Let $(J_v, y_v)$ be the partially $H$-labelled bip-graph that selects $\neigh{v}$ given by by Lemma~\ref{lem:complete_core_gadget}, and let $P_2$ be the path $(y_\L, z, y_\R)$. We fix the bipartition of $P_2$ such that $\partof[P_2](y_\L)$ and $\partof[G(J_\L)](y_\L)$ agree. Now, we take the dot product $(P_2, z) \odot (J_v, y_v)$, which we denote by $(J_E, z)$. In the second step, we distinguish in $J_E$ the vertices $y_\L$ and $y_\R$, which gives $\GadgetEdge$. We refer to Corollary~\ref{cor:dot_product_bip} for the used properties of the dot product. For every pair $(w_1, w_2)$ in $\selectSet[\L] \times \selectSet[\R]$, the number of homomorphisms $\numHomBip[\GadgetEdge, (H, w_1, w_2)]$ is given by the common neighbours of $v$, $w_1$, and $w_2$. If $w_1$ and $w_2$ are equal to $v$, then the number of common neighbours is $\deg(v) \not \equiv 0 \pmod p$. If $w_1$ is in $i_\L$ and $w_2$ is equal to $v$, then the number of common neighbours is $\abs{\npartof[K^{v,u_1}](v)} \not \equiv 0 \pmod p$. Similarly, if $w_1$ is equal to $v$ and $w_2$ is in $i_\R$ because $\abs{\npartof[K^{v,u_1}](v)} \not \equiv 0 \pmod p$. Lastly, if $(w_1, w_2)$ is in $i_\L \times i_\R$, then there cannot exist a common neighbour $z$ in $\neigh{v} \cap \neigh{w_1} \cap \neigh{w_2}$ because $U^{v,u_1}$ and $U^{v,u_2}$ are distinct.
\end{proof}

With these technical results at hand, we now show that every order~$p$ bip-reduced bip-graph $H$ that is \graphclass and of radius at most $2$ admits a $p$-hardness gadget if it is not complete bipartite. The proof is split into three lemmas that build upon each other in sequence. 

\begin{lemma}\label{lem:small_graphs_single_component}
	Let $p$ be a prime and $H$ be an order~$p$ bip-reduced bip-graph that is \graphclass{}, of radius at most $2$, and not complete bipartite. If there exists a vertex $v$ of $H$ such that $H$ is equal to $\twoneigh{v}$ and the split of $H$ at $v$ consists of a single component, then $H$ admits a $p$-hardness gadget.
\end{lemma}
\begin{proof}	
	The neighbourhood of $v$ is equal to $\npartof[H](v)$ as every vertex of $H$ is of distance at most $2$ from $v$.
	For every neighbour $u$ of $v$, the graph $H$ is equal to the component $U^{v,u}$ that we denote by $U$. The complete core $K^{v,u}$ is thus uniquely given, and we denote it by $K$. We recall Lemma~\ref{lem:weirdness}, and let $K$ be \isomorphic[bip] to $K_{a,b}$. We apply properties $3$ and $5$ of Lemma~\ref{lem:small_graphs_components_size_left}, which yield for the first count $a \in \sqBrackets{2;p-1}$, for the other count $b \geq 2$, and there exists a neighbour $u_1$ of $v$ with $\neigh{u_1} \not\subseteq K$. In particular, $\neigh{u_1}$ decomposes into $\partof[K](v)$ and the complement $\neigh{u_1} \setminus \partof[K](v)$ denoted $\L_1$. This complement $\L_1$ consists of only leaves and has cardinality $\abs{\L_1}$ in $\sqBrackets{p-1}$.
	
	We distinguish cases.
	\begin{enumerate}
		\item $b \not\equiv 1 \pmod p$.
		The case is illustrated with the first two figures in Figure~\ref{fig:small_graph_example_1}.
		
		We observe that $\twoneigh{u_1}$ is the thick $4$-vertex path $(v_0^{\abs{\L_1}},v_1^1,v_2^a,v_3^{b-1})$, which contains no entry congruent modulo $p$ to $0$. By Corollary~\ref{cor:gen_path_reduction_BIS} and Lemma~\ref{lem:2-neighbourhood_isolated_vertices_gadgetry}, $H$ admits a $p$-hardness gadget.
		
		\item $b \equiv 1 \pmod p$.
		The case is illustrated in the last figure in Figure~\ref{fig:small_graph_example_1}.
		
		Since $2 \leq b$, it follows that $b$ is at least $p+1$. By property~4 of Lemma~\ref{lem:small_graphs_components_size_left}, there exist a vertex $u_2$ in $\npartof[H](v)$ that is distinct from $u_1$, such that $\neigh{u_2} \not\subset K$. Similar to $u_1$, the neighbourhood of $u_2$ decomposes into $\partof[K](v)$ and the complement $\neigh{u_2} \setminus \partof[K](v)$ denoted $\L_2$, where the latter consists of only leaves and has cardinality $\abs{\L_2}$ in $\sqBrackets{p-1}$. 
		
		We recall Lemma~\ref{lem:complete_core_gadget} and construct the following $p$-hardness gadget:
		\begin{itemize}
			\item $\GadgetPart{\L}$ selects the set $\neigh{u_1}$ denoted by $\selectSet[\L]$;
			\item $\GadgetPart{\R}$ selects the set $\neigh{u_2}$ denoted by $\selectSet[\R]$;
			\item $\GadgetEdge$ selects walks of length $2$ by taking $G(J_E)$ to be the path $(y_\L, z, y_\R)$ and distinguishing $y_\L$ and $y_\R$ before we fix the bipartition of $J_E$ such that $\partof[G(J_E)](y_\L)$ and $\partof[G(J_\L)](y_\L)$ agree.
		\end{itemize}
		We set both $o_\L$ and $o_\R$ to be $\partof[K](v)$, the set $i_\L$ to be $\L_1$, and the set $i_\R$ to be $\L_2$. The cardinality $\abs{\lpart[K]}$ is $a$, which is not congruent modulo $p$ to $0$. Both the complement $\L_1$ and the complement $\L_2$ are of cardinality in $\sqBrackets{p-1}$. Now, every pair in $\L_1 \times \L_2$ is of distance larger than $2$. For every pair in $\selectSet[\L] \times \selectSet[\R]$ that is not in $\L_1 \times \L_2$, the number of common neighbours is $\abs{\npartof[K](v)}$, which is equal to $b$. Since $b\equiv 1\pmod p$, the gadget $\GadgetEdge$ yields that only the pairs in $i_\L \times i_\R$ are not in $\selectSet[E]$. \qedhere
	\end{enumerate}	
\end{proof}

\begin{figure}[t]
	\centering
	\includegraphics[scale=1]{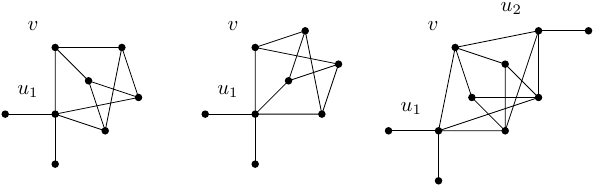}
	\caption{Illustration for Lemma~\ref{lem:small_graphs_single_component} with $p=3$. From left to right: illustration of the first case, rearrangement to $4$-vertex path, illustration of the second case.}
	\label{fig:small_graph_example_1}
\end{figure}

\begin{lemma}\label{lem:small_graphs_multiple_components_not_0_deg}
	Let $p$ be a prime and $H$ be an order~$p$ bip-reduced bip-graph that is \graphclass{} and of radius at most $2$. If there exists a vertex $v$ of $H$ such that the graph $H$ is equal to $\twoneigh{v}$ and $\deg(v)$ is not congruent modulo $p$ to $0$, then $H$ admits a $p$-hardness gadget.
\end{lemma}
\begin{proof}
	Let the split of $H$ at $v$ consist of the set of equivalence classes $\set{\eqclass{U_s}}_{s \in \sigma}$ with non-zero counts $\set{\alpha_s}_{s \in \sigma}$ and representatives $\set{U_s}_{s \in \sigma}$. If the split consists of only one component, then Lemma~\ref{lem:small_graphs_single_component} yields the result. We assume that the split contains at least two components, and we distinguish cases.
	\begin{enumerate}
		\item There exists a component $U^{v,u}$ given by a neighbour $u$ of $v$, such that the complete core satisfies $K^{v,u} \congbip K_{a, b}$ with $b \equiv  0 \pmod p$. The case is illustrated in the first figure in Figure~\ref{fig:small_graph_example_2}. 
		
		By property~4 of Lemma~\ref{lem:small_graphs_components_size_left}, we assume without loss of generality that $\neigh{u} \not \subseteq K$, and we denote the complement $\neigh{u} \setminus \partof[K^{v,u}](v)$ by $\L_u$. This complement consists of leaves and is of cardinality in $\sqBrackets{p-1}$.	
			
		We construct a \gadget{B,p}, where $B$ is an induced sub-bip-graph of $H$ and a thick $4$-vertex path. Let $\GadgetPart{\R}$ be the partially $H$-labelled bip-graph obtained by Lemma~\ref{lem:complete_core_gadget} that selects the set $\partof[K^{v,u}](u)$ denoted by $\selectSet[\R]$. We construct the partially $H$-labelled bip-graph $\GadgetPart{\L}$ in two steps.
		
		First, let $P$ be the bip-graph consisting of the single edge $(y_\L, x)$ and bipartition such that $\partof[P](y_\L)$ and $\partof[H](v)$ agree. We obtain the bip-graph $(P_2,y_\L)$ with one distinguished vertex, and claim that $(P_2, y_\L)$ selects the set $\set{v' \in \partof[H](v) \given \deg(v') \not \equiv 0 \pmod p}$. 		
		For every vertex $v'$ in $\partof[H](v)$, the number of homomorphisms $\numHomBip[(P_2, y_\L), (H,v')]$ is given by the number of possible images of $x$ in $\neigh{v'}$, thus $\numHomBip[(P_2, y_\L),(H,v')] = \deg(v')$. The claim follows. 
		
		Second, let $(J, y_\L)$ be the partially $H$-labelled bip-graph obtained by Lemma~\ref{lem:complete_core_gadget} that selects $\neigh{u}$. We apply the dot product, and we obtain the partially $H$-labelled bip-graph $(P_2, y_\L) \odot (J, y_\L)$ denoted by $\GadgetPart{\L}$. By Corollary~\ref{cor:dot_product_bip}, $\GadgetPart{\L}$ selects the set $\set{w \in \neigh{u} \given \deg(w) \not \equiv 0 \pmod p}$ denoted by $\selectSet[\L]$. By property~2 of Lemma~\ref{lem:small_graphs_components_size_left}, $v$ is the only vertex in $\neigh{u}$ with possible neighbours not in $K^{v,u}$. Due to the assumptions of $\deg(v) \not \equiv 0 \pmod p$ and $\partof[K^{v,u}](u) = b \equiv 0 \pmod p$, we obtain $\selectSet[\L] = \L_u \cupdot \set{v}$.
			
		We deduce that $B$ is the induced subgraph of $H$ given by $\selectSet[\L]$ and $\selectSet[\R]$. In particular, $B$ is the $4$-vertex path $(v_0^{\abs{\L_u}}, v_1^1, v_2^1, v_3^{b-1})$ given by the leaves $\L_u$, the vertex $u$, the vertex $v$, and the complement $\partof[K^{v,u}](u) \setminus \set{u}$. Since $b$ is larger than $1$, none of the entries in $B$ is congruent modulo $p$ to $0$. By Corollary~\ref{cor:gen_path_reduction_BIS}  and Corollary~\ref{cor:gadget_reduction}, $H$ admits a $p$-hardness gadget.
		
		\item There exists a pair of distinct components $U^{v,u_1}$ and $U^{v,u_1}$ given by a pair of distinct neighbours $u_1$ and $u_2$ of $v$, such that the complete cores satisfy $K^{v,u_1} \congbip K_{a_1, b_1}$ and $K^{v,u_1} \congbip K_{a_2, b_2}$ with $b_1, b_2 \not \equiv 0 \pmod p$ as well as $a_1 , a_2 > 1$.	The case is illustrated in the second figure in Figure~\ref{fig:small_graph_example_2}.
		
		By Lemma~\ref{lem:small_graphs_components_size_left} it follows that $a_1, a_2 \not \equiv 1 \pmod p$. Lemma~\ref{lem:gadget_adjacent_components} yields a $p$-hardness gadget.

		\item The first two cases do not apply. The case is illustrated in the third figure in Figure~\ref{fig:small_graph_example_2}.
		
		Since the first case does not apply, for every neighbour $u$ of $v$, the component $U^{v,u}$ has complete core $K^{v,u}$ that satisfies $K^{v,u} \congbip K_{a, b}$ with $b \not\equiv 0 \pmod p$. Since the second case does not apply, we have that there exist at most one component $U^{v,u_1}$, where $u_1$ is a neighbour of $\neigh{v}$, with complete core $K^{v,u_1}$ that satisfies $K^{v,u_1} \congbip K_{a_1,b_2}$ as well as $a_1 > 1$ and $b_1 \not\equiv 0 \pmod p$. In this case, it follows for the count $\alpha_1$ of the equivalence class of $U^{v,u_1}$ that $\alpha_1 = 1$. We briefly argue that every component $U$ not equal to $U^{v,u_1}$ is an edge.
		
		For a neighbour $u_2 $ of $v$ with $U=U^{v,u_2}$, the complete core $K^{v,u_2}$ is \isomorphic[bip] to $K_{a_2,b_2}$, where the counts $a_2$ and $b_2$ satisfy $b_2 \not\equiv 0 \pmod p$ and $a_2 \leq 1$. Hence, $a_2$ is equal to $1$. 
		
		By assumption, $H$ is not a star, and thus there exists one component $U_1$ that is not a star. Let $u_1$ be a neighbour of $v$ such that $U_1 = U^{v,u_1}$ and the complete core $K^{v,u_1}$ is \isomorphic[bip] to $K_{a_1, b_1}$. It follows $a_1 > 1$ and $b_1 \not \equiv 0 \pmod p$, and we denote $K^{v,u_1}$ by $K_1$. There are by assumption at least two components in the split of $H$ at $v$. We deduce that the number $\sigma$ of equivalence classes in the split satisfies $\sigma = 2$, where the second representative $U_2$ is \isomorphic[bip] to $K_{1,1}$. The components in the equivalence class $\eqclass{U_2}$ give a tuple of $\alpha_2$ leaves adjacent to $v$. Hence, the count $\alpha_2$ is in $\sqBrackets{p-1}$ by Lemma~\ref{lem:bound_tuple_same_neighbourhood}.
		
		By Lemma~\ref{lem:complete_core_gadget}, we construct the following \gadget{B,p} for $H$, where $B$ is an induced sub-bip-graph of $H$ and a thick $4$-vertex path.
		\begin{itemize}
			\item $\GadgetPart{\L}$ is the partially $H$-labelled bip-graph that selects the set $\partof[K_1](v)$ denoted $\selectSet[\L]$.
			\item $\GadgetPart{\R}$ is the partially $H$-labelled bip-graph that selects the set $\neigh{v}$ denoted $\selectSet[\R]$. 
		\end{itemize}
		In this way, $B$ is the $4$-vertex path $(v_0^{\alpha_2}, v_1^1, v_2^{b_1}, v_3^{a_1 -1})$ given by the $\alpha_2$ leaves adjacent to $v$, the vertex $v$, the $b_1$ neighbours of $v$ in $K_1$, and the remaining vertices in $\partof[K_1](v) \setminus \set{v}$. This is illustrated in the last figure in Figure~\ref{fig:small_graph_example_2}. By Lemma~\ref{lem:small_graphs_components_size_left}, $a_1 \leq p$, thus none of the counts is congruent to modulo $p$ to $0$. By  Corollary~\ref{cor:gen_path_reduction_BIS} and Corollary~\ref{cor:gadget_reduction}, the case follows.\qedhere
	\end{enumerate}
\end{proof}

\begin{figure}[t]
	\centering
	\includegraphics[scale=1]{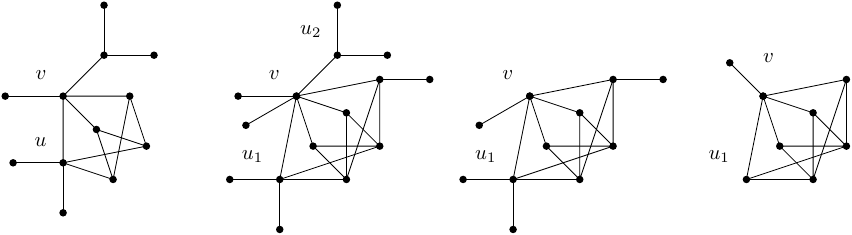}
	\caption{Illustration for Lemma~\ref{lem:small_graphs_multiple_components_not_0_deg} with $p=3$. From left to right: first case, second case, third case, and $4$-vertex path constructed in the third case.}
	\label{fig:small_graph_example_2}
\end{figure}

\begin{lemma}\label{lem:small_graphs_multiple_components_0_deg}
	Let $p$ be a prime and $H$ be an order~$p$ bip-reduced bip-graph that is \graphclass{} and of radius at most $2$. If there exists a vertex $v$ of $H$ such that the graph $H$ is equal to $\twoneigh{v}$ and $\deg(v)$ is congruent modulo $p$ to $0$, then $H$ admits a $p$-hardness gadget.
\end{lemma}
\begin{proof}
	Let the split of $H$ at $v$ consist of the set of equivalence classes $\set{\eqclass{U_s}}_{s \in \sigma}$ with non-zero counts $\set{\alpha_s}_{s \in \sigma}$ and representatives $\set{U_s}_{s \in \sigma}$. If the split consists of only one component, then Lemma~\ref{lem:small_graphs_single_component} yields the result. We assume that the split contains at least two components, and we distinguish cases.
	\begin{enumerate}
		\item There exists a neighbour $u \in \neigh{v}$ with complete core $K^{v,u}$, such that $K^{v,u}\congbip K_{a, b}$ with $a > 1$ and $b \not \equiv 0 \pmod p$. The case is illustrated in the first figure in Figure~\ref{fig:small_graph_example_3}.
		
		We shorten notation and denote $U^{v,u}$ by $U$ and $K^{v,u}$ by $K$. By property~2 of Lemma~\ref{lem:small_graphs_components_size_left}. it follows $a \not\equiv 1 \pmod p$. We recall Lemma~\ref{lem:complete_core_gadget} and construct the following \gadget{B,p} for $H$, where $B$ an induced sub-bip-graph of $H$ and also a $4$-vertex path.
		\begin{itemize}
			\item $\GadgetPart{\L}$ is a partially $H$-labelled bip-graph that selects the set $\partof[K](v)$ denoted $\selectSet[\L]$.
			\item $\GadgetPart{\R}$ is a partially $H$-labelled bip-graph that selects the set $\neigh{v}$ denoted $\selectSet[\R]$.
		\end{itemize}
		In this way, $B$ is the thick $4$-vertex path $(v_0^{\deg(v) - b}, v_1^1, v_2^b, v_3^{a-1})$ given by the neighbours of $v$ not in $K$, the vertex $v$, the neighbours of $v$ in $K$, and the complement $\partof[K](v) \setminus \set{v}$. The bip-graph $B$ is illustrated with the second figure in Figure~\ref{fig:small_graph_example_3}. None of the counts is congruent modulo $p$ to $0$ because the degree satisfies $\deg(v) \equiv 0 \pmod p$ and the counts satisfy $b \not\equiv 0 \pmod p$ and $a\not\equiv 1 \pmod p$. Corollary~\ref{cor:gen_path_reduction_BIS} and Corollary~\ref{cor:gadget_reduction} give the result.
		
		\item The first case does not apply. \\
		By the split of $H$ at $v$, we obtain
		\[
		\deg(v)\equiv\sum_{s\in \sigma} \alpha_s \cdot \deg_{U_s} (v) \pmod p.
		\]
		For every neighbour $u$ of $v$, the complete core $K^{v,u}$ with  $K^{v,u} \congbip K_{a,b}$ satisfies that, if $a > 1$, then $b \equiv 0 \pmod p$. The degree of $v$ in the component $U^{v,u}$ is equal to $b$. Hence, for every such component $U^{v,u}$, it follows that, if $a > 1$, then $\deg_{U^{v,u}} (v) \equiv 0 \pmod p$. 
		
		Every other component is a single edge. The reason is that, if $b \not\equiv 0 \pmod p$, then $a=1$. We deduce that $\deg(v)$ is congruent modulo $p$ to the count of components that are edges. Every such component corresponds to a leaf adjacent to $v$. By assumption, $\deg(v) \equiv 0 \pmod p$, and by Lemma~\ref{lem:bound_tuple_same_neighbourhood} we deduce that, for every neighbour $u$ of $v$, the complete core $K^{v,u}$ with $K^{v,u}\congbip K_{a,b}$ satisfies $b \equiv 0 \pmod p$ and $a \in \sqBrackets{2,p}$, where the latter follows from the definition of a component and Lemma~\ref{lem:small_graphs_components_size_left}. We distinguish the following two cases.
		
		\begin{enumerate}
			\item There exists a component $U^{v,u}$ with complete core $K^{v,u}$ that satisfies $K^{v,u} \congbip K_{a,b}$ with $a \neq p$. The case is illustrated in the third figure in Figure~\ref{fig:small_graph_example_3}. 
			
			We denote $K^{v,u}$ by $K$ and $U^{v,u}$ by $U$. By Lemma~\ref{lem:small_graphs_components_size_left} and $b\equiv 0 \pmod p$, the graph $U$ is not complete bipartite. We assume without loss of generality that the neighbourhood of $u$ satisfies $\neigh{u} \not\subseteq\lpart[K]$.
			Since $a$ is lesser than $p$, the component $U$ is by Lemma~\ref{lem:small_graphs_components_size_left} order~$p$ bip-reduced. Let $v'$ be a vertex in $\partof[U][v]$ different from $v$. We observe that $\twoneigh{v'}$ is equal to $U$ and the split of $U$ at $v'$ consists of the single component $U$. Hence, $U$ satisfies the prerequisites of Lemma~\ref{lem:small_graphs_single_component} and thus $U$ admits a $p$-hardness gadget. By Lemma~\ref{lem:2-neighbourhood_isolated_vertices_gadgetry}, $H$ admits a $p$-hardness gadget.
			
			\item Every component $U^{v,u}$ with complete core $K^{v,u}$ \isomorphic[bip] to $K_{a,b}$ satisfies $a = p$. The case is illustrated in the fourth figure in Figure~\ref{fig:small_graph_example_3}.
			
			Let $u$ be an arbitrary neighbour of $v$. We consider the component $U^{v,u}$ and the complete core $K^{v,u}$. The latter is \isomorphic[bip] to $K_{a,b}$, which by assumption satisfies $a=p$ and $b \equiv 0 \pmod p$. We denote $U^{v,u}$ by $U$ and $K^{v,u}$ by $K$. By Lemma~\ref{lem:small_graphs_components_size_left}, we assume without loss of generality that $\neigh{u} \not\subseteq \lpart[K]$. The complement $\neigh{u} \setminus \partof[K](v)$ is denoted $\L_u$, which has cardinality in $\sqBrackets{p-1}$ and consists of only leaves. This yields $\deg(u) = a + \abs{\L_u} \not \equiv 0 \pmod p$. We construct a partially $H$-labelled bip-graph $\GadgetPart{\R}$ that selects every neighbour of $v$ except those in $\partof[K](u) \setminus \set{u}$. The illustration in Figure~\ref{fig:small_graph_example_3} has vertices not selected by $\GadgetPart{\R}$ in tiny size.
			
			Let $G(J_\R)$ be the path $(y_\R, z, x)$ and $\tau_\R$ be the partial labelling $\tau_\R \colon x \mapsto u$, which induces a fixed a bipartition of $G(J_\R)$. The last figure in Figure~\ref{fig:small_graph_example_3} depicts $\GadgetPart{\R}$. For every vertex $w$ of $H$, every bip-homomorphism in $\HomBip[\GadgetPart{\R}, (H,w)]$ maps $z$ to $\neigh{u}$ and $y_\R$ to $\neigh{v}$. 			
			Hence, $\numHomBip[\GadgetPart{\R}, (H,w)]$ is determined by the number of common neighbours of $w$ and $u$, this number is $0$ if $w$ is not in $\neigh{v}$. For $w \in \neigh{v} \setminus K$, only $v$ is in the common neighbourhood. For $w \in \partof[K](u) \setminus \set{u}$, the common neighbourhood is given by $\partof[K](v)$, which is of cardinality $a$ with $a \equiv 0 \pmod p$. For $w$ equal to $u$, the common neighbourhood is $\neigh{u}$ and of cardinality $\deg(u)$ with $\deg(u) \not \equiv 0 \pmod p$. This yields the claim that $\GadgetPart{\R}$ selects the set $(\neigh{v} \setminus \rpart[K]) \cupdot \set{u}$ denoted $\selectSet[\R]$.
			
			For $\GadgetPart{\L}$ we take the trivial bip-graph consisting of the single vertex $y_\L$ without partial labelling, where we fix the bipartition by assigning $\partof[G(J_\L)](y_\L)$ to agree with $\partof[H](v)$. In conjunction with $\GadgetPart{\R}$ we obtain a \gadget{B,p}, where $B$ is the induced subgraph of $H$ obtained by deleting every vertex in $\rpart[K] \setminus \set{u}$. We claim that $B$ satisfies the prerequisites of Lemma~\ref{lem:small_graphs_multiple_components_not_0_deg}, which establishes the lemma.
			
			The degree of $v$ in $B$ is given by $ \deg(v) - b + 1 \not \equiv 0 \pmod p$. We recall that every component in the split of $H$ at $v$ different from $U$ has a $2$-connected complete core and there is at least one such component. All of these are also components in the split of $B$ at $v$. Hence, the only bip-automorphism of $B$ not in $\AutBip[H]$ is given by the $\deg(u)- 1$ leaves adjacent to $u$. It follows that $B$ is order~$p$ bip-reduced, the split of $B$ at $v$ contains at least two components, and $B$ is not complete bipartite. \qedhere
		\end{enumerate}	
	\end{enumerate}
\end{proof}

\begin{figure}[t]
	\centering
	\includegraphics[scale=1]{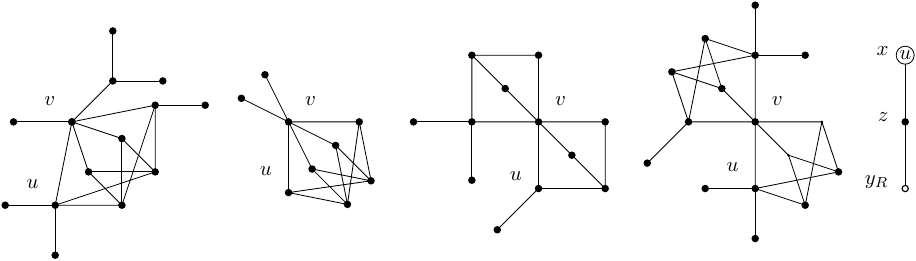}
	\caption{Illustration for Lemma~\ref{lem:small_graphs_multiple_components_0_deg} with $p=3$. From left to right: first case, constructed $4$-vertex path in the first case, second case a, second case b with vertices \enquote{cancelled} by $\GadgetPart{\R}$ in tiny size, and $\GadgetPart{\R}$ constructed in the third case.}
	\label{fig:small_graph_example_3}
\end{figure}

Lemma~\ref{lem:small_graphs_multiple_components_not_0_deg} and Lemma~\ref{lem:small_graphs_multiple_components_0_deg} cover all possible cases, which gives the main result of this subsection.
\begin{theorem}\label{thm:small graphs_weird_hard}
	Let $p$ be a prime and $H$ be an order~$p$ bip-reduced bip-graph that is \graphclass{}. If $H$ is of radius at most $2$ and not complete bipartite, then $H$ admits a $p$-hardness gadget.
\end{theorem} 

\subsection{Hardness for \Graphclass{} Graphs of Large Radius}
In the remainder of this section, we consider only connected graphs. Previously, this assumption was intrinsic when we studied graphs of radius at most $2$. Lemma~\ref{lem:bip_components} allows us to project hardness for connected components into hardness for the whole graph. We note that our structural findings and arguments do not need the assumption of a connected graph but the clarity benefits from it. Considering the sheer length of this section, we opted for this. Nevertheless, we encourage the curious reader to keep an open eye for arguments, which need connectivity.

Based on Theorem~\ref{thm:small graphs_weird_hard}, we use the following definition to identify graphs $H$ such that $\probNumHom{H}[p]$ is $\classNumP[p]$-complete.
\begin{definition}
	Let $p$ be a prime, $H$ be a bip-graph, and $v$ be a vertex of $H$. If there exists a possibly empty set $U$ of vertices in $\twoneigh{v} \setminus \set{v}$ such that, for $B$ the intersection graph $\bigcap_{u \in U} \twoneigh{u} \cap \twoneigh{v}$, the reduced form $\bipreduced{B}$ is not a collection of complete bipartite graphs, then $v$ is called a \emph{hard vertex}.
\end{definition}
By iteratively applying Lemma~\ref{lem:2-neighbourhood_isolated_vertices_gadgetry} and Theorem~\ref{thm:small graphs_weird_hard} we obtain that a hard vertex in $H$ implies a $p$-hardness gadget.
\begin{corollary}\label{cor:hard_vertex_hard}
	Let $p$ be a prime and $H$ be a connected order~$p$ bip-reduced bip-graph that is \graphclass{}. If $H$ contains a hard vertex, then $H$ admits a $p$-hardness gadget.
\end{corollary}

We recall the Definition~\ref{def:complete_core} of complete core, i.e., for a bipartite graph $H$ that is \graphclass{} and two adjacent vertices $u$ and $v$ of $H$, the component $U^{v,u}$ and the induced sub-bip-graph $K^{v,u}$. The order of the two vertices is important as $U^{v,u}$ and $U^{u,v}$ do not have to be isomorphic. In the definition of $K^{v,u}$, the case $\abs{\neigh{v} \cap U^{v,u}}=1$ was handled separately, which is the only case in which $K^{v,u}$ is not $2$-connected. Only in this case are $K^{u,v}$ and $K^{v,u}$ not necessarily identical. For an illustration we refer to Figure~\ref{fig:complete_core_example}.

The following pair of results are explicitly stated for bip-graphs $H$ that are not necessarily order~$p$ bip-reduced. This allows us to use the recursive gadgetry without rendering intermediate parameter bip-graphs order~$p$ bip-reduced.

\begin{observation}\label{obs:complete_core_same}
	Let $p$ be a prime and $H$ a connected bipartite graph with a pair of adjacent vertices $v$ and $u$. If $w$ is a vertex in $\partof[K^{v,u}](v)$, then $K^{w,u}=K^{v,u}$. If $w$ is a vertex in $\partof[K^{v,u}](u)$, then $K^{v,u}=K^{v,w}$. Further, if $K^{v,u}$ is $2$-connected, then $K^{u,v} \cong K^{v,u}$ by swapping the order of the bipartition.
\end{observation}

We obtain the following conditions for the existence of a hard vertex.
\begin{lemma}\label{lem:condition_new_hard_vertex}
	Let $p$ be a prime, $H$ be a connected bip-graph that is \graphclass{}, and $v$ be a vertex of $H$. Given a neighbour $u$ of $v$ and the complete core $K^{v,u}$ with $K^{v,u} \congbip K_{a,b}$.  If one of the following applies, then $H$ contains a hard vertex:
	\begin{enumerate}
		\item $a \not\equiv 1 \pmod p$, $b \not\equiv 0 \pmod p$, and there exists a vertex $w$ in $\partof[K^{v,u}](v)$ such that $\deg (w) \not \equiv b \pmod p$;
		\item $a \not\equiv 0 \pmod p$, $b \not\equiv 1 \pmod p$, and there exists a vertex $w$ in $\partof[K^{v,u}](u)$ with $\deg (w) \not \equiv a \pmod p$.
		\item $a,b \equiv 0 \pmod p$ and there exist two vertices $v'$ in $\partof[K^{v,u}](v)$ and $u'$ in $\partof[K^{v,u}](u)$ such that  $\deg (v'), \deg (u') \not \equiv 0 \pmod p$.
	\end{enumerate}
\end{lemma}
\begin{proof}
	We are going to prove the three statements in order.
	\begin{enumerate}
		\item By Observation~\ref{obs:complete_core_same}, we assume without loss of generality that $w$ is $v$. Let the set $U$ of neighbours of $v$ be $\partof[K^{v,u}](u)$ and $B$ be the induced sub-bip-graph $\twoneigh{v} \cap \, \bigcap_{u \in U} \twoneigh{u}$. We observe that $B$ consists of $K^{v,u}$ and $\deg(v) - \deg_{K^{v,u}} (v) \not\equiv 0 \pmod p$ leaves attached to $v$. By the assumptions $a \not\equiv 1 \pmod p$ and $b \not\equiv 0 \pmod p$ it follows that $\bipreduced{B}$ is not a collection of complete bipartite graphs.
		
		\item Since $b$ has to be larger than $1$ it follows that $K^{v,u}$ is $2$-connected. By Observation~\ref{obs:complete_core_same}, $K^{v,u} \cong K^{u,v}$. The case follows from the first case by swapping the bipartition.
		
		\item By Observation~\ref{obs:complete_core_same}, we assume without loss of generality that $v'$ is $v$ and $u'$ is $u$. Let $B$ be the induced sub-bip-graph $\twoneigh{w_l} \cap \twoneigh{w_r}$. It follows that $B$ consists of $K^{v,u}$ along with $\deg(v) - b$ leaves adjacent to $v$ and $\deg(u) - a$ leaves adjacent to $u$. By the assumptions $\deg (v), \deg (u) \not \equiv 0 \pmod p$, while $b\equiv 0 \pmod p$ and $a\equiv 0 \pmod p$, it follows that $\bipreduced{B}$ is not a collection of complete bipartite graphs. \qedhere
	\end{enumerate}
\end{proof}

With Corollary~\ref{cor:hard_vertex_hard} at hand, we are going to study the case that a graph $H$ does not contain a hard vertex. To this end, we study the case that Lemma~\ref{lem:condition_new_hard_vertex} does not apply. The analysis is branched with respect to the existence of a pair of vertices $v$ and $v'$ of $H$, such that the distance of $v$ and $v'$ is $2$ but the set of common neighbours satisfies $\abs{\neigh{v} \cap \neigh{v'}} \equiv 0 \pmod p$. If no such pair exists in $H$, then we call $H$ \emph{\nice}. For a \graphclass{} bipartite graph $H$, by Lemma~\ref{lem:weirdness} $H$ is \nice{} if and only if $H$ does not contain a pair of adjacent vertices $u$ and $v$ such that any part of the complete core $K^{v,u}$ is of cardinality congruent modulo $p$ to $0$.

\subsubsection{Bip-Graphs That are \Nice{}}\label{subsec:large_nice}

Our findings are given for the case of general graphs $H$.

\subsubsection*{Closed Walks}
Kazeminia and Bulatov~\cite{Kazeminia:19:Count_Homs_Square_Free_Mod_Prime}  gave a type of gadgetry exploiting the existence of cycles of length larger than $4$ in a square-free graph. Our study is however not restricted to square-free graphs and, even though in many cases cycles are sufficient, we are going to encounter cases in which we need to consider closed walks. We generalize the gadgetry and results to allow for walks and general graphs. As usual, the following terminology translates to paths as well.
We note in advance that we provide a more general adaptation to support future studies on $\probNumHom{H}[p]$ compared to what we apply later on.

\begin{definition}
	Let $H$ be a graph and $W$ be a walk in $H$. We call $W$ \emph{square-free} if, for every triple of consecutive vertices $(w, w', w'')$ in $W$, the common neighbourhood $\neigh{w} \cap \neigh{w''}$ contains only $w'$.
\end{definition}

We note that a square-free closed walk $(w_0, w_1, \dots, w_\ell, w_0)$ denoted $W$ also evades leaves if $W$ contains at least $3$ different vertices, i.e. vertices $w$, $w'$, and $w''$ that are pairwise different. The reason is, assuming $w_0$ is equal to $w_{2}$, that by definition $w_{1}$ is the only neighbour of $w_{0}$. We apply this iteratively, and thus $W$ contains at most $2$ different vertices. It follows that a closed walk with $4$ different vertices has to admit a square, and we obtain that a square-free walk in a bipartite graph $H$ contains either $2$ or at least $6$ different vertices.

We generalize the notion of a square-free walk by allowing weights in the same way as for $4$-vertex paths in Section~\ref{sec:gadgets}. Let $W$ be the walk $(w_0, w_1, \dots, w_\ell)$. Given a set of counts $\set{b_i}_{i \in \sqBrackets{0;\ell}}$, we construct a \emph{thick walk} $W'$ from $W$ by, iteratively for every index $i\in \sqBrackets{0;\ell}$, cloning $w_i$ $b_i$-times, where the set of introduced twins $\set{w_i^j}_{j \in \sqBrackets{b_i}}$ is denoted $w_i^{b_i}$. The thick walk $W'$ is denoted by $(w_0^{b_0}, w_1^{b_1}, \dots, w_\ell^{b_\ell})$. We refer to the walk $W$ as the \emph{twin-free form of $W'$} and the vertices of $W$ are \emph{representatives} of the vertices of $W'$. If every count $b_i$ is equal to $1$, then $W$ and $W'$ agree. In this case, we also refer to $W'$ as a thick walk for the sake of simplicity.
\begin{definition}
	Let $p$ be a prime, $H$ be a graph, and $W'$ be a thick walk $(w_0^{b_0}, w_1^{b_1}, \dots, w_\ell^{b_\ell})$ in $H$. If the twin-free form $W$ of $W'$ is square-free, then $W'$ is called \emph{square-free}.
	Further, $W'$ is \emph{$p$-square-free} if additionally, for every index $i \in \sqBrackets{0;\ell}$, the count $b_i$ is not equivalent to $0 \pmod p$.
\end{definition}
The notion of a thick walk that is \nice{} agrees with the notion that the graph given by the thick walk is \nice{}.

Let $H$ be a graph and $W'$ be a thick walk $(w_0^{b_0}, w_1^{b_1}, \dots, w_\ell^{b_\ell})$ in $H$. The vertices in $w_i^{b_i}$ do not have to be twins with respect to the graph $H$. A vertex $w_i^1$ in $W$ might have a neighbour not found in the neighbourhood of $w_i^2$. However, restricted to the walk $W'$ they are indeed twins; the common neighbourhood of $w_i^1$ and $w_i^2$ in $W'$ is equal.

For a prime $p$, a closed thick walk $W$ yields hardness if it is $p$-square-free and contains enough vertices. Figure~\ref{fig:2-hard_walk} illustrates the following definition.
\begin{definition}
	Let $p$ be a prime, $H$ be a graph, and $W$ be a thick closed walk with twin-free form $W$. If $W'$ is \nice{} and $W$ contains more than $3$ different vertices, then $W'$ is called \emph{$p$-hard}.
\end{definition}

\begin{figure}[t]
	\begin{center}
		\includegraphics[]{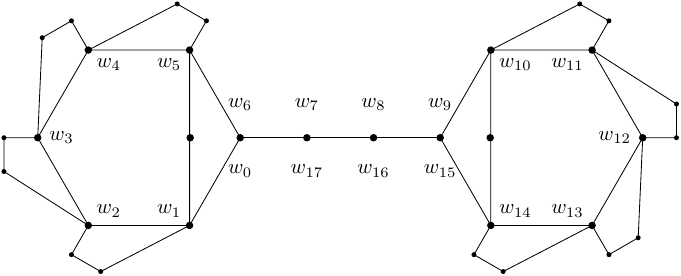}
	\end{center}
	\caption{For $p$ equal to $2$, example of a closed $2$-hard walk.}
	\label{fig:2-hard_walk}
\end{figure}

A thick path $P$ in a graph $H$ is a subgraph of $H$, which allows us to use the standard notation $\neigh{v}[P]$ for the neighbourhood in $P$ of a vertex $v$ of $P$. For instance, let $P$ be a path $(v_0, v_1, v_2, \dots, v_\ell)$ in $H$, which gives that $\neigh{v_1}[P]$ is equal to $\set{v_0, v_2}$. Let $W$ be a thick $(w_0^{b_0}, w_1^{b_1}, \dots, w_\ell^{b_\ell})$ in $H$. The thick walk $W$ is also a subgraph of $H$, but a vertex $w$ in $W$ can be contained in multiple sets $w_i^{b_i}$ and so the neighbourhood of $w$ restricted to $W$ is given by the neighbourhood of all these sets $w_i^{b_i}$ under edges in $W$. By the fixed ordering of $W$, in order to clarify the position of $w$ it suffices to provide the index $i \in \sqBrackets{0;\ell}$, which further restricts the neighbourhood. For any index $i \in \sqBrackets{0;\ell}$, the neighbourhood of a vertex $w_i$ in $w_i^{b_i}$ restricted to $W$ is then given by the vertices in $W$ visited by $W$ immediately before or after $w_i$. In particular, $\neigh{w_i}[W]$ is $w_{i - 1}^{b_{i-1}} \cup w_{i+1}^{b_{i+1}}$, conditional on the existence of $w_{i - 1}^{b_{i - 1}}$ and $w_{i + 1}^{b_{i + 1}}$, which is needed for the border cases in which $i$ is equal to $0$ or $k$. We note that these definitions agree if $W$ is a thick path.

With this notation at hand, it is straightforward to define parallel walks. We consider a thick walk $W$, where $W =(w_0^{b_0}, w_1^{b_1}, \dots, w_\ell^{b_\ell})$, and a walk $W'$ with $W'=(w'_0, \dots, w'_\ell)$. the walk $W'$ is \emph{parallel} to $W$ if, for every index $i \in \sqBrackets{0;\ell}$, the vertex $w'_i $ is in the complement $\neigh{w_i^{b_i}} \setminus \neigh{w_i^{b_i}}[W]$. An example is given in the left part of Figure~\ref{fig:cycle_gadget}.

\begin{lemma}\label{lem:walks_gadget}
	Let $p$ be a prime and $H$ be a connected order~$p$ bip-reduced bip-graph. If $H$ contains a $p$-hard thick walk $W$, where $W=(w_0^{b_0}, w_1^{b_1}, \dots, w_\ell^{b_\ell}, w_0^{b_0})$, then, for every index $i \in \sqBrackets{0;\ell}$, there exists a partially $H$-labelled bip-graph $(J, y)$ that selects the set $\selectSet$ that satisfies
	\[
		\selectSet = w^{b_{i-1}}_{i-1} \cupdot w^{b_{i+1}}_{i+1} \cupdot R_i,
	\]
	where the indices are taken modulo $k+1$ and the set of remainders $R_i$ consists of vertices adjacent to $w^{b_i}_i$ contained in a closed walk $\hat{W}$ parallel to $W$.
\end{lemma}
\begin{proof}
	For an arbitrary vertex $w_i$ in $W$, we construct $(J,y)$ before we show the desired property. Throughout this proof, the indices are taken modulo $k+1$.
	We recall Lemma~\ref{lem:complete_core_gadget}, which gives, for every index $j \in \sqBrackets{0;\ell}$, a partially $H$-labelled bip-graph $(J_j, y_j)$ that selects the set $\selectSet[j]$ equal to the common neighbourhood $\bigcap_{s \in \sqBrackets{b_j}} \neigh{w^s_j}$. Let $C$ be the bip-graph obtained from the cycle $(y_0, y_1, \dots, y_\ell, y_0)$ with bipartition of $C$ such that $\partof[C](y_0)$ agrees with $\npartof[W](w_0)$, where $w_0$ is a vertex in $w_0^{b_0}$. The partially $H$-labelled bip-graph $J$ is constructed from $C$ and the graphs $(J_j, y_j)$ by, iteratively for every $j \in \sqBrackets{0;\ell}$ starting with $j=0$, taking the dot product $(C, y_j) \odot (J_j, y_j)$ denoted $(J, y_j)$ and assigning $C$ equal to $J$. Lastly, we take for the distinguished vertex $y$ the vertex $y_i$. We refer to Corollary~\ref{cor:dot_product_bip} for the properties of the dot product used in the following. An illustration is depicted in Figure~\ref{fig:cycle_gadget}.
	
	\begin{figure}[t]
		\begin{center}
			\includegraphics[]{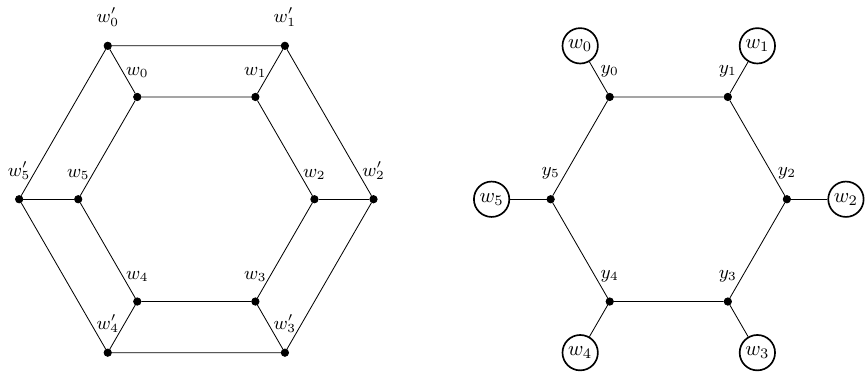}
		\end{center}
		\caption{Illustration of the construction employed to study Lemma~\ref{lem:walks_gadget}. The left figure depicts two parallel cycles $C=(v_0, \dots, v_5, v_0)$ and $C'=(v_0', \dots, v_5', v'_0)$. The right figure depicts the partially $H$-labelled bip-graph employed in the proof of Lemma~\ref{lem:walks_gadget}.}
		\label{fig:cycle_gadget}
	\end{figure}
	
	Towards the desired property we assume that every count $b_j$ is equal to $1$, i.e. we argue first about the twin-free form containing the representatives. This keeps notation tidy, We show that only for vertices $v$ in $\selectSet$ there exists a homomorphism in $\HomBip[(J,y)(H,v)]$. By re-introducing the counts we then argue for these vertices that $\numHomBip[(J,y)(H,v)][p]$ is not $0$.
	
	We note that $W$ contains at least $6$ different vertices because $W$ is \nice{} and closed. Therefore, for every index $j \in \sqBrackets{0;\ell}$, the common neighbourhood $\neigh{w_j} \cap \neigh{w_{j+2}}$ contains only $w_{j+1}$.	
	We assume without loss of generality that $i=0$, and let $f$ be a homomorphism in $\HomBip[(J,y)(H,v)]$, where $v$ is a vertex of $H$. By the construction of $(J,y)$, the vertex $v$ has to be in the neighbourhood $\neigh{w_0}$. If $v$ is $w_{\ell}$, then $f$ maps $y_1$ to the common neighbourhood $\neigh{w_\ell} \cap \neigh{w_1}$, that contains by assumption only the vertex $w_0$. We are going to prove by induction on $\ell$ that, for every index $j \in \sqBrackets{0;\ell}$, the homomorphism $f$ maps $y_j$ to $w_{j-1}$. The base case is already established. Towards the induction step, let $f(y_{j-1})$ be $w_{j-2}$. It follows that $f$ maps $y_j$ to the common neighbourhood $\neigh{w_j} \cap \neigh{w_{j-2}}$, that contains only $w_{\ell-1}$. Since $w_\ell$ and $w_{\ell-1}$ are adjacent, $f$ is indeed a valid homomorphism.
	
	We obtain the walk $W'$ by reversing the order, that is $W^{-1}=(w_0, w_\ell, w_{\ell-1}, \dots, w_1, w_0)$. By the same argumentation as for $W$, we obtain that if $f$ maps $y$ to $w_1$, then $f$ maps $y_\ell$ to $w_0$ and, for every index $j \in \sqBrackets{0; \ell-1}$, the image $f(y_{j})$ is $w_{j+1}$. This maps $y_1$ to $w_2$ and $f$ is indeed a valid homomorphism. We note that for both cases $v = w_\ell$ and $v= w_1$, there is exactly one homomorphism.
	
	Lastly, we assume that $f(y)$ is $v$ with $v$ not in $W$ and not in a closed walk $\hat{W}$ parallel to $W$. For any $l$ in $\sqBrackets{0;\ell}$, let $W_l$ be the subwalk $(w_0, \dots, w_l)$ of $W$. Further, let $l$ be maximal such that, for every index $j \in \sqBrackets{0;l}$, the homomorphism $f$ maps $y_j$ to a vertex $\hat{w}_{j}$ in a walk $\hat{W}_l$ parallel to $W_l$. This gives that $\hat{W}_l$ is $(\hat{w}_0, \hat{w}_1, \dots, \hat{w}_l)$. We note that $l < \ell$ because $f$ is a homomorphism and $(y_0, y_\ell)$ is an edge in $H$. It follows that $f$ maps $y_{l+1}$ to a vertex $v_{l+1}$ in the common neighbourhood $\neigh{w_{l+1}} \cap \neigh{\hat{w}_l}$, but $v_{l+1}$ unequal to $w_l$ and $w_{l+2}$ would give a larger parallel walk. If $v_{l+1}$ is $w_{l+2}$, then $\hat{w_l}$ is in the common neighbourhood $\neigh{w_l} \cap \neigh{w_{l+2}}$ but $\hat{w_l}$ is by assumption not $w_{l+1}$, a contradiction to $W$ being square-free. If $v_{l+1}$ is $w_{l}$, then by the same argumentation as employed for the case that $v$ is $w_\ell$, it follows that $f$ maps $y_\ell$ to $w_{\ell-1}$. Since $H$ is bipartite, the vertices $\hat{w}_0$ and $w_{\ell-1}$ cannot be adjacent, a contradiction to $f$ being a homomorphism. Therefore, if $f(y)$ is neither $w_{\ell}$ nor $w_1$, then $f(y)$ is in a closed walk parallel to $W$.
		
	We turn toward the thick walk $W$ with counts $\set{b_j}_{j \in \sqBrackets{0;\ell}}$, and we are going to show that, for vertices $v \in \selectSet$, the number of homomorphisms $\numHomBip[(J,y)(H,v)][p]$ is not $0$. It suffices to argue about vertices in $W$ as, for a vertex $\hat{w}_0$ in a walk parallel to $W$, the statement is unaffected by $\numHomBip[(J,y)(H,w)][p] = 0$. We recall that for the twin-free form and the cases that either $v= w_\ell$ or $v= w_1$, we proved the mapping of the only homomorphism. Extending this to counts, we obtain for $v$ in $w_\ell^{b_\ell}$, that the homomorphisms are given by, for every index $j \in \sqBrackets{0;\ell}$, the possible images $f(y_j)$ in $w^{b_{j-1}}_{j-1}$. Since $W$ is $p$-hard, no count $b_i$ is congruent modulo $p$ to $0$, and we deduce, for $v$ in $w_\ell^{b_\ell}$, that the number of homomorphisms $\numHomBip[(J,y)(H,v)]$ is equal to $ \prod_{j=0}^{k-1} b_j $, which is not congruent modulo $p$ to $0$. Analogue, for $v$ in $w_1^{b_1}$, the only homomorphisms are given by, for every index $j \in \sqBrackets{0;\ell}$, the possible images $f(y_j)$ in $w^{b_{j+1}}_{j+1}$. Therefore, the number of homomorphisms $\numHomBip[(J,y)(H,v)]$ is equal to $b_0 \cdot \prod_{j=2}^{k} b_j$, which is not congruent modulo $p$ to $0$.
\end{proof}

We discuss the necessity of the prerequisites of a walk $W$ for the gadgetry to work. If any count $b_i$ is congruent modulo $p$ to $0$, then the homomorphisms under study cancel out when counting modulo $p$. If $W$ is not square-free, then either $W$ consists of at most $4$ different vertices or we can recollect vertices into representatives such that we obtain a square-free walk $W'$. For the latter, e.g. if the common neighbourhood $\neigh{w_i} \cap \neigh{w_{i+2}}$ contains more than $\set{w_{i+1}}$, then we collect all vertices in the common neighbourhood $\neigh{w_i} \cap \neigh{w_{i+2}}$ to the set $\set{w_{i+1}^j}_{j \in \sqBrackets{b_{i+1}}}$. Lastly, if $W$ consists of at most $4$ different vertices, then $W$ trails a complete bipartite induced subgraph in $H$. This does not yield hardness.
\begin{lemma}\label{lem:walks_hard}
	Let $p$ be a prime and $H$ be a connected bip-graph. If there exists a $p$-hard thick walk in $H$, then there exists an induced sub-bip-graph $B$ of $H$ such that $H$ has a \gadget{B,p} and $\bipreduced{B}$ is not complete bipartite.
\end{lemma}
\begin{proof}
	We consider a $p$-hard thick walk $W$ in $H$, where $W= (w_0^{b_0}, w_1^{b_1}, \dots, w_\ell^{b_\ell}, w_0^{b_0})$. Since $H$ is bipartite and $W$ is square-free, it follows that $W$ contains $4$ different representative vertices not admitting a square. Without loss of generality, let $w_0$, $w_1$, $w_2$, and $w_3$ be different representatives that admit no square, i.e. the sub-bip-graph of $H$ induced by $\set{w_0, w_1, w_2, w_3}$ is the $4$-vertex path $(w_0, w_1, w_2, w_3)$.
	
	Let $\GadgetPart{\L}$ and $\GadgetPart{\R}$ be the partially $H$-labelled bip-graphs given by Lemma~\ref{lem:walks_gadget} such that $\GadgetPart{\L}$ selects the set $w^{b_{0}}_{0} \cupdot w^{b_{2}}_{2} \cupdot R_1$ denoted by $\selectSet[\L]$ and $\GadgetPart{\R}$ selects the set $w^{b_{1}}_{1} \cupdot w^{b_{3}}_{3} \cupdot R_2$ denoted by $\selectSet[\R]$, where for $i \in \set{1,2}$ the set $R_i$ consists of vertices adjacent to $w_i^{b_i}$ and in a closed walk parallel to $W$. We note that since $H$ is bipartite, no vertex in $R_1$ is adjacent to a vertex in the complement $\selectSet[\R] \setminus R_2$ and no vertex in $R_2$ is adjacent to a vertex in the complement $\selectSet[\L] \setminus R_1$. Further, by $W$ square-free it follows that a vertex in $R_1$ cannot be adjacent to $w^{b_{3}}_{3}$. Analogously, no vertex in $R_2$ is adjacent to $w^{b_{0}}_{0}$. Finally, $R_1$ and $R_2$ is given by walks parallel to $W$ and thus every vertex in $R_1$ is connected to at least one vertex in $R_2$ and vice versa.
	
	We conclude that the sub-bip-graph $B$ induced by $\selectSet[\R]$ and $\selectSet[\L]$ consists of the thick $4$-vertex path $(v_0^{b_0}, v_1^{b_1}, v_2^{b_2}, v_3^{b_3})$, which we denote $P$, and the vertices in $R_1$ and in $R_2$. Further, $B$ contains edges connecting $v_1^{b_1}$ to every vertex in $R_1$, edges connecting $v_2^{b_2}$ to every vertex in $R_2$, and edges connecting pairs in $R_1 \times R_2$ such that no vertex in $R_1$ and $R_2$ is a leaf. We deduce that for $B$ no vertex in $R_1$ and $R_2$ is \isomorphic[bip] to a vertex in $P$. Since $W$ is $p$-hard, no count is equivalent to $0$. Therefore, $\bipreduced{P}$ is still a thick $4$-vertex path. Since $R_1$ and $R_2$ do not introduce additional bip-automorphisms of $B$ involving $\bipreduced{P}$, the reduced form $\bipreduced{B}$ contains $\bipreduced{P}$. The lemma follows. 
\end{proof}

We have established by Theorem~\ref{thm:small graphs_weird_hard} that all order~$p$ bip-reduced bip-graphs $H$ of radius at most $2$ that are \graphclass{} admit a $p$-hardness gadget. This gives the following implication.
\begin{corollary}\label{cor:walks_hard}
	Let $p$ be a prime and $H$ be a connected order~$p$ bip-reduced bip-graph that is \graphclass{}. If $H$ contains a $p$-hard thick walk, then $H$ admits a $p$-hardness gadget.
\end{corollary} 

We are now studying graphs $H$ that do not admit a $p$-hardness gadget. By the previous results, these graphs have no induced cycle of length at least $6$ from which tree-like properties follow. The absence of a $p$-hardness gadget is going to yield the following structure associated with paths in a tree.

\subsubsection*{Hardness Paths}
Again we consider graphs that do not have to be \nice{} such that we can apply the results in the case of such graphs. As can be observed in Lemma~\ref{obs:hard_vertex}, one reason a given bip-graph $H$ does not contain a hard vertex is the existence of induced subgraphs isomorphic to complete bipartite graphs with one part of cardinality congruent modulo $p$ to $1$. In the case of a \nice bip-graph $H$, these induced subgraphs have in the other part a number of vertices that is not congruent modulo $p$ to $0$. We concatenate these subgraphs and obtain a structure formally defined as follows.
\begin{definition}\label{def:hardness-path}
	Let $p$ be a prime, $H$ be a bipartite graph, and $P$ be a thick path of even length $\ell \equiv 0 \pmod 2$ in $H$, where $P =(v_0^{b_0}, v_1^{b_1}, \dots, v_\ell^{b_\ell})$. We call $P$ a \emph{$p$-hardness path} if, for all indices $i \in \sqBrackets{0; \ell}$, the following conditions hold:
	\begin{enumerate}
		\item if $i$ is even, then $b_i$ is $1$;
		\item if $i$ is even and neither $0$ nor $\ell$, then $\deg(v_{i})$ is congruent modulo $p$ to the count $b_{i-1}$;
		\item if $i$ is odd, then $b_{i}$ is not congruent modulo $p$ to $0$;
		\item if $i$ is odd, then for the representative $v_{i}$ of $v_{i}^{b_{i}}$, the complete core $K^{v_{i-1}, v_{i}}$ is \isomorphic[bip] to $K_{a_{i}, b_{i}}$ with $a_{i}$ congruent modulo $p$ to $1$.
	\end{enumerate}
	$P$ is called \emph{symmetric} if additionally, for even indices $i \in \sqBrackets{2; \ell-2}$, the degree $\deg(v_{i})$ is congruent modulo $p$ to the other count $b_{i+1}$. 
\end{definition}
An example of a $p$-hardness path with its twin-free form is depicted in Figure~\ref{fig:generalized_hardness_path_2}. For a $p$-hardness path $P$, where $P=(v_0^{b_0}, v_1^{b_1}, \dots, v_\ell^{b_\ell})$, we denote by $v_i$ the \emph{representative} of $v_i^{b_i}$, where $i$ is an index $i \in \sqBrackets{\ell}$. Since $b_1$ is $1$ for odd indices, for such indices we also use $v_i$ instead of $v_i^{1}$. We note for $i$ neither $0$ nor $\ell$, both vertices $v_{i+1}$ and $v_{i-1}$ are in the complete core $K^{v_{i-1}, v_i}$. If and only if $P$ is symmetric, then the \emph{reverse path} $P^{-1}$ is also a $p$-hardness path, where $P^{-1} = ( v_{\ell}^{b_\ell}, v_{\ell-1}^{b_{\ell-1}}, \dots, v_0^{b_0})$. 
\begin{figure}[t]
	\begin{center}
		\includegraphics[]{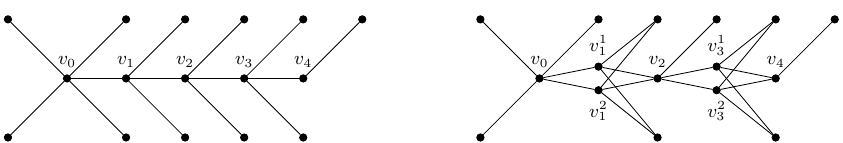}
	\end{center}
	\caption{Illustration of a $p$-hardness path for $p$ equal to $3$. The left figure depicts the $p$-hardness path $P'$ with $P'=(v_0^1,v_1^1,v_2^1,v_3^1,v_4^1)$. The right figure depicts the symmetric $p$-hardness path $P$ with $P=(v_0^1,v_1^2,v_2^1,v_3^2,v_4^1)$ and with twin-free form $P'$.}
	\label{fig:generalized_hardness_path_2}
\end{figure}

In the case of a \nice{} graph $H$, we always obtain the existence of a $p$-hardness path under the assumption that $H$ admits no $p$-hardness gadget.
\begin{lemma}\label{lem:existence_gen_hardness_path}
	Let $p$ be a prime and $H$ be a connected order~$p$ bip-reduced bip-graph that is \graphclass{} and \nice{}. If $H$ admits no $p$-hardness gadget, then there exists a $p$-hardness path in $H$.
\end{lemma}
\begin{proof}
	Due to the absence of a $p$-hardness gadget, we have by Corollary~\ref{cor:hard_vertex_hard} that $H$ contains no hard vertex.
	Furthermore, by Lemma~\ref{lem:bound_tuple_same_neighbourhood} there cannot exist a tuple of $p$ twins in $H$.
	
	We are going to show the existence of a pair of adjacent vertices $v_0, v_1$ such that the complete core $K^{v_0,v_1}$ is \isomorphic[bip] to $K_{a,b}$ with $a \equiv 1 \pmod p$ and $a>1$. This is sufficient for the lemma. The reason is that first, $b$ is not congruent modulo $p$ to $0$ because $H$ is \nice{}. We denote by $v_1^b$ the set of vertices in $\partof[K^{v_0,v_1}](v_1)$. Second, for any vertex $v_2$ in the intersection  $\neigh{v_1} \cap K^{v_0, v_1}$, the thick path $(v_0, v_1^b, v_2)$ is a $p$-hardness path because, for an endvertex $v$ of a $p$-hardness path $P$, there are no restrictions imposed by $P$ on the neighbourhood of $v$ outside of $P$.
	 
	By Theorem~\ref{thm:small graphs_weird_hard}, $H$ has radius larger than $2$. Let $v$ be an arbitrary vertex of $H$ and let $U_v$ be the set of vertices $v^\ast$ in $\twoneigh{v}$ with neighbourhood $\neigh{v^\ast}$ not contained in $\twoneigh{v}$. Since $H$ is of radius larger than $2$, it follows that $U_v$ contains at least one vertex.
	
	Let $v^\ast$ be a vertex in $U_v$ and $u$ be a common neighbour of $v$ and $v^\ast$. If the complete core $K^{v^\ast,u}$ is \isomorphic[bip] to $K_{a,b}$ and the cardinality $a$ is congruent modulo $p$ to $1$, then we are done because $a$ is larger than $1$ by the existence of $v$ and $v^\ast$. Hence, we assume that $a$ is not congruent modulo $p$ to $1$. The cardinality $b$ is not congruent modulo $p$ to $0$ since $H$ is \nice{}. It follows by Lemma~\ref{lem:condition_new_hard_vertex} that $\deg(v^\ast)$ is congruent modulo $p$ to $b$. Since $v^\ast$ is in $U_v$, there exists a neighbour $u^\ast$ of $v^{\ast}$ that is not in $\twoneigh{v}$. By Lemma~\ref{lem:bound_tuple_same_neighbourhood}, we assume without loss of generality that $u^\ast$ is not a leaf. We consider the complete core $K^{v^\ast, u^\ast}$ given by $v^\ast$ and $u^\ast$. Let the complete core $K^{v^\ast, u^\ast}$ be \isomorphic[bip] to $K_{a^\ast, b^\ast}$, for which we know that $a^\ast$ is larger than $1$ because $u^\ast$ is not a leaf. Since $H$ is \nice{}, it follows that $b^\ast$ is not congruent modulo $p$ to $0$. If $a^\ast \equiv 1 \pmod p$ then we are done, otherwise the prerequisites of Lemma~\ref{lem:gadget_adjacent_components} are met by $v^\ast$ and the two adjacent vertices $u$ and $u^\ast$. We obtain a $p$-hardness gadget, a contradiction.
\end{proof}

We call a $p$-hardness path $P$ in $H$ \emph{maximal} if there exists no extension of $P$ in $H$, that is, for $P= (v_0^{b_0}, v_1^{b_1}, \dots, v_\ell^{b_\ell})$, there exists no $p$-hardness path $P'=(v_0^{b_0}, v_1^{b_1}, \dots, v_\ell^{b_\ell}, v_{\ell+1}^{b_{\ell+1}} , v_{\ell+2}^{b_\ell+2})$.

The only vertices in the $p$-hardness path $P$, whose degrees are not determined, are the endvertices $v_0$ and $v_\ell$. Depending on the degree of the endvertices, we are going to obtain hardness by \enquote{connecting} the endvertices via the $p$-hardness path $P$. Maximality is used to further reduce the possible cases. 

We note that by definition a $p$-hardness path is $p$-square-free. A closed $p$-hardness path has to contain more than $4$ representative vertices since the path proceeds along complete cores. Hence, a closed $p$-hardness path is a $p$-hard walk and yields a $p$-hardness gadget by Corollary~\ref{cor:walks_hard}. With this argument, we show the following.

\begin{lemma}\label{lem:maximal_generalized_hardness_path}
	Let $p$ be a prime and $H$ be a connected order~$p$ bip-reduced bip-graph that is \graphclass{} and admits no $p$-hardness gadget. If $H$ contains a maximal $p$-hardness path $P$, where $P=(v_0^{b_0}, v_1^{b_1}, \dots, v_\ell^{b_\ell})$, then one of the following holds:
	\begin{enumerate}
		\item[(i)]   $\deg(v_{\ell}) \equiv b_{\ell-1} \pmod p$ and there exists a set of vertices $U$ in the complement $\neigh{v_{\ell}} \setminus \neigh{v_\ell}[P]$ such that the common neighbourhood satisfies $\abs{\bigcap_{u \in U} \neigh{u}} \not\equiv 1 \pmod p$;
		\item[(ii)] $\deg(v_{\ell}) \not\equiv b_{\ell-1} \pmod p$.
	\end{enumerate}
\end{lemma}
\begin{proof}
	By the definition of a $p$-hardness path, it follows that, for the representatives $v_{\ell-2}$ and $v_{\ell-1}$, the complete core $K^{v_{\ell-2}, v_{\ell-1}}$ is \isomorphic[bip] to $K_{k\cdot p+1, b_{\ell-1}}$, where $k$ is larger than $0$. By Lemma~\ref{lem:bound_tuple_same_neighbourhood}, we assume without loss of generality that $\deg(v_{\ell})$ is larger than the number $b_{\ell-1}$ of neighbours of $v_{\ell}$ in $K^{v_{\ell-2}, v_{\ell-1}}$. We assume toward contradiction that neither $(i)$ nor $(ii)$ holds, i.e. $\deg(v_{\ell})$ is congruent modulo $p$ to $b_{\ell-1}$ and there exists no set of vertices $U$ in the complement $\neigh{v_{\ell}}\setminus P$ such that $\abs{\bigcap_{u \in U} \neigh{u}} \not\equiv 1 \pmod p$.
	
	Since $\deg(v_{\ell})$ is congruent modulo $p$ to $b_{\ell-1}$, we deduce that the set $N$ with $N=\neigh{v_{\ell}} \setminus v_{\ell-1}^{b_{\ell-1}}$ has cardinality congruent modulo $p$ to $0$. Since $\deg(v_{\ell})$ is larger than $b_{\ell-1}$, we have that $N$ is not empty. We obtain from Lemma~\ref{lem:bound_tuple_same_neighbourhood} that $N$ has to contain a vertex $u$ of degree larger than $1$. Now we have by the maximality of $P$ that, if $u$ is not in $P$, then the complete core $K^{v_{\ell}, u}$, where $K^{v_{\ell}, u}$ is \isomorphic[bip]  to $K_{a', b'}$, satisfies that $a'$ is not congruent modulo $p$ to $1$. We suppose that $u$ is not in $P$ and let $U$ be the set of vertices in the part of $K^{v_{\ell}, u}$ that contains $u$. The common neighbourhood of the vertices in $U$ is exactly the other part of $K^{v_{\ell}, u}$, the part that contains $v$. Since $(i)$ does not apply, it follows that the cardinality $a'$ of the vertices in $\partof[K^{v_{\ell}, u}](v)$ is congruent modulo $p$ to $1$. We deduce that every vertex $u$ in $N$ with degree larger than $1$ has to be adjacent to a vertex $v$ in $P$ such that $v$ is in the complete core $K^{v_\ell, u}$. 
	
	Since $u$ is in $N$, it follows that $u$ is not in the neighbourhood $\neigh{v_{\ell-2}}$. The graph $H$ is bipartite, and we obtain that $v$ is equal to the vertex $v_i$ for some even index $i \in \sqBrackets{0;\ell}$. We note that $i$ is neither $\ell - 2$ nor $\ell$ because $u$ is not in the neighbourhood $\neigh{v_{\ell-2}}$. There might be two options for $i$ because $u$ might already be in $P$. In order to avoid squares, we take the maximal even index $i$. We obtain a closed thick path $P'$ given by $P'=(v_i, v_{i+1}^{b_{i+1}}, v_{i+2}^{b_{i+2}}, \dots, v_\ell^{b_\ell}, u^{b'}, v_i)$. By construction, no vertex in the set $u^{b'}$ is in the set $v_{i+1}^{b_{i+1}}$. Since $P$ traverses consecutive complete cores, we derive that $P'$ is square-free with more than $4$ representative vertices. None of the counts is congruent modulo $p$ to $0$. Therefore, the closed thick path $P'$ is $p$-hard. Corollary~\ref{cor:walks_hard} yields that $H$ admits a $p$-hardness gadget, a contradiction.
\end{proof}

\begin{figure}[t]
	\centering
	\includegraphics[]{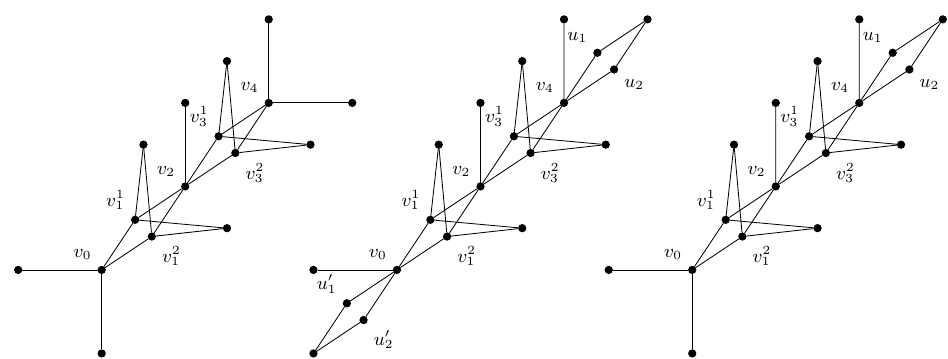}
	\caption{Illustration of the cases for a maximal $p$-hardness path according to Lemma~\ref{lem:maximal_generalized_hardness_path} for $p$ equal to $3$.}
	\label{fig:generalized_hardness_path_cases}
\end{figure}

The above lemma was stated in a way that enables an adaptation to general graphs $H$ without a necessary reformulation for the case that $H$ is not assured to be \graphclass{}.
\begin{remark}
	The application of Corollary~\ref{cor:walks_hard} is the only argument in the proof of Lemma~\ref{lem:maximal_generalized_hardness_path} that needed $H$ to be \graphclass{}. The statement of Lemma~\ref{lem:maximal_generalized_hardness_path} remains correct for general order~$p$ bip-reduced bip-graph $H$ if we add the third case that $H$ contains a $p$-hard walk, which allows an application of Lemma~\ref{lem:walks_hard}. However, for $H$ not \graphclass{}, we did not prove that this yields a $p$-hardness gadget. 
\end{remark}

We briefly argue that a maximal $p$-hardness path $P$, where $P=(v_0^{b_0}, v_1^{b_1}, \dots,  v_{\ell}^{b_\ell})$, contains a symmetric $p$-hardness path, for which we obtain the same result as in Lemma~\ref{lem:maximal_generalized_hardness_path} but for both endvertices. The reason is first, that there are no restrictions on the degree of an endvertex of $P$, and second, that an internal vertex $v_i$ of $P$ that prevents the reverse $P^{-1}$ from being a $p$-hardness path satisfies that $\deg(v_i)$ is not congruent modulo $p$ to the next count $b_{i+1}$. Therefore, in $P^{-1}$ the internal vertex $v_i$ satisfies case $(ii)$ of Lemma~\ref{lem:maximal_generalized_hardness_path}. Taking the maximal such index $i \in \sqBrackets{2;\ell-2}$, the thick subpath $P_1 = (v_{i}^{b_i}, v_{i+1}^{b_{i+1}}, \dots, v_{\ell}^{b_\ell})$ of $P$ is a symmetric $p$-hardness path.
\begin{corollary}\label{cor:maximal_generalized_hardness_path}
	Let $p$ be a prime and $H$ be a connected order~$p$ bip-reduced bip-graph that is \graphclass{} and admits no $p$-hardness gadget. If $H$ contains a maximal $p$-hardness path $P$, then there exists a thick subpath $P_1$ of $P$ that is a symmetric $p$-hardness path and both $P_1$ and its reverse $P_1^{-1}$ satisfy either case $(i)$ or case $(ii)$ of Lemma~\ref{lem:maximal_generalized_hardness_path}.
\end{corollary}

For the case of a \nice{} graph, we have by Lemma~\ref{lem:existence_gen_hardness_path} the existence of a maximal $p$-hardness path $P$. It remains to study the cases $(i)$ and $(ii)$ of Lemma~\ref{lem:maximal_generalized_hardness_path}. By Corollary~\ref{cor:maximal_generalized_hardness_path}, we assume $P$ to be symmetric. There are three different cases illustrated in Figure~\ref{fig:generalized_hardness_path_cases} from left to right: either $P$ and $P^{-1}$ satisfy both $(i)$ or both satisfy $(ii)$ or one satisfies $(i)$ and the other satisfies $(ii)$. We study these cases separately. For this, we analyse walks along $p$-hardness paths to derive gadgets, and use the following notation.
\begin{notation}\label{def:walks}
	Let $p$ be a prime and $H$ be a bip-graph. 
	\begin{itemize}
		\item Let $u$ and $v$ be a pair of vertices of $H$ and $k$ be a non-negative integer. The number of walks in $H$ of length $k$ from $u$ to $v$ is denoted by $\numWalks{u}{v}[k]$. 
		\item Let $P$ be a path in $H$, where $P=(v_0, v_1, \dots, v_\ell)$. For a pair of vertices $v'_0$ in $\neigh{v_0}$ and $v'_\ell$ in $\neigh{v_\ell}$, we denote by $\numWalks{v'_0}{v'_{\ell}}[\ell][P]$ the number of walks of length $\ell$ from $v'_0$ to $v'_\ell$ restricted to $P$, where we define $\numWalks{v'_0}{v'_{\ell}}[\ell][P]$ by
		\[
			\numWalks{v'_0}{v'_{\ell}}[\ell][P] = \abs[\big]{\set{ (v'_0, v'_1, \dots, v'_{\ell}) \given \text{for all } i \in \sqBrackets{\ell-1} ,\, v'_i \in \neigh{v_i}}}.
		\]
		For a vertex $v'_0$ not in $\neigh{v_0}$ or a vertex $v'_\ell$  not in $\neigh{v_\ell}$, we define $\numWalks{v'_0}{v'_{\ell}}[\ell][P]$ to be $0$.
		\item Let $P$ be a $p$-hardness path in $H$, where $P=(v_0^{b_0}, v_1^{b_1}, \dots, v_{\ell}^{b_\ell})$, and $P'$ be its twin-free form, where $P'=(v_0, v_1, \dots, v_\ell)$. For vertices $v'_0$ in $\neigh{v_0}$ and $v'_\ell$ in $\neigh{v_\ell}$, we denote by $\numWalks{v'_0}{v'_{\ell}}[\ell][P]$ the number of walks of length $\ell$ from $v'_0$ to $v'_\ell$ restricted to $P$ by using $P'$ and defining $\numWalks{v'_0}{v'_{\ell}}[\ell][P]$ to be equal to $\numWalks{v'_0}{v'_{\ell}}[\ell][P']$.
	\end{itemize}
\end{notation}
We extend this notation to also allow for walks of length $\ell+1$ and $\ell+2$. Let $P$ be a path or a $p$-hardness path with endvertices $v_0$ and $v_\ell$. For every pair of neighbours $v'_0 \in \neigh{v_0}$ and $v'_\ell \in \neigh{v_\ell}$, every walk contributing to $\numWalks{v'_0}{v'_{\ell}}[\ell][P]$ yields a walk of length $\ell+1$ from a neighbour $v$ of $v'_0$ to $v'_\ell$, a walk of length $\ell+1$ from $v'_0$ to a neighbour $v'$ of $v'_{\ell}$, and a walk of length $\ell +2$ from $v$ to $v'$. Therefore, we denote for vertices $v$ in $\twoneigh{v_0}$ and $v'$ in $\twoneigh{v_\ell}$
\begin{align}
	\label{eq:numwalks_neighbour_to_2-neighbour}
	\numWalks{v'_0}{v'}[\ell+1][P] &= \sum_{v'_\ell \in \neigh{v'} \cap \neigh{v_\ell}} \numWalks{v'_0}{v'_{\ell}}[\ell][P] ;\\
	\label{eq:numwalks_2-neighbour_to_neighbour}
	\numWalks{v}{v'_\ell}[\ell+1][P] &= \sum_{v'_0 \in \neigh{v} \cap \neigh{v_0}} \numWalks{v'_0}{v'_{\ell}}[\ell][P] ;\\
	\label{eq:numwalks_2-neighbour_to_2-neighbour}
	\numWalks{v}{v'}[\ell+2][P] &= \sum_{\substack{v'_0 \in \neigh{v} \cap \neigh{v_0} \\ v'_\ell \in \neigh{v'} \cap \neigh{v_\ell}}} \numWalks{v'_0}{v'_{\ell}}[\ell][P] .
\end{align}

\begin{lemma}\label{lem:path_gadget}
	Let $p$ be a prime and let $H$ be a connected order~$p$ bip-reduced bip-graph. If there exists a path $P$ in $H$, where $P= (v_0, \dots, v_\ell)$, then, for every length $\lambda$ in $\set{\ell, \ell+1, \ell +2}$, there exists a partially $H$-labelled bip-graph $\GadgetEdge$ such that, for every pair of vertices $v$ in $\twoneigh{v_0}$ and $v'$ in $\twoneigh{v_\ell}$,
	\begin{enumerate}
		\item if $\lambda \equiv \ell \pmod 2$, then the number of homomorphisms $\numHomBip[\GadgetEdge, (H, v, v')]$ is congruent modulo $p$ to $\numWalks{v}{v'}[\lambda][P]$;
		\item if $\lambda = \ell + 1$ and $v$ is in $\neigh{v_0}$, then the number of homomorphisms $\numHomBip[\GadgetEdge, (H, v, v')]$ is congruent modulo $p$ to  $\numWalks{v}{v'}[\lambda][P]$;
		\item if $\lambda = \ell + 1$ and $v'$ is in $\neigh{v_\ell}$, then the number of homomorphisms $\numHomBip[\GadgetEdge, (H, v, v')]$ is congruent modulo $p$ to $\numWalks{v}{v'}[\lambda][P]$.
	\end{enumerate}
\end{lemma}
\begin{proof}
	We are going to apply the dot product for the construction of the partially $H$-labelled bip-graphs, where the properties of the dot product are given by Corollary~\ref{cor:dot_product_bip}.
	
	We start with the case that $\lambda$ is $\ell$.	
	By Lemma~\ref{lem:complete_core_gadget}, for every index $i$ in $\sqBrackets{0; \ell}$, there exists a partially $H$-labelled bip-graph $(J_i, y_i)$ that selects $\neigh{v_i}$. 
	The partially $H$-labelled bip-graph $\GadgetEdge$ is constructed from the path $J_E$ of length $\ell$, where $J_E=(y_0, \dots, y_\ell)$, with bipartition such that $\partof[J_E](y_0)$ agrees with $\npartof[H](v_0)$. We take iteratively, for every index $i \in \sqBrackets{0; \ell}$, the dot product $(J_E, y_i) \odot (J_i, y_i)$, which we denote by $(G', y_i)$, and set $J_E$ to be $G'$. Afterwards, we take for the distinguished vertex $y_\L$ the vertex $y_0$ and for the distinguished vertex $y_\R$ the vertex $y_\ell$. In this way, a homomorphism $f$ in $\HomBip[\GadgetEdge, (H, v, v']$ has to map, for every index $i \in \sqBrackets{0; \ell}$, the vertex $y_i$ to a vertex in the neighbourhood of $v_i$. It follows that $\numHomBip[\GadgetEdge, (H, v, v']$ is equal to $\numWalks{v}{v'}[\lambda][P]$.
	
	Second, let $\lambda$ be equal to $\ell + 1$ and $v$ be a vertex in the neighbourhood of $v_0$. Additionally, let $\GadgetEdge$ be as constructed for the case that $\lambda$ is $\ell$. We extend $\GadgetEdge$ to account for walks of length $\ell+1$. Let $J$ be the single edge bip-graph consisting of the edge $(w_\L, w_\R)$ with the bipartition such that $\partof[J](w_\L)$ agrees with $\partof[J_E](y_\R)$. Utilizing the dot product we construct the partially $H$-labelled bip-graph $(J'_E, w_\L)$ by $(J'_E, w_\L) = (J_E, y_\R) \odot (J, w_\L)$. In this way, every homomorphism $f$ in $\HomBip[J'_E, H]$ has to map $w_\L$ to the neighbourhood $\neigh{v_\ell}$, and thus $f$ maps $w_\R$ to $\twoneigh{v_\ell}$. Due to $J_E$, the homomorphism $f$ maps, for every index $i \in \sqBrackets{0; \ell}$, the vertex $y_i$ to a vertex in the neighbourhood of $v_i$. By the definition of $\numWalks{v}{v'}[\lambda][P]$, the partially $H$-labelled bip-graph $(J'_E, y_\L, w_\R)$ has the desired property.
	
	The case that $\lambda$ is equal to $\ell +1$ and $v'$ is in the neighbourhood of $\neigh{v_\ell}$ is analogue to the previous case but with switched roles of vertices with subscript $\L$ and vertices with subscript $\R$. Lastly, let $\lambda$ be equal to $\ell + 2$. We extend the bip-graph $\GadgetEdge$ that we obtained for the case that $\lambda$ is $\ell$ in the same way as before but now in both directions. More precise, with the bip-graphs as in the previous cases we construct first the partially $H$-labelled bip-graph $(J', w'_\L)$ by $(J', w'_\L) =(J_E, y_\R) \odot (J, w_\L)$ and then the partially $H$-labelled bip-graph $(J'', w'_\R)$ by $(J'', w'_\R) =(J', y_\L) \odot (J, w_\R)$. Both individual constructions were used for the cases in which $\lambda$ is equal to $\ell + 1$. We deduce that the partially $H$-labelled bip-graph $(J'', w_\L, w_\R)$ has the desired property.
\end{proof}

We apply the same construction for a $p$-hardness path $P$. Let $P=(v_0^{b_0}, v_1^{b_1}, \dots, v_{\ell}^{b_\ell})$ and let $P'$ be the twin-free form of $P$, where $P'= (v_0, v_1, \dots, v_\ell)$. Any walk $(v'_0, v'_1, \dots, v'_{\ell})$ that contributes to $\numWalks{v'_0}{v'_{\ell}}[\ell][P]$ satisfies that, for every odd index $i \in \sqBrackets{\ell}$, the vertex $v'_i$ is in $\neigh{v_i^{b_i}}$. By Lemma~\ref{lem:complete_core_gadget}, there exists a partially $H$-labelled bip-graph $(J_i, y_i)$ that selects the common neighbourhood $\bigcap_{j \in \sqBrackets{b_i}} \neigh{v_i^j}$. This gives the following.
\begin{corollary}\label{cor:generalized_hardness_path_gadget}
	Let $p$ be a prime and let $H$ be a connected order~$p$ bip-reduced bip-graph. If there exists a $p$-hardness path $P$ in $H$, where $P=(v_0^{b_0}, v_1^{b_1}, \dots, v_{\ell}^{b_\ell})$, then, for every length $\lambda$ in $\set{\ell, \ell+1, \ell +2}$, there exists a partially $H$-labelled bip-graph $\GadgetEdge$ such that, for every pair of vertices $v$ in $\twoneigh{v_0}$ and $v'$ in $\twoneigh{v_\ell}$,
	\begin{enumerate}
		\item if $\lambda \equiv \ell \pmod 2$, then the number of homomorphisms $\numHomBip[\GadgetEdge, (H, v, v')]$ is congruent modulo $p$ to $\numWalks{v}{v'}[\lambda][P]$;
		\item if $\lambda = \ell + 1$ and $v$ is in $\neigh{v_0}$, then the number of homomorphisms $\numHomBip[\GadgetEdge, (H, v, v')]$ is congruent modulo $p$ to $\numWalks{v}{v'}[\lambda][P]$;
		\item if $\lambda = \ell + 1$ and $v'$ is in $\neigh{v_\ell}$, then the number of homomorphisms $\numHomBip[\GadgetEdge, (H, v, v')]$ is congruent modulo $p$ to $\numWalks{v}{v'}[\lambda][P]$.
	\end{enumerate}
\end{corollary}

By the structure of a square-free thick path $P$, a walk $W$ restricted to $P$ has to mainly follow $P$. 
\begin{lemma}\label{lem:path_gadget_follows_path}
	Let $H$ be a connected \graphclass{} bipartite graph and $P$ be a thick square-free path in $H$, where $P=(v_0^{b_0}, v_1^{b_1}, \dots,  v_{\ell}^{b_\ell})$. If $W$ is a walk that contributes to $\numWalks{v'_0}{v'_\ell}[\ell][P]$, where $W=(v'_0, \dots, v'_{\ell})$, then there exists at most one vertex $v'_i$ in $W$ such that $v'_i$ is not in the neighbourhood $\neigh{v_i}[P]$.
\end{lemma}
\begin{proof}
	We first show that $W$ cannot contain two consecutive vertices $v'_i$ and $v'_{i+1}$ not in $\neigh{v_i}[P]$ and $\neigh{v_{i+1}}[P]$, respectively. The reverse thick path of $P$ is also square-free as is every thick subpath of $P$ of even length at least $2$, that is every thick path $(v_i^{b_i}, v_{i+1}^{b_{i+1}}, \dots, v_{l}^{b_l})$, where $i \in \sqBrackets{0;\ell-2}$ and $l$ is an even index with $l \in \sqBrackets{2;\ell}$. Without loss of generality, let $\ell$ be $2$. We assume toward contradiction that the two consecutive vertices $v'_0$ and $v'_1$ are not in $\neigh{v_0}[P]$ and $\neigh{v_{1}}[P]$, respectively. Thus, the vertex $v'_0$ is not $v_1$ and the vertex $v'_1$ is neither $v_0$ nor $v_2$. Since $v'_0$ is in the neighbourhood $\neigh{v_0}$ and $v'_i$ is in the neighbourhood $\neigh{v_1^{b_1}}$, we deduce that $v'_0$ and $v'_1$ yield a square $(v_0, v_1, v'_1, v'_0)$. By the maximality of a complete core, it follows that $v'_0$ and $v'_1$ are in the complete core $K^{v_0, v_1}$, which is a complete bipartite graph by Lemma~\ref{lem:weirdness}. The vertex $v_2$ is also in the complete core $K^{v_0, v_1}$. Therefore, $v'_0$ is in the neighbourhood $\neigh{v_2}$, a contradiction to $P$ being square-free.
	
	Now, we show by induction on $\ell$ that $W$ contains at most one vertex $v'_i$ not in the neighbourhood $\neigh{v_i}[P]$. For the case that $\ell$ is $2$, it suffices to assume that $v'_0$ is not in $v_1^{b_1}$ because of the reverse thick path of $P$. We know from the first part of the proof that $v'_1$ is in the neighbourhood $\neigh{v_1}[P]$. The vertex $v'_1$ cannot be $v_2$ as this would give a square. It follows that $v'_1$ is $v_0$. Hence, $v'_2$ is in the common neighbourhood $\neigh{v_0} \cap \neigh{v_2}$, and we derive that $v'_2$ is in $v_1^{b_1}$.
	For the induction step, we assume the induction hypothesis for all even lengths smaller than $l$, where $l \in \sqBrackets{2;\ell}$. We observe the two thick square-free paths $P_2$ and $P_{-2}$ obtained from $P$ by cropping the first two vertices, that is $P_2 = (v_0, v_1^{b_1}, v_2)$ and $P'=(v_2, v_3^{b_3}, v_4, \dots, v_l)$. If $v'_2$ is in the neighbourhood $\neigh{v_2}[P_2]$, then we derive that $v'_2$ is in $v_1^{b_1}$ and thus $v'_2$ is not in the neighbourhood $\neigh{v_2}[P_{-2}]$. The other direction holds as well. The lemma follows by the induction hypothesis on $P_2$ and $P_{-2}$.
\end{proof}

Given a square-free thick path $P$ with endvertices $v_0$ and $v'_\ell$. Let $v'_0$ be a neighbour of $v_0$ and $v'_\ell$ be a neighbour of $v_\ell$. We apply Lemma~\ref{lem:path_gadget_follows_path} frequently in order to deduce that a walk $W$ that contributes to $\numWalks{v'_0}{v'_{\ell}}[\ell][P]$, where $v'_0$ is not in $P$, consists of $v'_0$ and the remaining vertices are in $P$. Let $W=(v'_0, v'_1, \dots, v'_{\ell})$. It follows for every index $i \in \sqBrackets{\ell}$ that the vertex $v'_i$ is in $v_{i-1}^{b_{i-1}}$. None of the counts $b_i$ is congruent modulo $p$ to $0$. We deduce that that the number of walks $\numWalks{v'_0}{v'_{\ell}}[\ell][P]$ is equal to $\prod_{i \in \sqBrackets{0; \ell}} b_{i} $, which is not congruent modulo $p$. Further, if $v'_\ell$ is not in the neighbourhood $\neigh{v_\ell}[P]$, then the number of walks $\numWalks{v'_0}{v'_{\ell}}[\ell][P]$ is $0$.
We initialize the construction of hardness gadgets with a lemma regarding (sub)walks in generalized paths. 

\begin{lemma}\label{lem:generalized_hardnes_path_gadget_1}
	Let $p$ be a prime and $H$ be a connected order~$p$ bip-reduced bip-graph that is \graphclass{}. If $H$ contains a $p$-hardness path, where $P=(v_0^{b_0}, v_1^{b_1}, \dots, v_{\ell}^{b_\ell})$, then there exists a non-negative integer $k$ such that, for every pair of vertices $v'_0$ in $v_1^{b_1}$ and $v'_{\ell}$ in $v_{\ell-1}^{b_{\ell-1}}$, it holds
	\[
	\numWalks{v'_0}{v'_{\ell}}[\ell][P] = k \not\equiv 0 \pmod p.
	\]
\end{lemma}
\begin{proof}
	Let $W$ be a walk that contributes to $\numWalks{v'_0}{v'_{\ell}}[\ell][P]$, where $W=(v'_0, v'_1, \dots,  v'_{\ell})$. By Lemma~\ref{lem:path_gadget_follows_path}, we deduce that $W$ contains at most one subwalk $(v'_{i-1}, v'_{i}, v'_{i+1})$ such that $v'_i$ is not in the neighbourhood $\neigh{v_i}[P]$ for some vertex $v_i$ in $P$.
	
	Supposing such a subwalk exists, it follows by $P$ being square-free that both $v'_{i-1}$ and $v'_{i+1}$ are in $v_i^{b_i}$. By the length $\ell$ of the walks under study, we deduce that $W$ contains exactly one such subwalk $(v'_{i-1}, v'_{i}, v'_{i+1})$.	We prove the lemma by induction on $\ell$.	
	
	The case that $\ell$ is $2$ follows from the complete core $K^{v_0, v_1}$. We know from Lemma~\ref{lem:weirdness} that $K^{v_0, v_1}$ is \isomorphic[bip] to $K_{a,b_1}$, where $a$ is congruent modulo $p$ to $1$ because $P$ is a $p$-hardness path. 	
	Regarding the induction step, we assume the induction hypothesis for all even lengths smaller than $l$, where $l \in \sqBrackets{2;\ell}$. Let $P_{-2}$ be the thick subpath $(v_0^{b_0}, v_1^{b_1}, \dots, v_{l-2}^{b_{l-2}})$, and let $v'_{l-2}$ be a vertex in $v_{l-3}^{b_{l-3}}$. By the induction hypothesis, there exists a positive integer $k$ such that the number of walks $\numWalks{v'_0}{v'_{l-3}}[l-2][P_2]$ is equal to $k$, and $k$ is not congruent modulo $p$ to $0$. Any walk $W$ that contributes to $\numWalks{v'_0}{v'_{l-2}}[l-2][P_2]$ can only be extended to a walk that contributes to $\numWalks{v'_0}{v'_{l}}[l][P]$ by adjoining the path $(v'_{l-2}, v_{l-2}, v'_{l} )$. There are $b_{l-3}$ choices for $v'_{l-2}$. Hence, the walks that visit a vertex in $v_{l-3}^{b_{l-3}}$ after $l-2$ steps contribute $k \cdot b_{l-3}$ to $\numWalks{v'_0}{v'_{l}}[l][P]$, where $k \cdot b_{l-3} \not \equiv 0 \pmod p$. 
	
	We claim that the remaining number of walks $W$ that contribute to $\numWalks{v'_0}{v'_{l}}[l][P]$ is congruent modulo $p$ to $0$. The walk $W$ does not visit $v_{l-3}^{b_{l-3}}$ after $l-2$ steps, i.e. $v'_{l-2}$ is not in $v_{l-3}^{b_{l-3}}$. By our initial observation, we know that $W$ contains exactly one subwalk $(v'_{i-1}, v'_{i}, v'_{i+1})$ such that $v'_i$ is not in the neighbourhood $\neigh{v_i}[P]$ for some vertex $v_i$ in $P$. We recall that both $v'_{i-1}$ and $v'_{i+1}$ are in $v_i^{b_i}$. The walk $W$ cannot have such a revision of a vertex with index $i \leq l-2$ because otherwise $v'_{l-2}$ has to be in $v_{l-3}^{b_{l-3}}$. The number of revisions of $v_{l-2}$ that do not invoke a revision of the set $v_{l-3}^{b_{l-3}}$ is equal to $\deg(v_{l-2}) - b_{l-3}$, which is congruent modulo $p$ to $0$ by $P$ being a $p$-hardness path. The number of revisions of the set $v_{l-1}^{b_{l-1}}$ via a vertex $v'$ in the common neighbourhood $\bigcap_{j \in \sqBrackets{b_{l-1}}} \neigh{v^j_{l-1}}$ that do not invoke a revision of $v_{l-2}$ is equal to the common degree $\abs{\bigcap_{j \in \sqBrackets{b_{l-1}}} \neigh{v^j_{l-1}}} - 1$, which is also congruent modulo $p$ to $0$.
\end{proof}

By Lemma~\ref{lem:generalized_hardnes_path_gadget_1} we obtain the first condition for a $p$-hardness gadget depending only on the degree of the end vertices. This corresponds to a symmetric $p$-hardness path $P$ such that $P$ and its reverse $P^{-1}$ both satisfy case $(ii)$ of Lemma~\ref{lem:maximal_generalized_hardness_path}.
\begin{lemma}\label{lem:nice_graphs_hardnes_path_hard_1}
	Let $p$ be a prime and $H$ be a connected order~$p$ bip-reduced bip-graph that is \graphclass{}. If $H$ contains a $p$-hardness path $P$, where $P=(v_0^{b_0}, v_1^{b_1}, \dots, v_{\ell}^{b_\ell})$, such that $\deg(v_0) \not \equiv b_1 \pmod p$ and $\deg(v_{\ell}) \not \equiv b_{\ell-1} \pmod p$, then $H$ admits a $p$-hardness gadget.
\end{lemma}
\begin{proof}
	Let $\GadgetPart{\L}$ and $\GadgetPart{\R}$ be the partially $H$-labelled bip-graphs given by Lemma~\ref{lem:complete_core_gadget} such that $\GadgetPart{\L}$ selects the set $\selectSet[\L]$ with $\selectSet[\L] = \neigh{v_0}$ and $\GadgetPart{\R}$ selects the set $\selectSet[\R]$ with $\selectSet[\R] = \neigh{v_{\ell}}$. Furthermore, let $\GadgetEdge$ be the partially $H$-labelled bip-graph given by Corollary~\ref{cor:generalized_hardness_path_gadget} such that, for every pair of vertices $v$ and $v'$, where $v$ is in $\neigh{v_0}$ and $v'$ is in $\neigh{v_{\ell}}$, the number of homomorphisms $\numHomBip[\GadgetEdge, (H, v, v')][p]$ is equal to the number of walks $\numWalks{v}{v'}[\lambda][P]$. We decompose the set $\selectSet[\L]$ into the pair of disjunct subsets $o_\L$ and $i_\L$, where we define $o_\L$ to be $v_1^{b_1}$ and $i_\L$ to be the complement $\neigh{v_0} \setminus o_\L$. Similarly, we decompose $\selectSet[\R]$ into the pair of disjunct subsets $o_\R$ and $i_\R$, where we define $o_\R$ to be $v_{\ell-1}^{b_{\ell-1}}$ and $i_\R$ to be the complement $\neigh{v_{\ell}} \setminus o_\R$. By the lemma's assumptions, we have that none of these four subsets has cardinality congruent modulo $p$ to $0$. It remains to show that, for the pair of vertices $v$ and $v'$ as given above, the number of homomorphisms $\numHomBip[\GadgetEdge, (H, v, v')][p]$ is equal to $0$ if and only if the pair $(v, v')$ is in $i_\L \times i_\R$.
	
	Let $W$ be a walk that contributes to $\numWalks{v}{v'}[\ell][P]$. By Lemma~\ref{lem:path_gadget_follows_path}, $W$ contains at most one vertex $v'_i$ that is not in $\neigh{v_i}[P]$. If the vertex $v$ is in $i_\L$, then $v$ gives such a vertex because $v$ is not in $P$. The same holds if the vertex $v'$ is in $i_\R$. It follows for any pair $(v,v')$ in $i_{\L} \times i_{\R}$, that there are no walks that contribute to $\numWalks{v}{v'}[\ell][P]$. Moreover, for any pair $(v,v')$ in $o_{\L} \times i_{\R}$, it follows by the distance restricted to $P$ that the number of walks $\numWalks{v}{v'}[\ell][P]$ is given by a product with factors in $\set{b_0, b_1, \dots, b_{\ell}}$, thus not congruent modulo $p$ to $0$. Similarly, for any pair $(v,v')$ in $i_{\L} \times o_{\R}$, we deduce that the number of walks $\numWalks{v}{v'}[\ell][P]$ is not congruent modulo $p$ to $0$. Finally, for any pair $(v,v')$ in $o_{\L} \times o_{\R}$, it follows by Lemma~\ref{lem:generalized_hardnes_path_gadget_1} that the number of walks $\numWalks{v}{v'}[\ell][P]$ is not congruent modulo $p$ to $0$.
\end{proof}

This is not the extent of the properties, that yield a $p$-hardness gadget. For a $p$-hardness path $P$ with endvertices $v_0$ and $v_\ell$, we study in the following the case that $\deg(v_\ell)$ is not congruent modulo $p$ to the count $b_{\ell-1}$. If $P$ is symmetric, then we show hardness in the case that $P$ satisfies case $(i)$ of Lemma~\ref{lem:maximal_generalized_hardness_path} and its reverse $P^{-1}$ satisfies case $(ii)$ of Lemma~\ref{lem:maximal_generalized_hardness_path}.

\begin{lemma}\label{lem:nice_graphs_hardnes_path_hard_3}
	Let $p$ be a prime and $H$ be a connected order~$p$ bip-reduced bip-graph that is \graphclass{}. Further, let $H$ contain a $p$-hardness path $P$, where $P=(v_0^{b_0}, v_1^{b_1}, \dots, v_{\ell}^{b_\ell})$, such that $\deg(v_0) \not \equiv b_1 \pmod p$ and $\deg(v_{\ell}) \equiv b_{\ell-1} \pmod p$. If there exists a set of vertices $U$ in $\neigh{v_{\ell}} \setminus \neigh{v_{\ell}}[P]$ such that $\abs{U} \not \equiv 0 \pmod p$ and the common neighbourhood satisfies $\abs{\bigcap_{u \in U} \neigh{u}} \not \equiv 1 \pmod p$, then $H$ admits a $p$-hardness gadget. 
\end{lemma}
\begin{proof}
	Let $\GadgetPart{\L}$ and $\GadgetPart{\R}$ be the partially $H$-labelled bip-graphs given by Lemma~\ref{lem:complete_core_gadget} such that $\GadgetPart{\L}$ selects the sets  $\selectSet[\L]$, where $\selectSet[\L] = \neigh{v_0}$, and $\GadgetPart{\R}$ selects the set $\selectSet[\R]$, where $\selectSet[\R] = \bigcap_{u \in U} \neigh{u}$.
	For the partially $H$-labelled bip-graph $\GadgetEdge$ we use walks of length $\ell+1$; let $\GadgetEdge$ be the bip-graph given by Corollary~\ref{cor:generalized_hardness_path_gadget} such that, for every pair of vertices $v$ and $v'$, where $v$ is in $\neigh{v_0}$ and $v'$ is in $\twoneigh{v_\ell}$, the number of homomorphisms $\numHomBip[\GadgetEdge, (H, v, v')][p]$ is equal to the number of walks $\numWalks{v}{v'}[\ell+1][P]$.	
	We decompose the set $\selectSet[\L]$ into the pair of disjunct subsets $o_\L$ and $i_\L$, where we define $o_\L$ to be $v_1^{b_1}$ and $i_\L$ to be the complement $\neigh{v_0} \setminus o_\L$.  Similarly, we decompose $\selectSet[\R]$ into the pair of disjunct subsets $o_\R$ and $i_\R$, where we define $o_\R$ to be $\set{v_{\ell}}$ and $i_\R$ to be the complement $\bigcap_{u \in U} \neigh{u} \setminus o_\R$. By the lemma's assumptions, none of these four subsets has cardinality congruent modulo $p$ to $0$. It remains to show that, for the pair of vertices $v$ and $v'$ as given above, the number of homomorphisms $\numHomBip[\GadgetEdge, (H, v, v')][p]$ is equal to $0$ if and only if $(v,v')$ is in $i_\L \times i_\R$.
	
	Let $W$ be a walk that contributes to $\numWalks{v}{v'}[\ell+1][P]$. First, for any $v$ in $\neigh{v_0}$ and $v'$ equal to $v_\ell$, we have by \eqref{eq:numwalks_neighbour_to_2-neighbour} \begin{align*}
		\numWalks{v}{v'}[\ell+1][P] &= \sum_{v'_\ell \in \neigh{v_\ell} \cap P} \numWalks{v}{v'_\ell}[\ell][P] + \sum_{v'_\ell \in \neigh{v_\ell} \setminus P} \numWalks{v}{v'_\ell}[\ell][P].\\
		\intertext{There are $b_{\ell-1}$ choices for a vertex $v_{\ell -1}$ in $\neigh{v_\ell}[P]$ because the latter is equal to $v_{\ell-1}^{b_{\ell-1}}$, and there are $\deg(v_{\ell}) - b_{\ell-1}$ choices for a vertex $\hat{v}_{\ell-1}$ in the complement $\neigh{v_\ell} \setminus  \neigh{v_\ell}[P]$. This gives}
		&= b_{\ell -1} \cdot \numWalks{v}{v_{\ell -1}}[\ell][P] + (\deg(v_{\ell}) - b_{\ell-1}) \cdot \numWalks{v}{\hat{v}_{\ell-1}}[\ell][P] \equiv b_{\ell -1} \cdot \numWalks{v}{v_{\ell -1}}[\ell][P] \pmod p,
	\end{align*}	
	where the equivalence follows from the degree of $v_{\ell}$. For $v$ not in $v_1^{b_1}$, by Lemma~\ref{lem:path_gadget_follows_path} and $P$ being a $p$-hardness path the number of walks $\numWalks{v}{v_{\ell -1}}[\ell][P]$ is not congruent modulo $p$ to $0$. By Lemma~\ref{lem:generalized_hardnes_path_gadget_1}, the same holds for $v$ in $v_1^{b_1}$.	
	Second, for $v$ in $\neigh{v_0}$ and $v'$ in $\bigcap_{u \in U} \neigh{u} \setminus \set{v_\ell}$, there are $\abs{U}$ choices for a vertex $u$ in $U$, and we obtain by \eqref{eq:numwalks_neighbour_to_2-neighbour}
	\begin{align*}
		\numWalks{v}{v'}[\ell+1][P] = \abs{U} \cdot \numWalks{v}{u}[\ell][P].
	\end{align*}
	By Lemma~\ref{lem:path_gadget_follows_path}, this is $0$ if $v$ is not in $v_1^{b_1}$, otherwise this is not congruent modulo $p$ to $0$.
\end{proof}

In the same spirit, for a $p$-hardness path $P$ of length $\ell$ with endvertices $v_0$ and $v_\ell$, we study the case that both $\deg(v_0)$ and $\deg(v_\ell)$ are congruent modulo $p$ to their respective counts $b_1$ and $b_{\ell-1}$ as given by case $(i)$ of Lemma~\ref{lem:maximal_generalized_hardness_path}. If $P$ is symmetric, then we show hardness by exploiting walks of length $\ell+2$.

\begin{lemma}\label{lem:nice_graphs_hardnes_path_hard_2}
	Let $p$ be a prime and $H$ be a connected order~$p$ bip-reduced bip-graph that is \graphclass{}. Further, let $H$ contain a $p$-hardness path $P$, where $P=(v_0^{b_0}, v_1^{b_1}, \dots, v_{\ell}^{b_{\ell}})$, such that  $\deg(v_0) \equiv b_1 \pmod p$ and $\deg(v_\ell) \equiv b_{\ell-1} \pmod p$. If there exists a pair of sets of vertices $U$ in $\neigh{v_0} \setminus \neigh{v_0}[P]$ and $U'$ in $\neigh{v_\ell} \setminus \neigh{v_\ell}[P]$ such that
	\begin{itemize}
		\item the cardinalities satisfy $\abs{U} \not \equiv 0 \pmod p$ and $\abs{U'} \not \equiv 0 \pmod p$;
		\item the common neighbourhoods satisfy $\abs{\bigcap_{u \in U} \neigh{u}} \not \equiv 1 \pmod p$ and $\abs{\bigcap_{u \in U'} \neigh{u}} \not \equiv 1 \pmod p$,
	\end{itemize}
	then $H$ admits a $p$-hardness gadget.
\end{lemma}
\begin{proof} 
	Let $\GadgetPart{\L}$ and $\GadgetPart{\R}$ be the partially $H$-labelled bip-graphs given by Lemma~\ref{lem:complete_core_gadget} such that $\GadgetPart{\L}$ selects $\selectSet[\L]$, where $\selectSet[\L] = \bigcap_{u \in U} \neigh{u}$, and $\GadgetPart{\R}$ selects $\selectSet[\R]$, where $\selectSet[\R] = \bigcap_{u \in U'} \neigh{u}$.	
	For the partially $H$-labelled bip-graph $\GadgetEdge$ we use walks of length $\ell+2$ as follows. Let $\GadgetEdge$ be the bip-graph given by Corollary~\ref{cor:generalized_hardness_path_gadget} such that, for every pair of vertices $v$ and $v'$, where $v$ is in $\twoneigh{v_0}$ and $v'$ is in $\twoneigh{v_\ell}$, the number of homomorphisms $\numHomBip[\GadgetEdge, (H, v, v')][p]$ is equal to the number of walks $\numWalks{v}{v'}[\ell+2][P]$.	
	We decompose the set $\selectSet[\L]$ into the pair of disjunct subsets $o_\L$ and $i_\L$, where we define $o_\L$ to be $\set{v_0}$ and $i_\L$ to be the complement $\bigcap_{u \in U} \neigh{u} \setminus o_\L$.  Similarly, we decompose $\selectSet[\R]$ into the pair of disjunct subsets $o_\R$ and $i_\R$, where we define $o_\R$ to be $\set{v_{\ell}}$ and $i_\R$ to be the complement $\bigcap_{u \in U'} \neigh{u} \setminus o_\R$. By the lemma's assumptions, none of these four subsets has cardinality congruent modulo $p$ to $0$. It remains to show that, for the pair of vertices $v$ and $v'$ as given above, the number of homomorphisms $\numHomBip[\GadgetEdge, (H, v, v')][p]$ is equal to $0$ if and only if $(v,v')$ is in $i_\L \times i_\R$.
	
	Let $W$ be a walk that contributes to $\numWalks{v}{v'}[\ell+2][P]$, where $W=(v, v'_0, \dots, v'_{\ell}, v')$.
	First, for $v$ in $\bigcap_{u \in U} \neigh{u} \setminus \set{v_0}$ and $v'$ equal to $v_\ell$, we observe that $v'_0$ has to be in $U$. We obtain by \eqref{eq:numwalks_2-neighbour_to_2-neighbour}, for any $u$ in $U$,
	\begin{align*}
		\numWalks{v}{v'}[\ell+2][P] &= \abs{U} \cdot \parenthesis[\Big]{ \sum_{v'_\ell \in \neigh{v_\ell}[P]} \numWalks{u}{v'_\ell}[\ell][P] + \sum_{v'_\ell \not\in \neigh{v_\ell}[P]} \numWalks{u}{v'_\ell}[\ell][P] }\\
		\intertext{There are $b_{\ell-1}$ options to choose a vertex $v'_{\ell}$ in $v_{\ell-1}^{b_{\ell-1}}$, which is equal to $\neigh{v_\ell}[P]$. By Lemma~\ref{lem:path_gadget_follows_path}, the second summand is equal to $0$ since $u$ is not in $v_1^{b_1}$. The same reasoning gives also the following second equality}
		&= \abs{U} \cdot b_{\ell -1} \cdot \numWalks{u}{v_{\ell -1}}[\ell][P] = \abs{U} \cdot \prod_{i \in \sqBrackets{0;\ell -1}} b_{i}.
	\end{align*}	
	By $P$ being a $p$-hardness path, it follows that $\numWalks{v}{v_{\ell }}[\ell+2][P]$ is not congruent modulo $p$ to $0$ for every choice of $v$ in $i_\L$.
	The analogue argument shows that, for $v$ equal to $v_0$ and $v'$ in $\bigcap_{u \in U} \neigh{u} \setminus \set{v_\ell}$, the number of walks $\numWalks{v}{v'}[\ell+1][P]$ is not congruent modulo $p$ to $0$.
	
	Second, for $v$ equal to $v_0$ and $v'$ equal to $v_\ell$, we employ again \eqref{eq:numwalks_2-neighbour_to_2-neighbour}. For this, we distinguish whether $v_0'$ is in $v_1^{b_1}$ and $v_\ell'$ is in $v_{\ell-1}^{b_{\ell-1}}$ or not. This gives, for the representative $v_1$ of $v_1^{b_1}$, any vertex $\hat{v}_1$ in the complement $\neigh{v_0}\setminus \neigh{v_0}[P]$, the representative $v_{\ell-1}$ of $v_{\ell-1}^{b_{\ell-1}}$, and any vertex $\hat{v}_{\ell-1}$ in the complement $\neigh{v_\ell}\setminus \neigh{v_\ell}[P]$,
	\begin{align*}
		\numWalks{v}{v'}[\ell+2][P] &=
		b_1 \cdot \parenthesis[\big]{ (\deg(v_{\ell}) - b_{\ell-1}) \cdot \numWalks{v_1}{\hat{v}_{\ell-1}}[\ell][P] + b_{\ell-1} \cdot \numWalks{v_1}{v_{\ell -1}}[\ell][P] } \\
		&\phantom{=}+ \parenthesis[\big]{\deg(v_{0}) - b_{1}} \cdot \parenthesis[\big]{ (\deg(v_{\ell}) - b_{\ell-1}) \cdot \numWalks{\hat{v}_1}{\hat{v}_{\ell -1}}[\ell][P] + b_{\ell-1} \cdot \numWalks{\hat{v}_1}{v_{\ell -1}}[\ell][P] }.\\
		\shortintertext{By the degree of $v_0$ and $v_\ell$,}
		\numWalks{v}{v'}[\ell+2][P] &\equiv b_1 \cdot b_{\ell-1} \cdot \numWalks{v_1}{v_{\ell -1}}[\ell][P] \pmod p.
	\end{align*}
	Lemma~\ref{lem:generalized_hardnes_path_gadget_1} yields now that $\numWalks{v}{v'}[\ell+1][P]$ is not congruent modulo $p$ to $0$.
	
	Finally, for $v$ in $\bigcap_{u \in U} \neigh{u} \setminus \set{v_0}$ and $v'$ in $\bigcap_{u' \in U'} \neigh{u'} \setminus \set{v_\ell}$, we observe that $W$ has to visit a vertex $u$ in $U$ and a vertex $u'$ in $U'$. By Lemma~\ref{lem:path_gadget_follows_path}, there are no such walks.
\end{proof}

With Lemma~\ref{lem:nice_graphs_hardnes_path_hard_1}, Lemma~\ref{lem:nice_graphs_hardnes_path_hard_3}, and Lemma~\ref{lem:nice_graphs_hardnes_path_hard_2} at hand, we covered all cases for a maximal $p$-hardness path according to Lemma~\ref{lem:maximal_generalized_hardness_path}. This concludes our study of $p$-hardness paths. 

\subsubsection*{Hardness for \Nice{} Large Bip-Graphs}
\begin{theorem}\label{thm:nice_graphs_hard}
	Let $p$ be a prime and $H$ be a connected order~$p$ bip-reduced bip-graph that is \graphclass{}. If $H$ is \nice{}, then $H$ admits a $p$-hardness gadget.
\end{theorem}
\begin{proof}
	We assume toward contradiction that $H$ admits no $p$-hardness gadget. By Lemma~\ref{lem:existence_gen_hardness_path}, $H$ contains a $p$-hardness path. In particular, due to the finiteness of $H$ we obtain a maximal $p$-hardness path $P$, where $P=(v_0^{b_0}, v_1^{b_1}, \dots, v_{\ell}^{b_\ell})$. By Corollary~\ref{cor:maximal_generalized_hardness_path}, we assume without loss of generality that $P$ is symmetric and both $P$ and its reverse $P^{-1}$ satisfy case $(i)$ or case $(ii)$ of Lemma~\ref{lem:maximal_generalized_hardness_path}.
	If both paths satisfy case $(ii)$, then by Lemma~\ref{lem:nice_graphs_hardnes_path_hard_1} we arrive at a contradiction. If $P$ satisfies case $(i)$, then we obtain a set $U$ of vertices in $\neigh{v_\ell}\setminus \neigh{v_\ell}[P]$ such that the cardinality $\abs{\bigcap_{u \in U} \neigh{u}}$ is not congruent modulo $p$ to $1$. In particular, this cardinality has to be at least $2$ due to the existence of $v_\ell$. By $H$ being \nice{}, we deduce that the cardinality of $U$ is not congruent modulo $p$ to $0$. The analogue holds if the reverse $P^{-1}$ satisfies $(i)$. Lemma~\ref{lem:nice_graphs_hardnes_path_hard_2} and Lemma~\ref{lem:nice_graphs_hardnes_path_hard_3} yield that $H$ admits a $p$-hardness gadget in any of the remaining three cases, a contradiction.
\end{proof}

\subsubsection{Bip-Graphs That are not \Nice{}}
\label{subsec:large_not_nice}
With Theorem~\ref{thm:nice_graphs_hard} at hand we restrict our attention to  bip-graphs $H$ that are not \nice{}, which we know to contain a pair of adjacent vertices $v$ and $u$ such that $K^{v,u}$ is \isomorphic[bip] to $K_{a,k\cdot p}$, where $a$ and $k$ are positive integers. From the definition of complete core and $p$ being at least $2$, it follows that $K^{v,u}$ is $2$-connected and $a$ is at least $2$. We note that the results concerning $p$-hardness paths and $p$-hard walks remain true. These results will also play a central role in the following analysis.

\subsubsection*{\texorpdfstring{\boldmath{$p$}-Mosaic Paths}{p-Mosaic Paths}}\label{sec:mosaic}
We are naturally tasked with studying chains of consecutive components, whose complete cores contain in at least one part a multiple of $p$ vertices. This motivates the following definition.

\begin{definition}
	Let $p$ be a prime and $H$ be a bip-graph that contains a path $Q$, where $Q=(w_0, \dots, w_\ell)$ and $\ell$ is at least $1$. We call $Q$ a \emph{$p$-mosaic path} if the following conditions are satisfied:
	\begin{enumerate}
		\item for all indices $i \in \sqBrackets{\ell}$, the complete core $K^{w_{i-1}, w_{i}}$ is \isomorphic[bip] to $K_{a_i, k_i \cdot p}$, where $a_i$ and $k_i$ are positive integers;
		\item for all indices $i \in \sqBrackets{\ell-1}$, the common neighbourhood $\neigh{w_{i-1}} \cap \neigh{w_{i+1}}$ is equal to  $\set{w_i}$.
	\end{enumerate}
	$Q$ is called \emph{partially hard} if $\deg(w_0)$ is not congruent modulo $p$ to $0$ and, for all indices $i \in \sqBrackets{\ell}$, the component $U^{w_{i-1}, w_{i}}$ does not contain a hard vertex.
\end{definition}

For a bip-graph $H$ and any vertex $v$ of $H$, by definition every pair of distinct components $U$ and $U'$ in the split of $\twoneigh{v}$ at $v$ intersects only in $v$. Hence, the second condition is satisfied by any path that switches components at every vertex. An example of a partially hard $p$-mosaic path is illustrated in the left part of Figure~\ref{fig:p-mosaic_path}.
\begin{figure}[t]
	\begin{center}
		\includegraphics[]{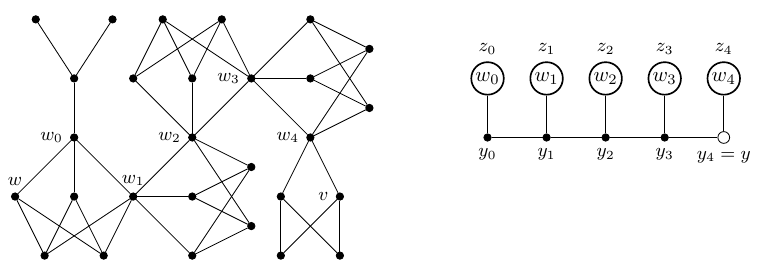}
	\end{center}
	\caption{The left picture depicts an example of a partially hard $p$-mosaic path for $p$ equal to $3$. The right picture depicts the gadget $(J,y)$ used in the proof of Lemma~\ref{lem:p-mosaic_gadget}.}
	\label{fig:p-mosaic_path}
\end{figure}

Similar to the criteria for a $p$-hardness path, a $p$-mosaic path yields hardness depending on the endvertices. By the restricted structure of a $p$-mosaic path, we construct gadgets that cancel out all vertices in the previous component visited except for the vertex previously visited by the path.
\begin{lemma}\label{lem:p-mosaic_gadget}
	Let $p$ be a prime and $H$ be a connected order~$p$ bip-reduced bip-graph. If $H$ contains a partially hard $p$-mosaic path $Q$, where $Q=(w_0, \dots, w_\ell)$, then, for every index $i \in \sqBrackets{\ell}$, there exists a partially $H$-labelled bip-graph $(J,y)$ that selects the set $\selectSet[\L]$ with $\selectSet[\L] = (\neigh{w_i} \setminus K^{w_{i-1}, w_i}) \cupdot \set{w_{i-1}}$.
\end{lemma}
\begin{proof}
	If an internal vertex $w_i$, where $i \in \sqBrackets{\ell-1}$, has $\deg(w_i)$ not congruent modulo $p$ to $0$, then the subpath $(w_i, \dots, w_\ell)$ of $Q$ starting at $w_i$ is also a partially hard $p$-mosaic path. The subpath $(w_i, \dots, w_\ell)$ does not contain another partially hard $p$-mosaic path if, for every index $j \in \sqBrackets{i+1; \ell}$, the vertex $w_j$ has $\deg(w_j)$ congruent modulo $p$ to $0$.
	Hence, the lemma follows from separating $Q$ into partially hard $p$-mosaic subpaths and proving the statement for these individually.
	
	For every internal vertex $w_i$ with $i \in \sqBrackets{\ell-1}$, we assume that $\deg(w_i)$ is congruent modulo $p$ to $0$. We note that this yields by Lemma~\ref{lem:condition_new_hard_vertex} and the absence of hard vertices, for all $i \in \sqBrackets{\ell-1}$, that the cardinality $\abs{\partof[K^{w_{i-1}, w_{i}}](w_{i-1})}$ is congruent modulo $p$ to $0$.
	
	We show the lemma by induction on the length $l$ of a subpath of $Q$ starting at $w_0$, where $l \in \sqBrackets{0;\ell}$, and initialize with the case that $l$ is $1$. For the path $(w_0, w_1)$, let $(J, y_\L, y_\R)$ be the partially $H$-labelled bip-graph obtained by Lemma~\ref{lem:path_gadget} that selects the set of vertices $\selectSet[1]$, where $\selectSet[1] = \set{v'_1 \in \neigh{w_1} \given \numWalks{w_0}{v'_1}[2][Q] \neq 0}$. By the degree of $w_0$, we have that $\numWalks{w_0}{w_0}[2][Q]$ is not congruent modulo $p$ to $0$. If $v'_1$ is of distance $2$ from $w_0$ and not in the complete core $K^{w_0, w_1}$, then there is only one common neighbour of $v'_1$ and $w_0$, that is $w_1$. It follows that $\numWalks{w_0}{v'_1}[2][Q]$ is $1$. Lastly, for $v'_1$ is of distance $2$ from $w_0$ and in the complete core $K^{w_0, w_1}$, then there are $\abs{\partof[K^{w_0, w_1}](w_1)}$ common neighbours of $w_0$ and $v'_1$. By definition, $\abs{\partof[K^{w_0, w_1}](w_1)}$ is a multiple of $p$, and it follows that $\numWalks{w_0}{v'_1}[2][Q]$ is congruent modulo $p$ to $0$. This concludes the case that $l$ is $1$.
	
	For the induction step with length $l \in \sqBrackets{\ell}$, let $Q_{l}$ be the subpath $(w_0, \dots,w_{l-1}, w_l)$ of $Q$. For the subpath $Q_{l-1}$, where $Q_{l-1}=(w_0, \dots,w_{l-1})$, let $(J', y')$ be a partially $H$-labelled bip-graph that selects the set $\selectSet[l-1]$, where $\selectSet[l-1]= (\neigh{w_{l-1}} \setminus K^{w_{l-2},w_{l-1}}) \cupdot \set{w_{l-2}}$. We extend $(J', y')$ similar to the proof of Lemma~\ref{lem:path_gadget}. For the path $Q'$ with $Q'=(w_{l-1}, w_l)$, let $\GadgetEdge$ be the partially $H$-labelled bip-graph obtained by Lemma~\ref{lem:path_gadget} such that, for pairs of vertices $(v'_{l-1}, v'_l)$ in $\neigh{w_{l-1}} \times \neigh{w_l}$, the number of homomorphisms $\numHomBip[\GadgetEdge, (H, w'_{l-1}, w'_l)]$ is equal to $\numWalks{v'_{l-1}}{v'_l}[1][Q']$, i.e. the number of homomorphisms is given by adjacent pairs in $\neigh{w_{l-1}} \times \neigh{w_l}$. By Lemma~\ref{lem:path_gadget}, there is no homomorphism in $\numHomBip[\GadgetEdge, (H, w'_{l-1}, w'_l)]$ if the pair $(v'_{l-1}, v'_l)$ is not in $\neigh{w_{l-1}} \times \neigh{w_l}$. We construct the partially $H$-labelled bip-graph $(J, y)$ by the dot product $(J', y') \odot (J_E, y_\L)$, which we define to be $(J, y')$, and distinguishing the vertex $y_\R$, which we define to be $y$. For an illustration, we refer to the right side of Figure~\ref{fig:p-mosaic_path}. It remains to show that the set $\selectSet[l]$ selected by $(J, y)$ satisfies that $\selectSet[l]$ is equal to $(\neigh{w_l} \setminus K^{w_l,w_{l-1}}) \cupdot \set{w_{l-1}}$.
	
	By the dot product used for the construction of $J$ and by Corollary~\ref{cor:dot_product_bip}, every homomorphism in $ \HomBip[J, H]$ maps $y_\L$ to $\selectSet[l-1]$ and $y$ to $\neigh{w_l}$. Additionally, by the properties of $(J', y')$, for every $w'_l$ in $\neigh{w_l}$, the number of homomorphisms $\numHomBip[(J, y), (H, w'_l)]$ is given by the number of vertices $w'_{l-1}$ in the set $\selectSet[l-1]$ selected by $(J', y')$ such that $w'_{l-1}$ is adjacent to $w'_l$. If $w'_l$ is $w_{l-1}$, then every vertex in $\selectSet[l-1]$ is adjacent to $w'_l$. The cardinality $\abs{\selectSet[l-1]}$ is equal to $\deg(w_{l-1}) - \abs{\partof[K^{w_{l-2},w_{l-1}}](w_{l-2})} + 1$. By the assumption on $\deg(w_{l-1})$ and $\abs{\partof[K^{w_{l-2},w_{l-1}}](w_{l-2})}$, it follows that $\abs{\selectSet[\ell-1]}$ is congruent modulo $p$ to $1$. If $w'_l$ is not in the complete core $K^{w_{l},w_{l-1}}$, then by the same argumentation as employed for the case that $l$ is $1$ we obtain that the only neighbour in $\selectSet[l-1]$ is $w_{l}$. Lastly, if $w'_l$ is in the complete core $K^{w_{l-1}, w_l} \setminus \set{w_{l-1}}$, then by the same argumentation as employed for the case that $l$ is $1$ the set of neighbours of $w'_l$ in $\selectSet[l-1]$ is equal to  $\partof[K^{w_{l-1}, w_l}](w_l)$. The latter is of cardinality congruent modulo $p$ to $0$, which concludes the proof.
\end{proof}

The following lemma allows us to derive the existence of a partially hard $p$-mosaic path.
\begin{lemma}\label{lem:construct_hard_path}
	Let $p$ be a prime and $H$ be a connected order~$p$ bip-reduced bip-graph that is \graphclass{}. Further, let $H$ contain a pair of adjacent vertices $v$ and $u$ such that the complete core $K^{v,u}$ is \isomorphic[bip] to $K_{a,k\cdot p}$, where $a$ and $k$ are positive integers. If $H$ admits no $p$-hardness gadget, then
	there exists a vertex $w$ in $\partof[K^{v,u}](u)$ with $\deg (w) > a$ such that one of the following cases holds:
	\begin{itemize}
		\item[(a)] $\deg(w) \equiv a \pmod p$ and there exists a maximal $p$-hardness path $P$ starting at $w$;
		\item[(b)] $\deg(w) \equiv a \pmod p$ and $\neigh{w}$ contains a vertex $v'$ that is not in the complete core $K^{v,u}$ such that the complete core $K^{w,v'}$ is \isomorphic[bip] to $K_{a',k'\cdot p}$, where $a'$ and $k'$ are positive integers;
		\item[(c)] $a \equiv 0 \pmod p$ and $\deg(w) \not \equiv 0 \pmod p$.
	\end{itemize}
\end{lemma}
\begin{proof}
	By Corollary~\ref{cor:hard_vertex_hard}, $H$ does not contain a hard vertex. First, we are going to show the existence of a suitable vertex $w$ in $\partof[K^{v,u}](u)$, and then we argue about the cases. Since the cardinality $\abs{\rpart[K^{v,u}]}$ is congruent modulo $p$ to $0$, it follows from Observation~\ref{obs:nb_outside_core} that there exists a vertex $w$ in $\partof[K^{v,u}](u)$ with $\deg (w)$ larger than $a$. We note that by Observation~\ref{obs:complete_core_same} the complete cores $K^{v,u}$ and $K^{v,w}$ are equal.
	
	If $a$ is not congruent modulo $p$ to $0$, then it follows by the absence of hard vertices in $H$ and Lemma~\ref{lem:condition_new_hard_vertex} that $\deg(w)$ is congruent modulo $p$ to $a$. By Lemma~\ref{lem:bound_tuple_same_neighbourhood}, there exists a vertex $v'$ in $\neigh{w} \setminus K^{v,w}$ such that $v'$ is not a leaf. Let the complete core $K^{w,v'}$ be \isomorphic[bip] to $K_{a',b'}$ for which we know that $a'$ is larger than $1$. If $b'$ is congruent modulo $p$ to $0$, then case (b) is met. If $b'$ is not congruent modulo $p$ to $0$ and $a'$ is congruent modulo $p$ to $1$, then case (a) is met. Else, we observe that because $K^{v,w}$ is $2$-connected we have by Observation~\ref{obs:complete_core_same} that $K^{v,w}$ is isomorphic to $K^{w,v}$ by swapping the bipartition, i.e. $K^{w,v}$ is \isomorphic[bip] to $K_{k\cdot p, a}$. Therefore, $w$ together with $v$ and $v'$ satisfy the prerequisites of Lemma~\ref{lem:gadget_adjacent_components}, which yields a $p$-hardness gadget, a contradiction.
	
	We assume now that $a$ is congruent modulo $p$ to $0$. If $\deg(w)$ is not congruent modulo $p$ to $0$, then we are in case (c). Otherwise, by the analogue argument as employed previously we deduce the existence of a neighbour $v'$ of $w$ such that $v'$ is not in the complete core $K^{w,v}$. Further, $v'$ is not a leaf and the complete core $K^{w,v'}$ is \isomorphic[bip] to $K_{a',b'}$, for which we know that $a'$ is larger than $1$. Now, by the absence of hard vertices and  Lemma~\ref{lem:condition_new_hard_vertex} it follows that $a'$ is congruent modulo $p$ to $1$ or $b'$ is congruent modulo $p$ to $0$.
\end{proof}

Towards the existence of a partially hard $p$-mosaic path we study a pair of adjacent vertices $v$ and $u$ as in Lemma~\ref{lem:construct_hard_path} such that there exists a maximal $p$-hardness path $P$ starting at $u$, where $P=(v_0^{b_1}, v_1^{b_1}, \dots, v_\ell^{b_\ell})$ and $u$ is equal to $v_0$. If $\deg(u) \equiv 0 \pmod p$, then it follows by the definition of a $p$-hardness path that $\deg(u) \not \equiv b_1 \pmod p$. We apply Lemma~\ref{lem:maximal_generalized_hardness_path} on $P$. If $\deg(v_\ell) \not \equiv b_{\ell-1} \pmod p$, then we obtain by Lemma~\ref{lem:nice_graphs_hardnes_path_hard_1} a $p$-hardness gadget. Otherwise, $\deg(v_\ell) \equiv b_{\ell-1} \pmod p$ and there exists a set $U'$ of neighbours $u'$ in $\neigh{v_\ell}$ such that the cardinality $\abs{\bigcap_{u' \in U'} \neigh{u'}}$ is not congruent modulo $p$ to $1$. We note that, for any $u'$ in $U'$, the complete core $K^{v_\ell, u'}$ is \isomorphic[bip] to $K_{a', b'}$, where $a'$ is $\abs{\bigcap_{u' \in U'} \neigh{u'}}$ and $b'$ is $\abs{U'}$. If $b' \not \equiv 0 \pmod p$, then we obtain by Lemma~\ref{lem:nice_graphs_hardnes_path_hard_3} a $p$-hardness gadget. If $b' \equiv 0 \pmod p$, then we obtain by $\deg(v_\ell) \not \equiv 0 \pmod p$ a partially hard $p$-mosaic path starting at $v_\ell$.

Similarly, if $\deg(u) \not \equiv 0 \pmod p$, then we obtain
$\abs{\lpart[K^{v,u}]} \not \equiv 0 \pmod p$ and the set of vertices $U$ given by $U=\lpart[K^{v,u}]$ satisfies $\abs{U} \not \equiv 0 \pmod p$. Further, the set of common neighbours $\bigcap_{u \in U} \neigh{u}$ is equal to $\partof[K^{v,u}](u)$. Hence, $\abs{\bigcap_{u \in U} \neigh{u}} \equiv 0 \pmod p$. Analogue to the deduction for $\deg(u) \equiv 0 \pmod p$, we obtain by Lemma~\ref{lem:nice_graphs_hardnes_path_hard_2} either a $p$-hardness gadget, or a partially hard $p$-mosaic path starting at $v_\ell$, or $\deg(v_\ell) \not \equiv b_{\ell-1} \pmod p$. In the latter case, $P$ has to be symmetric since otherwise, following the same argument as employed for Corollary~\ref{cor:maximal_generalized_hardness_path}, the reverse path $P^{-1}$ contains a $p$-hardness path starting at $v_{\ell}$ that satisfies the prerequisites of Lemma~\ref{lem:nice_graphs_hardnes_path_hard_1}. However, if $P$ is symmetric, then we obtain by $P^{-1}$ and Lemma~\ref{lem:nice_graphs_hardnes_path_hard_3} a $p$-hardness gadget. 

We conclude that the $p$-hardness path $P$ either yields a $p$-hardness gadget or a partially hard $p$-mosaic path starting at $v_\ell$.
\begin{corollary}\label{cor:construct_partially_hard_mosaic}
	Let $p$ be a prime and $H$ be a connected order~$p$ bip-reduced bip-graph that is \graphclass{} graph. Let $H$ admit no $p$-hardness gadget and contain a pair of adjacent vertices $v$ and $u$ such that the complete core $K^{v,u}$ is \isomorphic[bip] to $K_{a, k\cdot p}$, where $a$ and $k$ are positive integers. If there exists a $p$-hardness path starting at $u$ with endvertex $v_\ell$, then there exists a partially hard $p$-mosaic path starting at $v_\ell$ in $H$.
\end{corollary}

Now, we show the existence of a partially hard $p$-mosaic path provided the absence of a $p$-hardness gadget.
\begin{lemma}\label{lem:existence_partially_hard_p-mosaic}
	Let $p$ be a prime and $H$ be a connected order~$p$ bip-reduced bip-graph that is \graphclass{}. If $H$ admits no $p$-hardness gadget, then $H$ contains a partially hard $p$-mosaic path.
\end{lemma}
\begin{proof}
	By Corollary~\ref{cor:hard_vertex_hard}, $H$ contains no hard vertices, and by Theorem~\ref{thm:nice_graphs_hard}, $H$ is not \nice{}. Hence, there exists a pair of adjacent vertices $v$ and $u$ in $\vertexset[H]$ such that the complete core $K^{v,u}$ is \isomorphic[bip] to $K_{a, k\cdot p}$, where $a$ and $k$ are positive integers. We apply Lemma~\ref{lem:construct_hard_path} and claim that it suffices to argue on the case that $a$ is congruent modulo $p$ to $0$.
	
	Let $w$ be the vertex in $\partof[K^{v,u}](u)$ given by Lemma~\ref{lem:construct_hard_path}, we distinguish the cases according to said lemma. If case (c) holds, then the path $(w,v)$ is a partially hard $p$-mosaic path because by Observation~\ref{obs:complete_core_same} the complete core $K^{w,v}$ is \isomorphic[bip] to $K_{k\cdot p, a}$. If case (a) holds, then $H$ contains a maximal $p$-hardness path starting at $w$, which yields a partially hard $p$-mosaic path due to Corollary~\ref{cor:construct_partially_hard_mosaic}. If case (b) holds, then let $v'$ be the neighbour of $w$ given by the case. We have that $\deg(w)$ is congruent modulo $p$ to $a$ and the complete core $K^{w,v'}$ has cardinality $\abs{\partof[K^{w,v'}](v')}$ congruent modulo $p$ to $0$. The path $(w,v')$ is a partially hard $p$-mosaic path if $a$ is not congruent modulo $p$ to $0$.
	
	With the claim established, we assume without loss of generality that every complete core $K$ in $H$ with a multiple of $p$ vertices in one part has also a multiple of $p$ vertices in the other part. Further, as we have seen, in both parts of $K$ exists a vertex that is incident to another complete core $K'$ with a multiple of $p$ vertices in both parts. This incidence is given by an intersection at exactly one vertex. In this way, we are going to show the existence of a closed $p$-hard thick walk, a contradiction by Corollary~\ref{cor:walks_hard}.
	
	Let $w_0$ and $w_1$ be a fixed pair of adjacent vertices such that the complete core $K^{w_0, w_1}$ has a multiple of $p$ vertices in both parts. It follows that there exists a neighbour $w_2$ of $w_1$ such that the complete core $K^{w_1, w_2}$ has also a multiple of $p$ vertices in both parts. Iteratively repeating this construction, we obtain by the finiteness of $H$ a closed walk $W$. Without loss of generality, let $W$ be $W = (w_0, w_1, \dots, w_\ell, w_0)$. Since $W$ traverses distinct complete cores, $W$ has to contain more than $4$ different vertices. By the same reason, for every index $i \in \sqBrackets{0;\ell}$, the common neighbourhood $\neigh{w_i} \cap \neigh{w_{i+2}}$ is equal to $\set{w_{i+1}}$, where the indices are taken modulo $\ell+1$. Hence, $W$ is square-free. None of the vertices $w_i$ comes with a count $b_i$ larger than $1$. Therefore, $W$ is $p$-hard. An example was already given in Figure~\ref{fig:2-hard_walk} and further examples are given in Figure~\ref{fig:2-mosaic_cycle}. 
\end{proof}

\begin{figure}[t]
	\begin{center}
		\includegraphics[]{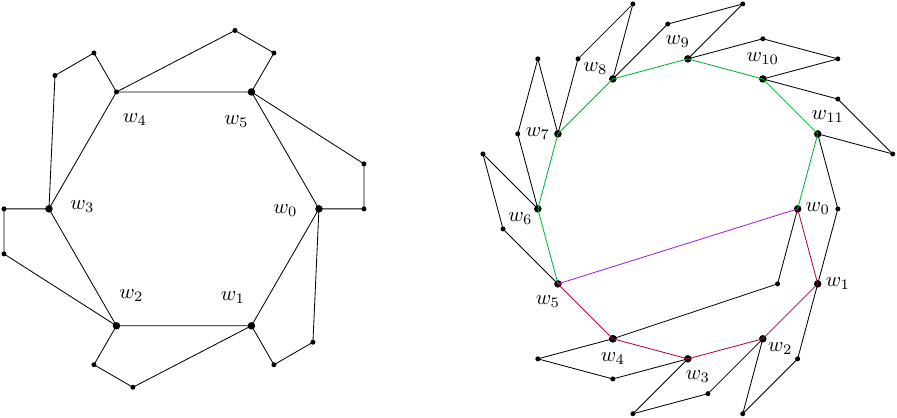}
	\end{center}
	\caption{For $p$ equal to $2$, illustration of closed $p$-mosaic paths, which yield $p$-hard walks. The closed walk given by the red path and the purple edge is not $p$-hard whereas the closed walk given by the green path and the purple edge is.}
	\label{fig:2-mosaic_cycle}
\end{figure}

With Lemma~\ref{lem:existence_partially_hard_p-mosaic} at hand, we have the existence of a partially hard $p$-mosaic path $Q$ in an order~$p$ bip-reduced bip-graph $H$ provided that it is \graphclass and does not admit a $p$-hardness gadget. Assuming $Q$ to be maximal we have by Lemma~\ref{lem:construct_hard_path} three cases: either both endvertices of $Q$ have degree not congruent modulo $p$ to $0$, or $Q$ yields a closed walk, or there exists a $p$-hardness path starting at the last vertex of $Q$. We initialize with the first case, in which we obtain a $p$-hardness gadget by Lemma~\ref{lem:p-mosaic_gadget}.

\begin{lemma}\label{lem:hard_p-mosaic}
	Let $p$ be a prime and $H$ be a connected order~$p$ bip-reduced bip-graph that is \graphclass{}. If $H$ contains a partially hard $p$-mosaic path $Q$, where $Q=(w_0, \dots, w_\ell)$, such that  $\abs{\partof[K^{w_{\ell-1}, w_\ell}](w_{\ell - 1})} \equiv 0 \pmod p$ but $\deg(w_\ell) \not \equiv 0 \pmod p$, then $H$ admits a $p$-hardness gadget.
\end{lemma}
\begin{proof}
	The lemma follows from the existence of a hard vertex by Lemma~\ref{lem:condition_new_hard_vertex} if $\ell$ is $1$. Otherwise, let $i \in \sqBrackets{2;\ell}$ be the minimal index such that the subpath $(w_0, \dots, w_i)$ satisfies the prerequisites of the lemma, that is modulo $p$ the cardinality $\abs{\partof[K^{w_{i-1}, w_i}](w_{i - 1})}$ is congruent to $0$ but $\deg(w_i)$ is not. It follows that the complete core $K^{w_{i-1}, w_i}$ has a multiple of $p$ vertices in both parts. We apply Lemma~\ref{lem:p-mosaic_gadget} to construct a \gadget{B,p} such that the induced sub-bip-graph $B$ of $H$ admits a $p$-hardness gadget.
	
	Let $\GadgetPart{\L}$ be a partially $H$ labelled bip-graph that selects the set $\selectSet[\L]$, where $\selectSet[\L] = (\neigh{w_{i-1}} \setminus K^{w_{i-2}, w_{i-1}}) \cupdot \set{w_{i-2}}$. By the assumption on the minimality of $i$, it follows that $\selectSet[\L]$ contains the vertices in $\partof[K^{w_{i-1}, w_i}](w_{i })$ and $\selectSet[\L]$ has cardinality not congruent modulo $p$ to $0$. Let $\GadgetPart{\R}$ be a partially $H$ labelled bip-graph given by Lemma~\ref{lem:complete_core_gadget} that selects the set $\selectSet[\R]$, where $\selectSet[\R] = \neigh{w_{i}}$.
	
	We conclude that the sub-bip-graph $B$ of $H$ induced by $\selectSet[\L]$ and $\selectSet[\R]$ consists of the complete core $K^{w_{i-1}, w_i}$ together with $\abs{\selectSet[\L]} - \abs{\partof[K^{w_{i-1}, w_i}](w_{i })}$ leaves adjacent to $w_{i-1}$ and $\deg(w_i) - \abs{\partof[K^{w_{i-1}, w_i}](w_{i-1 })}$ leaves adjacent to $w_i$. These numbers of leaves are not congruent modulo $p$ to $0$. It follows that the order~$p$ bip-reduced form $\bipreduced{B}$ is a bip-graph of radius at most $2$ that is not complete bipartite. By Theorem~\ref{thm:small graphs_weird_hard}, the lemma follows. 
\end{proof}	

Second, we show that also the case of a closed partially hard $p$-mosaic path yields a $p$-hardness gadget. It follows then that only the third case remains.

\begin{lemma}\label{lem:mosaic_gives_path}
	Let $p$ be a prime and $H$ be a connected order~$p$ bip-reduced bip-graph that is \graphclass{}. Further, let $H$ contain a vertex $w_0$ such that there exists a partially hard $p$-mosaic path starting at $w_0$. If $H$ admits no $p$-hardness gadget, then there exists a maximal $p$-mosaic path $Q$ starting at $w_0$ and a $p$-hardness path $P$ starting at a vertex $w$ in $Q$ with $w$ not equal to $w_0$.
\end{lemma}
\begin{proof}
	We assume toward contradiction that, for every partially hard $p$-mosaic path $Q$ starting at $w_0$, there exists no $p$-hardness path starting at a vertex $w$ in $Q$ such that $w$ is not $w_0$. 
	Let $Q$ be a maximal $p$-mosaic path starting at $w_0$, where $Q=(w_0, \dots, w_\ell)$. Further, let $w$ be the vertex in $\partof[K^{w_{\ell-1}, w_\ell}(w_\ell)]$ given by Lemma~\ref{lem:construct_hard_path}; we assume without loss of generality that $w$ is $w_\ell$. We obtain three cases for $w_\ell$ by Lemma~\ref{lem:construct_hard_path}, of which the first cannot apply due to the absence of a $p$-hardness path starting at $w_\ell$ and the third cannot apply because this would give a $p$-hardness gadget by Lemma~\ref{lem:hard_p-mosaic}. Hence, the second case applies, and we obtain a vertex $v'$ in $\neigh{w_\ell}$ that is not in the complete core $K^{w_{\ell-1}, w_\ell}$ such that the complete core $K^{w_\ell, v'}$ has a multiple of $p$ vertices in $\partof[K^{w_{\ell}, v'}](v')$. By the maximality of $Q$, it follows that the walk $(w_0,\dots, w_\ell, v')$ contains a closed path, that is $v'$ is equal to a vertex $w_i$ in $Q$. We claim that this yields a closed $p$-hard walk, a contradiction by Corollary~\ref{cor:walks_hard}.
	
	Since $v'$ and $w_{\ell-1}$ are in different components in the split of $\twoneigh{w_\ell}$ at $w_\ell$, it follows that $(w_0,\dots, w_\ell, v')$ contains more than $3$ different vertices. It suffices to show that this walk is \nice{}. Without loss of generality, let $v'$ be $w_0$. There are two cases. First, $w_\ell$ is not in the complete core $K^{w_0, w_1}$. The vertices $w_{\ell-1}$ and $w_{1}$ cannot be adjacent due to the maximality of the complete cores and therefore $(w_0,\dots, w_\ell, w_0)$ is \nice{}. For an illustration, we refer to the left part of Figure~\ref{fig:2-mosaic_cycle}.
	
	Second, $w_\ell$ is in the complete core $K^{w_0, w_1}$. Let $b_0$ be the cardinality of $\partof[K^{w_0, w_1}](w_0)$. If $b_0$ is not congruent modulo $p$ to $0$, then we obtain the thick closed path $(w_0^{b_0}, w_1, \dots, w_\ell, w_0^{b_0})$, which is \nice{}. Otherwise, we are allowed to apply Lemma~\ref{lem:construct_hard_path} on the complete core $K^{w_\ell, w_0}$. In particular, this yields a second maximal $p$-mosaic path $Q'$ starting at $w_0$ disjoint from $Q$, that is $Q' = (w'_0, w'_1, \dots, w'_r)$ and $w'_0$ is $w_0$. By the degree of $w_0$, it follows that also $Q'$ is partially $p$-hard. For an illustration, we refer to Figure~\ref{fig:2-hard_walk}.
	
	We briefly argue on the case that $w'_r$ is adjacent to a vertex in $Q$ as illustrated in the right part of Figure~\ref{fig:2-mosaic_cycle}. Let $w_j$ be the vertex in $Q$ adjacent to $w'_r$. It follows that $w'_r$ cannot be in both complete cores $K^{w_j, w_{j+1}}$ and $K^{w_{j}, w_{j-1}}$. In the first case, the closed walk obtained from concatenating $(w_j, \dots, w_\ell)$ with $ (w'_0, w'_1, \dots, w'_r, w_j)$ is \nice{}. In the second case, the closed walk obtained from concatenating $(w_j, \dots, w_0)$ with $ (w'_0, w'_1, \dots, w'_r, w_j)$ is \nice{} because the reverse of a $p$-mosaic path is \nice{}.
	
	Hence, $w'_r$ is not adjacent to a vertex in $Q$. Analogue to our argumentation on $Q$, we deduce that there exists a vertex $w'_j$ in $Q'$ such that $(w'_0, w'_1, \dots, w'_r, w'_j)$ contains a closed path. We join $Q$ and $Q'$ and obtain the closed walk $Q \cup Q' \cup (w'_j, \dots, w'_0)$, which we denote by $W$. This walk $W$ decomposes into the cycle $(w_0,\dots, w_\ell, w_0)$ followed by the path $(w'_0, \dots, w'_j)$ followed by the cycle $(w'_j, \dots, w'_\ell, w'_j)$ followed by the reverse path $(w'_j, \dots, w'_0)$. At any of the two concatenation points $w_0$ and $w'_i$, two distinct complete cores are concatenated. For instance, $w_\ell$ and $w'_0$ are not in the same component in the split of $\twoneigh{w_0}$ at $w_0$. Since $Q$ and $Q'$ along with their respective reverse paths are by construction \nice{}, it follows that $W$ is \nice{}.
\end{proof}

We conclude that if $H$ contains a partially hard $p$-mosaic path starting at a vertex $w_0$, then we find a $p$-mosaic path starting at $w_0$ such that starting at the last vertex there is a $p$-hardness path, and for every such $p$-hardness path, starting at its last vertex there is a partially hard $p$-mosaic path. By the finiteness of $H$, we obtain a closed walk and are going to show the existence of a $p$-hardness gadget by Corollary~\ref{cor:walks_hard}.

\subsubsection*{Hardness for Bip-Graphs That are not \Nice{}}
\begin{theorem}\label{thm:K_a,p_yields_hard}
	Let $p$ be a prime and $H$ be a connected order~$p$ bip-reduced bip-graph that is \graphclass{}. If $H$ is not \nice{}, then $H$ admits a $p$-hardness gadget.
\end{theorem}
\begin{proof}
	We assume toward contradiction that $H$ admits no $p$-hardness gadget. By Lemma~\ref{lem:existence_partially_hard_p-mosaic}, there exists a partially hard $p$-mosaic path $Q^1$ in $H$, where $Q^1=(w_0,w_1, \dots, w_\ell)$. By Lemma~\ref{lem:mosaic_gives_path}, we assume without loss of generality that there exists a $p$-hardness path $P^1$ starting at $w_\ell$ with endvertex $v_{\ell'}$. Let $P^1$ be $(v_0^{b_1}, v^{b_1}_1, \dots, v_{\ell'}^{b_{\ell'}})$, where we note that $v_0$ is equal to $w_\ell$. Corollary~\ref{cor:construct_partially_hard_mosaic} gives a partially hard $p$-mosaic path $Q^2$ starting at $v_{\ell'}$. Iteratively applying this argumentation, we obtain a walk $W$, where $W$ is the concatenation $Q^1 \, P^1 \, Q^2 \, P^2 \dots $ of partially hard $p$-mosaic paths and $p$-hardness paths. By the finiteness of $H$, this concatenation has to give a closed walk. Let $W$ be the minimal such closed walk of the form $W=( w_0, w_1, \dots ,w_{\ell-1}, w_\ell = v_0^{b_0}, v^{b_1}_1, \dots, v_{\ell'}^{b_{\ell'}} = w_0)$. We note that every $p$-mosaic path and every $p$-hardness path is a square-free thick walk. With this, we are going to show that $W$ is a $p$-hard thick walk, and Corollary~\ref{cor:walks_hard} yields a contradiction.
	
	We assume without loss of generality that $W= Q^1 \, P^1 \, \dots \, P^\ell	$, i.e. $W=( w_0, w_1, \dots ,w_{\ell-1}, w_\ell = v_0^{b_0}, v^{b_1}_1, \dots, v_{\ell'}^{b_{\ell'}} = w_0)$. The walk $W$ has to contain at least three vertices given by the minimal length of a $p$-hardness path. Since $W$ yields a cycle and $H$ is bipartite, it follows that the twin-free form of $W$ has to contain at least $4$ different vertices. Further, the paths concatenate such that $W$ continues to be square-free. More precise, for every concatenation point $(w_{\ell-1}, w_\ell = v_0^{b_0}, v^{b_1}_1)$ the vertex $w_{\ell-1}$ is in a different component of the split of $\twoneigh{v_0}$ at $v_0$ than the vertices in $v^{b_1}_1$. The reason is that, for the representative $v_1$ of $v_1^{b_1}$, the complete core $K^{v_0,v_1}$ has in neither parts a multiple of $p$ vertices whereas the complete core $K^{w_{\ell-1}, w_\ell}$ has a part with a multiple of $p$ vertices. None of the counts $b_i$ of the $p$-hardness paths is congruent modulo $p$ to $0$. Therefore, the walk $W$ is $p$-hard. 
\end{proof}

\subsection{Hardness for \Graphclass{} Graphs}
With these results at hand, we finally show the main theorem of this section. Here, we explicitly allow $H$ to be disconnected.
\HardnessGraphclass*
	\begin{proof}
		By Lemma~\ref{lem:bip_components} it suffices to consider only connected graphs $H$.
		By Corollary~\ref{cor:hardness gadget} it suffices to show that $H$ admits a $p$-hardness gadget. If $H$ is \nice{}, then by Theorem~\ref{thm:nice_graphs_hard} $H$ admits a $p$-hardness gadget. Otherwise, we obtain by Theorem~\ref{thm:K_a,p_yields_hard} that $H$ admits a $p$-hardness gadget. This establishes the result.
\end{proof}

\section{(Partially) Surjective Homomorphisms}
\label{sec:surjective_homomorphisms}
In the last section, we turn toward an application of our insights on quantum graphs in order to highlight the importance of a dichotomy for $\probNumHom{H}[p]$.
Let $H$ be a graph with set of distinguished vertices $\distVertices$ in $\vertexset[H]$ and set of distinguished edges $\distEdges$ in $\edgeset[H]$. For a graph $G$, a homomorphism $f$ in $\Hom[G, H]$ is called \emph{partially surjective} with respect to $\distVertices$ and $\distEdges$ if, 
\begin{itemize}
	\item for all vertices $v$ in $\distVertices$, there exists a vertex $x$ in $\vertexset[G]$ such that $f(x) = v$;
	\item for all edges $(u,v)$ in $\distEdges$, there exists an edge $(x,y)$ in $\edgeset[G]$ such that $f(x) = u$ and $f(y) =v$.
\end{itemize}

We implicitly assume the endvertices of distinguished edges to be distinguished because homomorphisms are defined as maps on vertices.
The set of partially surjective homomorphisms from $G$ to $H$ with respect to $\distVertices$ and $\distEdges$ is denoted by $\PartSurj[G, (H, \distVertices, \distEdges)]$, the number of such homomorphisms is denoted by $\numPartSurj[G, (H, \distVertices, \distEdges)]$, and for a modulus $p$ the number of such homomorphisms modulo $p$ is denoted by $\numPartSurj[G, (H, \distVertices, \distEdges)][p]$.

We study the following problem
\prob%
{$\probNumPartSurjHom{H, \distVertices, \distEdges}[p]$.}
{Prime $p$, Graph~$H$, and set of distinguished elements $\distVertices$ and $\distEdges$, where $\distVertices \subseteq \vertexset[H]$, $\distEdges \subseteq \edgeset[H]$.}
{Graph $G$.}
{$\numPartSurj[G, (H, \distVertices, \distEdges)][p]$.}

Similarly to Theorem~\ref{thm:reduction_cancel_order_p_auto}, we aim to reduce the problem to instances with an order~$p$ reduced target graph $H$. However, the additional property of a vertex being distinguished prevents us from applying the reduction from $\normreduced{H}$. Given a graph with distinguished vertices and edges $(H, \distVertices, \distEdges)$, in order to arrive at $\normreduced{H}$ we studied the fixed vertices under an arbitrary automorphism in $\Aut[H]$. However, this disregards distinguished vertices and edges. For our purposes, we are interested in the subgroup $\AutDist[H]$ of $\Aut[H]$ that consists of \emph{dist-automorphisms}, automorphisms that map distinguished elements bijectively to distinguished elements. This is formally defined by
\begin{align*}
	\AutDist[H] = \{\varrho \in \Aut[H] \SetSymbol &\text{for all }u \in \vertexset[H], \, \varrho(v) \in \distVertices \text{ iff } v \in \distVertices \text{ and } \\ &\text{for all } (u,v) \in \edgeset[H],\, (\varrho(u), \varrho(v)) \in \distEdges \text{ iff } (u,v)\in \distEdges \} .
\end{align*}
The composition of two automorphisms in $\AutDist[H]$ is again an element in $\AutDist[H]$, and thus $\AutDist[H]$ is a subgroup of $\Aut[H]$. Similarly, two graphs with distinguished elements $(H, \distVertices, \distEdges)$ and $(H_1, V_1^{\dist}, E_1^{\dist})$ are called \emph{\isomorphic[dist]} if there exists an isomorphism $\phi \colon H \to H_1$ that maps bijectively  $\distVertices$ to $V_1^{\dist}$ and $\distEdges$ to $E_1^{\dist}$. Such an isomorphism $\phi$ is called a \emph{$\dist$-isomorphism}. If $(H, \distVertices, \distEdges)$ and $(H_1, V_1^{\dist}, E_1^{\dist})$ are \isomorphic[dist], then we denote this by $(H, \distVertices, \distEdges) \congdist (H_1, V_1^{\dist}, E_1^{\dist})$.

When we apply the notion $\AutDist[H]$ we keep notation short and assume that the graph $H$ has sets of distinguished vertices and edges. For instance, if both sets are empty or identical to the whole sets of vertices and edges, then $\AutDist[H]$ is identical to $\Aut[H]$. For every vertex $v$ of $H$,
the subgroup $\AutDist[H]$ induces an orbit $\OrbDist[v]$ under automorphisms in $\AutDist[H]$ by $\OrbDist[v] = \set{\varrho(v) \in \vertexset[H] \given \varrho \in \AutDist[H]}$.

\begin{lemma}\label{lem:distinguishing_easy_elements_is_easy}
	Let $p$ be a prime and $H$ be a graph with set of distinguished vertices $\distVertices$ in $\vertexset[H]$ and set of distinguished edges $\distEdges$ in $\edgeset[H]$. Moreover, let $\varrho$ be an automorphism in $\AutDist[H]$ of order $p$ and $H^{\varrho}$ be the subgraph induced by the fixed points of $\varrho$. If there exists a vertex $v$ in $\distVertices$ such that $v$ is not in $\vertexset[H^{\varrho}]$, 
	then $\numPartSurj[G, (H, \distVertices, \distEdges)][p]$ is equal to $0$.
\end{lemma}
\begin{proof}
	If $\numPartSurj[G, (H, \distVertices, \distEdges)]$ is equal to $0$, then the claim holds trivially. Otherwise, let $f$  be a homomorphism in $\PartSurj[G, (H, \distVertices, \distEdges)]$ and let $\Aut_\varrho$ be the cyclic subgroup of $\AutDist[H]$ generated by $\varrho$.
	
	Let $v$ be a distinguished vertex in $\distVertices $ such that $v$ is not in $\vertexset[H^{\varrho}]$, hence $\varrho(v) \neq v$. We recall from Section~\ref{sec:prelims} that $\Orb[v][\varrho]$ denotes the set $\set{\varrho^i(v)}_{i\in \sqBrackets{p}}$. In fact, $\Orb[v][\varrho]$ is a subset of $\OrbDist[v]$. By the Orbit-Stabilizer Theorem~\ref{thm:orbit_stabilizer}, it follows that the cardinality $\abs{\Orb[v][\varrho]}$ divides the order $\abs{\Aut_\varrho}$, which is equal to $p$. Since $v$ is not fixed and $p$ is a prime, we deduce that $\abs{\Orb[v][\varrho]} =p$. For every homomorphism $f$ in $\PartSurj[G, (H, \distVertices, \distEdges)]$, we construct a set of $p$ homomorphisms $\set{f_i}_{i \in \sqBrackets{p}}$ in $\PartSurj[G, (H, \distVertices, \distEdges)]$ by
	\[
		f_i (x) = (\varrho^i \circ f)(x).
	\]
	In fact, for every index $i\in \sqBrackets{p}$ and every vertex $x$ in $\vertexset[G]$ such that $f(x)$ is in $H^{\varrho}$, it holds $\varrho^i(f(x)) =f(x)$, and hence $f(x) = f_i(x)$. It follows that $f$ and $f_i$ might only differ on the set of vertices not mapped to $H^\varrho$. We briefly argue that $f_i$  is indeed a homomorphism in $\PartSurj[G, (H, \distVertices, \distEdges)]$.
	
	By construction, $f_i$ is a homomorphism in $\Hom[G,H]$. Let $v$ be a distinguished vertex in $\distVertices$. Since $\varrho$ is an automorphism bijective on distinguished vertices, we have that $\varrho^{-i}(v)$ is in $\distVertices$ and there exists a vertex $x$ in $\vertexset[G]$ with $f(x) = \varrho^{-i}(v)$ because $f$ is in $\PartSurj[G, (H, \distVertices, \distEdges)]$. By the definition of $f_i$ it follows that $f_i(x) = \varrho^{i}(\varrho^{-i}(v))=v$. Since $\varrho$ is also bijective on the set of distinguished edges, the same argumentation shows that, for every distinguished edge $(u,v)$ in $\distEdges$, there exists an edge $(x,y)$ in $\edgeset[G]$ with $(f_i(x), f_i(y)) = (u,v)$.
	
	There exists by assumption a distinguished vertex $v$ in $\distVertices$ that is not in $H^{\varrho}$. We obtain that every homomorphism $f$ in $\PartSurj[G, (H, \distVertices, \distEdges)]$ yields a set of $p$ distinct homomorphisms $\set{f_i}_{i \in \sqBrackets{p}}$ in $\PartSurj[G, (H, \distVertices, \distEdges)]$. For two different homomorphisms $f$ and $f'$ in $\PartSurj[G, (H, \distVertices, \distEdges)]$, we have, for any pair of indices $i$ and $j \in \sqBrackets{p}$, that $\varrho^i \circ f$ is unequal to $\varrho^j \circ f'$. Therefore, the number of homomorphisms $\PartSurj[G, (H, \distVertices, \distEdges)]$ is congruent modulo $p$ to $0$.
\end{proof}

In the proof of Lemma~\ref{lem:distinguishing_easy_elements_is_easy}, it was crucial for the automorphism $\varrho$ to map bijectively distinguished elements to distinguished elements. Otherwise, the composition with a valid homomorphism $f$ in $\PartSurj[G, (H, \distVertices, \distEdges)]$ does not necessarily yield another valid homomorphism in $\PartSurj[G, (H, \distVertices, \distEdges)]$. We note that for an automorphism $\varrho$ in $\AutDist[H]$ of order $p$, the set of distinguished edges $\distEdges$ also has to be contained in ${H}^\varrho$. If a distinguished edge $e$ in $\distEdges$ is not contained in $\edgeset[{H}^\varrho]$, then we conclude by the absence of multiple edges that $\distVertices$ is not contained in $\vertexset[{H}^\varrho]$. Therefore, the classification of Lemma~\ref{lem:distinguishing_easy_elements_is_easy} concerning distinguished vertices is sufficient for our needs. An extension to distinguished edges would be useful for the setting of homomorphisms allowing multiple edges, which is beyond the scope of this paper.

For any prime $p$, it suffices to study graphs $H$ with distinguished vertices $\distVertices$ and edges $\distEdges$ such that no automorphism $\varrho$ in $\AutDist[H]$ of order~$p$ \enquote{cancels} a distinguished vertex due to Lemma~\ref{lem:distinguishing_easy_elements_is_easy}. If every automorphism of $H$ is also a dist-automorphism, then we adjust the binary relation $\!\relArrow[][p]\!$ and find by Theorem~\ref{thm:reduction_cancel_order_p_auto} and Theorem~\ref{thm:reduced_form_unique} that the \enquote{order~$p$ dist-reduced form} of $(H, \distVertices, \distEdges)$ is $(\normreduced{H}, \distVertices, \distEdges)$, which is unique up to dist-isomorphism. The reason for the quotation marks will become apparent after the following corollary. We provide some technical details.

Let $H_1$ be a graph with distinguished vertices $V_1^{\dist}$ and distinguished edges $E_1^{\dist}$. We denote $(H, \distVertices, \distEdges) \relArrow[\dist][p] (H_1, V_1^{\dist}, E_1^{\dist})$ if there exists an automorphism $\varrho$ in $\AutDist[H]$ of order $p$ such that $H_1 \cong H^{\varrho}$. By the assumptions on $\AutDist[H]$, it follows that $(H_1, V_1^{\dist}, E_1^{\dist})$ is dist-isomorphic to $(H^\varrho, \distVertices, \distEdges)$. For an arbitrary enumeration of $\AutDist[H]$, the terminal object of the relation $\!\relArrow[\dist][p]\!$ is by  Theorem~\ref{thm:reduced_form_unique} dist-isomorphic to $(\normreduced{H}, \distVertices, \distEdges)$.
\begin{corollary}\label{cor:mod-p_reduction_part_surj_homs}
	Let $p$ be a prime and $H$ be a graph with distinguished vertices $\distVertices$ in $\vertexset[H]$ and distinguished edges $\distEdges$ in $\edgeset[H]$. Further, let the groups $\AutDist[H]$ and $\Aut[H]$ be equal.
	\begin{itemize}
		\item If there exists an automorphism $\varrho$ in $\AutDist[H]$ of order $p$ such that $\distVertices$ is not contained in $\vertexset[H^\varrho]$, then $\probNumPartSurjHom{H, \distVertices, \distEdges}[p]$ is solvable in constant time;
		\item otherwise, for every graph $G$, the number of partially surjective homomorphisms $\numPartSurj[G, (H, \allowbreak \distVertices, \distEdges)]$ is congruent modulo $p$ to $\numPartSurj[G, (\normreduced{H}, \distVertices, \distEdges)]$.
	\end{itemize}
\end{corollary}

We deliberately stated this reduction for the special case in which $\Aut[H]$ and $\AutDist[H]$ are equal. In general, we can still apply the binary relation $\!\relArrow[\dist][p]\!$, but the uniqueness of the terminal object does not follow straightforwardly. In a nutshell, the chain of relations $\!\relArrow[\dist][p]\!$ terminates earlier compared to the chain of relations $\!\relArrow[p]\!$ provided that $\Aut[H]$ and $\AutDist[H]$ are unequal. On a technical level, the problem $\probNumPartSurjHom{H, \distVertices, \distEdges}$ relates to an asymmetric homomorphism-type problem in which the domain is more restricted than the range; any graph $G$ is equivalently a graph with an empty set of distinguished elements. Such an asymmetry hinders us from applying the same technique we employed in Section~\ref{sec:quantum_graphs}, for example, to obtain Observation~\ref{obs:p-wise_reduced_unique}. We see no reason to doubt that also in the general case the terminal object of $\!\relArrow[\dist][p]\!$ is unique up to dist-isomorphism -- leading to a proper definition of \enquote{the order~$p$ dist-reduced form}. Alas, for the sake of this paper, this result is not needed since we now translate the problem $\probNumPartSurjHom{H, \distVertices, \distEdges}[p]$ to a quantum-homomorphism problem.

For a target graph with distinguished vertices and edges $(H, \distVertices, \distEdges)$, we define the family $\family{D}(H)$ to consist of subgraphs $H'$ of $H$ such that $H'$ is obtainable from $H$ by deleting a set of distinguished elements, i.e. there exists a set of vertices $V'$ in $\distVertices$ and a set of edges $E'$ in $\distEdges$ such that $H'$, given by $(\vertexset[H'], \edgeset[H'])$, is equal to $(\vertexset[H] \setminus V',\,  \edgeset[H] \setminus E' )$.

Let the number of distinguished vertices be $n_d$ and the number of distinguished edges be $m_d$. The family $\family{D}(H)$ decomposes into the families $\family{D}_{i,j}(H)$, where the index $i$ is in $\sqBrackets{0; n_d}$, the index $j$ is in $\sqBrackets{0; m_d}$, and every element in $\family{D}_{i,j}(H)$ is obtainable from $H$ by deleting first $j$ distinguished edges and then $i$ distinguished vertices from $H$ such that no deleted vertex is incident to a remaining distinguished edge. Here, the deletion of an edge does not affect the set of vertices, but the deletion of a vertex affects the set of edges. Thus, we needed to fix the order of the deletion in order to obtain a decomposition of $\family{D}(H)$. This yields,
\begin{equation}\label{eq:decomposition_removing_dist_elements}
	\family{D}(H) = \bigcup_{i\in \sqBrackets{0;n_d}} \bigcup_{j\in \sqBrackets{0; m_d}} \family{D}_{i,j}(H) .
\end{equation}

\begin{example}
	Let $p$ be a prime and $H$ be an order~$p$ reduced graph containing an edge $e$ and a vertex $v$. We apply the inclusion-exclusion principle. Here, we denote set deletion by \enquote{$-$} for better readability. For any graph $G$, the set of partially surjective homomorphisms $\PartSurj[G, (H, \set{v}, \set{e})]$ is equal to
	\[
	\PartSurj[G, (H, \set{v})] - \parenthesis[\big]{ \PartSurj[G, (H, \set{v})] \cap \Hom[G, H - \set{e}] }.
	\]
	Since $\PartSurj[G, (H, \set{v})]$ is equal to $\Hom[G, H] - \Hom[G, H - \set{v}]$, we obtain
	\begin{align*}
	\PartSurj[G, (H, \set{v}, \set{e})] &= 
	\parenthesis[\big]{ \Hom[G, H] - \Hom[G, H - \set{v}]} \\
	&\phantom{=} - \parenthesis[\big]{ \Hom[G, H- \set{e}] - \Hom[G, H - \set{v, e}] } .
	\end{align*}
\end{example}

For a graph $H$ with set of $n_d$ distinguished vertices $\distVertices$ and set of $m_d$ distinguished edges $\distEdges$, the family $\family{D}(H)$ and the families $\family{D}_{i,j}(H)$, where $i \in \sqBrackets{0; n_d}$ and $j \in \sqBrackets{0; m_d}$, are further decomposed by isomorphism classes.
We denote by $\family{D}^\ast (H)$ the family of isomorphism classes in $\family{D}(H)$, and let $\family{D}^{\ast}_{i,j}(H)$ be the family of isomorphism classes in $\family{D}_{i,j}(H)$.

Chen et al. \cite{Chen:19:The_Exponential-Time_Complexity} gave a lemma similar to the following.
\begin{lemma}\label{lem:part_surj_quantum_graph_2}
	Let $H$ be a graph with $n_d$ distinguished vertices $\distVertices$ in $\vertexset[H]$ and $m_d$ distinguished edges $\distEdges$ in $\edgeset[H]$. Let the quantum graph $\quantum{F}$ consist of the family of constituents $\family{D}^\ast(H)$ and set of affiliated coefficients $\set{\alpha_F}_{F \in \family{D}^\ast(H)}$, where, for $F$ in $\family{D}^\ast_{i,j}(H)$ with $i \in \sqBrackets{0; n_d}$ and $j \in \sqBrackets{0; m_d}$, the coefficient $\alpha_F$ is $(-1)^{i+j} \cdot \abs{\set{F' \in \family{D}_{i,j}(H) \given F \cong F'}}$. For all graphs $G$,
	\[
	\numPartSurj[G, (H, \distVertices, \distEdges)] = \numHom[G, \quantum{F}] .
	\]
\end{lemma}
\begin{proof}
	The lemma is a result of the inclusion-exclusion principle.	We enumerate the set of distinguished vertices and edges. In this way, for $i \in \sqBrackets{0; n_d}$ and $j \in \sqBrackets{0; m_d}$, let $v_i$ be the $i$-th distinguished vertex and $e_j$ the $j$-th distinguished edge. 
	For every $i \in \sqBrackets{0; n_d}$ and $j \in \sqBrackets{0; m_d}$, let $\overline{A}_i$ be the set $\PartSurj[G, (H, \set{v_i})]$ of homomorphisms surjective on $v_i$ and $\overline{B}_j$ be the set $\PartSurj[G, (H , \set{e_j})]$ of homomorphisms surjective on $e_j$. Any homomorphism $f$ in $\PartSurj[G, (H, \distVertices, \distEdges)]$ has to be in the intersection of the sets $\overline{A}_i$ and $\overline{B}_j$. Therefore,
	\[
		\numPartSurj[G, (H, \distVertices, \distEdges)] = \abs[\Big]{ \parenthesis[\Big]{ \bigcap_{i \in \sqBrackets{n_d}} \overline{A}_i} \cap \parenthesis[\Big]{ \bigcap_{j \in \sqBrackets{m_d}} \overline{B}_j}}.
	\]
	The complement of $\overline{A}_i$ is the set $\Hom[G, H\setminus \set{v_i}]$ denoted by $A_i$. Similarly, the complement of $\overline{B}_j$ is the set $\Hom[G, H\setminus \set{e_j}]$ denoted by $B_j$. By \emph{De Morgan's law}, we obtain
	\begin{align*}
		\numPartSurj[G, (H, \distVertices, \distEdges)]	=
		& \sum_{i\in \sqBrackets{0;n_d}} \sum_{j\in \sqBrackets{0;m_d}} (-1)^{i + j} \cdot \sum_{{\substack{V_i \subseteq \distVertices \\ \abs{V_i}=i}}} \sum_{{\substack{E_j \subseteq \distEdges \\ \abs{E_j}=j}}} \numHom[G, (H\setminus (V_i \cup E_j))].
	\end{align*}
	
	Any term $\numHom[G, (H\setminus (V_i \cup E_j))]$ is equal to $\numHom[G, F]$, where $F$ is the induced subgraph in $H$ obtained by deleting $(V_i \cup E_j)$, thus $F$ is in $\family{D}_{i,j}(H)$. Collecting terms for pairwise isomorphic graphs, we obtain that every graph in $\set{F' \in \family{D}_{i,j}(H) \given F \cong F'}$
	contributes the term $\numHom[G, F]$. Let the constant $\alpha'_F$ be $\abs{\set{F' \in \family{D}_{i,j}(H) \given F \cong F'}}$. Therefore, 
	\begin{align*}
		\numPartSurj[G, (H, \distVertices, \distEdges)]	=\sum_{i \in \sqBrackets{0;n_d}} \sum_{j \in \sqBrackets{0;m_d}} (-1)^{i+j} \cdot \sum_{F \in \family{D}^\ast_{i,j}(H)}  \alpha'_F \cdot \numHom[G, F] ,
	\end{align*}
	which establishes the lemma.
\end{proof}

By Lemma~\ref{lem:part_surj_quantum_graph_2}, any problem $\probNumPartSurjHom{H, \distVertices, \distEdges}$ is equivalent to a quantum graph problem $\probNumHom{\quantum{F}}$, where $\probNumPartSurjHom{H, \distVertices, \distEdges}$ denotes the non-modular version of $\probNumPartSurjHom{H, \distVertices, \distEdges}[p]$. Moreover, for every constituent $F$ of $\quantum{F}$ the constant $\alpha_F$ is an integer. Therefore, for any prime $p$, the value $\alpha_F \pmod p$ is well-defined.
We deduce that also any problem $\probNumPartSurjHom{H, \distVertices, \distEdges}[p]$ is equivalent to a quantum graph problem $\probNumHom{\quantum{F}}[p]$.

We recall that constituents $F$ of $\quantum{F}$ do not have to be order~$p$ reduced and $\probNumHom{\quantum{F}}[p]$ is equivalent to $\probNumHom{\normreduced{\quantum{F}}}[p]$, where $\normreduced{\quantum{F}}$ is the order~$p$ reduced quantum graph obtained from $\quantum{F}$. Also, we recall that the constituents of $\normreduced{\quantum{F}}$ have constant $\alpha_F$ not congruent modulo $p$ to $0$. This quantum graph allows us to apply Theorem~\ref{thm:homs_to_quantum_graph_mod_p}. A dichotomy for $\probNumHom{H}[p]$ then implies a dichotomy for $\probNumPartSurjHom{H, \distVertices, \distEdges}[p]$ as follows.

Let us recall the conjecture on $\probNumHom{H}[p]$ we stated in the introduction.
\FabenJerrump*

We also recall that Bulatov and Kazeminia~\cite{Bulatov:22:Complexity_Classification_of_Homs_mod} proved this conjecture. This allows us to obtain a dichotomy for $\probNumPartSurjHom{H, \distVertices, \distEdges}[p]$. As noted in the Introduction, in the extended abstract of this paper~\cite{Lagodzinski:21:On_Counting_Quantum-Graph_Homomorphisms} the following results were stated conditional on the conjecture, which we adjusted now that the conjecture is proven.

\begin{corollary}\label{cor:partsurjhoms_equiv_homs}
	Let $p$ be a prime and $H$ a graph with set of distinguished vertices $\distVertices$ in $\vertexset[H]$ and set of distinguished edges $\distEdges$ in $\edgeset[H]$. Let $\quantum{F}$ be the quantum graph defined in Lemma~\ref{lem:part_surj_quantum_graph_2} with order~$p$ reduced form $\normreduced{\quantum{F}}$. If every constituent $F$ of $\normreduced{\quantum{F}}$ is a collection of complete bipartite graphs and reflexive complete graphs, then $\probNumPartSurjHom{H, \distVertices, \distEdges}[p]$ is solvable in polynomial time.
	Otherwise, $\probNumPartSurjHom{H, \distVertices, \distEdges}[p]$ is $\classNumP[p]$-hard.
\end{corollary}

Compared to \cite[Theorem~4]{Chen:19:The_Exponential-Time_Complexity} by Chen et al., who followed a similar line of argumentation for the non-modular problem $\probNumPartSurjHom{H, \distVertices, \distEdges}$, Corollary~\ref{cor:partsurjhoms_equiv_homs} does not state structural properties of $(H, \distVertices, \distEdges)$ testable for hardness. The reason is that Conjecture~\ref{conj:faben-jerrum_p} is only applicable to order~$p$ reduced graphs, i.e. to the order~$p$ reduced form $\normreduced{\quantum{F}}$ and not to the quantum graph $\quantum{F}$ itself. The reduction from $\normreduced{\quantum{F}}$ to $\quantum{F}$ might erase significant structure, even sources for hardness in the non-modular case. We discuss this further by the graph illustrated in Figure~\ref{fig:partsurj_reduction_example}.

\begin{figure}[t]
	\centering
	\includegraphics[]{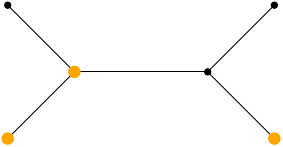}
	\caption{Example of a graph $H$ with distinguished vertices $\distVertices$ marked large and orange such that $(H, \distVertices)$ is order~$2$ reduced and $\probNumPartSurjHom{H, \distVertices, \distEdges}[2]$ is solvable in polynomial time.}
	\label{fig:partsurj_reduction_example}
\end{figure}
Indeed, the graph $(H, \distVertices)$ has no dist-automorphism of order~$2$ and is bipartite but not complete bipartite. By Chen et al.~\cite[Theorem~4]{Chen:19:The_Exponential-Time_Complexity}, the non-modular problem $\probNumPartSurjHom{H, \distVertices}$ is $\classNumP$-hard. However, regarding the modular version, the graph $H$ itself admits involutions. The constituents of the quantum graph $\quantum{F}$ such that $\probNumPartSurjHom{H, \distVertices, \distEdges}[2]$ is equivalent to $\probNumHom{\quantum{F}}[2]$ are given by the subgraphs $H'$ of $H$ obtained by deleting a (possibly empty) set of distinguished vertices. For every such subgraph $H'$, we observe that the order~$2$ reduced form $\normreduced{H'}[2]$ contains at most one vertex, and thus $\probNumHom{\quantum{F}}[2]$ is solvable in polynomial time.

The core of the difference between the non-modular version \cite[Theorem~4]{Chen:19:The_Exponential-Time_Complexity} and the modular version Corollary~\ref{cor:partsurjhoms_equiv_homs} is the possible difference between the sets $\AutDist$ and $\Aut$. For a graph with distinguished vertices and edges $(H, \distVertices, \distEdges)$, we recall that $\AutDist[H]$ and $\Aut[H]$ can be identical, for instance in the extreme cases that $\distVertices$ is empty or equal to $\vertexset[H]$, and $\distEdges$ is empty or equal to $\edgeset[H]$. If the sets $\AutDist[H]$ and $\Aut[H]$ are identical, then we are allowed to apply Corollary~\ref{cor:mod-p_reduction_part_surj_homs} and able to give precise structural statements in the modular case. These depend only on the set $\family{D}^\ast (H)$ compared to Corollary~\ref{cor:partsurjhoms_equiv_homs}, which requires a computation of the order~$p$ reduced quantum graph $\normreduced{\quantum{F}}$. We recall that the latter consists of constituents in $\family{D}^\ast (H)$ but not necessarily the whole set.

\begin{theorem}\label{thm:dichotomy_part_surj_homs}
Let $p$ be a prime and $H$ be a graph with order~$p$ reduced form $\normreduced{H}$, $n_d$ distinguished vertices $\distVertices$ in $\vertexset[H]$, and $m_d$ distinguished edges $\distEdges$ in $\edgeset[H]$. If the groups $\AutDist[H]$ and $\Aut[H]$ are equal, then
$\probNumPartSurjHom{H, \distVertices, \distEdges}[p]$ is solvable in polynomial time if
	\begin{enumerate}
		\item the set of distinguished vertices $\distVertices$ is not contained in $\vertexset[\normreduced{H}]$,
		\item or every graph $F$ in $\family{D}^\ast (\normreduced{H})$ is a collection of complete bipartite graphs and reflexive complete graphs. 
	\end{enumerate}
Otherwise, $\probNumPartSurjHom{H, \distVertices, \distEdges}[p]$ is $\classNumP[p]$-hard.
\end{theorem}
\begin{proof}
	We apply Corollary~\ref{cor:mod-p_reduction_part_surj_homs}, whose first case establishes the first case of this theorem. Regarding the second case, we obtain a parsimonious reduction from $\probNumPartSurjHom{\normreduced{H}, \distVertices, \distEdges}[p]$ to $\probNumPartSurjHom{H, \distVertices, \distEdges}[p]$. We assume in the following that $H$ is order~$p$ reduced.	
	
	We apply Corollary~\ref{cor:partsurjhoms_equiv_homs}, whose second case directs our attention to the constituents $F$ of the order~$p$ reduced quantum graph $\normreduced{\quantum{F}}$. In particular, $\quantum{F}$ is given by $\sum_{F \in \family{D}^\ast(H)} \alpha_F \cdot F$, where for a constituent $F$ in $\family{D}^\ast_{i,j}(H)$ with indices $i$ and $j$, the coefficient $\alpha_F$ is given by $\alpha_F = (-1)^{i+j} \cdot \abs{\set{F' \in \family{D}_{i,j}(H) \given F \cong F'}}$. Contrary to the more general Corollary~\ref{cor:partsurjhoms_equiv_homs}, we have more information about the constituents of $\normreduced{\quantum{F}}$ because $H$ itself is such a constituent with coefficient $\alpha_H$ equal to $1$. This is sufficient to obtain the second case.
	
	If $H$ itself is such that $\probNumHom{H}[p]$ is $\classNumP[p]$-hard, then $F$ equal to $H$ establishes hardness. 
	Otherwise, $H$ is such that $\probNumHom{H}[p]$ is solvable in polynomial time. By Conjecture~\ref{conj:faben-jerrum_p}, we deduce that $H$ is the disjoint union of connected components $C$ such that $C \cong K_{a,b}$ or $C \cong K^\circ_q$. Since $H$ is order~$p$ reduced, we obtain that $a$, $b$, and $q$ are in $\sqBrackets{0;p-1}$.	
	We show that only one graph $F$ in $\family{D}^\ast(H)$, such that $F$ is not a collection of complete bipartite graphs and reflexive complete graphs, is sufficient to imply that $\probNumPartSurjHom{H, \distVertices, \distEdges}[p]$ is $\classNumP[p]$-hard. To this end, we show that any such graph $F$ has an associated coefficient $\alpha_F$ that is not congruent modulo $p$ to $0$. Consequently, Corollary~\ref{cor:partsurjhoms_equiv_homs} yields the theorem
	
	The assumption that $\Aut[H]$ and $\AutDist[H]$ are equal gives that any connected component $C$, where $C \cong K_{a,b}$ or $C \cong K^\circ_q$ as above, either has no distinguished vertices or all vertices distinguished. For distinguished edges, the classification also distinguishes between loops and non-loops: either $C$ has no (non-)loops distinguished or all (non-)loops distinguished. Additionally, if a connected component $C'$ of $H$ is isomorphic to $C$, then $C'$ has the same number of distinguished vertices, distinguished loops, and distinguished non-loops as $C$ does. We distinguish cases.
	
	Let there exist a connected component $C$ that contains a distinguished edge $e$ in $\distEdges$ such that either $C \cong K_{a,b}$, where $a$ and $b$ are at least $2$, or $C\cong K^\circ_q$, where $q$ is at least $3$, or $C \cong K^\circ_2$ and $e$ is a loop. The subgraph $C^\ast$, where $C^\ast=C \setminus \set{e}$, is an order~$p$ reduced connected graph that is neither a complete bipartite graph nor a reflexive complete graph. The subgraph $F$ given by $F= H\setminus \set{e}$ contains $C^\ast$ as connected component. Further, $F$ is order~$p$ reduced because $H$ contains only connected components that are reflexive complete or complete bipartite. Hence, by Conjecture~\ref{conj:faben-jerrum_p} we obtain that $\probNumHom{F}[p]$ is $\classNumP[p]$-hard. By the construction of $F$, we assume without loss of generality that $F$ is in $\family{D}^{\ast}(H)$, i.e. $F$ is the representative of the respective isomorphism class in $\family{D}(H)$. It remains to show that the coefficient $\alpha_{F}$ is not congruent modulo $p$ to $0$, where we recall $\alpha_F =  (-1)^{1} \cdot \abs{\set{F' \in \family{D}_{0,1}(H) \given F\cong F'}}$. Since $a$, $b$, and $q$ are at most $p-1$, the connected component $C$ cannot contribute a multiple of $p$ to $\abs{\set{F' \in \family{D}_{0,1}(H) \given F \cong F'}}$. Similarly, for any other connected component $C'$ such that $C' \cong C$, the component $C'$ contributes the same number as $C$. Since $H$ is order~$p$ reduced, there are at most $p-1$ connected components $C'$ with $C' \cong C$, and thus these contribute an amount not congruent modulo $p$ to $0$. Finally, $F$ can only be constructed from $C$ by deleting one edge, and thus any component $C'$ with $C' \not\cong C$ cannot contribute to $\abs{\set{F' \in \family{D}_{0,1}(H) \given F \cong F'}}$. We obtain that $\alpha_{F}$ is not congruent modulo $p$ to $0$.
	
	Otherwise, every connected component $C$ that contains a distinguished edge $e$ in $\distEdges$ satisfies that either $C\cong K_{2}^\circ$ and $e$ is not a loop or $C \cong K_{a,b}$, where $a = 1$ and $b \in [p-1]$, i.e. $C$ is a star with at most $p-1$ leaves. Hence, for any pair of subsets $V'$ of distinguished vertices $\distVertices$ and $E'$ of distinguished edges $\distEdges$, the subgraph $C \setminus (V' \cup E')$ obtained from $C$ by deleting $(V' \cup E')$ is either a complete bipartite graph or a reflexive complete graph. The same holds for every connected component $C$ that contains no distinguished edge, that is, for any pair of subsets $V'$ of distinguished vertices $\distVertices$ and $E'$ of distinguished edges $\distEdges$, the subgraph $C \setminus (V' \cup E')$ is either a complete bipartite graph or a reflexive complete graph. Therefore, in any of the remaining cases, every subgraph $F$ in $\family{D}^\ast(H)$ is a collection of complete bipartite graphs and reflexive complete graphs. This concludes the proof.
\end{proof}

We conclude by Corollary~\ref{cor:partsurjhoms_equiv_homs} that the dichotomy for $\probNumHom{H}[p]$ yields dichotomies for the whole class of partially surjective homomorphisms. Additionally, under assumptions on the automorphism groups $\Aut[H]$ and $\AutDist[H]$, Theorem~\ref{thm:dichotomy_part_surj_homs} gives a dichotomy with a precise structural statement. Let us highlight this connection with a dichotomy for the computational problems of counting the number of vertex surjective homomorphisms and the number of compactions. We recall that vertex surjective homomorphisms to a graph $H$ are characterized by the set of distinguished vertices $\distVertices$ equal to $\vertexset[H]$. Therefore, the computational problem $\probNumVertSurjHom{H}[p]$ is equivalent to $\probNumPartSurjHom{H, \distVertices, \distEdges=\emptyset}[p]$. Compactions are characterized by the set of distinguished vertices $\distVertices$ equal to $\vertexset[H]$ and the set of distinguished edges $\distEdges$ equal to $\set{e \in \edgeset[H] \given e \text{ is not a loop}}$. In this way, the computational problem $\probNumComp{H}[p]$ is equivalent to $\probNumPartSurjHom{H, \distVertices, \distEdges}[p]$.
\begin{lemma}\label{lem:vertex_surjective_homs_and_compactions_are_okay}
	Let $p$ be a prime and $H$ be a graph with distinguished vertices $\distVertices$ and distinguished edges $\distEdges$. If $\distVertices$ is equal to $\vertexset[H]$ and $\distEdges$ is either empty or the set of non-loop edges in $\edgeset[H]$, then the groups $\AutDist[H]$ $\Aut[H]$ are equal.
\end{lemma}
\begin{proof}
	We briefly argue both cases, where we have $\distVertices = \vertexset[H]$. If no edge is distinguished, then any automorphism in $\Aut[H]$ maps bijectively distinguished elements to distinguished elements. The second case is established by the observation that any automorphism has to map loops to loops and non-loops to non-loops.
\end{proof}

It follows immediately from Lemma~\ref{lem:vertex_surjective_homs_and_compactions_are_okay} that the computational problems $\probNumVertSurjHom{H}[p]$ and $\probNumComp{H}[p]$ allow an application of Theorem~\ref{thm:dichotomy_part_surj_homs}. For these, we obtain the criteria analogous to the criteria in the non-modular setting given by Focke et al.~\cite{Focke:19:The_Complexity_of_Counting_Surjective_Homomorphisms_and_Compactions}. For the sake of comparability, we employ for this statement the notion of an \emph{irreflexive} graph, which denotes a graph without any loops.

\SurjCompDichotomyModp*

\section*{Acknowledgements}
The authors would like to thank Holger Dell for bringing~\cite{Chen:19:The_Exponential-Time_Complexity} to their attention. Our gratitude also goes to Jacob Focke and Marc Roth for their valuable insights on partially surjective homomorphisms and for pointing out a mistake in a previous version.

Andreas Göbel was funded by the project PAGES (project No. 467516565) of the German Research Foundation (DFG).

\bibliographystyle{plainurl}
\bibliography{bibliography}

\begin{thebibliography}{10}

\bibitem{Amini:12:Counting_Subgraphs_via_Homomorphisms}
Omid Amini, Fedor~V. Fomin, and Saket Saurabh.
\newblock Counting subgraphs via homomorphisms.
\newblock {\em {SIAM} J. Discret. Math.}, 26(2):695--717, 2012.
\newblock \href {https://doi.org/10.1137/100789403}
  {\path{doi:10.1137/100789403}}.

\bibitem{Artin:11:book:Algebra}
Michael Artin.
\newblock {\em Algebra}.
\newblock Pearson Prentice Hall, 2011.

\bibitem{Borgs:06:Counting_Graph_Homomorphisms}
Christian Borgs, Jennifer Chayes, L{\'a}szl{\'o} Lov{\'a}sz, Vera~T. S{\'o}s,
  and Katalin Vesztergombi.
\newblock Counting graph homomorphisms.
\newblock In {\em Topics in Discrete Mathematics}, pages 315--371, 2006.
\newblock \href {https://doi.org/10.1007/3-540-33700-8_18}
  {\path{doi:10.1007/3-540-33700-8_18}}.

\bibitem{Borgs:13:Convergence_Graphs_Bounded_Degree}
Christian Borgs, Jennifer~T. Chayes, Jeff Kahn, and L{\'{a}}szl{\'{o}}
  Lov{\'{a}}sz.
\newblock Left and right convergence of graphs with bounded degree.
\newblock {\em Random Structures and Algorithms}, 42(1):1--28, 2013.
\newblock \href {https://doi.org/10.1002/rsa.20414}
  {\path{doi:10.1002/rsa.20414}}.

\bibitem{Borgs:06:Graph_Limits_and_Parameter_Testing}
Christian Borgs, Jennifer~T. Chayes, L{\'{a}}szl{\'{o}} Lov{\'{a}}sz, Vera~T.
  S{\'{o}}s, Bal{\'{a}}zs Szegedy, and Katalin Vesztergombi.
\newblock Graph limits and parameter testing.
\newblock In {\em Proceedings of~{STOC} 2006}, pages 261--270, 2006.
\newblock \href {https://doi.org/10.1145/1132516.1132556}
  {\path{doi:10.1145/1132516.1132556}}.

\bibitem{Brightwell:99:Graph_Homomorpisms_and_Phase_Transitions}
Graham~R. Brightwell and Peter Winkler.
\newblock Graph homomorphisms and phase transitions.
\newblock {\em J. Comb. Theory, Ser. {B}}, 77(2):221--262, 1999.
\newblock \href {https://doi.org/10.1006/jctb.1999.1899}
  {\path{doi:10.1006/jctb.1999.1899}}.

\bibitem{Bulatov:13:The_Complexity_of_Counting_CSP}
Andrei~A. Bulatov.
\newblock The complexity of the counting constraint satisfaction problem.
\newblock {\em J. {ACM}}, 60(5):34:1--34:41, 2013.
\newblock \href {https://doi.org/10.1145/2528400} {\path{doi:10.1145/2528400}}.

\bibitem{Bulatov:05:The_Complexity_of_Partition_Functions}
Andrei~A. Bulatov and Martin Grohe.
\newblock The complexity of partition functions.
\newblock {\em Theor. Comput. Sci.}, 348(2-3):148--186, 2005.
\newblock \href {https://doi.org/10.1016/j.tcs.2005.09.011}
  {\path{doi:10.1016/j.tcs.2005.09.011}}.

\bibitem{Bulatov:22:Complexity_Classification_of_Homs_mod}
Andrei~A. Bulatov and Amirhossein Kazeminia.
\newblock Complexity classification of counting graph homomorphisms modulo a
  prime number.
\newblock In {\em Proceedings of~{STOC}~2022}, pages 1024--1037, 2022.
\newblock \href {https://doi.org/10.1145/3519935.3520075}
  {\path{doi:10.1145/3519935.3520075}}.

\bibitem{Cai:17:Counting_CSP_Complex_Weights}
Jin{-}Yi Cai and Xi~Chen.
\newblock Complexity of counting {CSP} with complex weights.
\newblock {\em J. {ACM}}, 64(3):19:1--19:39, 2017.
\newblock \href {https://doi.org/10.1145/2822891} {\path{doi:10.1145/2822891}}.

\bibitem{Cai:13:Graph_Homomorphisms_with_Complex_Values}
Jin{-}Yi Cai, Xi~Chen, and Pinyan Lu.
\newblock Graph homomorphisms with complex values: {A} dichotomy theorem.
\newblock {\em {SIAM} J. Comput.}, 42(3):924--1029, 2013.
\newblock \href {https://doi.org/10.1137/110840194}
  {\path{doi:10.1137/110840194}}.

\bibitem{Cai:20:Graph_Homs_Bounded_Degree_Complex}
Jin{-}Yi Cai and Artem Govorov.
\newblock Dichotomy for graph homomorphisms with complex values on bounded
  degree graphs.
\newblock In {\em Proceedings of {FOCS} 2020}, pages 1103--1111, 2020.
\newblock \href {https://doi.org/10.1109/FOCS46700.2020.00106}
  {\path{doi:10.1109/FOCS46700.2020.00106}}.

\bibitem{Cai:22:Perfect_Matchings}
Jin{-}Yi Cai and Artem Govorov.
\newblock Perfect matchings, rank of connection tensors and graph
  homomorphisms.
\newblock {\em Comb. Probab. Comput.}, 31(2):268--303, 2022.
\newblock \href {https://doi.org/10.1017/S0963548321000286}
  {\path{doi:10.1017/S0963548321000286}}.

\bibitem{Chandra:77:Optimal_Implementation_of_Conjunctive_Queries_in_Relational_Data_Bases}
Ashok~K. Chandra and Philip~M. Merlin.
\newblock Optimal implementation of conjunctive queries in relational data
  bases.
\newblock In {\em Proceedings of~{STOC} 1977}, pages 77--90, 1977.
\newblock \href {https://doi.org/10.1145/800105.803397}
  {\path{doi:10.1145/800105.803397}}.

\bibitem{Chen:19:The_Exponential-Time_Complexity}
Hubie Chen, Radu Curticapean, and Holger Dell.
\newblock The exponential-time complexity of counting (quantum) graph
  homomorphisms.
\newblock In {\em Proceedings~of~{WG}~2019}, pages 364--378, 2019.
\newblock \href {https://doi.org/10.1007/978-3-030-30786-8_28}
  {\path{doi:10.1007/978-3-030-30786-8_28}}.

\bibitem{Curticapean:17:Homomorphisms_Are_a_Good_Basis_for_Counting}
Radu Curticapean, Holger Dell, and D\'{a}niel Marx.
\newblock Homomorphisms are a good basis for counting small subgraphs.
\newblock In {\em Proceedings of~{STOC} 2017}, page 210–223, 2017.
\newblock \href {https://doi.org/10.1145/3055399.3055502}
  {\path{doi:10.1145/3055399.3055502}}.

\bibitem{Dalmau:04:The_Complexity_of_Counting_Homomorphisms_Seen_from_the_Other_Side}
V{\'{\i}}ctor Dalmau and Peter Jonsson.
\newblock The complexity of counting homomorphisms seen from the other side.
\newblock {\em Theor. Comput. Sci.}, 329(1):315--323, 2004.
\newblock \href {https://doi.org/10.1016/j.tcs.2004.08.008}
  {\path{doi:10.1016/j.tcs.2004.08.008}}.

\bibitem{Dyer:07:Counting_Homs_Directed_Acyclic_Graphs}
Martin~E. Dyer, Leslie~Ann Goldberg, and Mike Paterson.
\newblock On counting homomorphisms to directed acyclic graphs.
\newblock {\em J. {ACM}}, 54(6):27:1 -- 27:23, 2007.
\newblock \href {https://doi.org/10.1145/1314690.1314691}
  {\path{doi:10.1145/1314690.1314691}}.

\bibitem{Dyer:2000:Counting_Graph_Homs}
Martin~E. Dyer and Catherine Greenhill.
\newblock The complexity of counting graph homomorphisms.
\newblock {\em Random Structures and Algorithms}, 17(3--4):260--289, 2000.
\newblock \href
  {https://doi.org/10.1002/1098-2418(200010/12)17:3/4<260::AID-RSA5>3.0.CO;2-W}
  {\path{doi:10.1002/1098-2418(200010/12)17:3/4<260::AID-RSA5>3.0.CO;2-W}}.

\bibitem{Dyer:13:An_Effective_Dichotomy_for_CSP}
Martin~E. Dyer and David Richerby.
\newblock An effective dichotomy for the counting constraint satisfaction
  problem.
\newblock {\em {SIAM} J. Comput.}, 42(3):1245--1274, 2013.
\newblock \href {https://doi.org/10.1137/100811258}
  {\path{doi:10.1137/100811258}}.

\bibitem{Diaz:02:Counting_H-Colorings_of_Partial_k-Trees}
Josep Díaz, Maria Serna, and Dimitrios~M. Thilikos.
\newblock Counting {H}-colorings of partial k-trees.
\newblock {\em Theor. Comput. Sci.}, 281(1):291 -- 309, 2002.
\newblock Selected Papers in honour of Maurice Nivat.
\newblock \href {https://doi.org/https://doi.org/10.1016/S0304-3975(02)00017-8}
  {\path{doi:https://doi.org/10.1016/S0304-3975(02)00017-8}}.

\bibitem{Elenberg:15:Beyond_Triangles}
Ethan~R. Elenberg, Karthikeyan Shanmugam, Michael Borokhovich, and
  Alexandros~G. Dimakis.
\newblock Beyond triangles: {A} distributed framework for estimating 3-profiles
  of large graphs.
\newblock In {\em Proceedings of~{KDD} 2015}, pages 229--238, 2015.
\newblock \href {https://doi.org/10.1145/2783258.2783413}
  {\path{doi:10.1145/2783258.2783413}}.

\bibitem{Elenberg:16:Distributed_Estimation_of_Graph_4-Profiles}
Ethan~R. Elenberg, Karthikeyan Shanmugam, Michael Borokhovich, and
  Alexandros~G. Dimakis.
\newblock Distributed estimation of graph 4-profiles.
\newblock In {\em Proceedings of~{WWW} 2016}, pages 483--493, 2016.
\newblock \href {https://doi.org/10.1145/2872427.2883082}
  {\path{doi:10.1145/2872427.2883082}}.

\bibitem{Faben:12:thesis:Complexity_Modular_Counting_CSP}
John Faben.
\newblock {\em The Complexity of Modular Counting in Constraint Satisfaction
  Problems}.
\newblock PhD thesis, Queen Mary University of London, 2012.

\bibitem{Faben:15:Parity_Graph_Homs}
John Faben and Mark Jerrum.
\newblock The complexity of parity graph homomorphism: An initial
  investigation.
\newblock {\em Theory of Computing}, 11:35--57, 2015.
\newblock \href {https://doi.org/10.4086/toc.2015.v011a002}
  {\path{doi:10.4086/toc.2015.v011a002}}.

\bibitem{Feder:98:List_Homomorphisms_to_Reflexive_Graphs}
Tom{\'{a}}s Feder and Pavol Hell.
\newblock List homomorphisms to reflexive graphs.
\newblock {\em J. Comb. Theory, Ser. {B}}, 72(2):236--250, 1998.
\newblock \href {https://doi.org/10.1006/jctb.1997.1812}
  {\path{doi:10.1006/jctb.1997.1812}}.

\bibitem{Feder:98:The_Computational_Structure_of_Monotone_Monadic_SNP_and_Constraint_Satisfaction}
Tomás Feder and Moshe~Y. Vardi.
\newblock The computational structure of monotone monadic snp and constraint
  satisfaction: A study through datalog and group theory.
\newblock {\em SIAM J. Comput.}, 28(1):57--104, 1998.
\newblock \href {https://doi.org/10.1137/S0097539794266766}
  {\path{doi:10.1137/S0097539794266766}}.

\bibitem{Focke:21:Counting_Homomorphisms_to_K_4-Minor-Free_Graphs_Mod_2}
Jacob Focke, Leslie~Ann Goldberg, Marc Roth, and Stanislav Zivn{\'{y}}.
\newblock Counting homomorphisms to k\({}_{\mbox{4}}\)-minor-free graphs,
  modulo 2.
\newblock {\em {SIAM} J. Discret. Math.}, 35(4):2749--2814, 2021.
\newblock \href {https://doi.org/10.1137/20M1382921}
  {\path{doi:10.1137/20M1382921}}.

\bibitem{Focke:19:The_Complexity_of_Counting_Surjective_Homomorphisms_and_Compactions}
Jacob Focke, Leslie~Ann Goldberg, and Stanislav Zivn{\'{y}}.
\newblock The complexity of counting surjective homomorphisms and compactions.
\newblock {\em {SIAM} J. Discret. Math.}, 33(2):1006--1043, 2019.
\newblock \href {https://doi.org/10.1137/17M1153182}
  {\path{doi:10.1137/17M1153182}}.

\bibitem{Galanis:16:Approximately_Counting_H-Colorings}
Andreas Galanis, Leslie~Ann Goldberg, and Mark Jerrum.
\newblock Approximately counting {H}-colorings is ${\#}\mathrm{BIS}$-hard.
\newblock {\em {SIAM} J. Comput.}, 45(3):680--711, 2016.
\newblock \href {https://doi.org/10.1137/15M1020551}
  {\path{doi:10.1137/15M1020551}}.

\bibitem{Goebel:14:Cactus}
Andreas G{\"{o}}bel, Leslie~Ann Goldberg, and David Richerby.
\newblock The complexity of counting homomorphisms to cactus graphs modulo 2.
\newblock {\em {ACM} Trans. Comput. Theory}, 6(4):17:1--17:29, 2014.
\newblock \href {https://doi.org/10.1145/2635825} {\path{doi:10.1145/2635825}}.

\bibitem{Goebel:16:Square-Free}
Andreas G{\"{o}}bel, Leslie~Ann Goldberg, and David Richerby.
\newblock Counting homomorphisms to square-free graphs, modulo 2.
\newblock {\em {ACM} Trans. Comput. Theory}, 8(3):12:1--12:29, 2016.
\newblock \href {https://doi.org/10.1145/2898441} {\path{doi:10.1145/2898441}}.

\bibitem{Goebel:21:Counting_Homomorphisms_Trees}
Andreas G\"{o}bel, J.~A.~Gregor Lagodzinski, and Karen Seidel.
\newblock Counting homomorphisms to trees modulo a prime.
\newblock {\em ACM Trans. Comput. Theory}, 13(3), 2021.
\newblock \href {https://doi.org/10.1145/3460958} {\path{doi:10.1145/3460958}}.

\bibitem{Goldberg:10:A_Complexity_Dichotomy_for_Partition_Functions_with_Mixed_Signs}
Leslie~Ann Goldberg, Martin Grohe, Mark Jerrum, and Marc Thurley.
\newblock A complexity dichotomy for partition functions with mixed signs.
\newblock {\em {SIAM} J. Comput.}, 39(7):3336--3402, 2010.
\newblock \href {https://doi.org/10.1137/090757496}
  {\path{doi:10.1137/090757496}}.

\bibitem{Goldberg:14:The_Complexity_of_Approximately_Counting_Tree_Homomorphisms}
Leslie~Ann Goldberg and Mark Jerrum.
\newblock The complexity of approximately counting tree homomorphisms.
\newblock {\em {ACM} Trans. Comput. Theory}, 6(2):8:1--8:31, 2014.
\newblock \href {https://doi.org/10.1145/2600917} {\path{doi:10.1145/2600917}}.

\bibitem{Govorov:20:Dichotomy_Bounded_Degree_Homomorpisms}
Artem Govorov, Jin{-}Yi Cai, and Martin~E. Dyer.
\newblock A dichotomy for bounded degree graph homomorphisms with nonnegative
  weights.
\newblock In {\em Proceedings of {ICALP} 2020}, volume 168, pages 66:1--66:18,
  2020.
\newblock \href {https://doi.org/10.4230/LIPIcs.ICALP.2020.66}
  {\path{doi:10.4230/LIPIcs.ICALP.2020.66}}.

\bibitem{Grohe:07:The_Complexity_of_Homomorphisms_and_Constraint_Satisfaction_Problems}
Martin Grohe.
\newblock The complexity of homomorphism and constraint satisfaction problems
  seen from the other side.
\newblock {\em J. {ACM}}, 54(1):1:1--1:24, 2007.
\newblock \href {https://doi.org/10.1145/1206035.1206036}
  {\path{doi:10.1145/1206035.1206036}}.

\bibitem{Grohe:01:When_is_the_Evaluation_of_Conjunctive_Queries_Tractable}
Martin Grohe, Thomas Schwentick, and Luc Segoufin.
\newblock When is the evaluation of conjunctive queries tractable?
\newblock In {\em Proceedings of~{STOC} 2001}, pages 657--666, 2001.
\newblock \href {https://doi.org/10.1145/380752.380867}
  {\path{doi:10.1145/380752.380867}}.

\bibitem{Guo:11:Weighted_Bolean_CSP_Modulo}
Heng Guo, Sangxia Huang, Pinyan Lu, and Mingji Xia.
\newblock The complexity of weighted boolean {\#}csp modulo k.
\newblock In {\em Proceedings of~{STACS} 2011}, pages 249--260, 2011.
\newblock \href {https://doi.org/10.4230/LIPIcs.STACS.2011.249}
  {\path{doi:10.4230/LIPIcs.STACS.2011.249}}.

\bibitem{Hell:90:On_the_Complexity_of_H-Coloring}
Pavol Hell and Jaroslav Ne\v{s}et\v{r}il.
\newblock On the complexity of \emph{H}-coloring.
\newblock {\em J. Comb. Theory, Ser. {B}}, 48(1):92--110, 1990.
\newblock \href {https://doi.org/10.1016/0095-8956(90)90132-J}
  {\path{doi:10.1016/0095-8956(90)90132-J}}.

\bibitem{Kazeminia:19:Count_Homs_Square_Free_Mod_Prime}
Amirhossein Kazeminia and Andrei~A. Bulatov.
\newblock Counting homomorphisms modulo a prime number.
\newblock In {\em Proceedings of~{MFCS}~2019}, pages 59:1--59:13, 2019.
\newblock \href {https://doi.org/10.4230/LIPIcs.MFCS.2019.59}
  {\path{doi:10.4230/LIPIcs.MFCS.2019.59}}.

\bibitem{Ladner:75:On_the_Structure}
Richard~E. Ladner.
\newblock On the structure of polynomial time reducibility.
\newblock {\em J. {ACM}}, 22(1):155--171, 1975.
\newblock \href {https://doi.org/10.1145/321864.321877}
  {\path{doi:10.1145/321864.321877}}.

\bibitem{Lagodzinski:21:On_Counting_Quantum-Graph_Homomorphisms}
J.~A.~Gregor Lagodzinski, Andreas G\"{o}bel, Katrin Casel, and Tobias
  Friedrich.
\newblock {On Counting (Quantum-)Graph Homomorphisms in Finite Fields of Prime
  Order}.
\newblock In {\em Proceedings of~{ICALP} 2021}, pages 91:1--91:15, 2021.
\newblock \href {https://doi.org/10.4230/LIPIcs.ICALP.2021.91}
  {\path{doi:10.4230/LIPIcs.ICALP.2021.91}}.

\bibitem{Lovasz:12:book:Large_Networks_Graph_Limits}
L{\'{a}}szl{\'{o}} Lov{\'{a}}sz.
\newblock {\em Large Networks and Graph Limits}, volume~60 of {\em Colloquium
  Publications}.
\newblock AMS, 2012.

\bibitem{Peyerimhoff:21:Parameterized_Modular_Counting_Cayley_Graph_Expanders}
Norbert Peyerimhoff, Marc Roth, Johannes Schmitt, Jakob Stix, and Alina
  Vdovina.
\newblock {Parameterized (Modular) Counting and Cayley Graph Expanders}.
\newblock In {\em Proceedings of MFCS 2021}, pages 84:1--84:15, 2021.
\newblock \href {https://doi.org/10.4230/LIPIcs.MFCS.2021.84}
  {\path{doi:10.4230/LIPIcs.MFCS.2021.84}}.

\bibitem{Rossman:08:Homomorphism_Preservation_Theorems}
Benjamin Rossman.
\newblock Homomorphism preservation theorems.
\newblock {\em J. {ACM}}, 55(3):15:1--15:53, 2008.
\newblock \href {https://doi.org/10.1145/1379759.1379763}
  {\path{doi:10.1145/1379759.1379763}}.

\bibitem{Rossman:17:An_Improved_Homomorphism_Preservation_Theorem}
Benjamin Rossman.
\newblock An improved homomorphism preservation theorem from lower bounds in
  circuit complexity.
\newblock In {\em Proceedings of~{ITCS} 2017}, pages 27:1--27:17, 2017.
\newblock \href {https://doi.org/10.4230/LIPIcs.ITCS.2017.27}
  {\path{doi:10.4230/LIPIcs.ITCS.2017.27}}.

\bibitem{Roth:20:Counting_and_Finding_Homomorphisms_is_Universal_for_Parameterized_Complexity}
Marc Roth and Philip Wellnitz.
\newblock Counting and finding homomorphisms is universal for parameterized
  complexity theory.
\newblock In {\em Proceedings of~{SODA} 2020}, pages 2161--2180, 2020.
\newblock \href {https://doi.org/10.1137/1.9781611975994.133}
  {\path{doi:10.1137/1.9781611975994.133}}.

\bibitem{Shams:19:Counting_Homomorphisms_in_Bipartite_Graphs}
Shahab {Shams}, Nicholas {Ruozzi}, and Peter {Csikvári}.
\newblock Counting homomorphisms in bipartite graphs.
\newblock In {\em Proceedings of~{ISIT} 2019}, pages 1487--1491, 2019.
\newblock \href {https://doi.org/10.1109/ISIT.2019.8849389}
  {\path{doi:10.1109/ISIT.2019.8849389}}.

\bibitem{Sidorenko:91:Inequalities_for_Functionals}
Alexander~F. Sidorenko.
\newblock Inequalities for functionals generated by bipartite graphs.
\newblock {\em (Russian) Diskret. Mat}, 3(3):50--65, 1991.

\bibitem{Sidorenko:93:A_Correlation_Inequality}
Alexander~F. Sidorenko.
\newblock A correlation inequality for bipartite graphs.
\newblock {\em Graphs and Combin.}, 9(2-4):201--204, 1993.
\newblock \href {https://doi.org/10.1007/BF02988307}
  {\path{doi:10.1007/BF02988307}}.

\bibitem{Simonovits:84:Extremal_Graph_Problems}
Mikl\'{o}s Simonovits.
\newblock Extremal graph problems, degenerate extremal problems,and
  supersaturated graphs.
\newblock {\em Progress in Graph Theory}, pages 419–--437, 1984.

\bibitem{Toda:91:PP_is_as_Hard_as_the_Polynomial-Time_Hierarchy}
Seinosuke Toda.
\newblock {PP} is as hard as the polynomial-time hierarchy.
\newblock {\em {SIAM} J. Comput.}, 20(5):865--877, 1991.
\newblock \href {https://doi.org/10.1137/0220053} {\path{doi:10.1137/0220053}}.

\bibitem{Ugander:13:Subgraph_Frequencies}
Johan Ugander, Lars Backstrom, and Jon~M. Kleinberg.
\newblock Subgraph frequencies: mapping the empirical and extremal geography of
  large graph collections.
\newblock In {\em Proceedings of~{WWW} 2013}, pages 1307--1318, 2013.
\newblock \href {https://doi.org/10.1145/2488388.2488502}
  {\path{doi:10.1145/2488388.2488502}}.

\bibitem{Valiant:06:Accidental_Algorithms}
Leslie~G. Valiant.
\newblock Accidental algorithms.
\newblock In {\em Proceedings of~{FOCS} 2006}, pages 509--517, 2006.
\newblock \href {https://doi.org/10.1109/FOCS.2006.7}
  {\path{doi:10.1109/FOCS.2006.7}}.

\end{thebibliography}

\end{document}